\documentclass[12pt,reqno]{amsart}
\usepackage{setspace}
\usepackage{newtxtext,newtxmath}
\usepackage{comment}
\usepackage{amsmath, graphicx, url, mathtools,centernot, tikz}
\usepackage{xcolor}
\usepackage{wrapfig}
\usepackage{multirow} 
\usepackage{geometry}
\usepackage[authoryear,round]{natbib}
\usepackage[utf8]{inputenc} 
\usepackage[T1]{fontenc}    
\usepackage{hyperref}       
\usepackage{url}            
\usepackage{booktabs}       
\usepackage{amsfonts}       
\usepackage{nicefrac}       
\usepackage{microtype}      

\usepackage{caption}
\usepackage{subcaption}
\usepackage{ragged2e}

\newcommand{\RR}{\mathbb{R}}
\renewcommand{\P}{\mathbb{P}}

\newcommand{\IF}{{\mathbf{1}}}

\DeclareMathOperator{\diag}{diag}

\DeclareMathOperator{\tr}{tr}

\theoremstyle{plain}

\newtheorem{Assumption}{Assumption}

\theoremstyle{definition}

\newtheorem{remark}{Remark}

\newcommand{\E}{\mathbb{E}}
 \newcommand{\defeq}{\overset{\text{\tiny def}}{=}} 
\usepackage{amsmath}
\DeclareMathOperator*{\argmax}{arg\,max}
\DeclareMathOperator*{\argmin}{arg\,min}

\newtheorem{theorem}{Theorem}
\newtheorem{lemma}{Lemma}
        \newtheorem{corollary}{Corollary}[theorem]

\usepackage{mathrsfs}
\usepackage[lined,boxed]{algorithm2e}
        \usepackage{rotating}

\usepackage{graphicx} 
\graphicspath{{fig/}{../Code/Code-Quasi_bayes_Chernozhukov-Hong-2003/fig/}}

\title[]{Plausible GMM: A Quasi-Bayesian Approach}
\thanks{Date: \today{}.  }  
\author[]{Victor Chernozhukov$^a$,  Christian B. Hansen$^b$, Lingwei Kong$^c$, Weining Wang$^d$}  
 \thanks{$^a$: Department of Economics, Massachusetts Institute of Technology.}
 \thanks{$^b$: Booth School of Business, The University of Chicago.}
 \thanks{$^c$: Department of Economics, Econometrics and Finance, University of Groningen.}
 \thanks{$^d$: Department of Economics, University of  Bristol. weining.wang@bristol.ac.uk } 
\thanks{We thank  Bertille Antoine, Tom Boot, Peter Boswijk, Songxi Chen, \'Aureo de Paula, Art\={u}ras Juodis, Byunghoon Kang, Frank Kleibergen,  Michal Koles\'{a}r,  Jun Liu, Anna Mikusheva, Whitney K. Newey, Frank A. G. Windmeijer,  Taisuke Otsu,  Johannes Schmidt-Hieber, and Liyang Sun for thoughtful comments. All errors remain our own.}

\begin{document}

\begin{abstract}
Structural estimation in economics often makes use of models formulated in terms of moment conditions. While these moment conditions are generally well-motivated, it is often unknown whether the moment restrictions hold exactly. We consider a framework where researchers model their belief about the potential degree of misspecification via a prior distribution and adopt a quasi-Bayesian approach for performing inference on structural parameters. We provide quasi-posterior concentration results, verify that quasi-posteriors can be used to obtain approximately optimal Bayesian decision rules under the maintained prior structure over misspecification, and provide a form of frequentist coverage results. We illustrate the approach through empirical examples where we obtain informative inference for structural objects, allowing for substantial relaxations of the requirement that moment conditions hold exactly. 

\vspace{.1in}
\end{abstract}

\maketitle
 
\section{Introduction}
Moment restrictions are commonly used in the identification and estimation of structural or causal parameters in empirical economics. Prominent examples include instrument exclusion conditions, unconfoundedness assumptions, parallel trend assumptions, and nonlinear moment restrictions imposed in structural models. Economists typically use institutional knowledge and economic reasoning to argue for the validity of these restrictions in settings with observational data. Based on these arguments, classical estimation and inference, such as estimation and inference based on generalized method of moments (GMM), then proceed under the maintained assumption that the posited moment restrictions hold exactly. 

While the arguments employed to justify moment restrictions provide a basis for believing that the moment conditions are plausible, they are also generally debatable. That is, it is hard to know whether there are unmodeled sources of confounding or sources of misspecification that would result in moment conditions failing to hold exactly in any given empirical setting. Unfortunately, it is well-known that estimation and inference results obtained under the assumption that moment restrictions hold exactly can be substantially distorted in the sense of returning biased estimates and delivering unreliable conclusions. See, for example, \cite{hansen2001acknowledging, hansen2008robustness, hansen2010fragile}, \cite{hall2003large}, \cite{cheng2015uniform}, \cite{CHEN2024105653} and \cite{hansen2021inference}.

In this paper, we consider one approach to estimation and inference within a semiparametric structural model characterized by a set of moment restrictions, allowing for the possibility that the specified moment conditions do not hold exactly. We consider a setting where a researcher has access to an observable independent and identically distributed (i.i.d.) data stream $\{Z_t\}_{t=1}^{T}$ realized from unknown distribution $\mathbb{P}_{\mu_*}$.\footnote{While we maintain the assumption that the $Z_t$ are independent and identically distributed, some of our theoretical results hold more generally. We briefly comment on extensions to non-i.i.d. settings in Section \ref{simple}.}
The researcher specifies a structural model, defined in terms of an economically meaningful $k$-dimensional parameter vector $\theta_*$, that restricts the joint distribution via moment conditions $m(\theta_*)=\mathbb{E}_{\mathbb{P}_{\mu_*}}[g(Z_{t},\theta_*)] = \mu_*$ for some $q \geq k$ dimensional vector $\mu_*$.
Of course, informative inference about $\theta_*$ is impossible without beliefs about $\mu_*$.
Classical estimation and inference proceed under the dogmatic prior belief that $\mu_* \equiv 0$.

Rather than adopt dogmatic prior beliefs, we conceptualize the notion that moment restrictions are plausible---but not known to hold exactly---by assuming the researcher is able to place an informative prior distribution over $\mu_*$. The use of a proper prior over $\mu_*$ allows informative inference about model parameters to proceed while relaxing the usual restriction that $\mu_* \equiv 0$.  
By concentrating this prior over 0, we capture the notion that a researcher subjectively believes the structural moment restrictions are likely to be correct. The spread and shape of the prior away from 0 further captures the researchers' beliefs about likely economically motivated possible deviations from the baseline structural model. Thus, the use of a proper prior distribution over $\mu_*$ 
provides a way for researchers to explicitly encode their subjective beliefs over the plausibility of their structural model.

Given that we wish to only leverage structural moment conditions and choose to conceptualize the plausibility of these moment conditions by using a proper prior distribution, it is natural to consider estimation and inference based on approximate or quasi-Bayesian posteriors (QBP), as in \cite{kim:lilgmm} and \cite{chernozhukov2003mcmc}.\footnote{In \cite{chernozhukov2003mcmc}, estimators produced from quasi-Bayesian posteriors are referred to as Laplace-Type Estimators (LTEs). These approaches are also connected to ``probably approximately correct inference,'' e.g., \cite{catoni:learning}.} QBPs provide a tractable approach to approximate Bayesian estimation and inference in traditional semiparametric moment condition models where 
$\mu_* \equiv 0$ 
is imposed; see, e.g., \cite{kim:lilgmm}, \cite{gallant:inducedspace}, \cite{fs:bayesmoment}, and \cite{andrews2022optimal}. Outside of the Bayesian motivation, \cite{chernozhukov2003mcmc} demonstrate that these estimators have desirable frequentist properties within the moment condition framework when $\mu_* \equiv 0$.  
In addition, \cite{andrews2022optimal} verifies that quasi-Bayes decision rules approximate Bayes' optimal decision rules within the weakly identified GMM framework.  

In this paper, we extend the QBP framework to deal with settings where moment conditions are not assumed to hold exactly, i.e., to settings with non-dogmatic prior over 
$\mu_*$. We refer to estimation and inference within this setting as plausible GMM (PGMM). 
A central challenge in this setting is that $\theta_*$ and $\mu_*$ 
are not jointly identified, which implies that the impact of priors is not asymptotically negligible. We develop new technical results that address this challenge and show that, under suitable regularity conditions, key properties of the QBP framework extend to the PGMM setting. Building on \cite{andrews2022optimal}, we verify that quasi-Bayes decision rules approximate Bayes optimal decision rules given the provided subjective priors. We also generalize the results of \cite{chernozhukov2003mcmc} to verify that interval estimates from QBPs have a well-defined \textit{ex-ante} notion of frequentist coverage under a sampling process where nature first draws a degree of misspecification from the subjective prior for $\mu_*$ and then the model realizes conditional on this draw as in \cite{conley2012plausibly} and analogous to the coverage notion considered in \cite{andrews:pseudo}. Finally, we provide novel large sample approximations for quasi-posterior distributions within this partially identified framework, allowing for the dimension of both 
$\theta_*$ and $\mu_*$
to increase with the sample size. These results can be viewed as new Bernstein-von Mises type theorems that explicitly account for additional terms that arise when dealing with misspecification.

We illustrate the use of QBPs with proper priors over the degree of misspecification, 
$\mu_*$, via two empirical examples.  
A cost of allowing for potential misspecification by considering non-dogmatic beliefs about 
$\mu_*$  
is that inferential statements must be less precise than those obtained under dogmatic beliefs. The empirical applications demonstrate that one can still draw economically meaningful conclusions using our approach under what we consider to be sensible beliefs about model misspecification. The approach thus potentially enhances the credibility of the qualitative empirical results. We also use the empirical examples to discuss prior choice, illustrate the impact of prior choice on the resulting quasi-posteriors, and discuss sensitivity analysis focused on prior choice.   

There is a large existing literature on sensitivity analysis and partial identification. Much of this research focuses on establishing formal frequentist guarantees for estimating and performing inference on the identified set. See, e.g., \cite{Canay_Shaikh_2017} and \cite{molinari:review} for excellent reviews and \cite{norets2014semiparametric}, \cite{kline2016bayesian}, \cite{chib2018bayesian}, \cite{liao2019bayesian}, \cite{giacomini2021robust}, \cite{giacomini2022robust}, and \cite{kuangbayesian} for approaches leveraging Bayesian methods.

Within this literature, our work is closely related to \cite{armstrong2021sensitivity}. They consider moment condition models in which, under correct specification, $m(\theta_*) = 0$ for a true population parameter $\theta_*$. Misspecification is introduced by allowing $m(\theta_*) = C_T \neq 0$, where $C_T$ is unknown but constrained to lie in a known set; see also \cite{bonhomme2022minimizing}. \cite{armstrong2021sensitivity} focus on the case $C_T = c/\sqrt{T}$, where misspecification is local in the sense that it is of the same order as sampling uncertainty. Their framework, however, also accommodates settings in which misspecification does not vary with sample size. They develop a tractable frequentist approach to valid inference, show how sensitivity analysis can be conducted by varying the admissible set for $C_T$, and demonstrate that optimal GMM weighting matrices in this environment optimally trade off sampling variability against potential misspecification.

Our analysis delivers an analogous but conceptually distinct result from a quasi-Bayesian perspective. In particular, we show that QBPs are centered on a GMM estimator whose weighting matrix trades off moment precision with misspecification in a manner related to \cite{armstrong2021sensitivity}. This trade-off emerges endogenously from the quasi-posterior when Gaussian priors with variance proportional to $1/T$ are imposed, rather than being imposed through a minimax or sensitivity criterion. While the resulting weighting structure differs in detail, this connection highlights how robustness considerations arise naturally in quasi-Bayesian inference under local misspecification. \textcolor{black}{We further show that, in the Gaussian limit experiment, the Bayes credible interval produced by our procedure under a two-point least favorable prior coincides with the robust confidence interval of \cite{armstrong2021sensitivity}; see Supplementary Appendix Section~\ref{AK}.} This equivalence provides a formal link between the two approaches and clarifies the sense in which the quasi-Bayesian method recovers existing robust frequentist procedures in the least favorable case. It also helps explain why, under alternative priors, our confidence intervals are typically narrower than those obtained from the procedure of \cite{armstrong2021sensitivity}, as illustrated in the empirical application in Section~\ref{sec:empirical application}.

Another closely related paper is \cite{chen2018monte}, who use simulation from quasi-posteriors to construct confidence sets with frequentist coverage guarantees for identified sets in general settings, including moment condition models. They illustrate their approach in moment inequality models by augmenting the model with an auxiliary parameter $\mu_*$, imposing support restrictions implied by the moment inequalities, and profiling out $\mu_*$ rather than placing a prior on it. While their framework could in principle be extended to other settings with support restrictions on $\mu_*$, including the environment considered here, our focus is on a fundamentally different inferential regime in which a prior is imposed on $\mu_*$. This distinction is not innocuous: the presence of a prior has important implications for posterior concentration and alters the asymptotic behavior of the quasi-posterior in partially identified models.
Our resulting theory relies on arguments that differ from those in \cite{chen2018monte} and yields insights that are not directly available under profiling-based approaches. These theoretical differences are illustrated through concrete examples in Supplementary Appendix Section~\ref{chen}.
 
Our perspective is different from much of this previous work whose chief goal is establishing frequentist guarantees under partial identification as we wish to impose a proper subjective prior over $\mu_*$. That is, we mostly maintain a subjective Bayesian perspective as we believe there are settings where researchers will want to employ informative, subjective beliefs about potential misspecification. Our work thus aligns closely with the strand of Bayesian work on partial identification reviewed in \cite{gustafson:book}. An element of this work is establishing posterior concentration results. We contribute to this literature by providing such concentration results within the semiparametric moment condition framework where the source of partial identification is uncertainty about the potential misspecification. These concentration results also allow us to consider frequentist properties of posterior summaries and thus complement the broader literature on partial identification and sensitivity analysis.
 
Within the Bayesian literature on partial identification and misspecification, our setup resembles \cite{chib2018bayesian}. \cite{chib2018bayesian} consider a Bayesian semiparametric moment condition model that includes an auxiliary parameter equivalent to $\mu_*$
to capture the misspecification of some moment conditions. However, \cite{chib2018bayesian} focus on establishing posterior concentration on a well-defined pseudo-true value in the case of model misspecification, which requires that the number of free elements in $\mu_*$ is less than $q-k$. We instead allow for the possibility that all elements of $\mu_*$ are free, which precludes point identification of even a pseudo-true value and complicates establishing asymptotic concentration. Our formal results differ substantially in that priors have a non-negligible impact on asymptotic results, and posteriors do not generally concentrate on a unique pseudo-true value.
 
Our work also complements the recent contribution of \cite{andrews:pseudo}, who study Bayesian decision-making in over-identified population minimum distance problems under potential misspecification. By considering appropriately constructed misspecification-invariant statistics, they provide an approach for obtaining interval estimates that have a notion of \textit{ex ante} frequentist coverage under a class of distribution-invariant priors. Our approach applies in just-identified as well as over-identified settings under relatively general subjective priors over misspecification. Our motivation is chiefly Bayesian, but our approach provides a similar notion of \textit{ex ante} frequentist coverage. We thus view the two approaches as complementary. 

In summary, we develop a framework that treats potential misspecification in semiparametric moment condition models through a partial identification perspective induced by parameter over-parameterization. Rather than imposing the dogmatic benchmark that the moment restrictions hold exactly, we allow the researcher to place a proper prior on the degree of misspecification. This leads to a quasi-Bayesian approach to estimation, decision-making, and uncertainty quantification that remains informative under both global and local misspecification. We also show that informative prior restrictions can substantially sharpen the resulting identified set. More broadly, the quasi-posterior provides a tractable way to summarize uncertainty about both the parameter of interest and the extent of misspecification.

The remainder of the paper is organized as follows. In Section \ref{sec:The Approach: Main Ideas}, we more carefully discuss the main ideas and provide a convenient approximation result for the case when the prior for $\mu_*$ is taken to be Gaussian with precision proportional to the sample size---the case of ``local misspecification.'' We present the empirical illustrations in Section \ref{sec:empirical application},  and we provide formal results in Section \ref{sec:Theoretical results}. The Supplemental Appendix (indexed by the prefix ``SA'') and the Online Materials (indexed by the prefix ``OM'') provide additional details and proofs.

\section{The Approach: Main Ideas}\label{sec:The Approach: Main Ideas}

\subsection{\small Plausible Moment Restriction Model}
 
Suppose that we observe data $\{Z_{t}\}_{t=1}^{T}$ which are a realization from some unknown  distribution $\P_{\mu_*}$.  
Suppose that we also have a posited structural economic model, which provides a set of moment restrictions on the distribution $\P_{\mu_*}$ indexed by a $q$-dimensional parameter $\mu_* \in \mathcal{M}$. Specifically, suppose the structural model implies a set of $q \geq k$ equations for a $k$-dimensional parameter $\theta \in \Theta$
\begin{align*}
m(\theta)=\E_{\P_{\mu_*}}[g(Z_{t},\theta)],
\end{align*}
such that there exists a vector $\mu_*$ corresponding to a target parameter $\theta_*$ satisfying
\begin{align*}
m(\theta_*) = \mu_*.
\end{align*}

Of course, with no restrictions on the vector $\mu_*$, it is impossible to update beliefs about $\theta_*$ or the distribution $\P_{\mu_*}$ using the structural model. For any posited value of $\theta$ and distribution $\P_{\mu}$, we can always set $\mu = \E_{\P_{\mu}}[g(Z_{t},\theta)]$ such that the structural moment equation is satisfied.\footnote{Priors over $\theta_*$ and $\P_{\mu_*}$ produce restrictions over $\mu_*$, but the structural moment restriction adds no additional information if $\mu_*$ is left completely unrestricted.} Classical approaches to moment restriction models bypass this difficulty by assuming the vector $\mu_*$ is known to be a fixed, prespecified vector (without loss of generality $\mu_* \equiv 0$). This classical approach is equivalent to imposing the dogmatic prior that the researcher knows the structural moment equations hold exactly under $\P_{\mu_*}$---that is, the researcher has a dogmatic prior that the moment equations are correctly specified.

Unfortunately, it is hard to be fully confident that a set of structural moment restrictions hold exactly in many settings. For example, we may worry that there are unobserved confounding variables or that the functional form of the structural model is incorrect. We allow for departures from the dogmatic belief that the structural moment restrictions hold exactly by making use of a proper, non-degenerate prior distribution over $\mu_*$, denoted $\pi(\mu)$. The use of a proper prior over $\mu_*$ allows moment restrictions to be informative in updating beliefs about $\theta_*$ while falling short of imposing the often implausible restriction that moment restrictions hold exactly.

As a concrete example, consider the constant coefficient linear model
\begin{align*}
Y_t = X_t \theta_* + U_t,
\end{align*}
where $X_t$ is an observed variable with $\E_{\P_{\mu_*}}[X_t U_t] \neq 0$. Further, suppose we observe an additional variable $D_t$  that, based on economic reasoning or institutional knowledge, we believe satisfies the usual instrument exclusion restriction $\E_{\P_{\mu_*}}[D_t U_t] = 0$ for $t = 1,...,T$. Under this belief, we obtain the moment restriction $\E_{\P_{\mu_*}}[D_t U_t] = 0$ which can be used to identify the structural parameter $\theta_*$.

However, it is hard to know that the IV exclusion restriction holds exactly. For example, we might worry that there exists an unobserved confound, $M_t$, that covaries with both $Y_t$ and $D_t$ such that $U_t = M_t + V_t$, $\E_{\P_{\mu_*}}[D_t M_t]  =\mu_* \neq 0$ and $\E_{\P_{\mu_*}}[D_t V_t] = 0$. Imposing the moment restriction $\E_{\P_{\mu_*}}[D_t U_t] = \E_{\P_{\mu_*}}[D_t (Y_t - \theta X_t)] = 0$ and solving for $\theta$ produces
\begin{align*}
\theta &= \left(\E_{\P_{\mu_*}}[D_t X_t]\right)^{-1}\E_{\P_{\mu_*}}[D_t Y_t] = \theta_* +  \left(\E_{\P_{\mu_*}}[D_t X_t]\right)^{-1}\mu_* \neq \theta_*.
\end{align*}

Within the IV example, we might instead consider the restriction $\E_{\P_{\mu_*}}[D_t (Y_t - \theta_* X_t)] = \mu_*$ where we assume that $\mu_*$ is a fixed realization from a random variable $\mu$, e.g., $\mu \sim N(0,\sigma^2)$.
Here, the assumed distribution captures the notion that the researcher believes the instrument is ``close to'' being valid in that the prior mass for $\mu$ is concentrated around 0. The distribution also encapsulates that the researcher believes it is incredibly unlikely that the moment restriction is perfect as $\{\mu = 0\}$ occurs with zero probability under such a distribution. Finally, the researcher can control beliefs about the strength of the unobserved confound via the prior variance, $\sigma^2$, while technically allowing for $\mu_*$ to be unbounded. That is, the proper prior over $\mu_*$ allows a well-defined and concrete description of the moment restriction being plausibly, but not certainly, satisfied.

To summarize, we are interested in a formalized version of a ``plausible'' moment restriction model characterized by parameters $(\theta,\mu)$ such that 
\begin{align*}
m(\theta) = \mu
\end{align*}
and $\mu$ is governed by a prior distribution with density $\pi(\mu)$. We refer to $\mu$ as the ``plausibility characteristic,'' and we denote any root of the equation $m(\theta) = \mu$ as $\theta(\mu)$. For establishing some of the formal results in Section \ref{sec:Theoretical results}, we will assume that $\pi(\mu)$ places strictly positive mass over a region $\Gamma$ such that solutions $\theta(\mu)$ exist for $\mu \in \Gamma$. This prior restriction is essentially trivially satisfied for any prior when $q = k$; see, e.g. \cite{hall2003large}. However, satisfaction of this assumption is not guaranteed with $q > k$, suggesting that care should be taken in adding moment conditions about which a researcher has relatively weak prior beliefs unless the researcher is willing to use very diffuse priors.\footnote{From a frequentist perspective, \cite{armstrong2021sensitivity} note that usual overidentification statistics can be used to infer lower bounds on the magnitude of $\mu$. \cite{andrews:pseudo} consider an interesting different approach in overidentified settings that operationally inflates the size of confidence sets based on the magnitude of overidentification statistics.}

In the next section, we outline a quasi-Bayesian approach to perform inference on our main target parameter $\theta$.

\subsection{\small Quasi-Bayes for Plausible Moment Restrictions}
We adopt a quasi-Bayesian approach to performing inference within the plausible moment restriction model. Let 
\begin{align*}
\widehat m(\theta) := \frac{1}{T} \sum_{t=1}^T g(Z_t, \theta)
\end{align*}
be the average of $g(Z_t,\theta)$ against the empirical distribution at a given value of $\theta$. We can then define a continuous updating GMM-type criterion function for parameters $(\theta,\mu)$ as 
\begin{align}\label{eq: GMM criterion}
Q_{T}(\theta,\mu)=- {T}\left(\widehat{m}(\theta)-\mu\right)^{\top}\widehat{\Omega}_{T}(\theta)^{-1}\left(\widehat{m}(\theta)-\mu\right)
\end{align}
for $\widehat{\Omega}_{T}(\theta)$ a positive definite matrix approximating  $$\Omega(\theta)= \lim_{T\to \infty} \text{Var}\left(\sqrt{T}(\widehat{m}(\theta)-m(\theta))\right).$$ For example, it would make sense to use 
$$\widehat{\Omega}_T(\theta) = \frac{1}{T}\sum_{t=1}^{T} \left(g(Z_t, \theta)-\widehat{m}(\theta)\right)\left(g(Z_t,\theta)-\widehat{m}(\theta)\right)^{\top}$$
under the assumption that the $Z_t$ are i.i.d.

A quasi-posterior based on criterion \eqref{eq: GMM criterion} is then obtained as
\begin{align}\label{eq: Posterior}
p_{T}(\theta,\mu) = \frac{\exp\left(\frac{1}{2}Q_T(\theta,\mu)\right)\pi(\theta,\mu)}{\int_{\Xi} \exp\left(\frac{1}{2} Q_T(\theta,\mu)\right)\pi(\theta,\mu)d\mu d\theta}
\end{align}
where $\pi(\theta,\mu) = \pi(\theta|\mu)\pi(\mu)$ is the joint prior over $(\theta,\mu)$ and  $\Xi$  
is the corresponding joint prior support.  
This quasi-posterior is constructed directly from the sample criterion function together with the prior on $(\theta,\mu)$, and therefore remains well defined even when $(\theta,\mu)$ is not point identified. In many applications, researchers may choose a flat prior over $\theta$, i.e., $\pi(\theta|\mu) \propto 1$. At the same time, the quasi-Bayesian framework also provides a convenient way to incorporate economically motivated prior information about $\theta$ when appropriate. In general, $p_{T}(\theta,\mu)$ will not be available analytically but will need to be approximated using Markov Chain Monte Carlo (MCMC) or other approximation methods; see, e.g., \cite{MCMCbook} for a classic textbook introduction. We also provide a simple Gaussian approximation to the posterior in a setting where the prior for $\mu$ is taken to be normal with a small variance in Section \ref{sec: local}.

We obtain marginal posteriors for $\theta$ and $\mu$ as usual by integration:
\begin{align*}
p_T(\theta) = \int_{\mathcal{M}} p_{T}(\theta,\mu)d\mu \ \text{and} \ p_T(\mu) = \int_{\Theta} p_{T}(\theta,\mu)d\theta,
\end{align*}
where $\mathcal{M}$   
and $\Theta$ respectively denote the support of $\mu$ and $\theta$. $p_T(\theta)$ captures posterior information about the economically meaningful parameter $\theta$ and thus is the chief object of interest. $p_T(\mu)$ also potentially provides interesting information by summarizing posterior beliefs about the plausibility term $\mu$.\footnote{Note that, while $\theta$ and $\mu$ are not jointly identified, the imposition of a proper prior over either $\theta$ or $\mu$ will lead to posterior updating over both $\theta$ and $\mu$, even in settings with $q = k$.} 

Given the quasi-posterior, we can also immediately formulate optimal decisions under the quasi-posterior by minimizing the quasi-posterior expected risk. Specifically, let $\ell(\theta,\mu,d)$ be a loss function that depends on the underlying model parameters $(\theta,\mu)$ and a decision $d \in \mathcal{D}$.\footnote{In most applications, this loss function will depend only on the economically motivated parameter $\theta$, but we allow the loss to be over $\mu$ as well.} We can then define the expected loss minimizing decision under the quasi-posterior, denoted $s_T(p_T)$, as
\begin{align}\label{eq: decisions}
s_T(p_T) \in \arg\min_{d \in \mathcal{D}} \int_{\Xi} \ell(\theta,\mu,d) p_T(\theta,\mu) d\mu d\theta.
\end{align}

The quasi-posterior \eqref{eq: Posterior} and inferential objects obtained from it, such as optimal decisions or credible intervals, can be given an approximate Bayesian interpretation. \cite{fs:bayesmoment} and \cite{andrews2022optimal} study the classic semiparametric moment condition model under correct specification such that $m(\theta_*) \equiv 0$, i.e., under the dogmatic prior that $\mu \equiv 0$.  \cite{fs:bayesmoment} provide prior choices for the unknown model density such that the analog of \eqref{eq: Posterior} within this setting corresponds to the posterior for $\theta$ as the limit when the prior becomes diffuse. By augmenting the parameter space to include $\mu$, \eqref{eq: Posterior} can be obtained as the posterior over $(\theta,\mu)$ under the prior structure of \cite{fs:bayesmoment}. \cite{andrews2022optimal} study optimal decision rules in weakly identified moment condition models under correct specification such that $m(\theta_*) \equiv 0$. \cite{andrews2022optimal} establish that the analog of \eqref{eq: Posterior} for their setting, corresponding to \eqref{eq: Posterior} under the dogmatic prior that $\mu \equiv 0$, results as the limit of a sequence of posteriors under a specific choice of proper priors. The resulting quasi-Bayes decision rule then corresponds to the pointwise limit of the sequence of Bayes decision rules and can thus be motivated as approximating the optimal Bayes decision. We show that these results continue to apply in our setting with non-dogmatic prior over $\mu$ in Section \ref{sec:Theoretical results}. Of course, given the non-degenerate prior over $\mu$, the optimal Bayes decision depends explicitly on not only the prior for $\theta$ as in \cite{andrews2022optimal} but also on the prior for $\mu$. See also \cite{kim:lilgmm} and \cite{gallant:inducedspace} for additional Bayesian motivation and perspective. 

From a purely frequentist perspective, \cite{chernozhukov2003mcmc} verify that inference based on the quasi-posterior is asymptotically equivalent to inference based on efficient GMM in strongly identified settings under correct specification ($m(\theta_*) \equiv 0$). \cite{chernozhukov2003mcmc} further argue that basing frequentist estimation and inference on posterior summaries from \eqref{eq: Posterior}, such as using the posterior mean as a point estimator and posterior credibility interval as a confidence interval, may be desirable in settings where \eqref{eq: GMM criterion} is hard to optimize. However, credible intervals resulting from the quasi-posterior \eqref{eq: Posterior} within the partially identified setting where $\mu$ follows a non-degenerate prior no longer generally deliver usual frequentist coverage guarantees, though they do still have an approximate Bayes interpretation; see \cite{moon2012bayesian} and \cite{gustafson:book}.\footnote{\cite{andrews2022optimal} also show that quasi-posterior interval estimates generally do not provide correct frequentist coverage in weakly identified settings and suggest a procedure to obtain confidence sets with proper frequentist coverage. We choose to focus our coverage results on strongly identified settings while allowing for a non-dogmatic prior over $\mu$. In principle, the weak identification robust confidence set construction of \cite{andrews2022optimal} could also be incorporated in the present setting.} 

We extend the results from \cite{chernozhukov2003mcmc} by studying the frequentist properties of the quasi-Bayes posterior within our setting with a non-dogmatic prior over $\mu$ allowing for sequences with increasing $k$ and $q$ in Section \ref{sec:Theoretical results}. Within this setting, we provide new Bernstein-von Mises type posterior concentration results showing that the quasi-posterior converges to a mixture of Gaussian distributions where the mixture weights and components depend heavily on the specific prior $\pi(\mu)$. 
That is, the posterior aligns with our intuition about partial identification in that the prior for $\mu$ plays a key role in the shape of the posterior even in the limit. See also \cite{gustafson:book}. In the case that one has a dogmatic prior over $\mu$---e.g., $\mu \equiv \mu_*$---our concentration result reproduces \cite{chernozhukov2003mcmc}, in the sense that our posterior approximation collapses to a Gaussian random variable with center $\theta(\mu_*)$ and variance equal to the limiting variance of the efficient GMM estimator in this case.

Based on the posterior concentration results, we then have that usual Bayesian credible regions from the quasi-posterior \eqref{eq: Posterior} have correct frequentist coverage within a two-stage sampling thought experiment where each repeated sample corresponds to drawing a value $\mu_*$ from a random variable with density $\pi(\mu)$ and then generating data such that $m(\theta(\mu_*)) = \mu_*$. Alternatively, one can view this notion of coverage as providing an \emph{ex ante} coverage guarantee in a setting where a single value $\mu_*$ will be realized from a random variable with density $\pi(\mu)$. 

Given that the two-stage sampling notion of coverage is non-standard, we consider a final approach to leveraging \eqref{eq: Posterior} to provide a confidence set with a uniform frequentist coverage guarantee when the plausibility characteristic is taken to be some fixed vector, $\mu_0$, whose value is unknown, but where it is known that $\mu_0 \in C$ for some known compact set $C$. The basic idea is that we can use QBPs exactly as in \cite{chernozhukov2003mcmc} to obtain point and interval estimates that would be asymptotically equivalent to efficient GMM for any fixed $\mu_* \in C$ under strong identification. Letting $CI(\mu_*,\alpha)$ be the resulting $(1-\alpha)\times 100\%$ credible interval, it then follows that $\cup_{\mu_* \in C} CI(\mu_*,\alpha)$ has coverage at least $(1-\alpha)\%$ for $\theta(\mu_0)$. This approach mimics, e.g., the union of confidence intervals approach from \cite{conley2012plausibly} and the approach outlined in Remark 3.3 of \cite{armstrong2021sensitivity}.

\subsection{\small Simple Quasi-Bayes Inference using Gaussian Local Priors}\label{sec: local}

In this section, we outline an approach to doing approximate quasi-Bayes inference when the prior for the plausibility term $\mu$ is normal with a small variance. Specifically, suppose that our prior is 
\begin{align}\label{eq: local prior}
\mu \sim \mathcal{N}\left(\mu_0, \frac{\Lambda}{T}\right)
\end{align}
for some fixed $q$ dimensional vector $\mu_0$ and a fixed, full-rank 
$q \times q$ matrix $\Lambda$.\footnote{We could relax the restriction that $\Lambda$ is full-rank at the cost of a modest complication of notation.} 
A simple form of $\Lambda$ is a diagonal matrix with $\lambda_k$'s on the diagonal, where a small $\lambda_k$ indicates that there is little uncertainty about the plausibility of the $k$-th moment and larger values indicate higher uncertainty.

Intuitively, prior \eqref{eq: local prior} captures the case where misspecification is believed to be small but non-zero in the sense that we believe the moment conditions ``almost'' hold with $m(\theta_0) = \mu_0$. Considering a sequence of priors with variance of order $1/T$ means that prior uncertainty concentrates at the same rate as the sample moments, so neither will dominate as we consider large $T$ asymptotic approximations. Following the literature, we refer to \eqref{eq: local prior} as a local prior; see, e.g. \cite{conley2012plausibly} and \cite{armstrong2021sensitivity}.

We now provide an approximation to the quasi-posterior with $\pi(\mu)$ defined in \eqref{eq: local prior} and a flat prior over $\theta$. We assume $\mu_0$ is such that a solution $\theta(\mu_0)$ satisfying $m(\theta(\mu_0)) = \mu_0$ exists. For simplicity, we further set $\mu_0 = 0$. We assume strong identification in the sense that $m(\theta(\mu))=\mu$ has unique solution $\theta(\mu)$ for each $\mu$ and that the following linearization around $\theta_0 \equiv \theta(\mu_0)$ holds:
$$
m(\theta(\mu)) = G \left(\theta(\mu)- \theta_0\right) + o\left( \| \theta(\mu) - \theta_0\|\right),
$$
where $G=\partial \E(\widehat{m}(\theta))/\partial \theta |_{\theta = \theta_0}$ and $G^{\top}G$ has minimal eigenvalue bounded away from zero.
We focus here on the strong identification case because it yields a transparent approximation that captures important forces shaping the quasi-posterior. Further, define the weighting matrix $\widehat{A}_{T,\theta}$ and its population counterpart $A_\theta$:
\begin{align*}
&\widehat{A}_{T,\theta}=\widehat{\Omega}_T(\theta)^{-1}-\widehat{\Omega}_T(\theta)^{-1}[\Lambda^{-1}+\widehat{\Omega}_T(\theta)^{-1}]^{-1}\widehat{\Omega}_T(\theta)^{-1},\\
& A_\theta=\Omega(\theta)^{-1}-\Omega(\theta)^{-1}[\Lambda^{-1}+\Omega(\theta)^{-1}]^{-1}\Omega(\theta)^{-1}.
\end{align*}

Under these conditions, we show 
in Section \ref{sec:Theoretical results} that $p_{T}(\theta)$ is approximately proportional to 
$$
  \exp\left(-T\|\widehat m(\widehat \theta) + G (\theta - \widehat \theta) \|_{\widehat{A}_{T,\theta}}^{2}/2\right),
$$
where the quasi-posterior mode $\widehat{\theta}$ is the GMM estimator obtained using weighting matrix $\widehat{A}_{T,\theta}$. That is, we can approximate the quasi-posterior for $\theta$ as
\begin{align}\label{eq: normal approx}
\theta \approx \mathcal{N}\left(\widehat\theta, \frac{V}{T}\right), \quad V= \left(G^{\top} \widehat{A}_{T,\theta} G\right)^{-1}.
\end{align}
Note that the weighting matrix $A_{\theta_0}$ in the quasi-posterior is different from the standard efficient GMM weighting matrix $\Omega(\theta_0)^{-1}$.\footnote{As $\widehat{A}_{T,\theta}$ will converge to $A_{\theta_0}$ with $T\to \infty$, we provide intuition as if $A_{\theta_0}$ were known.} Specifically, $A_{\theta_0}$ reflects additional uncertainty brought by not having a fixed, known value for $\mu$. We do see that we recover the case of efficient GMM by letting $\Lambda \to 0$---i.e. by assuming that there is no uncertainty over the moment conditions. 

The quasi-posterior has several interesting features. The center of the quasi-posterior, $\widehat\theta$, corresponds to the classical GMM estimator that uses the weighting matrix $A_{\theta_0}$ rather than the efficient weighting matrix $\Omega(\theta_0)^{-1}$. This centering is intuitive as $\Omega(\theta_0)$ captures only sampling variation in the moments but does not reflect the additional uncertainty arising from the plausibility of the moments. The weighting matrix $A_{\theta_0}$ incorporates both sources of uncertainty, intuitively placing the most weight on moments about which the researcher is most confident in the sense that combined sampling variability and plausibility uncertainty is lowest. 

Looking at the quasi-posterior variance, there are two further noteworthy features. First, the variance matrix $V = (G^{\top}A_{\theta_0} G)^{-1} \geq (G^{\top} \Omega({\theta_0})^{-1} G)^{-1}$, where $(G^{\top}\Omega({\theta_0})^{-1} G)^{-1}$ is the usual asymptotic variance of the efficient GMM estimator. The variance matrix $V$ thus captures additional uncertainty, relative to efficient GMM, introduced by a lack of certainty over the validity of the moment restrictions. Second, the approximate sampling distribution of $\widehat\theta$ is 
$$
\sqrt{T}(\widehat{\theta} -\theta_0 )\rightarrow_d \mathcal{N}(0, \bar V), \quad \bar V = (G^{\top} A_{\theta_0} G)^{-1} G^{\top} A_{\theta_0} \Omega({\theta_0}) A_{\theta_0} G (G^{\top} A_{\theta_0} G )^{-1}, 
$$
where $\bar V \leq V$ because $A_{\theta_0} \Omega({\theta_0}) A_{\theta_0} \leq A_{\theta_0}$. Thus, the quasi-posterior variance is also larger than the asymptotic variance of the $A_{{\theta_0}}$-weighted GMM estimator. Again, this larger quasi-posterior variance arises because the sampling distribution of the $A_{{\theta_0}}$-weighted GMM estimator is obtained under the dogmatic belief that $\mu \equiv 0$ and thus does not reflect uncertainty about the validity of the moment restrictions outside of through reweighting the moment conditions.

To summarize, the approximation in \eqref{eq: normal approx} provides a very simple avenue to obtain approximate Bayesian inference under local Gaussian priors. While restrictive, it does seem like a Gaussian prior with small variance may provide a reasonable model for subjective beliefs about moment condition violations in some settings, and we illustrate the use of the approximation, along with illustrating simulation of the full quasi-posterior, in the empirical examples in the next section. More importantly, the approximation captures the clear intuition that there is no free lunch. Incorporating non-dogmatic priors over moment condition violations naturally results in less informative inference relative to the case where dogmatic priors are imposed---reflecting the researcher's uncertainty about the validity of the moment restrictions. This property seems desirable as the resulting inference likely more accurately reflects what can be learned in real empirical settings where model uncertainty exists.
  
\section{Empirical Applications}\label{sec:empirical application}

This section applies plausible GMM in two illustrative empirical applications. In the first, we revisit \cite{acemoglu2001colonial}, which uses linear instrumental variables (IV) to study the effect of institutions on economic output in a relatively small sample. In the second, we revisit \cite{CH:401krestat}, which uses IV quantile regression to examine the effects of 401(k) participation on measures of household assets.   

\subsection{\small Linear IV Example: Effect of Institutions on GDP}\label{subsec:Illustration Example 1} 

We start by illustrating our methodology by revisiting the classic study of \cite{acemoglu2001colonial}, which investigates the effect of institutions on economic performance. 
The outcome variable $Y_t$ is the log of PPP-adjusted GDP per capita in 1995 where $t = 1, \ldots, 64$ indexes a set of countries that are ex-European colonies. The main regressor of interest, $X_t$, is a ten-point index measuring protection against expropriation risk, serving as a proxy for institutional quality. We consider a baseline specification from \cite{acemoglu2001colonial} which includes normalized distance from the equator, $W_t$, as a control for geographic factors.

To address concerns about the endogeneity of $X_t$, \cite{acemoglu2001colonial} adopt an IV strategy. Following \cite{acemoglu2001colonial}, we consider two IV specifications. The first, denoted Linear IV(1), uses the log of settler mortality as the sole instrument. This specification corresponds to the baseline in the original study. The second specification  adds the proportion of the population of European descent in 1900 as an additional instrument as is done in a robustness exercise in the original paper. This specification allows us to illustrate our procedure in a setting with an overidentifying moment restriction. We refer to this specification as Linear IV(2).

Formally, we consider the linear IV model
$$
Y_t = \alpha + \beta_X X_t + \beta_W W_t + U_t,
$$
with parameter vector $\theta = (\alpha, \beta_X, \beta_W)^\top$ and moment condition
$$
g(Z_t, \theta) = (1, W_t, D_t^\top)^\top \left( Y_t - \alpha - X_t \beta_X - W_t \beta_W \right),
$$
where $D_t$ denotes the vector of instruments.

To implement PGMM, we must specify priors for the parameters $(\theta, \mu)$. We set the prior for $\theta$ as $\mathcal{N}(0, \text{diag}(100, 4, 64))$. We set the prior variances for the elements of $\theta$ via a loose argument based on economic intuition. For example, we know that the $X_t$ is measured on a 10-point scale with empirical 25th and 75th percentiles equal to 5.6 and 7.8, respectively, and an empirical range of $3.5$ to $10$. A coefficient of $\beta_X$ of approximately $.5$ would thus suggest moving from the 25th to 75th percentiles of $X_t$ is associated with around a one log unit change (around a 170\% change) in GDP, which seems economically quite large. We thus feel comfortable placing a relatively low prior probability on $\beta_X$ having a magnitude larger than 4. We use the same rationale for our choice of the prior over $\alpha$ and $\beta_W$. 

To specify the prior for $\mu$, we use reasoning based on the IV model. Specifically, we consider a benchmark where model misspecification arises from an omitted latent factor, $C_t$, that is correlated with the exogenous variables and may directly affect the outcome. That is, we consider an ``augmented'' model
$$
Y_t = \alpha + \beta_X X_t + \beta_W W_t + C_t + U_t,
$$
where we represent $C_t$ as $C_t=(1, W_t, D_t^\top)\pi$ to capture misspecification that is linearly related to the exogenous variables. In the case of Linear IV(1), this augmented structure implies that, for a given $\pi$, the moment equation becomes $\mathbb{E}[g(Z_t, \theta)] = \E[(1, W_t, D_t)^\top (1, W_t, D_t)] \pi = \mu$.

We can then capture subjective beliefs about misspecification by placing a proper prior over $\pi$. We start by centering this prior at 0, reflecting the subjective belief that the arguments for exclusion and exogeneity are compelling enough to center our beliefs on no direct effect of the instrument and no endogeneity in the control variable. We then focus on the element of $\pi$ associated with the excluded instrument. Given that $D_t$ is log mortality among Europeans several hundred years prior to 1995, one might reasonably believe there is limited scope for the direct effect of $D_t$ to be large. We benchmark our prior using the subjective belief that, with high probability, the elasticity of GDP with respect to settler mortality is no larger than 10\%, corresponding to the last entry of $\pi$ being no larger than 0.1. We encode these beliefs by specifying a mean-zero Gaussian prior for the entry of $\pi$ corresponding to $D_t$ with standard deviation $0.05$. We regard this as a sensible benchmark in this setting. It places substantial prior mass near 0, reflecting the view that the argument in \cite{acemoglu2001colonial} is broadly persuasive, while still allowing for moderate violations and assigning less mass to larger departures.
We complete the prior by assuming the other entries of $\pi$ behave similarly to the one corresponding to the instrument. Specifically, our baseline (denoted ``PGMM-g'') uses a Gaussian prior given by 
$$
\mu \sim \mathcal{N}(0, \Sigma_T \Omega_d \Sigma_T^\top),
$$
where $\Sigma_T = T^{-1} \sum_{t=1}^{T} (1, W_t, D_t^\top)^\top (1, W_t, D_t^\top)$ and $\Omega_d = 0.05^2 I_3$.

For Linear IV(2), we follow a similar approach. Here, the additional instrument is the proportion of the population of European descent. Assuming that its direct impact on the outcome is, with high probability, no greater than 0.01 (a semi-elasticity of 1\%), we extend the Gaussian prior construction used in Linear IV(1) by setting $\Sigma_T = T^{-1} \sum_{t=1}^{T} (1, W_t, D_t^\top)^\top (1, W_t, D_t^\top)$ and 
$$
\Omega_d = \text{diag}(0.05^2 I_3, 0.005^2),
$$
where the final diagonal entry corresponds to the new instrument.

Of course, it is important to gauge the sensitivity of the posterior to the prior specification. We thus report results using two additional simple prior settings. In the first, we consider a more diffuse prior for $\mu$ (denoted ``PGMM(d)-g''), given by $\mathcal{N}(0, c \Sigma_T \Omega_d \Sigma_T^\top)$ with $c=4$. As a second alternative, we also consider a uniform prior for $\mu$ (denoted ``PGMM-u''), distributed uniformly over the elliptical region $\mathcal{C} = \{ \left(\Sigma_T \Omega_d \Sigma_T^\top\right)^{1/2}c: c^{\top}c \leq \chi^2_{0.68}(q)\},$ where $q$ is the number of moment conditions, and $\chi^2_{0.68}(q)$ denotes the 0.68 quantile of $\chi^2(q)$. This prior thus allocates all probability mass to the 68\% highest density region of the Gaussian prior used in the ``PGMM-g'' case.

 \begin{figure}[htbp]
        \centering
        \includegraphics[width=\textwidth]{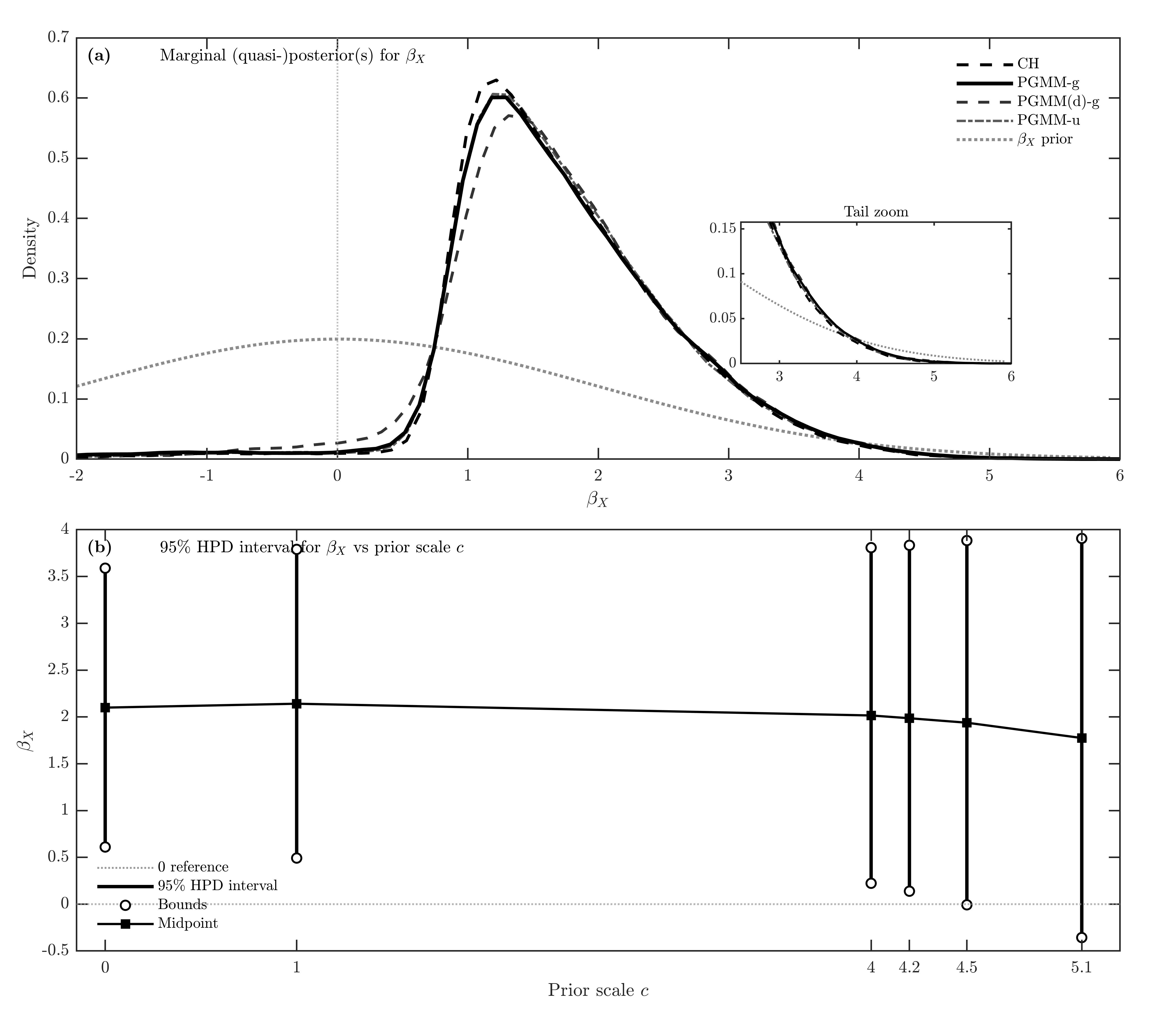} 
\caption{Upper panel (a): Linear IV (1) marginal (quasi-)posterior(s) for $\beta_X$ and its marginal prior (dotted curve). Lower panel (b): Linear IV(1) 95\% HPD intervals for $\beta_X$ resulting from the PGMM approach along various priors for $\mu$, i.e., $ \mathcal{N}(0,c\Sigma_T \Omega_d \Sigma_T^\top)$ for various values of $c$. 
 } \label{fig:ajr-empirical}
\end{figure}

We report the PGMM quasi-posterior obtained under our different priors, along with the quasi-posterior from \cite{chernozhukov2003mcmc} obtained under the dogmatic prior $\mu \equiv 0$ (labeled "CH"), for $\beta_X$ in the specification with one excluded instrument in the top panel of Figure \ref{fig:ajr-empirical}.\footnote{Quasi-posteriors in the Linear IV(2) case, shown in Figure \ref{fig:ajr-empirical:2}, exhibit similar patterns. We also present posteriors for elements of $\mu$, which roughly align with the corresponding priors, in Figure \ref{fig:ajr-empirical:3}.} We see that, in terms of $\beta_X$, the quasi-posteriors are relatively robust to the prior over $\mu$. As anticipated, we see that the quasi-posteriors become somewhat more diffuse as the prior dispersion increases from $\mu \equiv 0$ to PGMM-g to PGMM(d)-g, although the changes in dispersion are relatively small despite the large increase in prior dispersion for $\mu$ across these cases. Unsurprisingly given the design of the uniform prior, we also see that both the benchmark Gaussian prior (PGMM-g) and related uniform prior (PGMM-u) produce very similar quasi-posteriors for $\beta_X$.  

For additional insight, we provide 95\% HPD intervals for $\beta_X$ under $\mu \sim \mathcal{N}(0, c \Sigma_T \Omega_d \Sigma_T^\top)$ for additional values of $c$ in the bottom panel of Figure \ref{fig:ajr-empirical}. Here, we see that the lower bound of the HPD interval is relatively stable for values of $c \leq 4$. However, the lower bound then decreases relatively quickly as $c$ increases away from 4, crossing 0 at $c \approx 4.5$. That is, under our prior for $\theta$ and class of priors for $\mu$, posterior mass remains largely concentrated over positive effects even allowing for what appear to us be economically large deviations from correct specification. 

Finally, we observe that the left tails of the quasi-posteriors in Figure \ref{fig:ajr-empirical} are very similar to the upper tail of the maintained prior for $\beta_X$. That is, it appears that the behavior of the upper tail of the posteriors may be driven largely by the prior choice for $\theta$. Consequently, the upper bounds of the provided intervals are relatively insensitive to the prior for $\mu$. While unsurprising, we find this interplay between prior structure interesting, especially as researchers often have reasonable economic understanding about plausible values for structural parameters.

We report interval estimates obtained from a variety of procedures under both the Linear IV(1) and Linear IV(2) specification in Table \ref{tab:acemoglu2}. For frequentist methods, we report 95\% level confidence intervals, and we report 95\% level HPD regions for (quasi-)Bayesian procedures. Specifically, we consider intervals produced by applying the following: 
\begin{itemize}
\item \textbf{2SLS}: two-stage least squares with the usual asymptotic approximation.
\item \textbf{CUE}: continuous updating estimator with the usual asymptotic approximation.
\item \textbf{CH}: PGMM under $\mu \equiv 0$; the quasi-Bayes approach from \cite{chernozhukov2003mcmc}. 
\item \textbf{S}: inversion of the $S$ statistic; see Theorem 2 from \cite{stock2000gmm}. 
\item \textbf{AK}: robust intervals constructed using the local misspecification method of \cite{armstrong2021sensitivity} assuming a true value $\theta_*$ such that  
$g(\theta_*) =c/\sqrt{T}, c\in \mathcal{C}$ for $\mathcal{C}$ specified as in PGMM-u. 
\item \textbf{PGMM-u}: PGMM with uniform prior as previously specified.
\item \textbf{PGMM-g}: PGMM with baseline Gaussian prior.
\item \textbf{PGMM(d)-g}: PGMM with diffuse Gaussian prior. 
\item \textbf{Local Approx:} Gaussian limiting approximation for the $\beta_X$ marginal quasi-posterior of "PGMM-g" under the assumption of local misspecification as described in Equation (\ref{eq: normal approx}) and Theorem \ref{Theorem:1}.
\item \textbf{Local Approx (d):} Gaussian limiting approximation for the $\beta_X$ marginal quasi-posterior of "PGMM(d)-g" under the assumption of local misspecification as described in Equation (\ref{eq: normal approx}) and Theorem \ref{Theorem:1}.
\end{itemize} 
All approaches allow for heteroskedasticity. 2SLS, CUE, CH, and S maintain the assumption of correct specification. All other procedures allow for departures from correct specification by relaxing the constraint that $\mu \equiv 0$, with AK being frequentist valid and the remaining procedures having a (quasi-)Bayes interpretation. We note that S is formally valid under weak identification, while formal frequentist results for the other procedures are obtained assuming strong identification.

\begin{table}[htbp!]
\centering
\begin{tabular}{lcccc}
\hline\hline
\multicolumn{2}{l} {\textbf{Methods} }  & \textbf{Linear IV(1)} & \textbf{Linear IV(2)} \\
\hline\multicolumn{1}{l}{{Assuming correct specification, i.e.,  $\mu\equiv 0$}} &&& \\ 
2SLS && [0.56, 1.38] & [0.64, 1.23] \\
CUE && [0.56, 1.38] & [0.64, 1.21] \\  
CH  & &[0.61, 3.59] & [0.63, 3.56] \\
\hline
\multicolumn{1}{l}{{Weak identification robust}} &&& \\  
S  && [0.63, 3.23] & [0.63, 4.55] \\
\hline
\multicolumn{1}{l}{{Misspecification robust}} &&& \\ 
AK  && [-28.86, 30.80] & [-0.87, 2.66] \\
PGMM-u && [0.58, 3.65] & [0.53, 3.78] \\
PGMM-g & &[0.49, 3.79] & [0.54, 3.65] \\
PGMM(d)-g && [0.22, 3.81] & [0.30, 3.73] \\
Local Approx && [0.52, 1.42] & [0.53, 1.36] \\
Local Approx(d) && [0.42, 1.52] & [0.35, 1.57] \\
\hline\hline
\end{tabular}
\caption{95\% interval estimates for $\beta_X$, the effect of institutions on output, obtained from different procedures as described in the main text.}  
\label{tab:acemoglu2}
\end{table}

Table~\ref{tab:acemoglu2} shows that, with the exception of the AK interval, the qualitative conclusions are largely consistent across methods: The centers and lower bounds of the intervals lie above zero, suggesting a positive effect of institutions on output. More specifically, the interval estimates broadly fall within two groups, with 2SLS, CUE, Local Approx, and Local Approx(d) in one group and CH, S, PGMM-u, PGMM-g, and PGMM(d)-g in the other. 

Interestingly, the 2SLS, CUE, and local approximation intervals all rely on asymptotic approximations obtained under strong identification and are substantially narrower than the remaining intervals. One possible interpretation is that these approximations are less reliable in this application because of weak-instrument concerns (see, e.g., \cite{chernozhukov2008instrumental}, \cite{kleibergen2025double}).  
Relatively weak identification could also explain the lack of updating, relative to the prior on $\beta_X$, in the upper tail of the distribution. That is, while we specified what seemed like an economically diffuse prior, the data do not seem sufficiently informative to shift beliefs in the upper tail of the effect distribution. Of course, the more relevant issue from an economic perspective is the amount of posterior mass in the left tail below zero. By contrast, the CH and PGMM intervals are obtained directly from the (quasi-)posterior rather than from strong-identification asymptotic approximations. It is therefore interesting that these intervals are quite similar to the interval produced by the weak-identification-robust procedure. Although this similarity may be partly specific to this application, it suggests a potentially useful connection that may be worth exploring further.

Looking at the quasi-Bayes procedures specifically, recall that the CH method imposes the validity of the moment conditions, while the PGMM procedures relax this assumption. This relaxation, of course, results in the PGMM intervals being wider than CH as the PGMM intervals reflect the added uncertainty from accounting for potential misspecification. However, at least under the priors considered, the increase in width is relatively small and does not qualitatively change the conclusions that one would draw relative to CH.

Finally, we observe that the AK intervals lead to qualitatively different conclusions than those from the other approaches. This difference is particularly pronounced in the Linear IV(1) specification, where the AK interval is substantially wider than the intervals produced by the alternative methods. The most informative comparison is between AK and PGMM-u, as both restrict the plausibility term to lie within the same support. The key distinction is that the AK approach is designed to ensure valid frequentist coverage by focusing on least favorable directions within the misspecified model, given only the support restriction on the plausibility term. In contrast, PGMM-u imposes proper subjective priors on both the plausibility term and the structural parameters, $\theta$. In this example, the subjective priors place very little mass on models with economically extreme values of $\beta_X$, resulting in the quasi-posterior assigning negligible mass to much of the AK interval. This outcome illustrates how subjective priors can substantially shape the quasi-posterior in partially identified settings. Rather than targeting worst-case combinations, the quasi-posteriors reflect economically motivated beliefs about the joint distribution of the structural parameters and the plausibility term. Both approaches serve meaningful purposes, but we believe there are scenarios in which inference based on the quasi-posterior under subjective priors may offer more economically relevant insights.

\subsection{\small \small IV Quantile Regression Example: Effect of 401(k) Participation on Financial Assets}\label{subsec:Illustration Example 2}
In this subsection, we illustrate the use of PGMM in a non-linear model by using IV quantile regression (IVQR) to estimate the impact of 401(k) participation on quantiles of net financial assets as in \cite{CH:401krestat}. Specifically, we apply PGMM using the IVQR moment conditions from \cite{chernozhukov2005iv}: 
$$
g_{\tau}(Z_t, \theta_{\tau})=(1, W_t^\top, D_t)^\top\left(\tau-\IF{\left(Y_t \leqslant \alpha_\tau + X_t \beta_{X,\tau} + W_t^\top \beta_{W,\tau} \right)}\right)
$$
where $\tau$ denotes the quantile of interest; $Y_t$ is the outcome variable representing net total financial assets (in 1991 dollars); $X_t$ is a binary indicator for 401(k) participation; $W_t$ is a vector of control variables; and $D_t$ is a binary instrument indicating 401(k) eligibility.\footnote{The covariates are income, a quadratic in age, family size, four indicators of education categories, marital status, two-earner status, defined benefit pension status, IRA participation, and home ownership. For more details about the data, see, e.g., \cite{abadie2003semiparametric} and \cite{CH:401krestat}.} We report results for three quantiles, $\tau \in \{0.15,0.5,0.85\}$, to illustrate performance for a low, central, and upper quantile.

The basic argument for 401(k) eligibility being a valid instrument for participation is that eligibility is determined by employers and so may plausibly be taken as exogenous after conditioning on job relevant covariates. See, e.g., \cite{abadie2003semiparametric} for further discussion of the underlying exclusion restriction. Of course, there are reasons that one might worry that the exclusion restriction does not hold perfectly. For example, one might conjecture that firms that offered 401(k) plans were attractive to employees who prefer savings for other, unobserved, reasons. Motivated by such concerns, \cite{conley2012plausibly} explore sensitivity of linear IV estimates of the effect of 401(k) participation on financial assets. Our analysis extends this line of work by investigating the sensitivity of quantile treatment effect estimates. We also note that examining quantile effects may be of substantive economic interest given the strong asymmetry of financial asset holdings. 

Of course, implementing PGMM requires specification of priors for the structural parameters, $\theta_{\tau}$, and plausibility terms, $\mu_\tau$. In this example, we use the same diffuse prior for each $\theta_{\tau}$---$\theta_{\tau} \sim \mathcal{N}\left(0, \frac{10^{10}}{5} I_{14}\right)$---for all reported results. Given the magnitude of the outcome variable and units of the input variables,\footnote{For example, the 0.15 and 0.85 quantiles of $Y_t$ are approximately $-2751$ and $36,303$, respectively.} this prior seems to be extremely uninformative.

The more delicate choice is the prior over the local misspecification parameter $\mu_{\tau}$. As in the previous example, we assess sensitivity to this choice by considering both zero-mean Gaussian priors and zero-mean uniform priors for each of our three values of $\tau$. We set baseline priors using a stylized model for misspecification. Specifically, we construct a bound on the moment conditions, denoted $\delta_\tau$, under the assumption that any misspecification arises from a direct effect of $D_t$ on $Y_t$, capped at 2000 dollars in absolute value. To simulate this bound, we compute
$$
\max_{\gamma \in \{-2000,2000\} } \left|\frac{1}{T}\sum_{t=1}^T  (1, W_t^\top, D_t)^\top\left(\tau-\IF{\left(\epsilon_t +D_t \gamma \leqslant 0 \right)}\right)\right|,
$$ 
where the $\epsilon_t$ are generated as the residuals from the linear IV analog of our quantile models. Using the resulting $\delta_\tau$, we define priors for $\mu_{\tau}$ as $c \cdot \mathcal{N}\left(0, \text{diag}(\delta_\tau/3)^2\right)$ or as uniform priors over $[-c\delta_\tau/3, c\delta_\tau/3]$. To explore varying degrees of prior concentration, we consider $c \in \{0, 0.5, 0.9, 1.0\}$, where $c = 0$ corresponds to the dogmatic prior that maintains correct specification.  

\begin{figure}[htbp!] 
    \centering
\includegraphics[width=\textwidth]{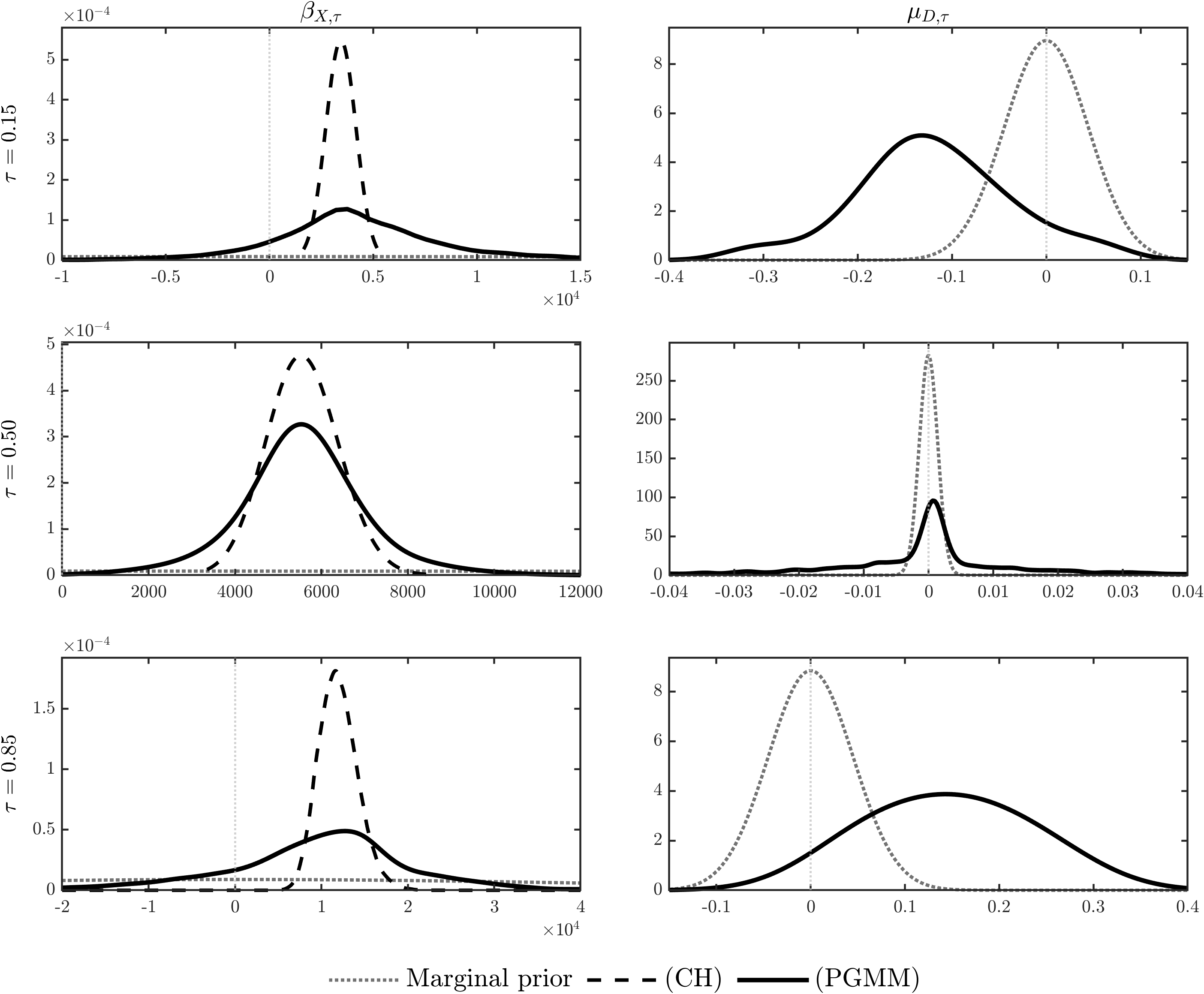}
  \caption{CH (dashed black curves, $\mu\equiv 0$) and PGMM marginal quasi-posteriors (solid black curves) for $\beta_{X,\tau}$ and $\mu_{D,\tau}$. Examples shown use Gaussian priors over $\mu$ that vary with $\tau$ with $c=1$ as described in the main text. The light dotted curves represent the marginal prior density curves for the displayed parameters.   
}
    \label{fig/401k-posterior_gauss_unif.png} 
\end{figure}

Figure \ref{fig/401k-posterior_gauss_unif.png} displays marginal quasi-posteriors for both $\beta_{X,\tau}$ and the component of $\mu_\tau$ associated with the IV moment condition, denoted by $\mu_{D,\tau}$, under the Gaussian prior for $\mu_\tau$ with $c = 1$.\footnote{We provide quasi-posterior plots for the remaining settings in Figures \ref{401k_posterior_figure_c05onlyPGMM}-\ref{401k_posterior_figure_c1onlyPGMMu}.} For comparison, we also provide the marginal quasi-posterior for $\beta_{X,\tau}$ under the assumption of correct specification ($c = 0$) in dashed curves. We see that the quasi-posteriors for the quantile effect obtained under correct specification concentrate over positive values for each value of $\tau$, suggesting a robust positive effect of 401(k) participation on net financial assets. The quasi-posteriors under correct specification are suggestive of larger quantile treatment effects at higher quantiles. 

Looking at the PGMM quasi-posteriors, we see that allowing for departures from correct specification according to our prior leads to substantially more diffuse quasi-posteriors. For $\tau = 0.15$ and $\tau = 0.85$, the quasi-posterior for $\beta_{X,\tau}$ now places substantial mass on both large positive and large negative values. This substantial spread in the posterior implies that it becomes difficult to draw reliable conclusions about the lower and upper quantile treatment effect of 401(k) participation once we allow for plausible violations of the exclusion restriction consistent with the instrument having up to a \$2,000 direct effect on savings. The resulting intervals for low and high quantiles from ``PGMM-g'' in Figure \ref{ExampleII-empirical} span from -7.88 to 18.09 and -13.66 to 33.23 (in units of thousands of dollars), respectively. In contrast, the quasi-posterior for the median remains concentrated over positive values, suggesting a relatively robust positive median treatment effect of 401(k) participation. 

Interestingly, Figure~\ref{fig/401k-posterior_gauss_unif.png} reveals that the marginal quasi-posterior distributions of $\mu_{D,\tau}$ differ noticeably from the prior. In particular, the quasi-posterior for $\tau = 0.15$ ($\tau = 0.85$) is shifted to the left (right) relative to the prior. The quasi-posterior for the median remains centered approximately over the prior center, but has substantially thicker tails than the marginal prior. This pattern suggests that deviations from the baseline model are more likely at the lower and upper quantiles, indicating that the moment conditions for the tails are “less plausible”---in the sense that their quasi-posteriors are not centered at zero---than those for the median. We find it interesting that, at least viewed through the lens of the marginal quasi-posteriors over $\mu_{D,\tau}$, the combination of data and moments leads to updating in the direction of model misspecification. That is, the marginal quasi-posterior over the sensitivity term for the IV moment restriction is less concentrated around zero (correct specification) than the initial prior. We take this updating as further evidence that a researcher should be hesitant to trust results that dogmatically maintain correct specification in this example.

\begin{figure}[htbp]
    \centering
\includegraphics[width=0.9\textwidth]{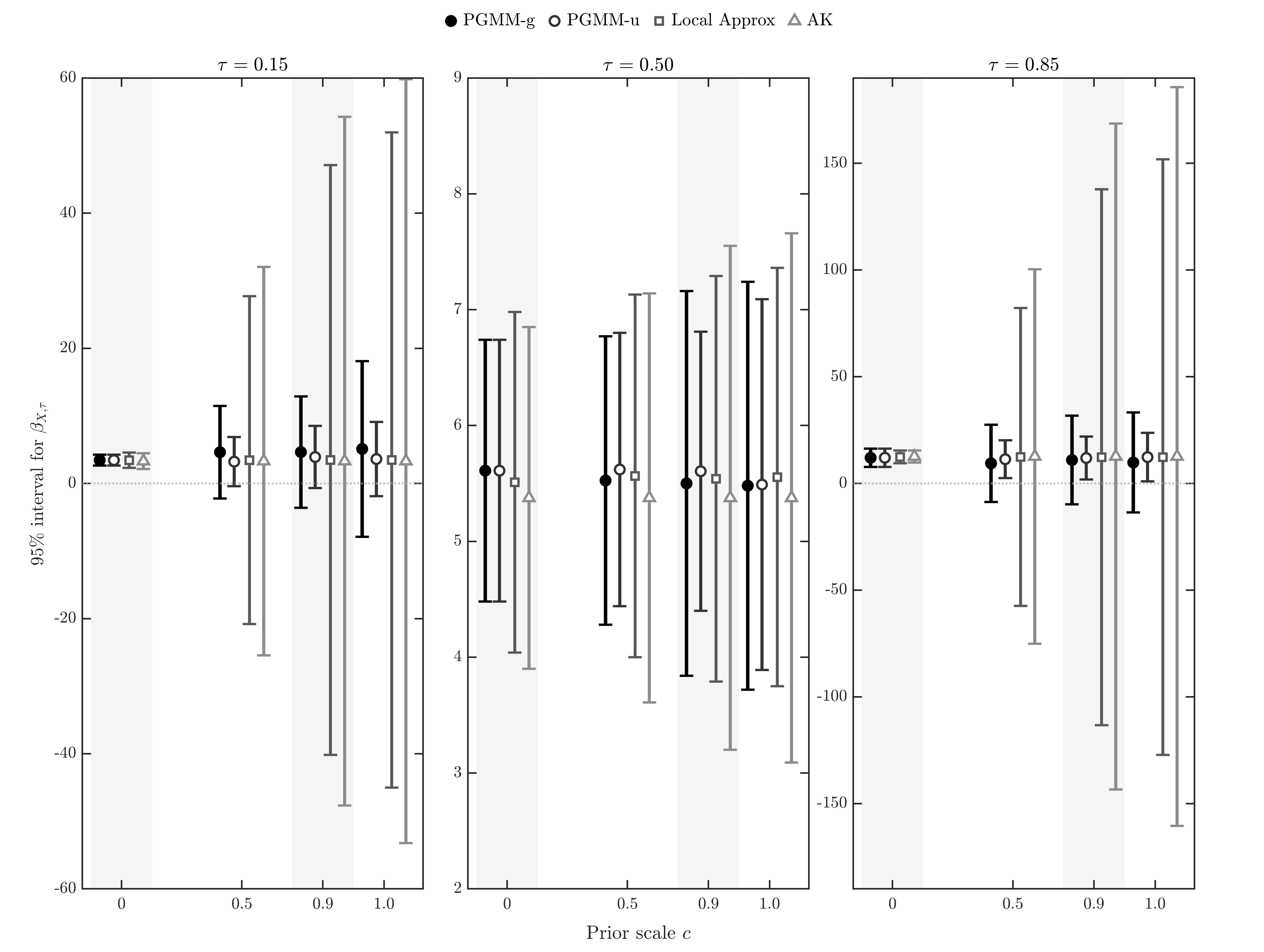}
\caption{
This figure shows the 95\% intervals constructed for the treatment effect parameter, $\beta_{X,\tau}$, using the limiting approximation indicated by Theorem \ref{Theorem:1} (Local Approx), PGMM (PGMM-u denotes the cases with uniform priors on $\mu$ while PGMM-g denotes the cases with Gaussian priors on $\mu_\tau$) 
and AK (AK) for IVQR with $\tau=0.15, 0.50, 0.85$  and  $c=0, 0.5,0.9, 1.0$.
}
\label{ExampleII-empirical}
\end{figure}

Figure~\ref{ExampleII-empirical} reports 95\% highest quasi-posterior density intervals constructed using the PGMM method under our full set of prior specifications for $\mu_\tau$. In addition to the PGMM intervals obtained from simulating the full quasi-posterior, we also report intervals based on the local limiting approximation described in Section~\ref{sec: local} and Theorem~\ref{Theorem:1}, as well as frequentist intervals constructed using the method of \cite{armstrong2021sensitivity} (AK). The AK intervals are derived under a local misspecification framework in which the true parameter value $\theta_{\tau,X}$ is assumed to satisfy $\sqrt{T} m(\theta_{\tau, X}) \in \mathcal{C}_\tau$. For comparability, we define the restriction set $\mathcal{C}_\tau$ to match the support of the corresponding uniform priors used in PGMM-u for a given constant $c$: 
$$
\mathcal{C}_\tau = \left\{ \sqrt{T} \cdot \text{diag}(c \delta_\tau / 3) a : a \in \mathbb{R}^q,\ |a|_\infty \leq 1 \right\},
$$
where $|\cdot|_\infty$ denotes the $\ell_\infty$ norm. The parameter $c$ in Figure~\ref{ExampleII-empirical} thus has a different interpretation for the different methods. For PGMM with a Gaussian prior on $\mu_\tau$ (PGMM-g), it determines the standard deviations of the prior distribution. For PGMM with a uniform prior (PGMM-u), it sets the upper and lower bounds of the uniform prior support. For the local approximation intervals (Local Approx), it indexes the scale of the local Gaussian prior $\mathcal{N}(0, \Lambda_c / T)$, where $\Lambda_c = T c^2 \cdot \text{diag}((\delta_\tau / 3)^2)$, consistent with the PGMM-g case. For the AK intervals, $c \delta_\tau$ parameterizes the local misspecification set $\mathcal{C}_\tau$ as defined above.

As shown in Figure \ref{ExampleII-empirical}, the results for the lower and upper quantiles appear relatively sensitive to both the value of $c$ and the method used to obtain the interval estimate. As expected, the AK intervals---which are designed to ensure asymptotic frequentist coverage under worst-case local misspecification---are strictly wider than PGMM-u intervals in all cases. For these lower and upper quantiles, we see that the AK intervals are much wider than the corresponding PGMM intervals and, interestingly, are tracked relatively closely by the intervals obtained from the local approximation to the quasi-posterior. An interesting feature of this example is that the moment corresponding to plausibility term $\mu_{D,\tau}$ has natural support restrictions. The full quasi-Bayes procedure updates such that values that violate these support restrictions have very little posterior mass. This updating does not occur in either the local approximation or the AK intervals, which may explain some of the discrepancy between the procedures, especially for larger values of $c$. 

In contrast, the intervals for the median effect are notably more stable across methods. All approaches yield similar interval estimates, and the lower bounds remain above zero even under relatively diffuse priors on the misspecification term. This stability suggests that inference about the median treatment effect is more robust to the specification of priors and the choice of estimation method.

As in the previous example, we find that examining quasi-posteriors under non-dogmatic priors on the degree of misspecification offers valuable insight into the identification and plausibility of the estimated effects. Estimates of quantile effects in the upper and lower tails are relatively sensitive to assumptions about model specification, with this sensitivity manifesting as instability across methods and prior choices. In contrast, the estimated median effects appear considerably more robust, yielding qualitatively similar results across all considered specifications. Finally, the updating from the prior over $\mu_\tau$ to the quasi-posterior is particularly interesting. In all cases, the quasi-posteriors place more mass away from $\mu_\tau = 0$ than the prior, indicating quasi-posterior evidence against correct specification. This shift suggests that researchers should be cautious about imposing the assumption of correct specification too rigidly in this setting.

\section{Theoretical results}\label{sec:Theoretical results}
 
This section is organized as follows. We first define notation in Subsection \ref{sec:notation}. Subsection \ref{sec:opt} introduces a Bayesian optimal decision-theoretic motivation for the proposed quasi-Bayes procedure. Subsection \ref{simple} presents a Bernstein-von Mises (BvM) theorem in a fixed-dimensional setting under local misspecification and establishes theoretical coverage guarantees for the highest quasi-posterior region. Subsection \ref{main} extends the BvM result to the high-dimensional case. Finally, subsection \ref{main2} provides a frequentist justification for the coverage of the Bayesian credible set.

 \subsection{\small Notation}\label{sec:notation}
For a vector $v=(v_1,\ldots,v_d)\in\RR^d$ and $q>0$, we denote $|v|_q=\left(\sum_{i=1}^d|v_i|^q\right)^{1/q}$,{$|v|_{\infty}=\max_{1\le i\le d}|v_i|$}, and $\|v\|=|v|_2$.  
For a vector $v$ and a conformable non-negative definite matrix $A$, define $\|v\|_{A}:=\sqrt{v^{\top}Av} \geq 0$. For two positive number sequences $(a_T)$ and $(b_T)$, we say  
$a_T\lesssim b_T$ (resp. $a_T\asymp b_T$) if there exists $C>0$ such that $a_T/b_T\le C$ (resp. $1/C\le a_T/b_T\le C$) for all large $T$. We denote 
$a_T \ll b_T$ if $a_T/b_T\rightarrow 0$ as $T\rightarrow\infty$, and write $a_T\gg b_T$ if $a_T/b_T\rightarrow\infty$ as $T$ diverges.
Let $\nu(\mu)$ denote a point sufficiently close to $\mu$ such that a unique solution $\theta(\nu(\mu))$ of $m(\theta(\nu(\mu))) = \nu(\mu)$ exists.
Denote a total variation of moments (TVM) type norm of ${\kappa}$ for a real-valued measurable function $g$ on $\Theta\times \mathcal{M}$ by  $\left\Vert g \right\Vert_{TVM(\kappa)}
=\int_{\theta \in \Theta, \mu \in \mathcal{M}}
(1+\|\sqrt{T}(\theta-\theta(\nu(\mu)))\|^{\kappa})|g(\theta,\mu)|d\theta d\mu$ for ${\kappa}\geq 0$.

We use $\propto$ to denote ``proportional to''. We use the subscript $p$ to denote statements with respect to the outer measure $\mathbb{P}^*$ of a given probability $\mathbb{P}$. We use $\rightarrow_d$ to denote convergence in distribution. We set $(X_T)$ and $(Y_T)$ as two sequences of random variables. Write $X_T=O_{p}(Y_T)$ if $\forall \epsilon>0$, there exists $C>0$ such that $\mathbb{P}^*(|X_T/Y_T|\le C)>1-\epsilon$ for all large $T$. We denote $X_T=o_{p}(Y_T)$ if $X_T/Y_T\rightarrow_p 0$ as $T\rightarrow\infty$.
We limit ourselves to situations in which, given $\mu$, observations are a random sample from a distribution $\mathbb{P}_{\mu}$ with $\mathbb{P}_{\mu}$ being the conditional law of the random sample given $\mu$; and probability statements under $\mathbb{P}$ are made relative to the joint distribution of the random sample and $\mu$, given a fixed latent distribution $F_\mu$ over $\mu$.

\subsection{\small Link to Bayes optimal decisions}\label{sec:opt}
 
Given the quasi-posterior, we can formulate optimal decisions under the quasi-posterior by minimizing the quasi-posterior expected risk. Specifically, let ${\ell}(\theta,\mu,d)$ be a loss function that depends on the parameters $\theta, \mu$ and a decision $d\in\mathcal{D}$. 
 Policymakers may want to choose a decision $d$ to minimize the loss ${\ell}\left(\theta,\mu,d\right)$ that depends on both $\theta$ and $\mu$. The expected loss-minimizing decision under the quasi-posterior, denoted by $s_T(p_T)$, then takes the usual form:
 \begin{eqnarray}
s_T(p_T) \in \arg \min_{d \in \mathcal{D} 
 }{\int {\ell}(\theta,\mu,d)p_T(\theta,\mu) d\mu d\theta}.
 \end{eqnarray}
We note that this framework encompasses the setting where loss depends only upon $\theta$ as a special case.

These quasi-Bayes decision rules can be motivated as approximations to fully Bayesian rules. In the weak identification setting, \cite{andrews2022optimal} show that the quasi-posterior based on continuously updated GMM can be obtained as the limit of a sequence of posteriors under proper priors. The corresponding quasi-Bayes decision rule can then be interpreted as the pointwise limit of the associated Bayes decision rules.

In the case where parameters are low-dimensional, the results of \cite{andrews2022optimal} can readily be adapted to the PGMM framework. Specifically, we have that,  under regularity conditions, e.g., Assumptions \ref{a1} and \ref{a3}.ii) stated in Section \ref{simple}, the process $\sqrt{T}\widehat{m}(\cdot) - \sqrt{T} m(\cdot)$ converges in distribution to a mean-zero Gaussian process with covariance function $\Sigma(\cdot, \cdot)$ and mean function satisfying $m(\theta(\mu))=\mu$ on $\mu\in \Gamma$. It then follows that we can construct a likelihood as in \cite{andrews2022optimal} by properly substituting their parameter $\theta^*$ with the pair $(\theta(\mu), \mu)$. As a result, the optimal quasi-Bayesian decision rule under model misspecification retains the desirable properties established by their analysis. We provide the supporting technical details in Supplementary Appendix Section \ref{appendix:optimalrule}.

\subsection{\small Gaussian quasi-posterior approximation under local misspecification}\label{simple}
This subsection considers a local misspecification setting in which the prior on $\mu$ is Gaussian with variance $\Lambda/T$.
We show that, under this specification, the quasi-posterior distribution is asymptotically Gaussian and coincides with the results in \citet{chernozhukov2003mcmc} in the special case where $\Lambda = 0$---i.e., in the case that $\mu \equiv 0$. The formal result in this section, Theorem \ref{Theorem:1}, provides justification for the arguments presented in Section~\ref{sec:The Approach: Main Ideas}.

We start by presenting the technical assumptions under which we establish the Gaussian approximation. Throughout this section, we set $\mu_0 = 0$ and define
$$
G(\theta(\mu)) = \frac{\partial \E_{P_\mu}\left[g(Z_t, \theta(\mu))\right]}{\partial \theta(\mu)}.
$$
We also let $\theta(\mu_0) = \theta_0$, $A_{\theta(\mu_0)} = A_{\theta_0}$, and $\Pi_T(\theta, \mu_0)$ be a density function with
$$
\Pi_T(\theta, \mu_0) \propto \exp\left( -\frac{T}{2} \|\theta - \widehat{\theta}\|^2_{G(\theta_0)^\top A_{\theta_0} G(\theta_0)} \right).
$$
 
\begin{Assumption}[Plausibility characteristic]\label{a1} The plausibility characteristic $\mu \in \mathcal{M} \subset \mathbb{R}^{q}$.  
$\Gamma\subseteq \mathcal{M} $ is a nonempty set containing values for $\mu$ such that $\theta(\mu)$ satisfying $m(\theta(\mu)) = \mu$ uniquely exists.  
Each $\theta(\mu)$ belongs to the interior of a compact convex subset $\Theta$ of the Euclidean space $\mathbb{R}^{k}$.  
\end{Assumption}

Assumption \ref{a1} defines the support of the plausibility characteristic and, importantly, a set $\Gamma$ within the support such that the moment condition has a unique solution for each value of $\mu \in \Gamma$. In Assumption \ref{ass:prior1} below, we require that the prior for $\mu$ has positive mass over at least one point in $\Gamma$, which ensures the quasi-Bayes procedure is (asymptotically) well-behaved.

\begin{Assumption}[GMM estimator]\label{a2} 
Assume $m(\theta)$ is first order differentiable in $\theta\in \Theta$. Let 
$$\widehat{\theta} = \argmin_{\theta \in \Theta} \widehat{m}\left(\theta\right)^\top \widehat{A}_{T,\theta} \widehat{m}\left(\theta\right)$$ 
be the GMM estimator using weighting matrix $\widehat{A}_{T,\theta}$. Assume $\widehat{\theta}$ has expansion 
$$\widehat{\theta} =\theta(\mu_{0})+J_A(\theta(\mu_{0}))^{-1}{\Delta}_{T}\left(\theta(\mu_{0})\right)+o_p(1/\sqrt{T}),$$ 
where $J_A(\theta(\mu_{0})) = G\left(\theta_0\right)^{\top}A_{\theta_0}G\left(\theta_0\right)$ and $\Delta_{T}\left(\theta_0\right) =-G\left(\theta_0\right)^{\top}A_{\theta_0}(\widehat{m}\left(\theta_0\right)-\mu_0).$
\end{Assumption}

\begin{Assumption}[Expansion] 
\label{a3}
i) Assume $J_{A}\left(\theta\right)$ is positive definite for all $\theta\in\Theta$, and $J_{A}\left(\theta\right)$  is continuous in $\theta$. Further, assume 
$G(\theta)$ and $\Omega(\theta)$ are continuous and full rank for all $\theta \in 
\Theta$. ii) Assume $$\sqrt{T}\Delta_{T}\left(\theta(\mu_0)\right)=-\sqrt{T}G(\theta(\mu_0))^{\top}A_{\theta(\mu_0)}(\widehat{m}\left(\theta(\mu_0)\right)-\mu_0)\rightarrow_{\mathrm{d}}\mathcal{N}\left(0,\tilde{V}\left(\theta(\mu_0)\right)\right),$$ where $\tilde{V}(\theta(\mu_0))=G(\theta(\mu_0))^{\top}A_{\theta(\mu_0)}\Omega(\theta(\mu_0))A_{\theta(\mu_0)}G(\theta(\mu_0))$ and $$\Omega(\theta(\mu_0))  = \lim_{T\to \infty} \text{Var}(\sqrt{T}(\widehat{m}(\theta(\mu_0))-{m}(\theta(\mu_0)))).$$ 
\end{Assumption}

\begin{Assumption}[Modulus of continuity and identification.]\label{a4}
Let $${r_{T}(m,\theta)}=\sqrt{T}\left\|\left(\widehat{m}(\theta)-\widehat{m}\left(\theta_0\right)\right)-\left(\E\widehat{m}(\theta)-\E\widehat{m}\left(\theta_0\right)\right)\right\|.$$ 
For a sufficiently small positive constant $\delta>0$, 
assume 
$$\sup_{\theta:\|\theta-\theta_0\|\leq\delta}r_{T}(m,\theta)/\left([1\text{\ensuremath{\vee}}\sqrt{T}\|\theta-\theta_0\|]\right)=o_p(r(\delta)),$$
and $r(\delta)\to_{p}0$ if $\delta\to 0$. 
Further, assume that with probability approaching one,
$$\inf_{\theta:\|\theta-\theta_0\|\geq \delta} \|\left(\widehat{m}(\theta)-\widehat{m}\left(\theta_0\right)\right)\|\geq \delta .$$
\end{Assumption}

Assumption \ref{a2} concerns the properties of the GMM estimator with a weighting matrix that incorporates prior uncertainty as outlined in Section \ref{sec: local}. It specifically imposes that the resulting GMM estimator has a linear expansion dominated by its leading term. This assumption rules out weak identification. Nonetheless, the main results, together with analogous insights about misspecification, could in principle be extended to settings with weak instruments along the lines of \cite{andrews2022optimal}, albeit at the cost of substantial additional technical work. Assumption \ref{a3} then characterizes the asymptotic behavior of the leading term in this expansion. Assumptions \ref{a2} and \ref{a3} are analogous to standard assumptions that align, for example, with Assumption 4 and conditions (ii) and (iii) in Proposition 1 in \cite{chernozhukov2003mcmc}.

Assumption \ref{a4} is a modulus-of-continuity condition similar to condition (iv) of Proposition 1 in \cite{chernozhukov2003mcmc}, used there to handle non-smooth criterion functions. It requires the remainder term to be bounded in a neighborhood of $\theta_0$ and is satisfied when the moments are sufficiently smooth. The final condition in Assumption \ref{a4} ensures that $\theta_0$ is asymptotically well-identified.

The next assumption imposes restrictions on the prior that are sufficient for verifying approximate normality of the quasi-posterior.

\begin{Assumption}[Prior]\label{ass:prior1}
$\pi(\mu,\theta)=\pi(\mu)\pi(\theta)$,
where $\pi(\mu)$ is a Gaussian prior centered at $\mu_{0}\in {\Gamma}$ with covariance matrix $\Lambda/T$, and { $\pi(\theta)$}  is bounded and continuously differentiable around an open ball of  $\theta_0 \in \Theta$. $\lambda_{\max}(\Lambda)=o(T)$  
    and 
 $1 \lesssim \lambda_{\max}(\Lambda)/\lambda_{\min}(\Lambda)\lesssim 1$, where $\lambda_{\min}(\Lambda)$ and $\lambda_{\max}(\Lambda)$ denote the minimum and the maximum eigenvalue of a matrix $\Lambda$ respectively. 
\end{Assumption}

The main substantive restriction of Assumption \ref{ass:prior1} is that the prior for $\mu$ is Gaussian with variance of the same order as sampling variation. We further impose that the prior variance is full rank and that, \emph{a priori}, $\theta$ and $\mu$ are independent. 
The requirement that $\Lambda$ be of full rank can be relaxed.\footnote{For example, let $B\in\mathbb{R}^{q\times\tilde q}, \tilde q<q,$
and let $\Lambda_x\in\mathbb{R}^{\tilde q\times\tilde q}$ be full rank. Consider the prior for $\mu$ generated from 
$\mu = Bx, x\sim \mathcal{N}\bigl(0,\;T^{-1}\Lambda_x\bigr).$
Following the same arguments as used to establish Theorem \ref{Theorem:1}, we can establish the quasi-posterior density $p_T(\theta)$ is, for large $T$, approximately Gaussian with covariance matrix $A_\theta\;=\;
\Omega(\theta)^{-1}
\;-\;
\Omega(\theta)^{-1}\,B\,
\bigl(\Lambda_x^{-1}+B^\top\,\Omega(\theta)^{-1}B\bigr)^{-1}
\,B^\top\,\Omega(\theta)^{-1}.$}  

We now present the first main theorem, which shows that the quasi-posterior density $p_{T}(\theta)$ converges in the TVM norm to a Gaussian density. With slight abuse of notation, we denote $\bar{p}_T(\widehat{m}(\theta))=p_T(\theta)$ and obtain $\bar{p}_{T}\left(\widehat{m}(\widehat{\theta})-G(\theta_0)(\widehat{\theta}-\theta)\right)$ by replacing $\widehat{m}(\theta)$ in $\bar{p}_T(\widehat{m}(\theta))$ with $\widehat{m}(\widehat{\theta})-G(\theta_0)(\widehat{\theta}-\theta)$.
\begin{theorem}[Convergence in TVM norm]\label{Theorem:1}
Under Assumptions \ref{a1}--\ref{ass:prior1}, for any $0 \le \kappa < \infty$,
\[
\begin{aligned}
&\Big\Vert 
\bar{p}_{T}\!\left(
\widehat{m}(\widehat{\theta}) 
- G(\theta_0)(\widehat{\theta}-\theta)
\right)
-
\Pi_{T}(\theta,\mu_{0})
\Big\Vert_{TVM(\kappa)}
\\
&\quad =
\int_{\Theta}
\left(1+\|\sqrt{T}(\theta-\theta_0)\|^{\kappa}\right)
\left|
\Pi_{T}(\theta,\mu_{0})
-
\bar{p}_{T}\!\left(
\widehat{m}(\widehat{\theta})
-
G(\theta_0)(\widehat{\theta}-\theta)
\right)
\right|
\,d\theta \to_p 0,
\\[0.8em]
&\Big\Vert 
p_{T}(\theta) - \Pi_{T}(\theta,\mu_0)
\Big\Vert_{TVM(\kappa)}
\\
&\quad =
\int_{\Theta}
\left(1+\|\sqrt{T}(\theta-\theta_0)\|^{\kappa}\right)
\left|
p_{T}(\theta)-\Pi_{T}(\theta,\mu_{0})
\right|
\,d\theta  \xrightarrow{p} 0.
\end{aligned}
\]
\end{theorem}\begin{proof}
    See Appendix \ref{Proof: theorem1}. 
\end{proof}
Theorem \ref{Theorem:1} demonstrates that $p_T(\theta)$ can be asymptotically approximated by a Gaussian density function $\Pi_{T}(\theta, \mu_0)$ under a sequence of Gaussian priors over misspecification that concentrate at the same rate as sampling error. Further, the theorem confirms the expected result that $p_T(\theta)$ concentrates in a $1 / \sqrt{T}$ neighborhood of $\theta_0$ under local misspecification. This result differs from the related approximation result in \cite{chernozhukov2003mcmc} in that the quasi-posterior depends on the prior for the plausibility characteristic, even asymptotically. 
We do note that the approximation result reproduces the result from \cite{chernozhukov2003mcmc} under $\mu_0 =0$ and $\Lambda= 0$. 
 
We note that Theorem \ref{Theorem:1}, along with Theorems \ref{Theorem:2}-\ref{th4} presented below, could be established without requiring the data stream $\{Z_t\}_{t=1}^{T}$ to be i.i.d. Rather, we could work with moments defined as $m_T(\theta(\mu)) = T^{-1}\sum_{t=1}^{T}\E_{\mathbb{P}_{\mu,t}}[g(Z_{t},\theta(\mu))] = \mu$ where $\mathbb{P}_{\mu,t}$ is the marginal distribution for $Z_t$. Within this structure, results could be established under suitable restrictions on dependence and heterogeneity. We do not pursue this direction formally to avoid further complicating notation.

\subsection{\small General quasi-posterior approximation results.}  
\label{main}
This section discusses an extension of Theorem \ref{Theorem:1} by allowing relatively general choice of prior for $\mu$. The key result is that, as in the previous section, the prior for $\mu$ matters even in the limit. However, under a general prior structure, we do not obtain a limiting Gaussian approximation. Rather, we have that, \emph{conditional} on $\mu$, the limiting approximation is Gaussian. Thus, the limiting approximation to the posterior is a Gaussian mixture where mixture weights depend heavily on the prior. We establish the formal results under sequences that allow the dimensions $k$ and $q$ to grow with the sample size $T$, which offers a technical extension of some results even in the case where a dogmatic prior is placed over $\mu$.
 
To accommodate a broader family of weighting matrices, we allow the GMM-type criterion that serves to define our quasi-posterior, $Q_{T}(\theta,\mu)$, to be formed with any positive-definite weight matrix $\widehat W_{T}(\theta)$. That is, we now consider 
$$
Q_{T}(\theta,\mu)=- {T}\left(\widehat{m}(\theta)-\mu\right)^{\top}\widehat{W}_{T}(\theta)\left(\widehat{m}(\theta)-\mu\right), 
$$
where setting $\widehat W_{T}(\theta) = \widehat\Omega_{T}(\theta)^{-1}$ corresponds to the leading case discussed in previous sections. Within this more general formulation, we use $W(\theta)$ to denote the population analog of $\widehat{W}_{T}(\theta)$ in the same manner that $\Omega(\theta)$ serves as the population counterpart of $\widehat{\Omega}_{T}(\theta)$.

In stating the formal results in this section, we make use of additional notation. 
Assume that, for each $\mu \in \mathcal{M}$, there exists a unique 
$\nu(\mu)\in\Gamma$ such that
\[
\|\mu-\nu(\mu)\|
=
\inf_{\gamma\in\Gamma}\|\mu-\gamma\|.
\]
Let 
$h(\theta,\mu)
=
G\big(\theta(\nu(\mu))\big)
\big(\theta-\theta(\nu(\mu))\big)
-\mu+\nu(\mu).$ 
For $\varepsilon\asymp C_\epsilon q\log T$, define the 
$\varepsilon/\sqrt{T}$-expansion of 
$\{(\theta(\mu),\mu):\mu\in\Gamma\}$ by
\[
\begin{aligned}
B_{\varepsilon}
=
\Big\{
(\theta,\mu)\in\Theta\times\mathcal{M}: 
\sqrt{T}\|h(\theta,\mu)\|\leq\varepsilon, 
\sqrt{T}\|\mu-\nu(\mu)\|\leq\varepsilon
\Big\}.
\end{aligned}
\]
Thus, $B_{\varepsilon}$ is an $\varepsilon/\sqrt{T}$-neighborhood containing 
pairs $(\theta,\mu)$ that are close to pairs 
$(\theta(\mu),\mu)$ with $\mu\in\Gamma$. Let
\[
\begin{aligned}
V_T(h(\cdot),\theta,\mu)
=
&\;
2T h^{\top}
W\big(\theta(\nu(\mu))\big)
\big(\widehat{m}(\theta(\nu(\mu)))-\nu(\mu)\big)
\\
&\;
+
T h^{\top}
W\big(\theta(\nu(\mu))\big)h
+
C(\mu),
\end{aligned}
\]
where we abbreviate $h(\theta,\mu)$ as $h$, and
\[
\begin{aligned}
C(\mu)
=
&\;
T\big(\widehat{m}(\theta(\nu(\mu)))-\nu(\mu)\big)^{\top}
W\big(\theta(\nu(\mu))\big)
\big(\widehat{m}(\theta(\nu(\mu)))-\nu(\mu)\big)
\\
&\;-2\log\big(\pi(\theta(\nu(\mu)))\big)
+
2\sqrt{T}\|\mu-\nu(\mu)\|.
\end{aligned}
\] 
We now state sufficient conditions for establishing our limiting approximation to the quasi-posterior.

\begin{Assumption}[Identification, smoothness, and tails]\label{assum33} Assume the following: (i) $m(\theta)$ is first order differentiable. (ii) For any $\theta\in \Theta$,  $G(\theta)$ has singular values bounded from below and above and $W(\theta)$ has eigenvalues bounded from below by $c_w\frac{1}{k}$ and above for $c_w$ a positive constant. (iii) There exists a positive constant $C>0$ such that $\sup_{1\leq j \leq q, (\theta,\mu)\in B_\varepsilon }\E\left[ \|e_j^\top (g(Z_t, \theta) -\mu)\|^{2}\right] <C,$ and $\sup_{\theta \in \Theta}\|\widehat{m}(\theta)-{m}(\theta)\|\lesssim_p \frac{\sqrt{q}}{\sqrt{T}}$.
(iv) For $\theta, \theta'\in \Theta$, $\mu \in \Gamma$, $\|G(\theta)-G(\theta')\|\lesssim  \|\theta-\theta'\|,  \|W(\theta)-W(\theta')\|\lesssim  \|\theta-\theta'\|$, and  $\|\widehat{W}_{T}(\theta)- W(\theta)\| = o_p(1)$.
\end{Assumption}

Assumption \ref{assum33} imposes regularity conditions that are sufficient for good behavior of the limiting criterion function. Importantly, these conditions ensure that, at fixed values of $\mu$, $\theta$ is strongly identified through the restrictions on $G(\theta)$; see, e.g., \cite{hansen2010instrumental}.  

The next assumption restricts priors, importantly requiring that $\theta$ and $\mu$ are \emph{a priori} independent and that the marginal priors place positive mass uniformly over the corresponding parameter space. Define a neighborhood for $\mu$,
 \[
\mathcal{B}_\delta(\mu) := \left\{ \mu' \in \mathcal{M}: \| \mu' - \mu \| < \delta \right\},
\]
and similarly define a neighborhood $\mathcal{B}_\delta(\theta)$ for a point $\theta \in \Theta$.

\begin{Assumption}[Prior]\label{assummu1}  
Assume the prior density $\pi(\theta,\mu)=\pi(\theta)\pi(\mu)$ with $\pi(\mu)>0$ for $\mu \in \mathcal{M}$ and $\pi(\theta)>0$ for $\theta \in {\Theta}$.  
Assume $\pi(\theta)$ is bounded and Lipschitz continuous on $\Theta$. There exists \(C_\mu>0\) with \(\pi(\mu)\le C_\mu\) for all \(\mu\in \mathcal{M}\setminus \Gamma \), and there exists a $\delta>0$ and a plausible pair \((\mu,\theta(\mu))\) such that for all $0<\delta' \leq \delta$, $\min_{\theta \in \mathcal{B}_{\delta'}(\theta(\mu))} \pi(\theta) \geq   (\frac{1}{C})^{k}$ and $\int_{\mathcal{B}_{\delta'}(\mu)} \pi(\mu) d\mu \geq c (\frac{\delta'}{C})^q$ with $C>0$. 
\end{Assumption}
When \(q\) and \(k\) are of fixed dimension, the above assumption imposes a strictly positive lower bound on \(\pi(\theta(\mu))\) and a positive lower bound on the probability mass that \(\pi(\mu)\) assigns to a value \(\mu \in \Gamma\). This condition allows the prior on \(\mu\) to be very informative---for example, it may reduce to a single point within \(\Gamma\)---whereas the prior on $\theta$ must remain sufficiently diffuse to satisfy the required lower bound.

Before stating the next assumption, define
\[
R_{T}(\theta,\mu)
:=\frac{1}{2}\bigl(Q_T(\theta,\mu)
+V_T(h(\cdot),\theta,\mu)\bigr)
+\log\pi(\theta),
\]
which measures the discrepancy between the posterior 
and its Gaussian approximation. This assumption 
collects high-level rate conditions, which we verify under more primitive conditions in Lemmas~\ref{RT} and~\ref{verifybound} 
under Assumptions~\ref{assum9}--\ref{2}.

\begin{Assumption}[Empirical process rates]\label{rT}
The following properties hold with probability 
tending to $1$.

\begin{itemize}

\item[i)] \textbf{Local rate.} On the ball $B_\varepsilon$,
\[
\sup_{(\theta,\mu)\in B_\varepsilon}
\frac{T|R_T(\theta,\mu)|}
{\|\sqrt{T}h(\theta,\mu)\|^2+k(\log T)^2}
\lesssim_p
\frac{\sqrt{k}(\log T)^2}{\sqrt{T}}
\vee\frac{q}{\sqrt{kT}},
\]
\label{eq:cont}

\item[ii)] \textbf{Tail bound.} There exists a constant 
$1/2<C_0<1$ such that for all 
$(\theta,\mu)\in B_\varepsilon^{c}$,
\begin{equation}\label{eq:bound}
R_T(\theta,\mu)
\leq
-C_0\sqrt{T}\|h(\theta,\mu)\|_{W(\theta(\nu(\mu)))}\varepsilon
+\frac{C_0\varepsilon^2}{2}
+\frac{T\|h(\theta,\mu)\|^2_{W(\theta(\nu(\mu)))}}{2}
+\sqrt{T}\|\mu-\nu(\mu)\|.
\end{equation}

\end{itemize}
\end{Assumption}
Assumption \ref{rT} i) is verified by Lemma~\ref{RT} under 
Assumptions~\ref{assum9}--\ref{2}.
Assumption \ref{rT} ii) is verified by Lemma~\ref{verifybound} under 
Assumptions~\ref{assum33} and~\ref{assum10}.
 
The condition  
$$
\sup_{\theta,\mu \in B_{\varepsilon}} \frac{T|R_{T}(\theta,\mu)|}{\|\sqrt{T}h(\theta,\mu)\|^2 + k(\log T)^2} \to_p 0.
$$  
arises from needing to control a modulus of continuity. It ensures the oscillatory behavior of the empirical process $ R_T(\theta, \mu) $ is mild. While the condition is high level, Lemma \ref{RT} shows that this condition is satisfied with differentiable moments. Assumption \ref{eq:cont} effectively imposes an identification requirement for large values of $ \theta $ and a smoothness condition for smaller values of $\theta$ on the term  $\sup_{(\theta,\mu) \in B_{\varepsilon}^{c}} TR_T(\theta,\mu).$ The assumption resembles the finite-sample bound in Lemma A.16 of \cite{spokoiny2019accuracy}, and it enables the derivation of a tail bound outside the ball $B_{\varepsilon}$ using a Gaussian integral argument.
\begin{Assumption}[Dimension and moment conditions]\label{rates}
The dimensions $k$, $q$, $\kappa$, and sample size 
$T$ satisfy 
\[
\frac{k^{(\kappa\vee 2)+1}(\log T)^{2(\kappa\vee 2)+4}}{T}
\to 0
\qquad\text{and}\qquad
\frac{q^{5/2}(\log T)^2}{k\sqrt{T}}\to 0.
\]
\end{Assumption}
The first condition is used in Lemma~\ref{gai} to 
control the moment terms of order $\kappa$, and 
reduces to $k^3(\log T)^8/T\to 0$ when $\kappa=2$.
The second condition implies both $q^2(\log T)^2/T\to 0$ 
and $q^2/(kT)\to 0$ since $q\geq k\geq 1$, 
and is needed to control the weighting matrix 
interaction term in Lemma~\ref{RT}. 
It also ensures $q^2(\log T)^2\to\infty$, 
which guarantees the tail decay 
$k^{\kappa/2}\exp(-cq^2(\log T)^2)\to 0$ 
required in Lemma~\ref{gai}(ii).

Now, define $$N_{T}(\theta,\mu)=\frac{\exp\left\{ -\frac{1}{2}[V_T(h(.),\theta,\mu)]\right\} \pi(\mu)}{\int_{\mu \in \mathcal{M}}\int_{\theta \in \Theta}\exp\left\{ -\frac{1}{2}[V_T(h(.),\theta,\mu)]\right\} \pi(\mu)d\theta d\mu}.$$ 
Under the stated assumptions, we obtain the following Bernstein-von Mises-type result establishing that the quasi-posterior is asymptotically approximated by $N_T(\theta,\mu)$.
\begin{theorem}\label{Theorem:2}
Under Assumptions \ref{a1} and \ref{assum33}--\ref{rates}, for any $0 \le \kappa < \infty$,
\begin{equation}\label{eq:tvm_joint_convergence}
\begin{aligned}
\big\| p_T - N_T \big\|_{TVM(\kappa)}
&:=
\int_{\mathcal M}\int_{\Theta}
\left(1+
\left\|\sqrt{T}\big(\theta-\theta(\nu(\mu))\big)\right\|^{\kappa}
\right)
\\
&\qquad\qquad \times
\left|p_T(\theta,\mu)-N_T(\theta,\mu)\right|
\,d\theta\,d\mu
\;\xrightarrow{p}\;0 .
\end{aligned}
\end{equation}
\end{theorem}
\begin{proof} See Appendix \ref{Proof: theorem2} \end{proof}

The above theorem indicates that, conditional on fixed values of $ \mu $, the quasi-posterior distribution of $ \theta $ can be well approximated by a Gaussian distribution, thus facilitating practical inference via conditional sampling, as demonstrated in Theorems \ref{th3}-\ref{th4}. In contrast to Theorem \ref{Theorem:1}, Theorem \ref{Theorem:2} relaxes the prior specification on $\mu$ by not imposing a Gaussian prior, thereby extending the applicability of the result. While the joint limiting distribution $ N_T(\theta, \mu) $ is not Gaussian in general, it becomes Gaussian when conditioning on $ \mu $. This insight has practical implications: one may select a representative set of $ \mu $ values, compute the corresponding conditional quasi-posterior distributions of $ \theta $, and aggregate the highest quasi-posterior density regions. This approach mirrors the strategy employed by \cite{conley2012plausibly} for constructing robust inference under partial identification.

Let $PR_T(\alpha)$ denote the $(1- \alpha)$\% ($0<\alpha< 1$) highest quasi-posterior density region for $\theta$ obtained from the quasi-posterior $p_T(\theta)$.\footnote{For a positive constant $c$ and density $p_T(\theta)$, $PR_T(\alpha) = \{\theta \in \Theta: p_T(\theta)> c\}$ such that $\int_{PR_T(\alpha)} p_T(\theta) d \theta = 1- \alpha$.} Lemma \ref{coverageprt} shows that $PR_T(\alpha)$ asymptotically provides valid weighted average frequentist coverage in large samples if we envision a world where nature draws $\mu$ from the prior $\pi(\mu)$, in which case $\pi(\mu)$ coincides with the fixed latent distribution $F_\mu$ in the data generating mechanism. 

\begin{lemma}[Weighted average coverage rate of $PR_T(\alpha)$] \label{coverageprt}
Assume that  $\pi(\mu)$ coincides with $F_\mu$ and that Assumptions \ref{a1}, \ref{assum33}-\ref{rates} hold. Let $\partial F_{\mu}(u)/\partial \mu = f_{\mu}(u)$ and
fix $\alpha \in (0,1)$. Assume that $W(\theta(\mu))^{-1} =\Omega(\theta(\mu))$ and the distribution of $\widehat{m}(\theta(\mu)) - \mu$ under $ \mathbb{P}_\mu$ can be well-approximated by a Gaussian distribution with mean zero and covariance matrix $\Omega(\theta(\mu))$. Then $PR_T(\alpha)$ satisfies the following $\pi$-weighted average coverage rate, which also corresponds to the coverage rate under $\mathbb{P}$, in large samples:
$$\int_{\mu} \mathbb{P}_{\mu}(\theta(\mu) \in PR_{T}(\alpha)) \pi(\mu) d \mu = \mathbb{P}(\theta(\mu) \in PR_{T}(\alpha))  = 1-\alpha+o(1).$$ 
\end{lemma}
\begin{proof}
    See Appendix \ref{Proof: lemma1}. 
\end{proof}

\subsection{\small Using quasi-posteriors to provide frequentist inference under support restrictions}\label{main2} 
In this section, we provide an approach to obtain regions for $\theta$ that deliver valid frequentist coverage guarantees under support restrictions over $\mu$. The regions are constructed by taking unions of quasi-posterior credible regions for fixed values of $\mu$. Frequentist validity of this approach relies on properties of quasi-posterior intervals obtained from the posterior distribution $p_{T}(\theta,\mu)$ established in Theorems \ref{th3}-\ref{th4} in this section.  

In the following, we suppose that one is interested in a continuously differentiable function $\eta(\theta): \mathbb{R}^{k}\to \mathbb{R}$. To state our results, we define the following quantities at fixed, given values of $\mu$: $$J_{\Omega,W}\left(\theta(\mu) \right) = G(\theta(\mu))^{\top}W(\theta(\mu))\Omega(\theta(\mu))  W(\theta(\mu))G(\theta(\mu)),$$
$$U_{T}(\mu)={J}_{W}\left(\theta(\mu)\right)^{-1}{\Delta}_{T,W}\left(\theta(\mu)\right), \; {J}_{W}\left(\theta(\mu)\right) = G(\theta(\mu))^{\top}W(\theta(\mu))G(\theta(\mu)), $$
$$\Delta_{T,W}\left(\theta(\mu)\right)=G(\theta(\mu))^{\top}W(\theta(\mu))(\widehat{m}(\theta(\mu))-\mu).$$  We first introduce a high level assumption for an estimator of $\theta(\mu)$ which is defined at a fixed value of $\mu$. While we focus on quasi-Bayes estimators in this paper, we note that the estimator in this section can be relatively generic. For example, it could be a LTE conditional on a value of $\mu$, $\widehat{\theta}(\mu) =\mbox{argmin}_{d\in \mathcal{D}} \int_{\theta} \ell(\theta,d) p_T(\theta,\mu)/p_{T}(\mu)d\theta$, 
or a CUE $\widehat{\theta}(\mu) =\mbox{argmin}_{\theta \in \Theta} Q_T(\theta,\mu)$, among many others.

\begin{Assumption}[Asymptotic normality]\label{Gaussian}
For each fixed $\mu \in \Gamma$, $\widehat{\theta}(\mu)$ admits the linearization
\[
\Big\|
\big(\widehat{\theta}(\mu)-\theta(\mu)\big)
-
U_T(\mu)
\Big\|
=
o_p\!\left(T^{-1/2}\right).
\]
Moreover,
\[
\sqrt{T}\,
\sigma_{\eta,\mu}^{-1}
\left(
\frac{\partial \eta(\theta(\mu))}{\partial\theta}
\right)^{\top}
\big(\widehat{\theta}(\mu)-\theta(\mu)\big)
\xrightarrow{d}
\mathcal{N}(0,1),
\]
where
\[
\begin{aligned}
\sigma_{\eta,\mu}^2
=
\left(
\frac{\partial \eta(\theta(\mu))}{\partial\theta}
\right)^{\top}
J_W\big(\theta(\mu)\big)^{-1}
J_{\Omega,W}\big(\theta(\mu)\big)
J_W\big(\theta(\mu)\big)^{-1}
\left(
\frac{\partial \eta(\theta(\mu))}{\partial\theta}
\right).
\end{aligned}
\]
\end{Assumption}

The expansion assumed in Assumption \ref{Gaussian} is readily justified for GMM estimators; see, e.g., Corollary \ref{coro:Gaussian}. The proof of Theorem \ref{Theorem:2} implies that, for any $\mu \in \Gamma$, the CUE admits the same first-order linearization as the LTE with symmetric loss functions analyzed in \cite{chernozhukov2003mcmc}, when considering the quasi-posterior conditional on $\mu$; see also Theorem 2 of \cite{chernozhukov2003mcmc} for the fixed-$k$ case. When the weight matrix satisfies the generalized information equality (Equation \ref{inform}), Assumption \ref{Gaussian} further implies that the asymptotic variance of the leading term $(\partial \eta(\theta(\mu))/\partial\theta)^{\top} U_{T}(\mu)$ in the expansion of $\eta(\widehat{\theta}(\mu)) - \eta(\theta(\mu))$ is given by
$$
T^{-1} \left(\frac{\partial \eta(\theta(\mu))}{\partial\theta}\right)^{\top} {J}_{\Omega}\left(\theta(\mu)\right)^{-1} \left(\frac{\partial \eta(\theta(\mu))}{\partial\theta}\right), \; {J}_{\Omega}\left(\theta(\mu)\right) = G(\theta(\mu))^{\top}\Omega(\theta(\mu))^{-1} G(\theta(\mu)).
$$

We now state two theorems that make use of different features of quasi-posteriors obtained conditional on fixed values of $\mu$ to produce interval estimates for $\eta(\theta(\mu))$. The first main result in each theorem verifies that the resulting interval estimates have asymptotically correct frequentist coverage under the assumption that the fixed value of $\mu$ corresponds to the value of $\mu$ defining the conditional distribution from which data were realized. As a consequence, we can obtain frequentist confidence regions with correct coverage under the prior support condition that $\mu$ belongs to a known set $\mathcal{M}$ without requiring a completely specified prior by taking a union of confidence intervals produced at each $\mu \in \mathcal{M}$. This approach is analogous to the union of confidence intervals approach in \citet{conley2012plausibly} and the approach outlined in Remark 3.3 of \citet{armstrong2021sensitivity}.
\begin{theorem}[Frequentist properties of posterior quantiles]\label{th3}
Assume that $\eta(\theta)$ has uniformly bounded first derivatives, that is, 
for some constant $C>0$,
\[
\left\|
\frac{\partial \eta(\theta(\mu))}{\partial \theta}
\right\|
\le C.
\] 
Assume further that
\begin{equation}\label{inform}
\lim_{T\to\infty}
J_{\Omega}\big(\theta(\mu)\big)
J_{\Omega,W}\big(\theta(\mu)\big)^{-1}
=
I_k,
\end{equation}
where $I_k$ denotes the $k\times k$ identity matrix. For each $\mu\in\Gamma$, define the conditional quasi-posterior distribution function of $\eta(\theta)$ by
\begin{equation}\label{eq:conditional_cdf_eta}
F_{\eta,T}(x,\mu)
=
\int_{\{\theta\in\Theta:\,\eta(\theta)\le x\}}
\frac{p_T(\theta,\mu)}{p_T(\mu)}
\,d\theta,
\end{equation}
and define its $\alpha$-quantile as
\[
c_{\eta,T}(\alpha,\mu)
=
\inf\left\{
x\in\mathbb{R}:
F_{\eta,T}(x,\mu)\ge \alpha
\right\}.
\]
Let
\[
\sigma_\eta^2(\mu)
=
\left(
\frac{\partial \eta(\theta(\mu))}{\partial \theta}
\right)^{\top}
J_{\Omega}\big(\theta(\mu)\big)^{-1}
\left(
\frac{\partial \eta(\theta(\mu))}{\partial \theta}
\right).
\]
Under Assumptions \ref{a1} and \ref{assum33}--\ref{Gaussian}, for any $\alpha\in(0,1)$,
\begin{equation}\label{eq:posterior_quantile_expansion}
c_{\eta,T}(\alpha,\mu)
=
\eta\big(\widehat{\theta}(\mu)\big)
+
q_{\alpha}\frac{\sigma_\eta(\mu)}{\sqrt{T}}
+
o_p\!\left(T^{-1/2}\right),
\end{equation}
where $q_\alpha$ denotes the $\alpha$-quantile of the standard normal distribution. Define
\[
\mathrm{CI}_T(\mu)
=
\left[
c_{\eta,T}(\alpha/2,\mu),
c_{\eta,T}(1-\alpha/2,\mu)
\right].
\]
Then
\begin{equation}\label{eq: cond cov 1}
\lim_{T\to\infty}
\mathbb{P}_{\mu}^*
\left\{
\eta\big(\theta(\mu)\big)\in \mathrm{CI}_T(\mu)
\right\}
=
1-\alpha,
\end{equation}
and
\begin{equation}\label{eq: union cov 1}
\lim_{T\to\infty}
\mathbb{P}^*
\left\{
\eta\big(\theta(\mu)\big)
\in
\cup_{\mu'} \mathrm{CI}_T(\mu'),
\;
\forall \mu\in\Gamma 
\right\}
\ge
1-\alpha.
\end{equation}
\end{theorem}
\begin{proof} See Appendix \ref{proof: theorem3}. \end{proof} 

Theorem \ref{th3} shows that quantiles induced by the limiting conditional posterior density 
$\frac{N_T(\theta,\mu)}{N_T(\mu)}$ yield an asymptotically valid frequentist approximation to the distribution of 
$\sqrt{T}\big(\eta(\widehat{\theta}(\mu))-\eta(\theta(\mu))\big)$. 
The coverage result is driven by the generalized information equality in \eqref{inform}, 
and is in line with existing results for point-identified scalar parameters and partially identified models; 
see, for example, \cite{chernozhukov2003mcmc} and \cite{chen2018monte}.

The first main result of Theorem \ref{th3}, equation \eqref{eq: cond cov 1}, verifies that posterior quantiles obtained from the quasi-posterior constructed conditional on fixed value of $\mu$ have asymptotically correct coverage under the corresponding conditional measure. Equation \eqref{eq: union cov 1}, showing valid coverage of the union of intervals obtained under a support restriction, then immediately follows under the assumption that the value of $\mu$ under which the data were generated belongs to the specified support $\mathcal{M}.$ The results of Theorem \ref{th3} critically depend on choosing $\widehat{W}_T(\theta(\mu))$ such that \eqref{inform} holds.  
Motivated by Theorem 4 of \citet{chernozhukov2003mcmc}, the next result proposes an alternative procedure for using features of the quasi-posterior to construct intervals with frequentist coverage guarantees in settings where $\widehat{W}_T(\theta(\mu))$ is specified in such a way that \eqref{inform} does not hold.  
 
\begin{theorem}[Frequentist properties of intervals based on Gaussian approximations]\label{th4}
Suppose Assumptions \ref{a1} and \ref{assum33}--\ref{Gaussian} hold, and suppose that 
$\eta(\theta)$ has bounded derivatives. Let $c>0$ be small enough such that
\[
\int_{\mathcal M}
1_{\{p_T(\mu)\leq c\}}\,p_T(\mu)\,d\mu
=
o_p(1).
\]
We apply the theorem to points $\mu$ with $p_T(\mu)>c$. Let
\[
\widehat{J}_{T}^{-1}\big(\widehat{\theta}(\mu)\big)
=
\int_{\Theta}
T\big(\theta-\widehat{\theta}(\mu)\big)
\big(\theta-\widehat{\theta}(\mu)\big)^{\top}
\left[
\frac{p_T(\theta,\mu)}{p_T(\mu)}
\right]
\,d\theta .
\]
Assume that there exists an estimator 
$\widetilde{J}_{\Omega,W}\big(\theta(\mu)\big)$ such that
\[
\left\|
\widetilde{J}_{\Omega,W}\big(\theta(\mu)\big)
J_{\Omega,W}\big(\theta(\mu)\big)^{-1}
-
I_{k\times k}
\right\|
\to_p 0 .
\]

\[
\left\|
\widehat{J}_{T}\big(\widehat{\theta}(\mu)\big)
J_W\big(\theta(\mu)\big)^{-1}
-
I_{k\times k}
\right\|
\to_p 0 ,
\]
where $I_{k\times k}$ denotes the $k\times k$ identity matrix. For $\mu\in\Gamma$, let
\[
\begin{aligned}
\widetilde{c}_{\eta,T}(\alpha,\mu)
\defeq\;
&\eta\big(\widehat{\theta}(\mu)\big)
+
q_{\alpha}\frac{1}{\sqrt{T}}
\\
&\times
\Bigg[
\left(
\frac{\partial \eta\big(\widehat{\theta}(\mu)\big)}
{\partial \theta}
\right)^{\top}
\widehat{J}_{T}\big(\widehat{\theta}(\mu)\big)^{-1}
\widetilde{J}_{\Omega,W}\big(\theta(\mu)\big)
\\
&\qquad\qquad\qquad
\times
\widehat{J}_{T}\big(\widehat{\theta}(\mu)\big)^{-1}
\left(
\frac{\partial \eta\big(\widehat{\theta}(\mu)\big)}
{\partial \theta}
\right)
\Bigg]^{1/2},
\end{aligned}
\]
and
\[
\widetilde{\mathrm{CI}}(\mu)
=
\left[
\widetilde{c}_{\eta,T}(\alpha/2,\mu),
\widetilde{c}_{\eta,T}(1-\alpha/2,\mu)
\right].
\]
Then
\begin{equation}\label{eq: cond cov 2}
\lim_{T\to\infty}
\mathbb{P}_{\mu}^{*}
\left\{
\eta\big(\theta(\mu)\big)
\in
\widetilde{\mathrm{CI}}(\mu)
\right\}
=
1-\alpha,
\end{equation}
and, for a small enough $c$,
\begin{equation}\label{eq: union cov 2}
\lim_{T\to\infty}
\mathbb{P}^{*}
\left\{
\eta\big(\theta(\mu)\big)
\in
\cup_{\mu'}
\widetilde{\mathrm{CI}}(\mu'),
\;
\forall \mu\in\Gamma 
\right\}
\ge
1-\alpha .
\end{equation}
\end{theorem}
\begin{proof} See Appendix \ref{proof: theorem4}. \end{proof}

Theorem \ref{th4} verifies that intervals constructed making use of a normal approximation constructed conditional on fixed value of $\mu$ also have asymptotically correct coverage under the corresponding conditional measure. In practice, $\widehat{J}_{T}\left(\widehat{\theta}(\mu)\right)^{-1}$ in Theorem \ref{th4} can be computed by multiplying the variance-covariance matrix of the MCMC
sequence $S=\left(\theta^{(1)},\theta^{(2)},\ldots,\theta^{(B)}\right)$, 
where $B$ denotes the simulation sample size, by $T$ in settings where MCMC is used to approximate the quasi-posterior at a fixed value of $\mu$. As with Theorem \ref{th3}, it is then immediate that a union of intervals obtained using this approach under a support restriction on $\mu$ delivers valid frequentist inference.

Note that we state these results for completeness and to verify that the quasi-Bayes approach can be used to deliver frequentist valid inference under only support restrictions as in \cite{armstrong2021sensitivity}. However, our chief interest is in using quasi-Bayes approaches in settings where informative prior information is of use. If valid frequentist inference under a support restriction is the goal, it is not clear there is much advantage to adopting the framework presented in this paper. Nevertheless, our posterior density assigns probability mass to both $\theta$ and $\mu$, enabling a ranking across different parameter values. This offers richer information than a simple interval estimate.

\section{Conclusion}

In this paper, we introduce Plausible GMM (PGMM), a quasi-Bayesian framework for inference in moment condition models that allows for potential misspecification. By placing a proper prior over the degree of misspecification, PGMM provides a flexible and transparent way to incorporate researchers' subjective beliefs about the plausibility of structural assumptions. This approach extends classical GMM by acknowledging that moment conditions are often credible but not exact, enabling more credible inference in the presence of model uncertainty.

Our theoretical contributions include posterior concentration results and new Bernstein-von Mises type approximations under partial identification for quasi-Bayes procedures allowing diverging dimensions of parameters and moments. In addition, we also provide decision-theoretic guarantees under some special cases. While not our main goal, we also provide an approach and results for using quasi-posteriors to obtain asymptotically valid frequentist inferential statements under support restrictions for the degree of misspecification.

Empirical applications illustrate the use of PGMM. In these examples, we see that PGMM intervals remain informative while allowing for subjective, but empirically motivated deviations away from dogmatic identifying assumptions. PGMM may thus offer a useful tool for applied researchers who wish to retain the structure of moment-based models while explicitly allowing for uncertainty about moments being perfectly satisfied.

\section{Appendix: Proofs}
 \textbf{Notation} We denote the max norm by $|A|_{\text{max}}=\max_{i,j}|a_{i,j}|$,  the spectral norm by $\|A\|=\sqrt{\lambda_{\max}(A^TA)}$, and the Frobenius norm by $\|A\|_F$.  The trace of an $ n \times n $ square matrix $A$ is defined as $\mbox{tr}(A) = \sum_i a_{i,i}$ and its determinant is denoted by $\det(A)$.
 For positive semi-definite matrices $A, B$, we write $A\geq B$ if $A-B$ is positive semi-definite.
 ``w.p.a.1.'' to denote ``with probability approaching one''.
 For $s>0$ and a random vector $X$, we say $X\in\mathcal{L}^s$ if $\| X\|_s=[\E(|X|^s)]^{1/s}<\infty.$ 
\subsection{\small Proof of Theorem  \ref{Theorem:1}}\label{Proof: theorem1}
	\[
	p_{T}(\theta,\mu)\propto\pi(\theta)\exp(-\frac{T}{2}\|\widehat{m}(\theta)-\mu\|_{\widehat{\Omega}_T(\theta)^{-1}}^2)\exp(-(T \mu^{\top}\Lambda^{-1}\mu)/2).
	\]
We shall now prove that it is indeed true that,
	\[
	p_{T}(\theta)\propto\exp(-\frac{T}{2}\|\widehat{m}(\theta)-\mu_0\|_{\widehat{A}_{T,\theta}}^{2}).
	\]
 Without loss of generality, we prove for $\mu_0 = 0$.
Let $\mu\sim N(\mu_0,T^{-1}\Lambda).$
Let $$C_w^{2}=T(\Lambda^{-1}+\widehat{\Omega}_T(\theta)^{-1}),\;  C_wC_g=T\widehat{\Omega}_T(\theta)^{-1}\widehat{m
}(\theta).$$ 
Then we have the following,
\begin{align*}
	p_{T}(\theta,\mu) & \propto \pi(\theta)\exp(-\frac{1}{2}\|\widehat{m
}(\theta) -\mu  \|_{T\widehat{\Omega}_T(\theta)^{-1}}^{2}-\|{\mu}\|_{T\Lambda^{-1}}^{2}/2)\\
	& \propto \pi(\theta) \exp(-\frac{1}{2}\|\widehat{m}(\theta)\|_{T\widehat{\Omega}_T(\theta)^{-1}}^{2}+T{\mu}^{\top}\widehat{\Omega}_T(\theta)^{-1}\widehat{m
}(\theta)-\|{\mu}\|_{T\widehat{\Omega}_T(\theta)^{-1}}^{2}/2-\|{\mu}\|^2_{T\Lambda^{-1}}/2)\\
	& \propto \pi(\theta)  \exp(-\frac{1}{2}(\|\widehat{m}
(\theta)\|_{T\widehat{\Omega}_T(\theta)^{-1}}^{2})\exp(+T\mu^{\top}\widehat{\Omega}_T(\theta)^{-1}\widehat{m
}(\theta)-\|\mu\|_{(\Lambda^{-1}+\widehat{\Omega}_T(\theta)^{-1})T}^{2}/2)\\
	& \propto \pi(\theta) \exp(-\frac{1}{2}(\|\widehat{m
}(\theta)\|_{T\widehat{\Omega}_T(\theta)^{-1}}^{2})\exp(+\mu^{\top}C_w C_g-\|\mu\|_{C_w^{2}}^{2}/2-C_g^{\top}C_g/2)\exp(+C_g^{\top}C_g/2)\\
	& \propto \pi(\theta) \exp(-\frac{1}{2}(\|\widehat{m
}(\theta)\|_{T\widehat{\Omega}_T(\theta)^{-1}}^{2})\exp(-(C_w\mu-C_g)^{\top}(C_w\mu-C_g)/2)\exp(C_g^{\top}C_g/2)\\
	& \propto \pi(\theta)  \exp(-\frac{1}{2}(\|\widehat{m
}(\theta)\|_{T\widehat{\Omega}_T(\theta)^{-1}}^{2})\exp(-(\mu-C_w^{-1}C_g)^{\top}C_w^{2}(\mu-C_w^{-1}C_g)/2)\exp(C_g^{\top}C_g/2).
\end{align*}
Plugging in the definition of $C_g$, we have that
\begin{align*}
	\int p_{T}(\theta,\mu)d\mu & \propto \pi(\theta)  \exp(-\frac{1}{2}(\|\widehat{m}(\theta)\|_{T\widehat{\Omega}_T(\theta)^{-1}}^{2}))\exp(T^2\widehat{m
}(\theta)^{\top}\widehat{\Omega}_T(\theta)^{-1}C_w^{-2}\widehat{\Omega}_T(\theta)^{-1}\widehat{m
}(\theta)/2) \\
& \qquad *\sqrt{2\pi{\det{(C_w^{2})}}}\\
	& \propto \pi(\theta) 
 \exp(-\frac{1}{2}(\|\widehat{m
}(\theta)\|_{T(\widehat{\Omega}_T(\theta)^{-1}-\widehat{\Omega}_T(\theta)^{-1}TC_w^{-2}\widehat{\Omega}_T(\theta)^{-1})}^{2})).
\end{align*} 
Then, the results directly follow the proof of Theorem 1 in \cite{chernozhukov2003mcmc} for fixed $\Lambda$.  
In the scenario where $c \ll \Lambda \ll T$ ($c$ is a positive constant), we need to replace $M$ in the proof of Theorem 1 in \cite{chernozhukov2003mcmc} by a slowly increasing $M_T$, e.g., $M_T=M\|A_{\theta(\mu_0)}\|^{-\frac{1}{2}}$, and the remaining proofs proceed with a similar rationale. 
\subsection{\small Proof of Theorem \ref{Theorem:2}}\label{Proof: theorem2}

We slightly modify the definition of the TVM norm used in the proof. 
Throughout this proof, we work with
\[
\int_{\mathcal M}\int_{\Theta}
\left(
1+\|\sqrt{T}h(\theta,\mu)\|^\kappa
\right)
\left|
p_T(\theta,\mu)-N_T(\theta,\mu)
\right|
\,d\theta\,d\mu .
\]
Establishing the result under this norm is equivalent to proving the result under
\[
\int_{\mathcal M}\int_{\Theta}
\left(
1+\|\sqrt{T}(\theta-\theta(\nu(\mu)))\|^\kappa
\right)
\left|
p_T(\theta,\mu)-N_T(\theta,\mu)
\right|
\,d\theta\,d\mu ,
\]
except for the case treated separately in Lemma~\ref{lemmaextra}. Note that Assumption \ref{assum33} implies
\[
\operatorname{tr}
\left(
G(\theta)^\top W(\theta)G(\theta)
\right)
\lesssim k,
\]
\[
\operatorname{tr}
\left\{
T
G(\theta)^\top
W(\theta)
\big(\widehat m(\theta)-\mu\big)
\big(\widehat m(\theta)-\mu\big)^\top
W(\theta)
G(\theta)
\right\}
=
O_p(k),
\]
and
\[
0
<
\lambda_{\min}
\left(
W(\theta)G(\theta)G(\theta)^\top W(\theta)
\right)
\le
\lambda_{\max}
\left(
W(\theta)G(\theta)G(\theta)^\top W(\theta)
\right)
\lesssim 1.
\]

Define
\[
C_{m,\mu}
=
G^\top\big(\theta(\nu(\mu))\big)
W\big(\theta(\nu(\mu))\big)
\big(
\widehat m(\theta(\nu(\mu)))-\nu(\mu)
\big),
\]
\[
C_{w,\mu}
=
G^\top\big(\theta(\nu(\mu))\big)
W\big(\theta(\nu(\mu))\big)
G\big(\theta(\nu(\mu))\big),
\]
and
\[
\begin{aligned}
2C_\mu
=
&-2\log\pi(\mu)
-2\log\pi\big(\theta(\nu(\mu))\big)
\\
&+
T
\big(
\widehat m(\theta(\nu(\mu)))-\nu(\mu)
\big)^\top
W\big(\theta(\nu(\mu))\big)
\big(
\widehat m(\theta(\nu(\mu)))-\nu(\mu)
\big)
\\
&+
2\sqrt{T}\|\mu-\nu(\mu)\|.
\end{aligned}
\]

We next analyze $V_T(h(\cdot),\theta,\mu)$. Abbreviating $h(\theta,\mu)$ by $h$, we have
\begin{align*}
&V_T(h(\cdot),\theta,\mu)-2\log\pi(\mu)
\\
={}&
2T h^\top
W\big(\theta(\nu(\mu))\big)
\big(
\widehat m(\theta(\nu(\mu)))-\nu(\mu)
\big)
+
T h^\top
W\big(\theta(\nu(\mu))\big)h
\\
&\quad
-2\log\pi(\mu)
-2\log\pi\big(\theta(\nu(\mu))\big)
\\
&\quad
+
T
\big(
\widehat m(\theta(\nu(\mu)))-\nu(\mu)
\big)^\top
W\big(\theta(\nu(\mu))\big)
\big(
\widehat m(\theta(\nu(\mu)))-\nu(\mu)
\big)
\\
&\quad
+
2\sqrt{T}\|\mu-\nu(\mu)\|
\\
={}&
2T h^\top
W\big(\theta(\nu(\mu))\big)
\big(
\widehat m(\theta(\nu(\mu)))-\nu(\mu)
\big)
+
T h^\top
W\big(\theta(\nu(\mu))\big)h
+
2C_\mu .
\end{align*}

Conditioning on $\mu$, and completing the square in $\theta$,
\[
\exp\!\left\{
-\frac{1}{2}
\big(
V_T(h(\cdot),\theta,\mu)-2\log\pi(\mu)
\big)
\right\}
\]
is, as a function of $\theta$, proportional to the density of
\[
\mathcal N
\left(
\begin{aligned}
&\theta(\nu(\mu))
-
C_{w,\mu}^{-1}C_{m,\mu}
\\
&\qquad
+
C_{w,\mu}^{-1}
G\big(\theta(\nu(\mu))\big)^\top
W\big(\theta(\nu(\mu))\big)
\big(\mu-\nu(\mu)\big),
\end{aligned}
\;
(TC_{w,\mu})^{-1}
\right).
\]

By Assumption \ref{assum33},
\begin{align*}
\operatorname{tr}\big(C_{m,\mu}C_{m,\mu}^\top\big)
&=
\big(
\widehat m(\theta(\nu(\mu)))-\nu(\mu)
\big)^\top
W\big(\theta(\nu(\mu))\big)
G\big(\theta(\nu(\mu))\big)
G^\top\big(\theta(\nu(\mu))\big)
\\
&\qquad\qquad\times
W\big(\theta(\nu(\mu))\big)
\big(
\widehat m(\theta(\nu(\mu)))-\nu(\mu)
\big)
\\
&=
O_p\!\left(
\operatorname{tr}
\left[
G^\top\big(\theta(\nu(\mu))\big)
W\big(\theta(\nu(\mu))\big)
\right.\right.
\\
&\qquad\qquad\qquad\left.\left.
\times
\mathbb E
\left\{
\big(
\widehat m(\theta(\nu(\mu)))-\nu(\mu)
\big)
\big(
\widehat m(\theta(\nu(\mu)))-\nu(\mu)
\big)^\top
\right\}
\right.\right.
\\
&\qquad\qquad\qquad\left.\left.
\times
W\big(\theta(\nu(\mu))\big)
G\big(\theta(\nu(\mu))\big)
\right]
\right)
\\
&=
O_p\!\left(\frac{k}{T}\right).
\end{align*}
and $\lambda_{\max}(G(\theta(\nu(\mu)))G^{\top}(\theta(\nu(\mu))))=\lambda_{\max}(G^{\top}(\theta(\nu(\mu)))G(\theta(\nu(\mu))))\lesssim 1$  for all $\mu$.

To derive the conclusion, we shall divide the proof into the following steps:
	\begin{align*}
&\int_{\mathcal{M}}\int_{\Theta}\left(1+\|\sqrt{T}h(\theta,\mu)\|^{\kappa}\right)\left|p_{T}(\theta,\mu)-N_{T}(\theta,\mu)\right|\mathrm{d}\theta d\mu\\
=&\int_{B_{\varepsilon}}\left(1+\|\sqrt{T}h(\theta,\mu)\|^{\kappa}\right)\left|p_{T}(\theta,\mu)-N_{T}(\theta,\mu)\right|\mathrm{d}\theta d\mu
\\&+\int_{B_{\varepsilon}^{c}}\left(1+\|\sqrt{T}h(\theta,\mu)\|^{\kappa}\right)\left|p_{T}(\theta,\mu)-N_{T}(\theta,\mu)\right|\mathrm{d}\theta d\mu
		=\mathcal{R}_{1,T}+\mathcal{R}_{2,T}.
	\end{align*}

To look at $\mathcal{R}_{1,T}$, we look at
$\int_{B_{\varepsilon}}\left(1+\|\sqrt{T}h(\theta,\mu)\|^{\kappa}\right)N_{T}(\theta,\mu)\left|p_{T}(\theta,\mu)/N_{T}(\theta,\mu)-1\right|\mathrm{d}\theta d\mu.$  Let the integral ratio be
	\begin{eqnarray*}
&&c^*=\frac{\int_{ \Xi}\exp(-V_T(h(.),\theta,\mu)/2+\log(\pi(\mu)))d\theta d\mu}{\int_{ \Xi}\exp(\frac{1}{2}Q_T(\theta,\mu)+\log\pi(\mu,\theta))d\theta d\mu}.
	\end{eqnarray*}	
 By Lemma \ref{smallc}, under Assumptions \ref{assummu1}-\ref{rT}, we have $c^* \rightarrow_{p} 1$.
	 
 We define $N_{T}(\theta|\mu) \propto \frac{N_{T}(\theta,\mu)}{\int_\Theta 
 N_{T}(\theta,\mu) d\theta}$ and $N_T(\mu)\propto \int_\Theta 
 N_{T}(\theta,\mu) d\theta$  for $\mu \in \mathcal{M}$. It is not hard to see that, conditional on $\mu$, the density function
$N_{T}(\theta|\mu)$ is a multivariate Gaussian density in $\theta$ with mean 
$\theta(\nu(\mu)) - C_{w,\mu}^{-1}C_{m,\mu} 
+ C_{w,\mu}^{-1}G(\theta(\nu(\mu)))^\top W(\theta(\nu(\mu)))(\mu - \nu(\mu))$ 
and variance $T^{-1}C_{w,\mu}^{-1}$. 
And we define $\E_{N_T(\theta|\mu)}(.)$ as taking expectation under the measure corresponding to the density function $N_T(\theta|\mu)$.  
We denote $\E_{\gamma}(.)$ as taking expectation under a standard multivariate Gaussian distribution with an identity variance covariance matrix. And $\mathbb{P}_{\gamma}(.)$ is the probability corresponding to $\gamma$.
Recall that 
\[
W_{\mu}=W(\theta(\nu(\mu)))\in\mathbb{R}^{q\times q}
\]
is symmetric positive definite and 
\[
G_{\mu}=G(\theta(\nu(\mu)))\in\mathbb{R}^{q\times k}
\]
has full column rank ($k\le q$). 
Define
\[
\theta_\mu:=\theta(\nu(\mu)), 
\qquad 
x_\mu:=\theta-\theta_\mu,
\qquad 
d_\mu:=\mu-\nu(\mu).
\]

Denote
\[
r_\mu:=\widehat m(\theta_\mu)-\mu
\]

Since $G_\mu$ has full column rank and $W_\mu$ is positive definite,
\[
C_{w,\mu}:=G_\mu^\top W_\mu G_\mu
\]
is positive definite and therefore invertible. 
We define the $W_\mu$–orthogonal projector onto 
$\operatorname{col}(G_\mu)$ by
\begin{equation}
P_{G,\mu}
:=G_\mu (G_\mu^\top W_\mu G_\mu)^{-1} G_\mu^\top W_\mu.
\label{eq:PG}
\end{equation}

The mapping $h:\mathbb{R}^k\times\mathbb{R}^q\to\mathbb{R}^q$ is
\begin{equation}
h(\theta,\mu)
:=G_\mu(\theta-\theta_\mu)+\nu(\mu)-\mu
=G_\mu x_\mu - d_\mu.
\label{eq:h-def}
\end{equation}

Define the orthogonal complement of 
$\operatorname{col}(G_\mu)$,  $W_\mu$, by
\[
\operatorname{col}(G_\mu)^{\perp_{W_\mu}}
:=\{y\in\mathbb{R}^q: G_\mu^\top W_\mu y=0\}.
\]

Using the identity $I_q=P_{G,\mu}+(I_q-P_{G,\mu})$, we decompose
\begin{align}
h(\theta,\mu)
&= \big(G_\mu x_\mu - P_{G,\mu} d_\mu\big)
\;-\;
\big(I_q-P_{G,\mu}\big)d_\mu.
\label{eq:h-decomp}
\end{align}
The first term lies in $\operatorname{col}(G_\mu)$, while the second lies in 
$\operatorname{col}(G_\mu)^{\perp_{W_\mu}}$. 

Indeed, $P_{G,\mu}$ is idempotent and $W_\mu$–self–adjoint:
\[
P_{G,\mu}^2=P_{G,\mu},
\qquad
P_{G,\mu}^\top W_\mu = W_\mu P_{G,\mu},
\]
which implies
\[
G_\mu^\top W_\mu (I_q-P_{G,\mu})=0.
\]
Consequently,
\[
\big(G_\mu x_\mu - P_{G,\mu} d_\mu\big)^\top
W_\mu
\big((I_q-P_{G,\mu}) d_\mu\big)
=0.
\]

Finally, letting $K:=W_\mu^{1/2}G_\mu$ and 
$\Pi:=K(K^\top K)^{-1}K^\top$, we have
\begin{equation}
P_{G,\mu}=W_\mu^{-1/2}\Pi W_\mu^{1/2},
\label{eq:PG-alt}
\end{equation}
so $P_{G,\mu}$ is the similarity transform of the usual orthogonal projector 
$\Pi$ onto $\operatorname{col}(K)$.
We shall then decompose $h(\theta,\mu)$ accordingly. 
By Lemma~\ref{NT}, the marginal kernel in $\mu$ satisfies
\begin{equation}
N_T(\mu)\ \propto\ 
|C_{w,\mu}|^{-1/2}
\exp\!\left\{
-\frac{T}{2}
r_\mu^\top
(I-P_{G,\mu})^\top W_{\mu} (I-P_{G,\mu})
r_\mu
\right\}
\pi(\mu)\,\pi(\theta_\mu)
\exp\!\big(-\sqrt{T}\|\mu-\nu(\mu)\|\big).
\end{equation}

Define
\[
x_{d,\mu}:=C_{w,\mu}^{-1}G_{\mu}^\top W_\mu d_\mu,\qquad
h_\parallel:=G_{\mu}(x_\mu-x_{d,\mu})\in\operatorname{col}(G_{\mu}),\qquad
h_\perp:=-(I-P_{G,\mu})d_{\mu}\in\operatorname{col}(G_{\mu})^{\perp_{W_\mu}},
\]

so that $h(\theta,\mu)=h_\parallel+h_\perp$ and $h_\parallel^\top W_\mu h_\perp=0$.
Note that,
	$$\frac{p_{T}(\theta,\mu)}{N_{T}(\theta,\mu)}=\frac{\exp\left(\frac{1}{2}Q_T(\theta,\mu)+V_T(h(.),\theta,\mu)/2+\log(\pi(\theta))\right) }{c^*}=\frac{\exp\left( R_T(\theta,\mu)\right)}{c^*},$$  let $a(\theta,\mu)=\exp(R_{T}(\theta,\mu))/c^*-1$.

\begin{align*}
\mathcal{R}_{1,T}
=&\int_{B_{\varepsilon}}
\left(1+\|\sqrt{T}h(\theta,\mu)\|^{\kappa}\right)
N_{T}(\theta,\mu)
\left|a(\theta,\mu)\right|
\mathrm{d}\theta\,d\mu\\
=&\int_{B_{\varepsilon}}
\left(1+\|\sqrt{T}h(\theta,\mu)\|^{\kappa}\right)
N_{T}(\theta|\mu)
\left|a(\theta,\mu)\right|
\mathrm{d}\theta\,N_T(\mu)\,d\mu.
\end{align*}
By Assumption~\ref{rT} and Lemma~\ref{smallc}, on $B_\varepsilon$,
\[
|a(\theta,\mu)| 
\lesssim |R_T(\theta,\mu)| 
\lesssim 
\frac{1}{T}\left(\frac{\sqrt{k}(\log T)^2}{\sqrt{T}} \vee \frac{q}{\sqrt{kT}}\right)
\left(\|\sqrt{T}h(\theta,\mu)\|^2 + k(\log T)^2\right).
\]
Since $h = h_\parallel + h_\perp$ with $W_\mu$-orthogonality,
\[
1 + \|\sqrt{T}h\|^\kappa 
\lesssim 
1 + \|\sqrt{T}h_\parallel\|^\kappa + \|\sqrt{T}h_\perp\|^\kappa,
\]
and
\[
\|\sqrt{T}h\|^2 \lesssim \|\sqrt{T}h_\parallel\|^2 + \|\sqrt{T}h_\perp\|^2.
\]
On $B_\varepsilon$, $\|\sqrt{T}h_\perp\| \leq \varepsilon$, so
$\|\sqrt{T}h_\perp\|^\kappa \leq \varepsilon^\kappa$ and
$\|\sqrt{T}h_\perp\|^2 \leq \varepsilon^2 \lesssim k(\log T)^2$.
Therefore
\begin{align*}
\mathcal{R}_{1,T}
&\lesssim 
\frac{1}{T}\left(\frac{\sqrt{k}(\log T)^2}{\sqrt{T}} \vee \frac{q}{\sqrt{kT}}\right)
\int_{\mu:\,\sqrt{T}\|h_\perp\|\leq\varepsilon}
(1 + \varepsilon^\kappa)\,
\\
&\qquad\times
\E_{N_T(\theta|\mu)}\!\left[
\left(1 + \|\sqrt{T}h_\parallel\|^\kappa\right)
\left(\|\sqrt{T}h_\parallel\|^2 + k(\log T)^2\right)
\mathbf{1}_{B_\varepsilon}
\right]
N_T(\mu)\,d\mu.
\end{align*}
Under the Gaussian change of variables 
$\gamma = \sqrt{T}\,C_{w,\mu}^{1/2}(x_\mu - x_{d,\mu} + C_{w,\mu}^{-1}C_{m,\mu})$, 
the conditional density $N_T(\theta|\mu)$ becomes the standard 
Gaussian density in $\gamma$, and
\[
\sqrt{T}\,h_\parallel 
= \sqrt{T}\,G_\mu(x_\mu - x_{d,\mu}) 
= G_\mu C_{w,\mu}^{-1/2}(\gamma - \sqrt{T}\,C_{w,\mu}^{-1/2}C_{m,\mu}).
\]
The constraint $\|\sqrt{T}h_\parallel\| \leq \varepsilon$ translates to 
$\gamma$ restricted to a ball of radius $O(\varepsilon)$ 
(up to constants depending on $C_{w,\mu}$). 
By Lemma~\ref{gai}, all moments of $\gamma$ under the standard 
Gaussian are bounded, so
\[
\E_\gamma\!\left[
\left(1 + \|\sqrt{T}h_\parallel\|^\kappa\right)
\left(\|\sqrt{T}h_\parallel\|^2 + k(\log T)^2\right)
\right]
\lesssim k^{\kappa/2+1}(\log T)^2.
\]
Combining,
\[
\mathcal{R}_{1,T}
\lesssim 
\frac{1}{T}\left(\frac{\sqrt{k}(\log T)^2}{\sqrt{T}} \vee \frac{q}{\sqrt{kT}}\right)
\cdot k^{\kappa/2+1}(\log T)^2
\cdot (1 + \varepsilon^\kappa)
\to 0
\]
under Assumption~\ref{rates}.

  Denote $\mathbb{P}_{N_{T}(\theta|\mu)}(\cdot)$ as the conditional probability measure function 
of $\theta- \theta(\nu(\mu))$ conditioning on a fixed value of $\mu$ corresponding to the density $N_{T}(\theta|\mu)$.

On $B_{\varepsilon}^{c}$, by Assumption \ref{rT},  
\[
\left|\exp(R_{T}(\theta,\mu))/c^*-1\right|
\leq 
\exp\!\Big(
-C_0\|h(\theta,\mu)\|\varepsilon\sqrt{T}
+ C_0\varepsilon^2/2
+ T\|h(\theta,\mu)\|^2_{W_{\mu}}/2
+\sqrt{T}\|\mu -\nu(\mu)\|
\Big)/c^*
+1 .
\]

Define
\[
B_{\varepsilon}^c 
=
\left\{ (\theta, \mu) \in \Theta \times \mathcal{M} :
\sqrt{T} \| h(\theta, \mu) \| > \varepsilon
\ \text{or}\ 
\sqrt{T} \| \mu - \nu(\mu) \| > \varepsilon
\right\}.
\]

Accordingly define
\[
B_{\varepsilon}^{c1}
=
\left\{
\|\sqrt{T}h_{\parallel}\|>\frac{\varepsilon}{2}
\right\},
\]

\[
B_{\varepsilon}^{c2}
=
\left\{
\|\sqrt{T}h_{\perp}\|>\frac{\varepsilon}{2}
\right\},
\]

\[
B_{\varepsilon}^{c3}
=
\left\{
\sqrt{T}\|P_{G,\mu}(\mu-\nu(\mu))\|>\frac{\varepsilon}{2} 
\right\}.
\]

Hence
\[
B_{\varepsilon}^{c}
=
B_{\varepsilon}^{c1}
\cup
B_{\varepsilon}^{c2}
\cup
B_{\varepsilon}^{c3}.
\]

The contribution on
\[
B_{\varepsilon}^{c3}
\cap
(B_{\varepsilon}^{c1})^c
\cap
(B_{\varepsilon}^{c2})^c
\]
is handled separately in Lemma~\ref{lemmaextra}, which studies the region
\[
\overline B
=
\{\|\sqrt{T}h\|\le \varepsilon,\;
\sqrt{T}\|\mu-\nu(\mu)\|>\tfrac{\varepsilon}{2} \}.
\]
It therefore suffices to control the integral on $B_{\varepsilon}^{c1}$ and
$B_{\varepsilon}^{c2}$. We first present the argument for $B_{\varepsilon}^{c1}$;
the argument for $B_{\varepsilon}^{c2}$ is given afterwards.

Let
\[
\mathcal{B}_T(h(\theta,\mu))
=
-C_0\|h(\theta,\mu)\|\varepsilon\sqrt{T}
+ C_0\frac{\varepsilon^2}{2}
+ \frac{T}{2}\|h(\theta,\mu)\|^2_{W_{\mu}} .
\]

Moreover, we have the decomposition
\begin{equation}\label{eq:BT-sum}
R_T(\theta,\mu)
\le
\mathcal{B}_T\!\big(h_\parallel(\theta,\mu)\big)
+
\mathcal{B}_T\!\big(h_\perp(\theta,\mu)\big)
+
\sqrt{T}\|\mu-\nu(\mu)\|.
\end{equation}

To analyze the inner integral on $B_{\varepsilon}^{c1}$, introduce the Gaussian
change of variables
\[
\gamma
=
\sqrt{T}\,C_{w,\mu}^{1/2}
\Bigl(x_\mu - x_{d,\mu} + C_{w,\mu}^{-1}C_{m,\mu}\Bigr),
\]
so that
\[
h_\parallel
=
G_\mu C_{w,\mu}^{-1/2}
\left(
T^{-1/2}\gamma
-
C_{w,\mu}^{-1/2}C_{m,\mu}
\right).
\]
Define the corresponding region
\[
B_\gamma
=
\left\{
\gamma\in\mathbb{R}^k:
\|\gamma - \sqrt{T}\,C_{w,\mu}^{-1/2}C_{m,\mu}\|
>
c'\varepsilon
\right\},
\]
where $c'>0$ is determined by the lower bound on the singular values of $G_\mu$ 
from Assumption~\ref{assum33}.

Then
\[
\left\{
\gamma:
\left\|
G_\mu C_{w,\mu}^{-1/2}
\big(
T^{-1/2}\gamma
-
C_{w,\mu}^{-1/2}C_{m,\mu}
\big)
\right\|
>
\frac{c\varepsilon}{\sqrt{T}}
\right\}
\subseteq
B_\gamma .
\]
Indeed, by Assumption~\ref{assum33} the singular values of $G_\mu$ are bounded 
from below, so $\|G_\mu v\| \geq c_0\|v\|$ for some $c_0>0$, which gives 
the inclusion with $c' = c \cdot c_0$.
This converts the inner integral into an expectation over a standard Gaussian
vector $\gamma$, and the resulting tail region $B_\gamma$ can be controlled
using Gaussian concentration bounds.

Next we shall control $\mathcal{R}_{2,T}(B_{\varepsilon}^{c1})$.
We write
\begin{align*}
\mathcal{R}_{2,T}(B_{\varepsilon}^{c1})
&=
\int_{B_{\varepsilon}^{c1}}
\left(1+\|\sqrt{T}h(\theta,\mu)\|^{\kappa}\right)
N_T(\theta,\mu)
\left|\exp(R_T(\theta,\mu))/c^*-1\right|
\,d\theta d\mu .
\end{align*}

Using Assumption~\ref{rT}, we obtain
\[
\left|\exp(R_T(\theta,\mu))/c^*-1\right|
\lesssim
1+
\exp\!\Big(
-C_0\|h(\theta,\mu)\|\varepsilon\sqrt T
+ C_0\varepsilon^2/2
+ T\|h(\theta,\mu)\|_{W_{\mu}}^2/2
+\sqrt T\|\mu-\nu(\mu)\|
\Big).
\]
Hence
\[
\mathcal{R}_{2,T}(B_{\varepsilon}^{c1})
\lesssim
\mathcal{R}_{21,T}
+
\mathcal{R}_{22,T},
\]
where
\begin{align*}
\mathcal{R}_{21,T}
&=
\int_{B_{\varepsilon}^{c1}}
\left(1+\|\sqrt{T}h(\theta,\mu)\|^{\kappa}\right)
N_T(\theta,\mu)
\,d\theta d\mu,
\\
\mathcal{R}_{22,T}
&=
\int_{B_{\varepsilon}^{c1}}
\exp(R_T(\theta,\mu))
\left(1+\|\sqrt{T}h(\theta,\mu)\|^{\kappa}\right)
N_T(\theta,\mu)
\,d\theta d\mu .
\end{align*}
\noindent
\textbf{Step 1: control of $\mathcal{R}_{21,T}$.}

Using
\[
\|h(\theta,\mu)\|
\le
\|h_{\parallel}(\theta,\mu)\|
+
\|h_{\perp}(\theta,\mu)\|,
\]
we have
\[
1+\|\sqrt T h(\theta,\mu)\|^\kappa
\lesssim
1+\|\sqrt T h_\parallel\|^\kappa
+\|\sqrt T h_\perp\|^\kappa .
\]

Therefore
\begin{align*}
\mathcal{R}_{21,T}
&\lesssim
\int_{\mathcal M}
\E_{N_T(\theta|\mu)}
\Big[
(1+\|\sqrt T h_\parallel\|^\kappa)
\mathbf 1\{\|\sqrt T h_\parallel\|>\varepsilon/2\}
\Big]
\\
&\qquad\qquad\times
(1+\|\sqrt T h_\perp\|^\kappa)
\,N_T(\mu)\,d\mu .
\end{align*}
According to the definition of $B_{\gamma}$,
\begin{align*}
\mathcal{R}_{21,T}
&\lesssim
\int_{\mathcal M}
\E_\gamma
\Big[
(1+\|G_\mu C_{w,\mu}^{-1/2}(\gamma-\sqrt T C_{w,\mu}^{-1/2}C_{m,\mu})\|^\kappa)
\mathbf 1\{\gamma\in B_\gamma\}
\Big]
\\
&\qquad\qquad\times
(1+\|\sqrt T h_\perp\|^\kappa)
\,N_T(\mu)\,d\mu .
\end{align*}

Since $\|G_\mu v\| \lesssim \|v\|$ by Assumption~\ref{assum33} and 
$B_\gamma$ corresponds to a Gaussian tail region with threshold
$\varepsilon \asymp q\log T$, Lemma \ref{gai} implies that
\[
\sup_{\mu\in\mathcal M}
\E_\gamma
\Big[
(1+\|G_\mu C_{w,\mu}^{-1/2}(\gamma-\sqrt T C_{w,\mu}^{-1/2}C_{m,\mu})\|^\kappa)
\mathbf 1\{\gamma\in B_\gamma\}
\Big]
\to0 .
\]

Hence $\mathcal{R}_{21,T}\to0$.

\noindent
\textbf{Step 2: control of $\mathcal{R}_{22,T}$.}

Write
\begin{align*}
\mathcal{R}_{22,T}
&\lesssim
\int_{\mathcal M}
\E_{N_T(\theta|\mu)}
\Big[
\exp(R_T(\theta,\mu))
(1+\|\sqrt T h(\theta,\mu)\|^\kappa)
\mathbf 1\{\theta\in B_{\varepsilon}^{c1}\}
\Big]
N_T(\mu)d\mu .
\end{align*}

Using the bound
\[
R_T(\theta,\mu)
\le
\mathcal B_T(h_\parallel)
+
\mathcal B_T(h_\perp)
+
\sqrt T\|\mu-\nu(\mu)\|,
\]
we obtain
\begin{align*}
\mathcal{R}_{22,T}
&\lesssim
\int_{\mathcal M}
\E_{N_T(\theta|\mu)}
\Big[
\exp(\mathcal B_T(h_\parallel))
(1+\|\sqrt T h_\parallel\|^\kappa)
\mathbf 1\{\theta\in B_{\varepsilon}^{c1}\}
\Big]
\\
&\qquad\qquad\times
\exp(\mathcal B_T(h_\perp)+\sqrt T\|\mu-\nu(\mu)\|)
(1+\|\sqrt T h_\perp\|^\kappa)
N_T(\mu)d\mu .
\end{align*}

After the Gaussian change of variables and using the same tail region
$B_\gamma$, Gaussian tail bounds imply
\[
\E_\gamma
\Big[
\exp(\mathcal B_T(\cdot))
(1+\|\cdot\|^\kappa)
\mathbf 1\{\gamma\in B_\gamma\}
\Big]
\leq C \exp(-(\varepsilon\vee \sqrt{k})) .
\]
where the bound is by Lemma \ref{Gaussintegral} as Gaussian integral lemma and implied by our Assumptions \ref{assum33} and  \ref{rT} .
Moreover the factor
\[
\exp(\mathcal B_T(h_\perp)+\sqrt T\|\mu-\nu(\mu)\|)N_T(\mu)
\]
remains integrable due to the cancellation of the
$\sqrt T\|\mu-\nu(\mu)\|$ terms in $N_T(\mu)$.

Therefore
$\mathcal{R}_{22,T}\to0.$
Combining the two steps yields
$\mathcal{R}_{2,T}(B_{\varepsilon}^{c1})\to 0.$

Furthermore we shall discuss the region on \(B_{\varepsilon}^{c2}\)

Consider the region
\[
B_{\varepsilon}^{c2}
=
\left\{
\|\sqrt{T}h_{\perp}\|>\varepsilon/2
\right\}
\]
We show that the contribution from this region is asymptotically negligible.

Write
\[
\mathcal R_{2,T}(B_{\varepsilon}^{c2})
:=
\int_{B_{\varepsilon}^{c2}}
\left(1+\|\sqrt{T}h(\theta,\mu)\|^\kappa\right)
N_T(\theta,\mu)
\left|\exp(R_T(\theta,\mu))/c^*-1\right|
\,d\theta d\mu .
\]
By Assumption~\ref{rT},
\[
\left|\exp(R_T(\theta,\mu))/c^*-1\right|
\lesssim
1+
\exp\!\Big(
-C_0\|h(\theta,\mu)\|\varepsilon\sqrt T
+ C_0\varepsilon^2/2
+ T\|h(\theta,\mu)\|_{W_{\mu}}^2/2
+\sqrt T\|\mu-\nu(\mu)\|
\Big).
\]
Hence
\[
\mathcal R_{2,T}(B_{\varepsilon}^{c2})
\lesssim
\mathcal R_{21,T}^{(2)}
+
\mathcal R_{22,T}^{(2)},
\]
where
\begin{align*}
\mathcal R_{21,T}^{(2)}
&:=
\int_{B_{\varepsilon}^{c2}}
\left(1+\|\sqrt{T}h(\theta,\mu)\|^\kappa\right)
N_T(\theta,\mu)\,d\theta d\mu,\\
\mathcal R_{22,T}^{(2)}
&:=
\int_{B_{\varepsilon}^{c2}}
\exp(R_T(\theta,\mu))
\left(1+\|\sqrt{T}h(\theta,\mu)\|^\kappa\right)
N_T(\theta,\mu)\,d\theta d\mu.
\end{align*}

\noindent
\textbf{Step 1: control of \(\mathcal R_{21,T}^{(2)}\).}

Since \(h=h_\parallel+h_\perp\),
\[
1+\|\sqrt T h\|^\kappa
\lesssim
1+\|\sqrt T h_\parallel\|^\kappa+\|\sqrt T h_\perp\|^\kappa.
\]
Therefore
\begin{align*}
\mathcal R_{21,T}^{(2)}
&\lesssim
\int_{\mathcal M}
\E_{N_T(\theta\mid\mu)}
\Big[
1+\|\sqrt T h_\parallel\|^\kappa
\Big]
\mathbf 1\{\|\sqrt T h_\perp\|>\varepsilon/2\}
(1+\|\sqrt T h_\perp\|^\kappa)
N_T(\mu)\,d\mu.
\end{align*}

Now the inner conditional moment is uniformly bounded under Assumption~\ref{assum33}, so
\[
\sup_{\mu\in\mathcal M}
\E_{N_T(\theta\mid\mu)}
\Big[
1+\|\sqrt T h_\parallel\|^\kappa
\Big]
<\infty.
\]
Hence
\[
\mathcal R_{21,T}^{(2)}
\lesssim
\int_{\mathcal M}
\mathbf 1\{\|\sqrt T h_\perp\|>\varepsilon/2\}
(1+\|\sqrt T h_\perp\|^\kappa)
N_T(\mu)\,d\mu.
\]

Using Lemma \ref{NT}, 
\[
N_T(\mu)
\propto
|C_{w,\mu}|^{-1/2}
\exp\!\left(
-\frac{T}{2}
r_\mu^\top (I-P_{G,\mu})^\top W_\mu(I-P_{G,\mu})r_\mu
\right)
\pi(\mu)\pi(\theta_\mu)
\exp\!\big(-\sqrt T\|\mu-\nu(\mu)\|\big).
\]
Since $(I-P_{G,\mu})r_\mu = h_\perp + (I-P_{G,\mu})(\widehat{m}(\theta_\mu)-\nu(\mu))$ 
and $\|(I-P_{G,\mu})(\widehat{m}(\theta_\mu)-\nu(\mu))\| = O_p(\sqrt{q/T})$, 
on $B_\varepsilon^{c2}$ where $\sqrt{T}\|h_\perp\| > \varepsilon/2 \asymp q\log T \gg \sqrt{q}$, 
we have 
\[
\|(I-P_{G,\mu})r_\mu\|_{W_\mu}^2 \geq c\|h_\perp\|_{W_\mu}^2 \geq c'\|h_\perp\|^2
\]
w.p.a.1.\ for some constants $c,c'>0$, where the last inequality uses the 
eigenvalue bounds on $W_\mu$ from Assumption~\ref{assum33}. 
Therefore, on $B_\varepsilon^{c2}$,
\[
\mathcal R_{21,T}^{(2)}
\lesssim
\int_{\mathcal M}
\mathbf 1\{\|\sqrt T h_\perp\|>\varepsilon/2\}
(1+\|\sqrt T h_\perp\|^\kappa)
|C_{w,\mu}|^{-1/2}
e^{-c'T\|h_\perp\|^2/2}
\pi(\mu)\pi(\theta_\mu)\,d\mu.
\]
Since $\sqrt{T}\|h_\perp\| > \varepsilon/2$ on $B_\varepsilon^{c2}$,
\[
e^{-c'T\|h_\perp\|^2/2}
\le
e^{-c'\varepsilon^2/8}.
\]
Thus
\[
\mathcal R_{21,T}^{(2)}
\lesssim
e^{-c'\varepsilon^2/8}
\int_{\mathcal M}
(1+\|\sqrt T h_\perp\|^\kappa)
|C_{w,\mu}|^{-1/2}
\pi(\mu)\pi(\theta_\mu)\,d\mu
=o(1),
\]
provided the last integral is $O(1)$.
 
~\\
\noindent
\textbf{Step 2: control of \(\mathcal R_{22,T}^{(2)}\).}

Using
\[
R_T(\theta,\mu)
\le
\mathcal B_T(h_\parallel)
+
\mathcal B_T(h_\perp)
+
\sqrt T\|\mu-\nu(\mu)\|,
\]
we obtain
\begin{align*}
\mathcal R_{22,T}^{(2)}
&\lesssim
\int_{\mathcal M}
\E_{N_T(\theta\mid\mu)}
\Big[
\exp(\mathcal B_T(h_\parallel))
(1+\|\sqrt T h_\parallel\|^\kappa)
\Big]
\\
&\qquad\qquad\times
\exp(\mathcal B_T(h_\perp)+\sqrt T\|\mu-\nu(\mu)\|)
\mathbf 1\{\|\sqrt T h_\perp\|>\varepsilon/2 \}
(1+\|\sqrt T h_\perp\|^\kappa)
N_T(\mu)\,d\mu .
\end{align*}

Again, the conditional expectation in \(h_\parallel\) is uniformly bounded. Hence
\[
\mathcal R_{22,T}^{(2)}
\lesssim
\int_{\mathcal M}
\exp(\mathcal B_T(h_\perp)+\sqrt T\|\mu-\nu(\mu)\|)
\mathbf 1\{\|\sqrt T h_\perp\|>\varepsilon/2\}
(1+\|\sqrt T h_\perp\|^\kappa)
N_T(\mu)\,d\mu .
\]
Now
\[
\mathcal B_T(h_\perp)
=
-C_0\varepsilon\sqrt T\,\|h_\perp\|_{W_\mu}
+\frac{C_0\varepsilon^2}{2}
+\frac{T}{2}\|h_\perp\|_{W_\mu}^2.
\]
Combining this with the expression for $N_T(\mu)$, the factor
$\exp(-\sqrt T\|\mu-\nu(\mu)\|)$ in $N_T(\mu)$ cancels
$\exp(\sqrt T\|\mu-\nu(\mu)\|)$. For the quadratic terms,
$\mathcal B_T(h_\perp)$ contains $\frac{T}{2}\|h_\perp\|_{W_\mu}^2$
while $N_T(\mu)$ contains $-\frac{T}{2}\|(I-P_{G,\mu})r_\mu\|_{W_\mu}^2$.
Since $(I-P_{G,\mu})r_\mu = h_\perp + (I-P_{G,\mu})(\widehat{m}(\theta_\mu)-\nu(\mu))$,
their sum equals
\[
-T h_\perp^\top W_\mu (I-P_{G,\mu})(\widehat{m}(\theta_\mu)-\nu(\mu))
- \frac{T}{2}\|(I-P_{G,\mu})(\widehat{m}(\theta_\mu)-\nu(\mu))\|_{W_\mu}^2,
\]
which is $O_p(\sqrt{Tq}\cdot\sqrt{T}\|h_\perp\|) + O_p(q)$.
This is absorbed by the linear term $-C_0\varepsilon\sqrt{T}\|h_\perp\|_{W_\mu}$
since $\varepsilon \asymp q\log T \gg \sqrt{q}$. Therefore
\begin{align*}
\mathcal R_{22,T}^{(2)}
&\lesssim
\int_{\mathcal M}
\exp\!\left(
-c\varepsilon\sqrt T\,\|h_\perp\|
+C\varepsilon^2
\right)
\mathbf 1\{\|\sqrt T h_\perp\|>\varepsilon/2\}
\\
&\qquad\qquad\times
(1+\|\sqrt T h_\perp\|^\kappa)
|C_{w,\mu}|^{-1/2}
\pi(\mu)\pi(\theta_\mu)\,d\mu .
\end{align*}

On $B_{\varepsilon}^{c2}$,
$\sqrt T\,\|h_\perp\|>\varepsilon/2$,
so
\[
-c\varepsilon\sqrt T\,\|h_\perp\|
+C\varepsilon^2
<0.
\]
Hence
\[
\mathcal R_{22,T}^{(2)}
\lesssim
\int_{\mathcal M}
(1+\|\sqrt T h_\perp\|^\kappa)
|C_{w,\mu}|^{-1/2}
\pi(\mu)\pi(\theta_\mu)\,d\mu
=o(1),
\]
provided the remaining integral is $O(1)$.

Combining the bounds for $\mathcal R_{21,T}^{(2)}$ and $\mathcal R_{22,T}^{(2)}$, 
we conclude that the contribution from $B_{\varepsilon}^{c2}$ is asymptotically negligible.

\section{Declaration of generative AI and AI-assisted technologies in the writing process}

During the preparation of this work, the authors used OpenAI’s ChatGPT, including GPT-4o, OpenAI o3, and GPT-5-series models to suggest edits to some lengthy passages for concision and clarity. After using these tools, the authors reviewed and edited the content as needed and take full responsibility for the content of the published article.

\clearpage

\setcounter{section}{0}\setcounter{subsection}{0}\setcounter{subsubsection}{0}
\setcounter{equation}{0}\setcounter{figure}{0}\setcounter{table}{0}
\setcounter{theorem}{0}\setcounter{lemma}{0}\setcounter{corollary}{0}
\makeatletter
\@ifundefined{c@Theorem}{}{\setcounter{Theorem}{0}}
\@ifundefined{c@Lemma}{}{\setcounter{Lemma}{0}}
\@ifundefined{c@Corollary}{}{\setcounter{Corollary}{0}}
\@ifundefined{c@Assumption}{}{\setcounter{Assumption}{0}}
\@ifundefined{c@algocf}{}{\setcounter{algocf}{0}}
\makeatother

\renewcommand{\thesection}{SA.\arabic{section}}
\renewcommand{\theequation}{SA.\arabic{equation}}
\renewcommand{\thefigure}{SA.\arabic{figure}}
\renewcommand{\thetable}{SA.\arabic{table}}
\renewcommand{\thecorollary}{SA.\arabic{corollary}}
\renewcommand{\thelemma}{SA.\arabic{lemma}}
\renewcommand{\thetheorem}{SA.\arabic{theorem}}
\renewcommand{\theAssumption}{SA.\arabic{Assumption}}
\renewcommand{\thealgocf}{SA.\arabic{algocf}}
\providecommand{\theHsection}{\thesection}\renewcommand{\theHsection}{SA.\arabic{section}}
\providecommand{\theHequation}{\theequation}\renewcommand{\theHequation}{SA.\arabic{equation}}
\providecommand{\theHfigure}{\thefigure}\renewcommand{\theHfigure}{SA.\arabic{figure}}
\providecommand{\theHtable}{\thetable}\renewcommand{\theHtable}{SA.\arabic{table}}
\providecommand{\theHtheorem}{\thetheorem}\renewcommand{\theHtheorem}{SA.\arabic{theorem}}
\providecommand{\theHlemma}{\thelemma}\renewcommand{\theHlemma}{SA.\arabic{lemma}}
\providecommand{\theHcorollary}{\thecorollary}\renewcommand{\theHcorollary}{SA.\arabic{corollary}}
\providecommand{\theHAssumption}{\theAssumption}\renewcommand{\theHAssumption}{SA.\arabic{Assumption}}

\begin{center}
{\Large\bfseries Supplemental Appendix for\\ ``Plausible GMM: A Quasi-Bayesian Approach''\par}
\vspace{1em}
{\normalsize Victor Chernozhukov$^a$, Christian B. Hansen$^b$, Lingwei Kong$^c$, Weining Wang$^d$\par}
\vspace{0.5em}
{\footnotesize $^a$ Department of Economics, Massachusetts Institute of Technology.\\
$^b$ Booth School of Business, The University of Chicago.\\
$^c$ Department of Economics, Econometrics and Finance, University of Groningen.\\
$^d$ Department of Economics, University of Bristol.}
\end{center}
\thispagestyle{empty}
\clearpage
\pagestyle{plain}
\section{Appendix: additional figures}
\begin{figure}[htbp!]
    \centering 
\includegraphics[width=0.8\textwidth,height=0.35\textheight]{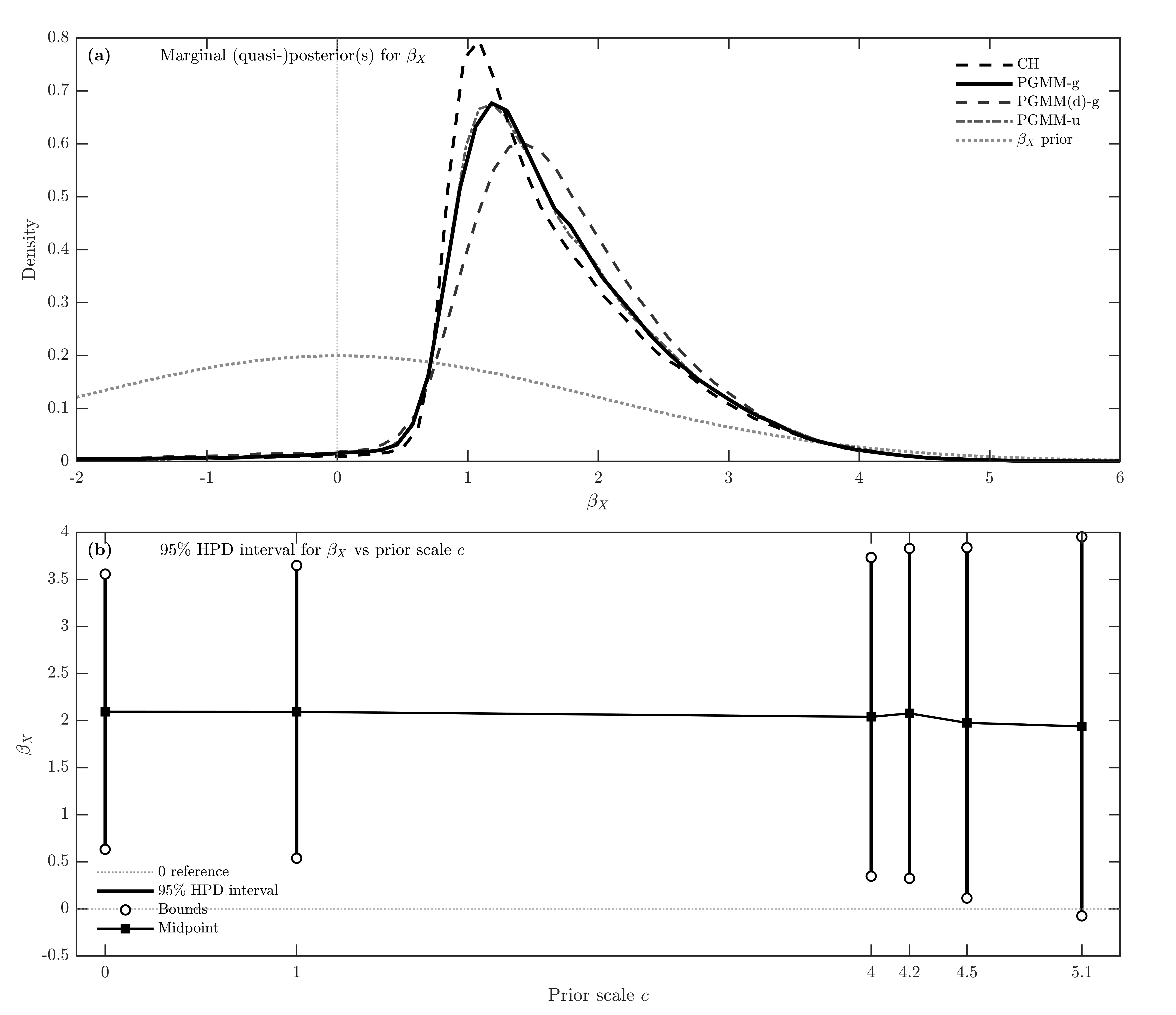}
\caption{ Upper panel (a): Linear IV (2) marginal (quasi-)posterior(s) for $\beta_X$ and its marginal prior (light dotted curve); lower panel (b): Linear IV(2) 95\% HPD intervals for $\beta_X$ resulting from the PGMM approach along various priors for $\mu$, i.e., $ \mathcal{N}(0,c\Sigma_T \Omega_d \Sigma_T^\top)$ for various values of $c$. 
 } \label{fig:ajr-empirical:2}
\end{figure}  
\begin{figure}[htbp!]
    \centering 
\includegraphics[width=0.8\textwidth,height=0.35\textheight]{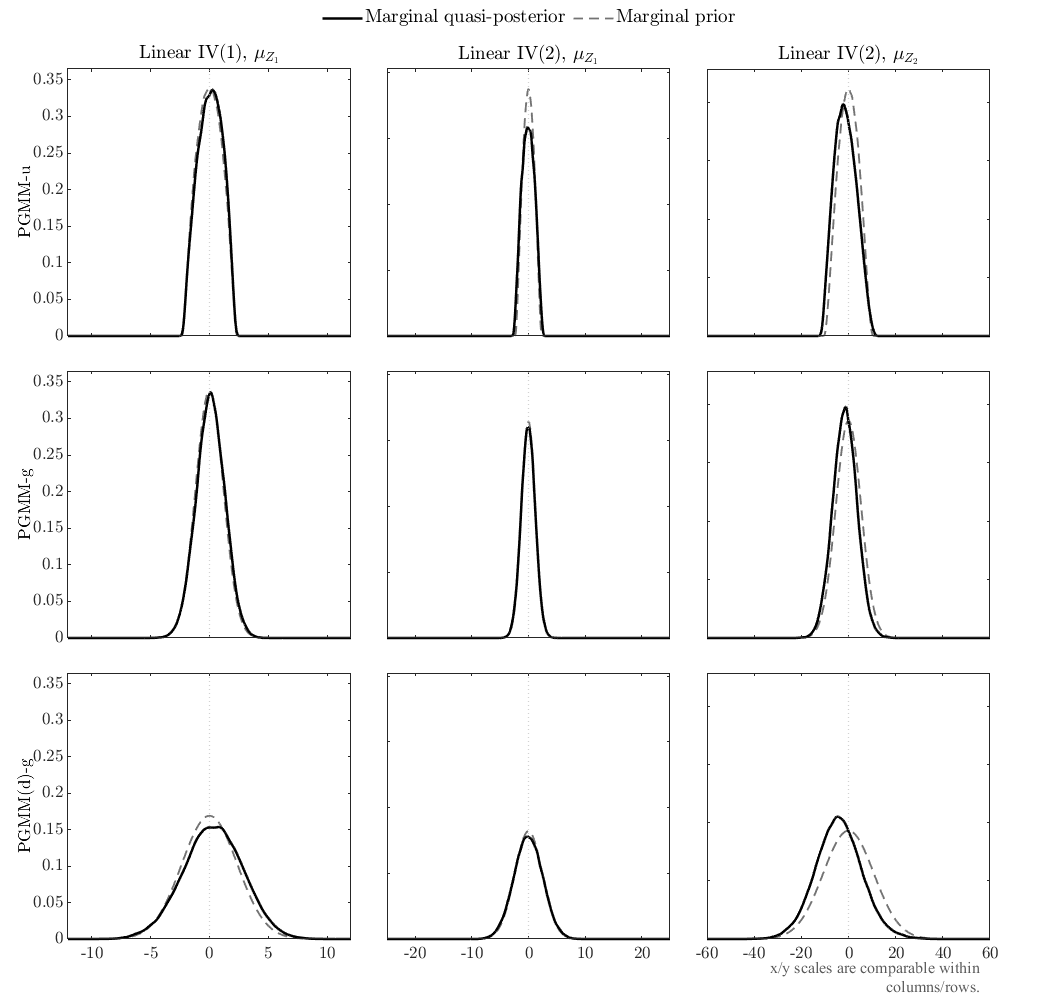}
\caption{ Marginal (quasi-)posteriors  (solid curves) and priors (dashed curves) for selected $\mu$'s (entries corresponding to moments constructed with IVs) in Linear IV(1) and Linear IV(2). 
 } \label{fig:ajr-empirical:3}
\end{figure}  

\begin{figure}[htbp!] 
    \centering
\includegraphics[width=\textwidth,height=0.35\textheight]{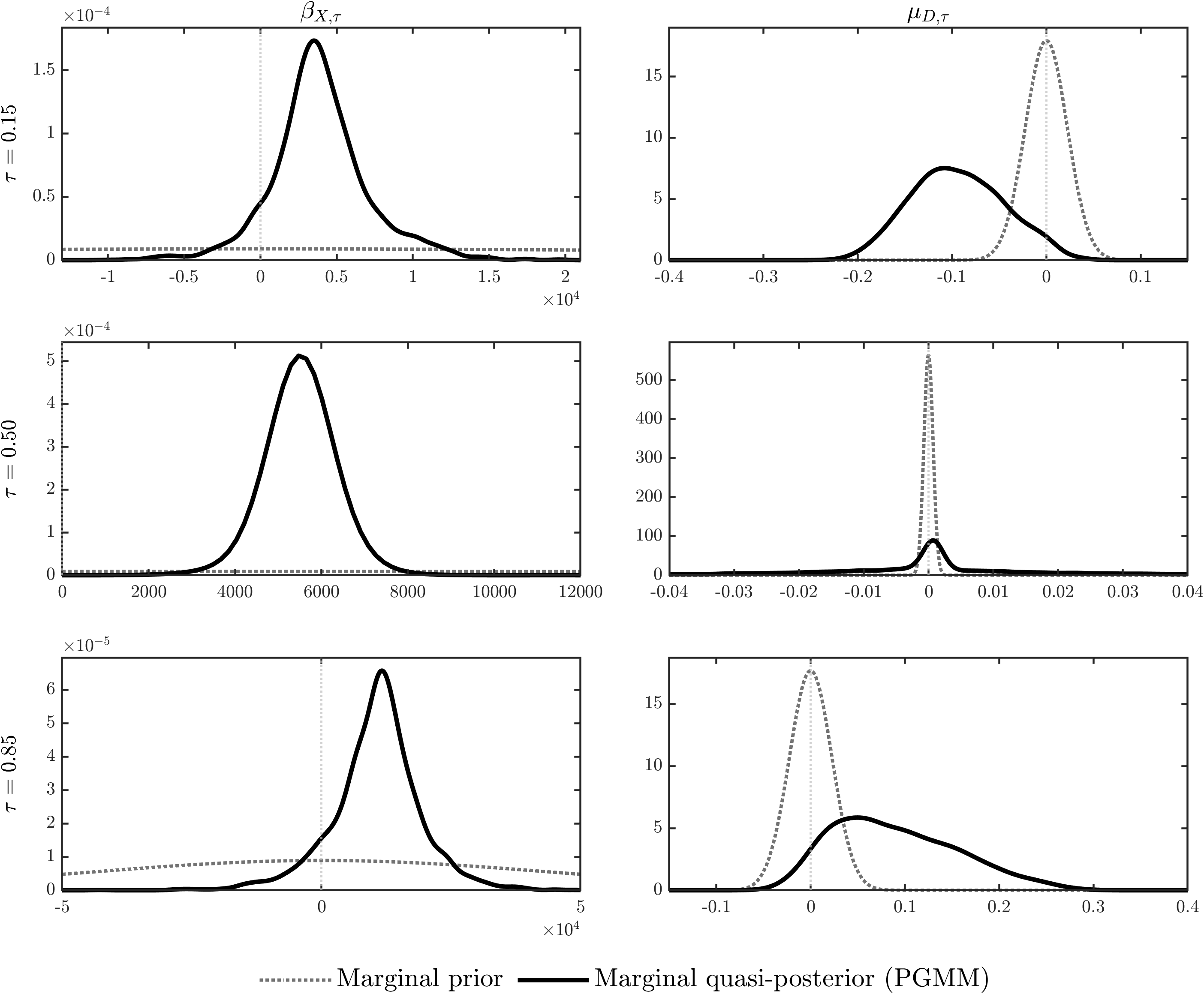} 
  \caption{PGMM marginal quasi-posteriors (solid black curves) of $\beta_{X,\tau}$ for the Gaussian prior over $\mu$ across various values of $\tau$ and $c=0.5$. The light dotted curves represent the marginal prior density curves for these parameters.  
} \label{401k_posterior_figure_c05onlyPGMM}
\end{figure} 

\begin{figure}[htbp!] 
    \centering
\includegraphics[width=\textwidth,height=0.35\textheight]{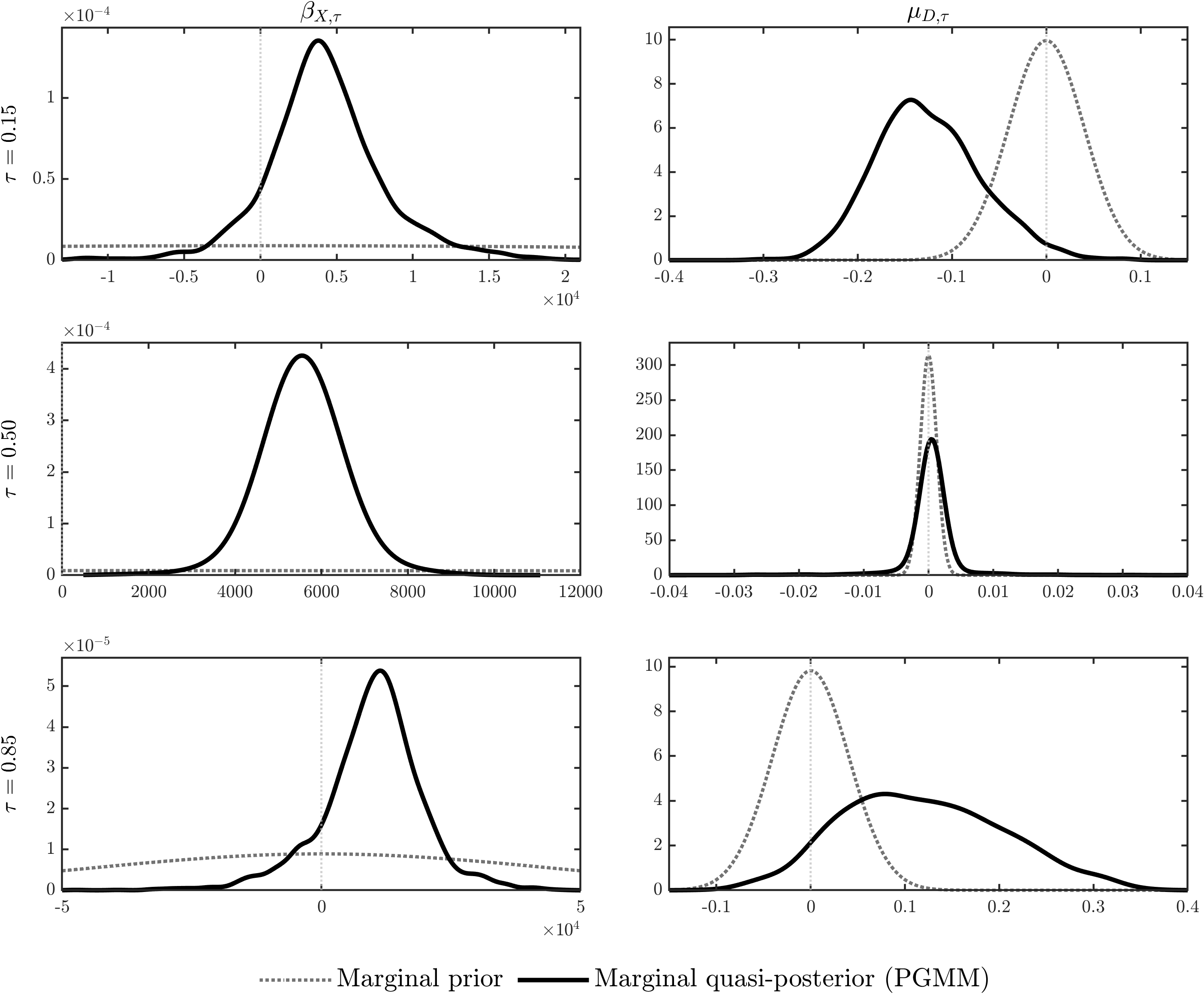} 
  \caption{PGMM marginal quasi-posteriors (solid black curves) of $\beta_{X,\tau}$ for the Gaussian prior over $\mu$ across various values of $\tau$ and $c=0.9$. The light dotted curves represent the marginal prior density curves for these parameters. 
} 
\end{figure} 

\begin{figure}[htbp!] 
    \centering
\includegraphics[width=\textwidth,height=0.35\textheight]{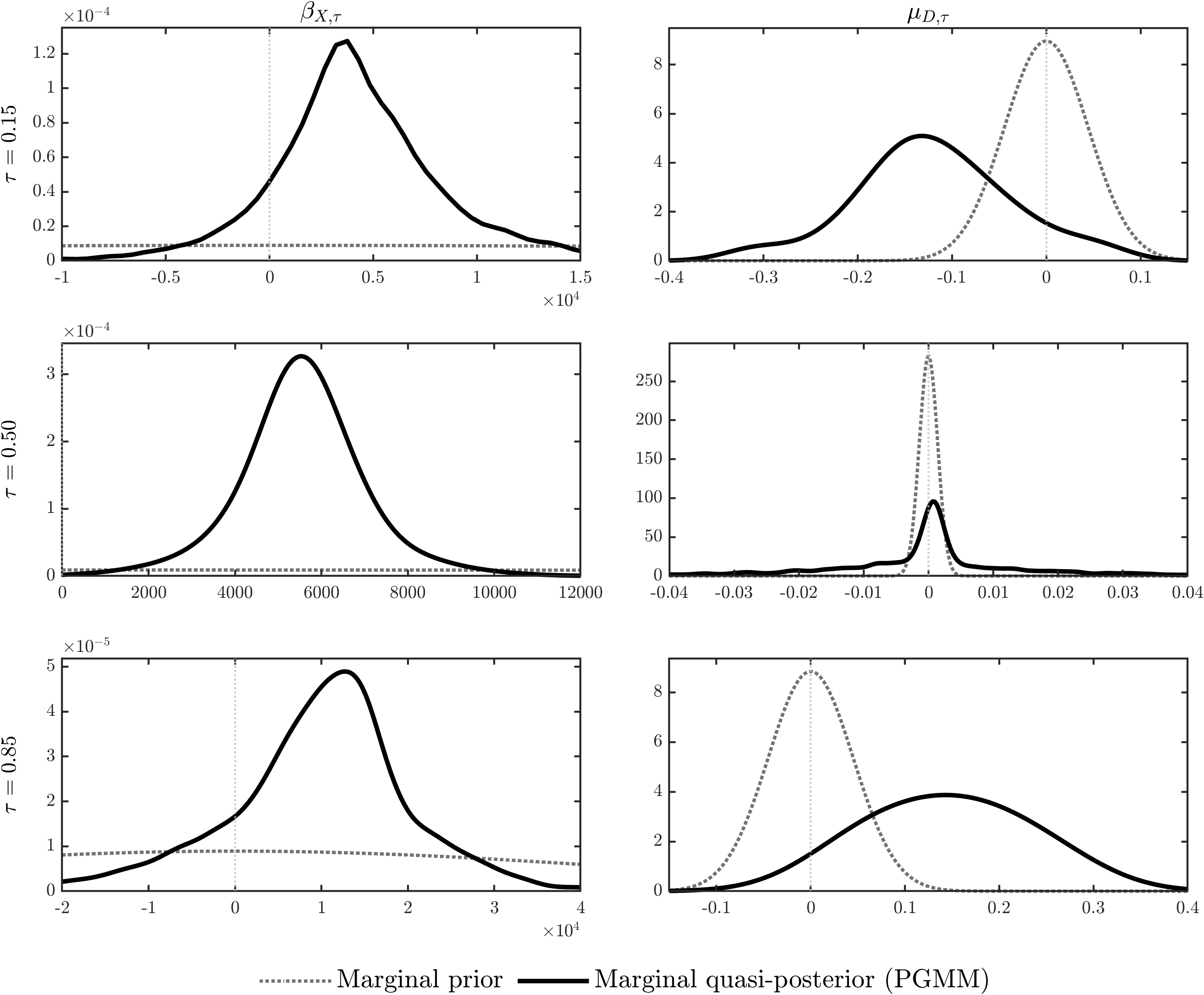} 
  \caption{PGMM marginal quasi-posteriors (solid black curves) of $\beta_{X,\tau}$ for the Gaussian prior over $\mu$ across various values of $\tau$ and $c=1$. The light dotted curves represent the marginal prior density curves for these parameters. 
} 
\end{figure}

\begin{figure}[htbp!] 
    \centering
\includegraphics[width=\textwidth,height=0.35\textheight]{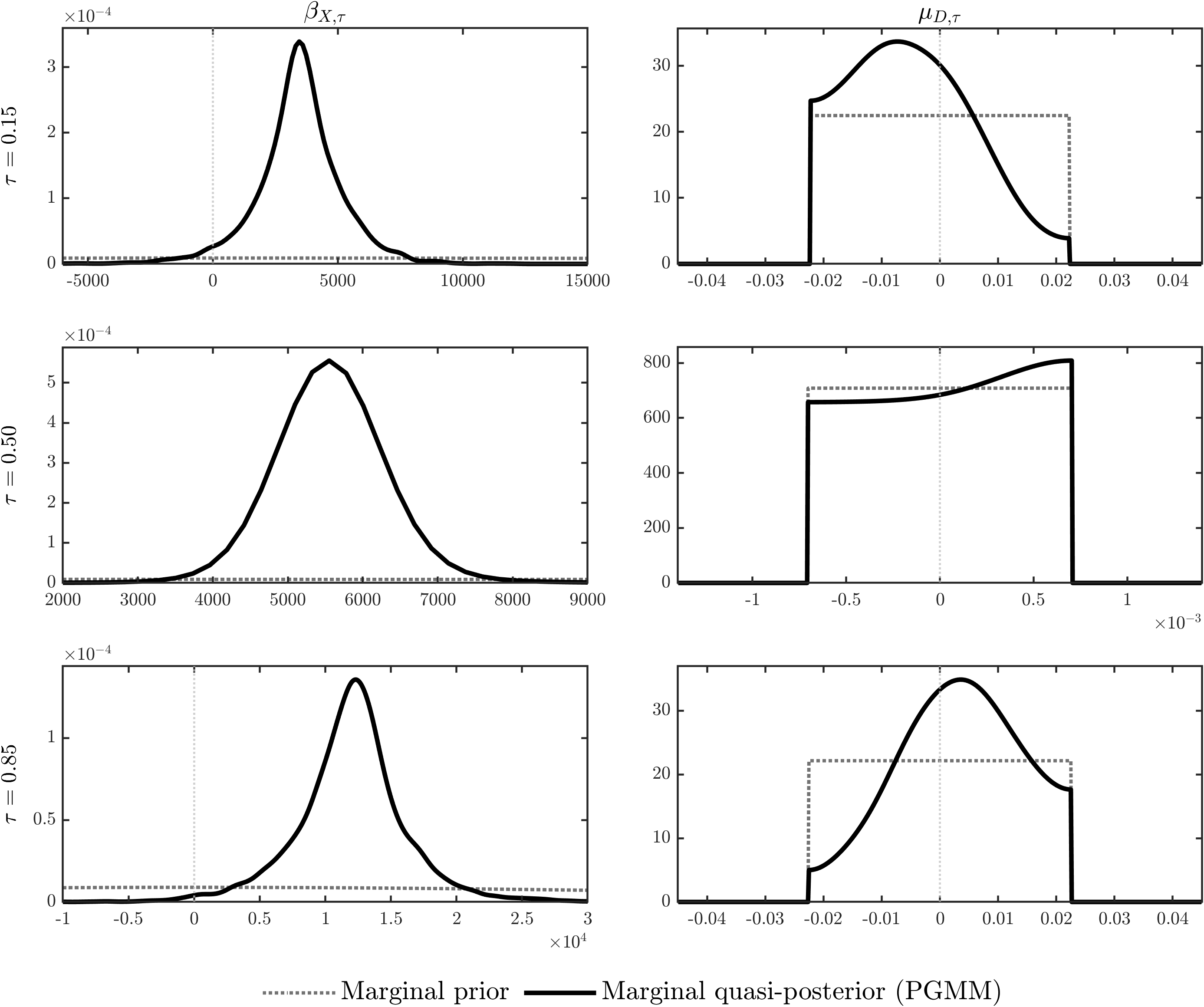} 
  \caption{PGMM marginal quasi-posteriors (solid black curves) of $\beta_{X,\tau}$ for the uniform prior over $\mu$ across various values of $\tau$ and $c=0.5$. The light dotted curves represent the marginal prior density curves for these parameters. 
} 
\end{figure} 

\begin{figure}[htbp!] 
    \centering
\includegraphics[width=\textwidth,height=0.35\textheight]{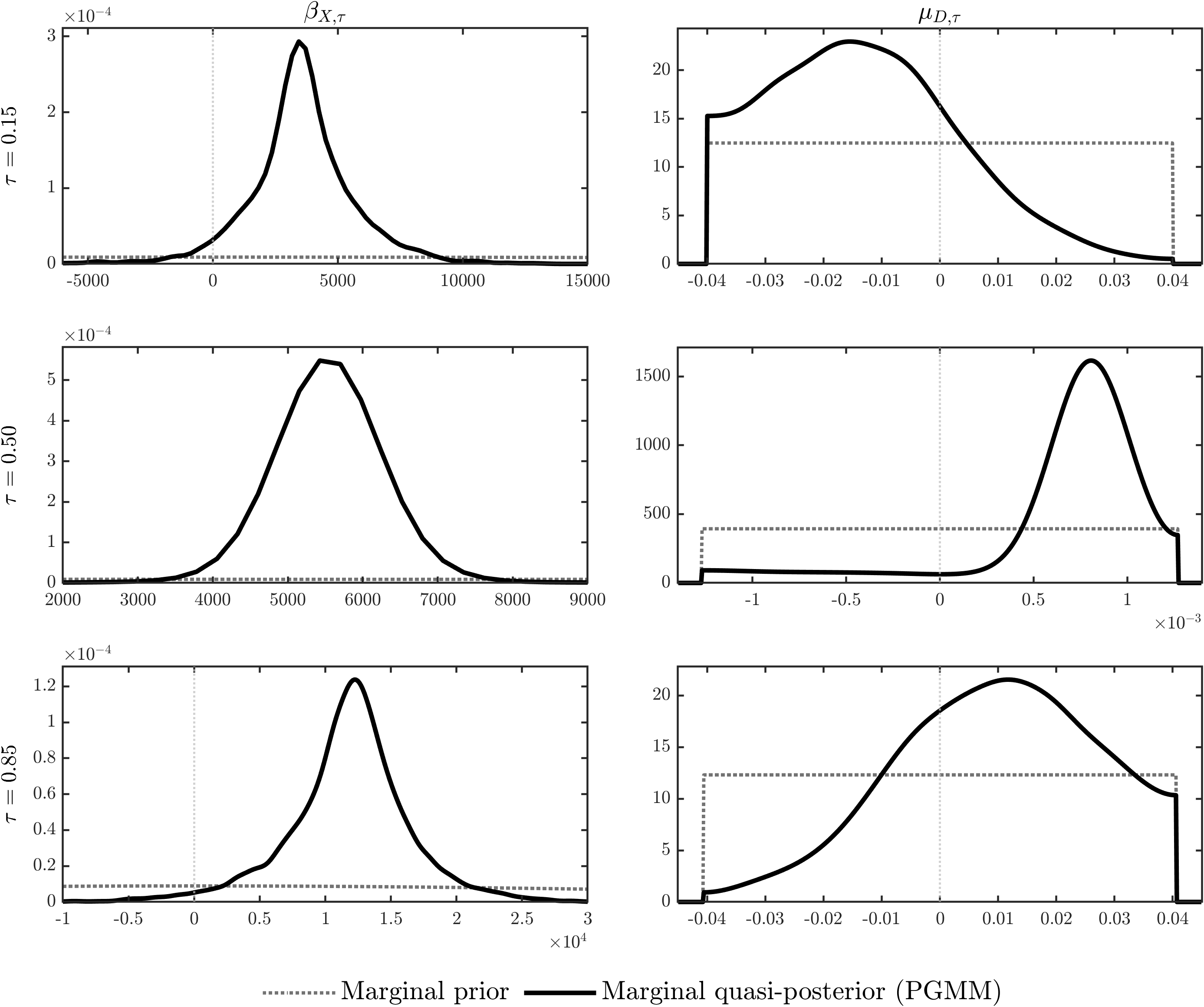} 
  \caption{PGMM marginal quasi-posteriors (solid black curves) of $\beta_{X,\tau}$ for the uniform prior over $\mu$ across various values of $\tau$ and $c=0.9$. The light dotted curves represent the marginal prior density curves for these parameters. 
} 
\end{figure} 

\begin{figure}[htbp!] 
    \centering
\includegraphics[width=\textwidth,height=0.35\textheight]{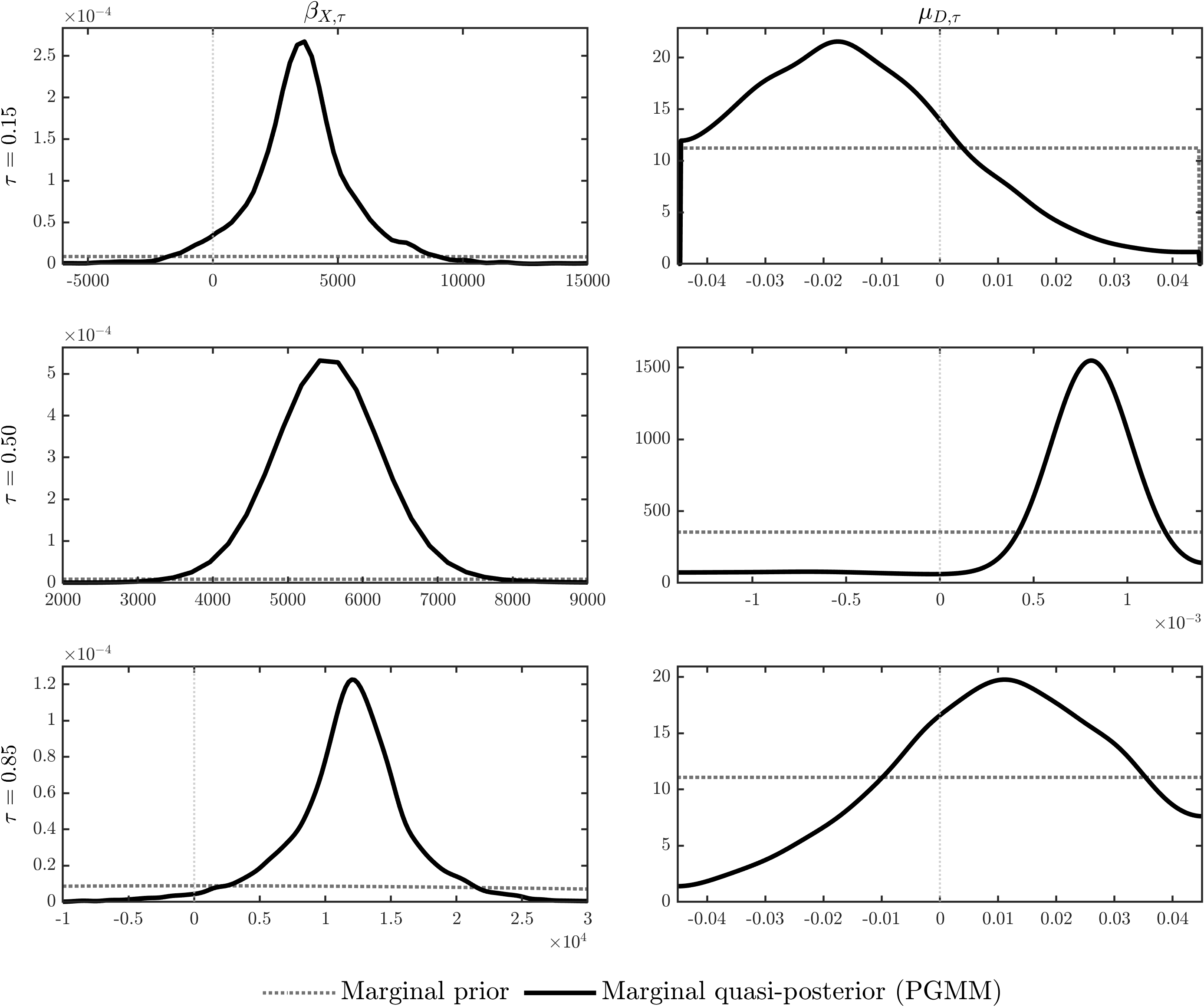}
  \caption{PGMM marginal quasi-posteriors (solid black curves) of $\beta_{X,\tau}$ for the uniform prior over $\mu$ across various values of $\tau$ and $c=1$. The light dotted curves represent the marginal prior density curves for these parameters. 
} 
\label{401k_posterior_figure_c1onlyPGMMu}
\end{figure}

 ~\\ \clearpage
\section{Auxiliary Results and Additional Proofs} 
\subsection{Intermediate results for the proof of Theorem \ref{Theorem:2}}

\subsubsection{Step of the Gaussian integral}
In this section, we derive a lemma regarding the Gaussian integral and the tail probability involved in the above main theorem.
{
\begin{lemma}\label{gai}
Under Assumptions \ref{assummu1}--\ref{rT}, let $0\leq\kappa<\infty$ 
and $\varepsilon\asymp q\log T$.
Then the following two statements hold.

\medskip
\noindent\textbf{(i)} The weighted integral over the local region 
satisfies
\[
\int_{\mathcal{M}}\int_{B_{\varepsilon,\theta|\mu}}
N_T(\theta|\mu)
\left(1+[\sqrt{T}\|\theta-\theta(\nu(\mu))\|]^\kappa\right)
d\theta\; N_T(\mu)\,d\mu
\;\lesssim_p\; k^{\kappa/2},
\]
where $B_{\varepsilon,\theta|\mu} 
:= \{\theta : \sqrt{T}\|G(\theta(\nu(\mu)))(\theta-\theta(\nu(\mu)))\|
\leq\varepsilon\}$ is the local ball in $\theta$ for fixed $\mu$, 
and $N_T(\mu)$ is the normalized marginal density satisfying 
$\int_{\mathcal{M}} N_T(\mu)\,d\mu = 1$.

\medskip\noindent\textbf{(ii)}
Recall that $\gamma\sim N(0,I_k)$ is a standard $k$-dimensional 
Gaussian random vector, and define
\[
X_\mu := C_{w,\mu}^{-1/2}\gamma - \sqrt{T}\,C_{w,\mu}^{-1}C_{m,\mu}.
\]
The set $B_\gamma$ corresponds to the region outside $B_\varepsilon$ 
under the Gaussian parametrization of $\theta$ given $\mu$, in the 
sense that $\{\theta\notin B_{\varepsilon,\theta|\mu}\}
\subseteq\{\gamma\in B_\gamma\}$, where
\[
B_\gamma := \left\{\gamma\in\mathbb{R}^k:\|\gamma\|>a_{T,\mu}\right\},
\qquad
a_{T,\mu} := \frac{c\varepsilon - \sqrt{T}\|C_{w,\mu}^{-1/2}C_{m,\mu}\|}
{\|C_{w,\mu}^{-1/2}\|}.
\]
The tail expectation satisfies
\[
\sup_{\mu\in\mathcal{M}}\,
\mathbb{E}_\gamma\!\left[
\left(1+\|X_\mu\|^\kappa\right)
\mathbf{1}\{\gamma\in B_\gamma\}
\right]
\;\lesssim_p\;
k^{\kappa/2}\exp(-c\varepsilon^2),
\]
where $c>0$ is a constant and with $\varepsilon\asymp q\log T$,
\[
\exp(-c\varepsilon^2)=\exp(-cq^2(\log T)^2)\to 0
\quad\text{as }T\to\infty.
\]
\end{lemma}
}

\begin{proof}
\noindent\textbf{Proof of (i).}It suffices to show it for $\kappa>0$. 
It is well known that for a positive continuous random variable $X$, 
$\mathbb{E}(X)=\int_{0}^{\infty}\mathbb{P}(X>x)dx$.
Recall that $\Gamma_\varepsilon$ is the set of $\mu$ corresponding to 
$B_{\varepsilon}$, where $\varepsilon \asymp q\log T$.
We also have the fact that
$\int_{B_{\varepsilon}}d\theta\,d\mu
\leq\int_{\Gamma}\int_{B_{\varepsilon,\theta|\mu}}d\theta\,d\mu,$
with $B_{\varepsilon,\theta|\mu}
\defeq\{\theta:\sqrt{T}\|G(\theta(\nu(\mu)))
(\theta-\theta(\nu(\mu)))\|\leq\varepsilon\}$
for a point $\mu\in\Gamma_\varepsilon$,
therefore, we may start by looking at
\[
\int_{B_{\varepsilon,\theta|\mu}}N_{T}(\theta|\mu)
\left(1+\|\sqrt{T}(\theta-\theta(\nu(\mu)))\|^{\kappa}\right)
\mathrm{d}\theta.
\]

Since for $u\geq 0$, $u=\int_{z\geq 0}\mathbf{1}(z\leq u)dz$,
\begin{align*}
&\int_{B_{\varepsilon,\theta|\mu}}
N_{T}(\theta|\mu)
\left(1+\|\sqrt{T}(\theta-\theta(\nu(\mu)))\|^{\kappa}\right)
d\theta\\
=\;&\int_{B_{\varepsilon,\theta|\mu}}\int_{0}^\infty
\mathbf{1}\!\left(x\leq 1+\|\sqrt{T}(\theta-\theta(\nu(\mu)))\|^{\kappa}\right)
dx\,N_{T}(\theta|\mu)\,d\theta\\
=\;&\int_{0}^\infty\int_{B_{\varepsilon,\theta|\mu}}
\mathbf{1}\!\left(1+\|\sqrt{T}(\theta-\theta(\nu(\mu)))\|^{\kappa}\geq x\right)
N_{T}(\theta|\mu)\,d\theta\,dx.
\end{align*}

As in the proof of Theorem~\ref{Theorem:2}, 
$\mathbb{P}_{N_T(\theta|\mu)}(\cdot)$ is the probability measure 
conditioning on $\mu$ corresponding to the density $N_T(\theta|\mu)$, 
and $\gamma$ is a standard $k$-dimensional multivariate Gaussian 
random variable with the associated probability measure 
$\mathbb{P}_{\gamma}$. For a fixed $\mu$, using the Gaussian parametrisation 
$\theta = \theta_\mu - x_{\star,\mu} + T^{-1/2}C_{w,\mu}^{-1/2}\gamma$
from Lemma~\ref{NT}, so that 
$\theta - \theta(\nu(\mu)) 
= T^{-1/2}C_{w,\mu}^{-1/2}\gamma - x_{\star,\mu}$,
we can proceed as follows,
\begin{align*}
&\int_{0}^{\infty}
\mathbb{P}_{N_T(\theta|\mu)}\!\left(
\|\sqrt{T}(\theta-\theta(\nu(\mu)))\|^{\kappa}>x-1,\,
\|\sqrt{T}G(\theta(\nu(\mu)))
(\theta-\theta(\nu(\mu)))\|^{2}
\leq\varepsilon^{2}\,\big|\,\mu
\right)dx\\
=\;&\int_{0}^{\infty}\mathbb{P}_{\gamma}\!\left(
\|C_{w,\mu}^{-1/2}\gamma - \sqrt{T}\,C_{w,\mu}^{-1}C_{m,\mu}\|
>((x-1)\vee 0)^{1/\kappa},\right.\\
&\qquad\qquad\left.
\|G(\theta(\nu(\mu)))
(C_{w,\mu}^{-1/2}\gamma - \sqrt{T}\,C_{w,\mu}^{-1}C_{m,\mu})\|^{2}
\leq\varepsilon^{2}\,\big|\,\mu
\right)dx\\
\leq\;&\int_{0}^{\infty}\mathbb{P}_{\gamma}\!\left(
\|C_{w,\mu}^{-1/2}\gamma\|
>((x-1)\vee 0)^{1/\kappa}
-\sqrt{T}\,\|C_{w,\mu}^{-1}C_{m,\mu}\|\,\big|\,\mu
\right)dx,
\end{align*}
where the inequality uses $\mathbb{P}(A\cap B)\leq\mathbb{P}(A)$ 
and the triangle inequality 
$\|u - v\|\geq\|u\|-\|v\|$.

By Assumption~\ref{assum33},
$\sup_{\mu\in\mathcal{M}}\|C_{w,\mu}^{-1/2}\|\lesssim_p 1$
and $\sqrt{T}\|C_{w,\mu}^{-1}C_{m,\mu}\|\lesssim_p\sqrt{k}$.
Define the split point
$x^{*}:=1+\bigl(\sqrt{k}+\sqrt{T}\|C_{w,\mu}^{-1}C_{m,\mu}\|\bigr)^{\kappa}$.
Note that
\begin{align*}
&\int_{0}^{\infty}\mathbb{P}_{\gamma}\!\left(
\|C_{w,\mu}^{-1/2}\gamma\|
>((x-1)\vee 0)^{1/\kappa}
-\sqrt{T}\|C_{w,\mu}^{-1}C_{m,\mu}\|\,\big|\,\mu
\right)dx\\
=\;&\int_{0}^{x^{*}}
\mathbb{P}_{\gamma}\!\left(
\|C_{w,\mu}^{-1/2}\gamma\|
>((x-1)\vee 0)^{1/\kappa}
-\sqrt{T}\|C_{w,\mu}^{-1}C_{m,\mu}\|\,\big|\,\mu
\right)dx\\
&+\int_{x^{*}}^{\infty}
\mathbb{P}_{\gamma}\!\left(
\|C_{w,\mu}^{-1/2}\gamma\|
>(x-1)^{1/\kappa}
-\sqrt{T}\|C_{w,\mu}^{-1}C_{m,\mu}\|\,\big|\,\mu
\right)dx.
\end{align*}

Since every probability is at most $1$,
\[
\int_{0}^{x^{*}}1\,dx
=1+\bigl(\sqrt{k}+\sqrt{T}\|C_{w,\mu}^{-1}C_{m,\mu}\|\bigr)^{\kappa}
\lesssim_p 1+k^{\kappa/2}
\lesssim k^{\kappa/2},
\]
where the second inequality uses
$\sqrt{T}\|C_{w,\mu}^{-1}C_{m,\mu}\|\lesssim_p\sqrt{k}$
and hence
$(\sqrt{k}+\sqrt{T}\|C_{w,\mu}^{-1}C_{m,\mu}\|)^{\kappa}
\lesssim_p k^{\kappa/2}$.

For $x>x^{*}$, the threshold satisfies
$(x-1)^{1/\kappa}-\sqrt{T}\|C_{w,\mu}^{-1}C_{m,\mu}\|>\sqrt{k}>0$.
Substitute $u=(x-1)^{1/\kappa}-\sqrt{T}\|C_{w,\mu}^{-1}C_{m,\mu}\|$,
so that $x-1=(u+\sqrt{T}\|C_{w,\mu}^{-1}C_{m,\mu}\|)^{\kappa}$
and $dx=\kappa(u+\sqrt{T}\|C_{w,\mu}^{-1}C_{m,\mu}\|)^{\kappa-1}du$,
with lower limit $u=\sqrt{k}$.
Since $\|C_{w,\mu}^{-1/2}\|\lesssim_p 1$,
there exists a constant $C_{3}>0$ such that
$\mathbb{P}_{\gamma}(\|C_{w,\mu}^{-1/2}\gamma\|>u)
\leq\mathbb{P}_{\gamma}(\|\gamma\|>u/C_{3})$,
and Lemma~\ref{tail} gives
$\mathbb{P}_{\gamma}(\|\gamma\|>u/C_{3})\leq\exp(-c_{0}u^{2})$
for all $u\geq\sqrt{k}$ and a constant $c_{0}>0$.
Therefore, using $u+\sqrt{T}\|C_{w,\mu}^{-1}C_{m,\mu}\|\lesssim_p u$
for $u\geq\sqrt{k}$ and setting $z=c_{0}u^{2}$,
\begin{align*}
\int_{x^{*}}^{\infty}
\mathbb{P}_{\gamma}\!\left(
\|C_{w,\mu}^{-1/2}\gamma\|>u\,\big|\,\mu
\right)
\kappa(u+\sqrt{T}\|C_{w,\mu}^{-1}C_{m,\mu}\|)^{\kappa-1}du
&\lesssim_p
\int_{\sqrt{k}}^{\infty}
u^{\kappa-1}e^{-c_{0}u^{2}}du\\
&=\frac{1}{2c_{0}^{\kappa/2}}
\int_{c_{0}k}^{\infty}
z^{\kappa/2-1}e^{-z}dz\\
&\lesssim_p k^{\kappa/2},
\end{align*}
where the last bound uses the standard incomplete gamma function
estimate $\int_{s}^{\infty}z^{\kappa/2-1}e^{-z}dz\lesssim s^{\kappa/2-1}e^{-s/2}
\lesssim k^{\kappa/2}$ for $s=c_{0}k$.\\
\noindent\textbf{Proof of (ii).}
\medskip
\noindent\textbf{Step 1:}
Define
\[
X_\mu := C_{w,\mu}^{-1/2}\gamma - \sqrt{T}\,C_{w,\mu}^{-1}C_{m,\mu},
\qquad \gamma\sim N(0,I_k).
\]

\medskip
\noindent\textbf{Step 2: Positivity of $a_{T,\mu}$ and lower bound.}
By Assumption~\ref{assum33}, uniformly over $\mu\in\mathcal{M}$,
\[
\sup_{\mu\in\mathcal{M}}\|C_{w,\mu}^{-1/2}\|\lesssim_p 1,
\qquad
\sup_{\mu\in\mathcal{M}}\sqrt{T}\|C_{w,\mu}^{-1/2}C_{m,\mu}\|
\lesssim_p \sqrt{k}.
\]
Since $\varepsilon\asymp q\log T$, there exist constants 
$C_1, C_2>0$ such that with probability approaching one,
\[
c\varepsilon \geq c\cdot C_1 q\log T 
\gg C_2\sqrt{k} 
\geq \sqrt{T}\|C_{w,\mu}^{-1/2}C_{m,\mu}\|,
\]
so $a_{T,\mu}>0$ for all $\mu\in\mathcal{M}$ with probability 
approaching one. Moreover, since 
$\|C_{w,\mu}^{-1/2}\|\lesssim_p 1$,
\[
a_{T,\mu} 
= \frac{c\varepsilon - \sqrt{T}\|C_{w,\mu}^{-1/2}C_{m,\mu}\|}
{\|C_{w,\mu}^{-1/2}\|}
\geq \frac{c\varepsilon - C_2\sqrt{k}}{C_3}
\geq c_1\varepsilon
\]
for some constant $c_1>0$ and all large $T$, where the last 
inequality uses $\sqrt{k}\ll\varepsilon$.
Hence $\inf_{\mu\in\mathcal{M}}a_{T,\mu}\geq c_1\varepsilon$ 
with probability approaching one.

\medskip
\noindent\textbf{Step 3: Gaussian concentration.}
By Lemma~\ref{tail} (Corollary A.3), for $\gamma\sim N(0,I_k)$,
\[
\Pr(\|\gamma\|\geq\sqrt{k}+\sqrt{2x})\leq e^{-x}.
\]
Setting $\sqrt{k}+\sqrt{2x}=c_1\varepsilon$ gives 
$x\geq (c_1\varepsilon-\sqrt{k})^2/2\geq c_2\varepsilon^2$ 
for some constant $c_2>0$ (using $\varepsilon\asymp q\log T$ 
so that $c_1\varepsilon\gg\sqrt{k}$ for large $T$). Therefore
\[
\sup_{\mu\in\mathcal{M}}\mathbb{P}_\gamma(B_\gamma)
\leq \Pr(\|\gamma\|>c_1\varepsilon)
\leq \exp(-c_2\varepsilon^2),
\]
where $c_2>0$ depends only on $c_1$ and hence only on the 
constants in Assumption~\ref{assum33}.

\medskip
\noindent\textbf{Step 4: Cauchy--Schwarz decomposition.}
Write
\[
\mathbb{E}_\gamma\!\left[
(1+\|X_\mu\|^\kappa)\mathbf{1}\{\gamma\in B_\gamma\}
\right]
=
\underbrace{\mathbb{P}_\gamma(B_\gamma)}_{=:I}
+
\underbrace{\mathbb{E}_\gamma\!\left[
\|X_\mu\|^\kappa\mathbf{1}\{\gamma\in B_\gamma\}
\right]}_{=:II}.
\]
For term $II$, apply Cauchy--Schwarz:
\[
II \leq 
\left(\mathbb{E}_{\gamma}\|X_\mu\|^{2\kappa}\right)^{1/2}
\mathbb{P}_\gamma(B_\gamma)^{1/2}.
\]

\medskip
\noindent\textbf{Step 5: Bounding $\mathbb{E}_\gamma\|X_\mu\|^{2\kappa}$.}
By the moment inequality applied to 
$X_\mu = C_{w,\mu}^{-1/2}\gamma - \sqrt{T}C_{w,\mu}^{-1}C_{m,\mu}$,
\[
\mathbb{E}_\gamma\|X_\mu\|^{2\kappa}
\lesssim
\mathbb{E}_\gamma\|C_{w,\mu}^{-1/2}\gamma\|^{2\kappa}
+\|\sqrt{T}\,C_{w,\mu}^{-1}C_{m,\mu}\|^{2\kappa}.
\]
By the same arguments as (i)  applied 
with $2\kappa$ in place of $\kappa$, both terms are 
$\lesssim_p k^\kappa$. Therefore
\[
\sup_{\mu\in\mathcal{M}}\mathbb{E}_\gamma\|X_\mu\|^{2\kappa}
\lesssim_p k^\kappa.
\]

Combining Steps 3, 4, and 5,
\[
\sup_{\mu\in\mathcal{M}}
\mathbb{E}_\gamma\!\left[
(1+\|X_\mu\|^\kappa)\mathbf{1}\{\gamma\in B_\gamma\}
\right]
\lesssim_p
\exp(-c_2\varepsilon^2)
+k^{\kappa/2}\exp(-c_2\varepsilon^2/2).
\]
Since $\exp(-c_2\varepsilon^2)\leq k^{\kappa/2}\exp(-c_2\varepsilon^2/2)$ 
for large $T$, we obtain
\[
\sup_{\mu\in\mathcal{M}}
\mathbb{E}_\gamma\!\left[
(1+\|X_\mu\|^\kappa)\mathbf{1}\{\gamma\in B_\gamma\}
\right]
\lesssim_p
k^{\kappa/2}\exp(-c\varepsilon^2),
\]
where $c=c_2/2>0$. Since $\varepsilon\asymp q\log T$, the right-hand side satisfies
\[
k^{\kappa/2}\exp(-c\varepsilon^2)
=k^{\kappa/2}\exp(-cq^2(\log T)^2)\to 0
\]
as $T\to\infty$ for any fixed $\kappa\geq 0$.
This completes the proof of (ii).

\end{proof}

\subsubsection{Step of \texorpdfstring{$c^*$}{TEXT}}
In this subsection, we study the term $c^*$,
\begin{align*}
	c^* & =\int_{ \Xi}\exp(-\frac{1}{2}V_T(h(.),\theta,\mu)+\log\pi(\mu))d\theta d\mu /\int_{ \Xi}\exp(\frac{1}{2}Q_T(\theta,\mu)+\log\pi(\theta, \mu))d\theta d\mu
\end{align*}
\begin{lemma}\label{smallc}
Under Assumptions \ref{assummu1}-\ref{rT}, we have
$c^* \rightarrow_{p} 1$.
\end{lemma}
\begin{proof}

Define $$c_1^*= {\int_{B_{\varepsilon}}\exp(-V_T(h(.),\theta,\mu)/2+\log(\pi(\mu)))d\theta d\mu},$$ and $${c_2^*}= {\int_{ B_{\varepsilon}}\exp(\frac{1}{2}Q_T(\theta,\mu)+\log\pi(\mu,\theta))d\theta d\mu}.$$
  	\begin{eqnarray*}
&&c^* =\frac{\int_{ \Xi}\exp(-V_T(h(.),\theta,\mu)/2+\log(\pi(\mu)))d\theta d\mu}{\int_{\Xi}\exp(\frac{1}{2}Q_T(\theta,\mu)+\log\pi(\mu,\theta))d\theta d\mu}\\
&&=\frac{\int_{B_{\varepsilon}}\exp(-V_T(h(.),\theta,\mu)/2+\log(\pi(\mu)))d\theta d\mu}{\int_{ B_{\varepsilon}}\exp(\frac{1}{2}Q_T(\theta,\mu)+\log\pi(\mu,\theta))d\theta d\mu}+ o_p(1)=\frac{c_1^*}{c_2^*}+ o_p(1)\\
&&=\frac{c_1^*}{\int_{ B_{\varepsilon}}\exp(-V_T(h(.),\theta,\mu)/2+\log(\pi(\mu)) + R_T(\theta,\mu))d\theta d\mu}+ o_p(1)
\\
&&=\frac{c_1^*}{\int_{ B_{\varepsilon}}\exp(-V_T(h(.),\theta,\mu)/2+\log(\pi(\mu))) [\exp( R_T(\theta,\mu))-1]d\theta d\mu+c_1^*}
\\&&+ o_p(1).
	\end{eqnarray*}

 If the above term is of order 
 $\frac{\int_{B_{\varepsilon}}\exp(-V_T(h(.),\theta,\mu)/2+\log(\pi(\mu)))d\theta d\mu}{\int_{ B_{\varepsilon}}\exp(-V_T(h(.),\theta,\mu)/2+\log(\pi(\mu)))d\theta d\mu(1+o_p(1))}
+ o_p(1),$ then we reach the conclusion.

By Assumption~\ref{rT},
\[
\sup_{(\theta,\mu)\in B_{\varepsilon}}
\frac{T|R_{T}(\theta,\mu)|}
{\|\sqrt{T}h(\theta,\mu)\|^2+k(\log T)^2}
\lesssim_p
\frac{\sqrt{k}(\log T)^2}{\sqrt{T}}\vee\frac{q}{\sqrt{kT}}
\to 0.
\]
Rearranging, we obtain
\[
\sup_{(\theta,\mu)\in B_{\varepsilon}}|R_{T}(\theta,\mu)|
\lesssim_p
\frac{1}{T}\left(\frac{\sqrt{k}(\log T)^2}{\sqrt{T}}\vee\frac{q}{\sqrt{kT}}\right)
\left(\|\sqrt{T}h(\theta,\mu)\|^2+k(\log T)^2\right).
\]
Since on $B_\varepsilon$ we have $\|\sqrt{T}h(\theta,\mu)\|^2\leq\varepsilon^2
\asymp q^2(\log T)^2$ and $k(\log T)^2 \leq q^2(\log T)^2$ (as $k \leq q$), 
it follows that
\[
\sup_{(\theta,\mu)\in B_{\varepsilon}}|R_{T}(\theta,\mu)|
\lesssim_p
\frac{q^2(\log T)^2}{T}
\cdot\left(\frac{\sqrt{k}(\log T)^2}{\sqrt{T}}\vee\frac{q}{\sqrt{kT}}\right)
=
\frac{\sqrt{k}\,q^2(\log T)^4}{T^{3/2}}\vee\frac{q^3(\log T)^2}{k^{1/2}T^{3/2}}
\to 0,
\]
under Assumption~\ref{rates}, since $q^2(\log T)^2/T \to 0$ and 
$q/(k^{1/2}T^{1/2}) \to 0$ are both implied by the second condition 
of Assumption~\ref{rates}.
 
\end{proof}

\subsubsection{Properties of $N_T(\mu)$}

\begin{lemma}[$N_T(\mu)$] \label{NT}
Let $W_\mu$ and $G_\mu$ satisfy Assumption \ref{assum33}.  
Recall that
\[
\theta_\mu := \theta(\nu(\mu)), 
\qquad 
x_\mu := \theta - \theta_\mu \in \mathbb{R}^k,
\qquad 
r_\mu := \widehat{m}(\theta(\nu(\mu))) -\mu \in \mathbb{R}^q.
\]
Recall that 
\begin{equation}
\label{eq:VT-def}
V_T(h(.),\theta,\mu) =V_T(\theta,\mu)
= 
T\,(G_\mu x_\mu + r_\mu)^\top W_\mu (G_\mu x_\mu + r_\mu)
- 2\log \pi(\mu)
- 2\log \pi(\theta_\mu)
+ 2\sqrt{T}\|\mu - \nu(\mu)\|.
\end{equation}
Let
\begin{equation}
\label{eq:xstar}
x_{\star,\mu}
:= C_{w,\mu}^{-1} G_\mu^\top W_\mu r_\mu,
\qquad
C_{w,\mu} := G_\mu^\top W_\mu G_\mu,
\qquad
P_{G,\mu} := G_\mu C_{w,\mu}^{-1} G_\mu^\top W_\mu.
\end{equation}

\begin{enumerate}
\item[(i)] The conditional distribution of $\theta$ given $\mu$ is Gaussian:
\begin{equation}
\label{eq:theta-cond}
\theta \mid \mu
\sim \mathcal N\!\left(
\theta_\mu - x_{\star,\mu},
\;
(T\,C_{w,\mu})^{-1}
\right),
\end{equation}
with normalized density
\[
N_T(\theta\mid\mu)
=
\frac{T^{k/2}\,|C_{w,\mu}|^{1/2}}{(2\pi)^{k/2}}
\exp\!\left(
-\frac{T}{2}
(\theta - \theta_\mu + x_{\star,\mu})^\top
C_{w,\mu}
(\theta - \theta_\mu + x_{\star,\mu})
\right).
\]

\item[(ii)] The (unnormalized) marginal density in $\mu$ is
\begin{align}
\label{eq:pT-mu-alt}
N_T(\mu)
\;\propto\;&
|C_{w,\mu}|^{-1/2}
\\
&\times
\exp\!\left\{
-\frac{T}{2}
\big(\widehat m(\theta(\nu(\mu))) - \mu\big)^\top
(I-P_{G,\mu})^\top
W_\mu
(I-P_{G,\mu})
\big(\widehat m(\theta(\nu(\mu))) - \mu\big)
\right\}
\nonumber
\\
&\times
\pi(\mu)\,\pi(\theta(\nu(\mu)))
\exp\!\big(-\sqrt{T}\|\mu-\nu(\mu)\|\big).
\nonumber
\end{align}
\end{enumerate}
\end{lemma}

\begin{proof}

Start from \eqref{eq:VT-def} and expand the quadratic term:
\[
(G_\mu x_\mu + r_\mu)^\top W_\mu (G_\mu x_\mu + r_\mu)
=
x_\mu^\top C_{w,\mu} x_\mu
+
2 x_\mu^\top G_\mu^\top W_\mu r_\mu
+
r_\mu^\top W_\mu r_\mu.
\]

Complete the square using $x_{\star,\mu}$ defined in \eqref{eq:xstar}.  
Since
\[
C_{w,\mu} x_{\star,\mu}
=
G_\mu^\top W_\mu r_\mu,
\]
we obtain
\[
x_\mu^\top C_{w,\mu} x_\mu
+
2 x_\mu^\top G_\mu^\top W_\mu r_\mu
=
(x_\mu + x_{\star,\mu})^\top
C_{w,\mu}
(x_\mu + x_{\star,\mu})
-
r_\mu^\top W_\mu P_{G,\mu} r_\mu.
\]

Therefore,
\begin{equation}
\label{eq:quad-split}
(G_\mu x_\mu + r_\mu)^\top W_\mu (G_\mu x_\mu + r_\mu)
=
(x_\mu + x_{\star,\mu})^\top
C_{w,\mu}
(x_\mu + x_{\star,\mu})
+
r_\mu^\top W_\mu (I - P_{G,\mu}) r_\mu.
\end{equation}

Substituting \eqref{eq:quad-split} into $-V_T(\theta,\mu)/2$ gives
\begin{align}
\exp\{-V_T(\theta,\mu)/2\}
=\;&
\exp\!\left(
-\frac{T}{2}
(x_\mu + x_{\star,\mu})^\top
C_{w,\mu}
(x_\mu + x_{\star,\mu})
\right)
\nonumber
\\
&\times
\exp\!\left(
-\frac{T}{2}
r_\mu^\top W_\mu (I-P_{G,\mu}) r_\mu
\right)
\nonumber
\\
&\times
\pi(\mu)\,\pi(\theta_\mu)
\exp\!\big(-\sqrt{T}\|\mu-\nu(\mu)\|\big).
\label{eq:kernel-factor}
\end{align}

\textbf{Proof of (i).}

The $\theta$-dependent term in \eqref{eq:kernel-factor} is a Gaussian density with
precision matrix $T C_{w,\mu}$ and center $\theta_\mu - x_{\star,\mu}$.  
Normalizing yields \eqref{eq:theta-cond} and the stated density.

\textbf{Proof of (ii).}

Integrating \eqref{eq:kernel-factor} over $\theta \in \mathbb{R}^k$ gives
\[
\int
\exp\!\left(
-\frac{T}{2}
(x_\mu + x_{\star,\mu})^\top
C_{w,\mu}
(x_\mu + x_{\star,\mu})
\right)
d\theta
=
(2\pi)^{k/2}
\left(T\right)^{-k/2}
|C_{w,\mu}|^{-1/2}.
\]

The remaining $\mu$-dependent terms give
\begin{align*}
N_T(\mu)
\propto\;&
|C_{w,\mu}|^{-1/2}
\exp\!\left(
-\frac{T}{2}
r_\mu^\top W_\mu (I-P_{G,\mu}) r_\mu
\right)
\\
&\times
\pi(\mu)\,\pi(\theta(\nu(\mu)))
\exp\!\big(-\sqrt{T}\|\mu-\nu(\mu)\|\big).
\end{align*}

Finally, since
\[
P_{G,\mu}^\top W_\mu = W_\mu P_{G,\mu},
\qquad
P_{G,\mu}^2 = P_{G,\mu},
\]
we have
\[
r_\mu^\top W_\mu (I-P_{G,\mu}) r_\mu
=
r_\mu^\top
(I-P_{G,\mu})^\top
W_\mu
(I-P_{G,\mu})
r_\mu,
\]
which yields \eqref{eq:pT-mu-alt}.  
This completes the proof.

\end{proof}

 {

	Recall that $$B_{\varepsilon}^{c1} =\{\|\sqrt{T}h_{\parallel}\|>\frac{\varepsilon}{2}\} $$
	$$B_{\varepsilon}^{c2} =\{\|\sqrt{T}h_{\perp}\|>\frac{\varepsilon}{2}\} $$
	$$B_{\varepsilon}^{c3} =\{\sqrt{T} \|P_{G,\mu}( \mu - \nu(\mu)) \| > \frac{\varepsilon}{2}\} .$$
	Then we have $B_{\varepsilon}^{c} = B_{\varepsilon}^{c1}\cup B_{\varepsilon}^{c2}\cup B_{\varepsilon}^{c3}.$
	We shall prove the case for  $$\bar{B}\equiv B_{\varepsilon}^{c3}\cap (B_{\varepsilon}^{c1c}\cap B_{\varepsilon}^{c2c}).$$
	
	\begin{lemma}\label{lemmaextra}
		Suppose Assumptions \ref{assum33} and \ref{assummu1} hold,   
		then
		\[
		\int_{\bar{\mathcal{B}}} \left(1 + \|\sqrt{T}h(\theta,\mu)\|^\kappa\right)
		\left|p_T(\theta,\mu)-N_T(\theta,\mu)\right|\,d\theta\,d\mu
		= o_p(1),
		\]
		where  $\bar{\mathcal{B}} = \{\|\sqrt{T}h(\theta,\mu)\|\le \varepsilon,\;
		\sqrt{T}\|\mu-\nu(\mu)\|>\tfrac{\varepsilon}{2} \}$ and $\varepsilon = q\log T$.
	\end{lemma}
	
	\begin{proof}
		Recall that
		\[
		N_T(\theta,\mu)
		=
		\frac{
			\exp\left\{-\tfrac{1}{2}V_T(h(\cdot),\theta,\mu)\right\}\,\pi(\mu)
		}
		{
			\displaystyle \int_{\mathcal{M}}\int_{\Theta}
			\exp\left\{-\tfrac{1}{2}V_T(h(\cdot),\theta,\mu)\right\}\,\pi(\mu)\,d\theta\,d\mu
		},
		\]
		with $
		V_T(h,\theta,\mu)
		=
		- 2Q_{h,T}(h,\theta,\mu)
		- 2\log\pi(\theta(\nu(\mu)))
		+ 2\sqrt{T}\,\|\mu-\nu(\mu)\|$, 
		and
		\[
		Q_{h,T}(h,\theta,\mu)
		=
		- \frac{T}{2}\left(\widehat m(\theta(\nu(\mu)))-\nu(\mu)+h(\theta,\mu)\right)^\top 
		W(\theta(\nu(\mu)))
		\left(\widehat m(\theta(\nu(\mu)))-\nu(\mu)+h(\theta,\mu)\right).
		\]

		To prove the claim, we first show
		\[
		\int_{\bar{\mathcal{B}}}
		\left(1+\|\sqrt{T}h(\theta,\mu)\|^\kappa\right)
		N_T(\theta,\mu)\,d\theta\,d\mu
		= o_p(1).
		\] 
		
		Let the plausible pair \((\mu',\theta(\mu'))\) with \(\mu'\in {\mathcal{M}}\) be the one specified by Assumption \ref{assummu1} such that there exists \(\delta>0\) on the ball, \(  B_{\delta/\sqrt{T}}(\theta(\mu))\), \(\pi(\theta)\) are bounded below. Over the neighborhood $ B_{\delta/\sqrt{T}}((\mu,\theta(\mu))) $, \(\sqrt{T}\|\mu-\nu(\mu)\| = O(1)\), so with probability approaching one,
		\[
		\exp\left\{-\tfrac{1}{2}Q_{h,T}(h,\theta,\mu)\right\}
		\gtrsim \exp(-C_q q)
		\]
		for sufficiently large \(C_q\) that $C_q q\rightarrow \infty$, e.g., $C_q=\log (\log T)/c$ for $c>0$. Hence the term $$\int_{\mathcal{M}}\int_{\Theta}
		\exp\left\{-\tfrac{1}{2}V_T(h(\cdot),\theta,\mu)\right\}\,\pi(\mu)\,d\theta\,d\mu$$ is lower bounded (with probability approaching one) by \(\exp(-C_q q) {(cT)}^{-(k+q)/2}\) with $c$ being a positive number.
		
		\medskip
		
		The constraints on \(\bar{\mathcal{B}}\) implies that
		\[
		1+\|\sqrt{T}h(\theta,\mu)\|^\kappa \lesssim 1 +  \sqrt{T}^\kappa d(\mu,\Gamma)^\kappa.
		\]
		where we let \(d(\mu,\Gamma) := \|\mu-\nu(\mu)\|\) and note that the above inequality also holds for $\sqrt{T}\|\theta-\theta(\nu(\mu))\|$.

		Using   \(\pi(\theta)\le C_\theta\) and the fact that $\int_{\sqrt{T}d(\mu,\Gamma) \geq \varepsilon} (\sqrt{T}d(\mu,\Gamma))^\kappa \exp(-\sqrt{T}d(\mu,\Gamma)) \pi(\mu) d\mu \leq \varepsilon^\kappa  e^{-\varepsilon}$, the integral with probability approaching one is (up to constant multiplied by $\exp(C_q q) {(cT)}^{(k+q)/2}$), and is up to the following magnitude
			\[
			T^{-k/2}\varepsilon^{k/2}  e^{-\varepsilon} + 	T^{-k/2}\varepsilon^{k/2}  \varepsilon^\kappa  e^{-\varepsilon}.  
			\]  Note that if we let $C_q=\log(\log T) /3$ under the rate assumption and $\varepsilon = q\log T, q\geq k$, we have that 
		\begin{align*}
			& \exp(C_q q) {(cT)}^{(k+q)/2} T^{-k / 2} \varepsilon^{k / 2+\kappa} e^{-\varepsilon}  \\	\lesssim &  ((\log T)^\frac{q}{3} (T \log T )^{(k+q)/2} T^{-k / 2} q^{k/2+\kappa} (\log T)^{k/2+\kappa} T^{-q} \\\lesssim &  ((\log T)^4q^2/T)^\frac{q}{2}   =o(1).
		\end{align*}
		Therefore, $
		\int_{\bar{\mathcal{B}}}
		\left(1+\|\sqrt{T}h(\theta,\mu)\|^\kappa\right)
		N_T(\theta,\mu)\,d\theta\,d\mu
		= o_p(1).$ The conclusion is proved once we show that $$
		\int_{\bar{\mathcal{B}}} \left(1 + \|\sqrt{T}h(\theta,\mu)\|^\kappa\right)
		|p_T(\theta,\mu)|\,d\theta\,d\mu
		= o_p(1),$$, and this is done in a similar manner to the previous arguments. Note that following a similar argument, we have that with probability approaching one, the above integral  (up to a constant  multiplied by $\exp(C_q q) {(cT)}^{(k+q)/2}$) is up to the following magnitude 
		\[
			T^{-k/2}\varepsilon^{k/2}  e^{-(k/c_w)^{-1} \varepsilon^2} + 	T^{-k/2}\varepsilon^{k/2}  \varepsilon^\kappa  e^{-(k/c_w)^{-1} \varepsilon^2} ,
		\]		 
		which multiplied by $\exp(C_q q) {(cT)}^{(k+q)/2}$ is also $o(1)$ as indicated by the rate assumption. This completes the proof.
	\end{proof}
}

\subsection{Proof of Lemma \ref{coverageprt}.}\label{Proof: lemma1}
\begin{proof}
We aim to establish that
\begin{align}
\int_\mu \mathbb{P}_{\mu}\left(\theta(\mu) \in PR_T \right) \pi(\mu) d\mu = 1 - \alpha +o_p(1). \label{eq:lemma1-1}
\end{align}

We begin by considering the case under the conditions of Theorem \ref{Theorem:1}, as the argument for Theorem \ref{Theorem:2} proceeds analogously.

We know from Theorem \ref{Theorem:1} that $PR_T(\alpha)=\left\{\theta: T(\widehat{m}(\theta))^\top  A_{\theta_0} (\widehat{m}(\theta)) \leq {Z_{\alpha}}\right\}$ such that $Z_\alpha$ satisfies the following condition 
\begin{align*}
     \P_{T,\theta}\left(T(\widehat{m}(\theta)^\top  A_{\theta_0} \widehat{m}(\theta)) \leq {Z_{\alpha}}\right)  =1-\alpha,
\end{align*}
where $\mathbb{P}_{T,\theta}$ denotes the probability measure of $\widehat{m}(\theta)-\mu+ \mu$ corresponding to the marginal posterior $p_T(\theta)$, which in this case is Gaussian with mean $0$ and covariance matrix $A_{\theta_0}^{-1}$. Therefore, $Z_\alpha$ is the $(1 - \alpha)$ quantile of a chi-squared distribution with $q$ degrees of freedom. To verify equation~\eqref{eq:lemma1-1}, it suffices to show that
\begin{align}
\label{eq:lemma 1 linear}
    \int \mathbb{P}_{\mu}\left(T(\widehat{m}({\theta}(\mu))-\mu+ \mu)^\top  A_{\theta_0}(\widehat{m}(\theta(\mu))-\mu+\mu)) \leq {Z_{\alpha}}\right)f_{\mu}(\mu)d\mu =1-\alpha + o_p(1),
\end{align}
where $A_{\theta_0} = \Lambda^{-1} - \Lambda^{-1} \left( \Omega(\theta_0)^{-1} + \Lambda^{-1} \right)^{-1} \Lambda^{-1}$, so that $A_{\theta_0}^{-1} = \Omega(\theta_0) + \Lambda$ by the Woodbury matrix identity.
 
{Note that under $\mathbb{P}$, $m(\theta(\mu)) =\mu$ following the distribution specified by $f_\mu(\cdot)$ corresponding to the local Gaussian prior, while for a given $\mu$, under $\mathbb{P}_\mu$, $\widehat{m}(\theta(\mu))-\mu$ is asymptotically Gaussian with mean zero and variance $T^{-1}\Omega(\theta_0)$. Consequently, the sum of $\widehat{m}(\theta(\mu))-\mu$ and $\mu $ is Gaussian with mean zero and variance $A_{\theta_0}^{-1}$ under $\mathbb{P}$.
Thus by the definition of the quantile $Z_{\alpha}$, the following is true:
\begin{align*}
    \int \mathbb{P}_{\mu}\left(T(\widehat{m}(\theta(\mu))- \mu+\mu))^\top A_{\theta_0} (\widehat{m}(\theta(\mu))- \mu+\mu) \leq Z_\alpha   \right)f_{\mu}(\mu)d\mu =1-\alpha,
\end{align*}
which then implies equation (\ref{eq:lemma 1 linear}) upon noting that $\|\mu - \mu_0 \|=O_p(\frac{\sqrt{q}}{\sqrt{T}})$ in this local misspecification case.  }


 
We now generalize the preceding argument to the nonlocal case under the conditions of Theorem~\ref{Theorem:2}. Here, we have that by definition $PR_T(\alpha)=\left\{\theta:  \int_{\mu}p_T(\theta,\mu)d\mu \geq {Z_{\alpha}} \right\}=\left\{\theta:  p_T(\theta) \geq {Z_{\alpha}} \right\}$ where $Z_\alpha$ is chosen to satisfy
\begin{align*}
     \P_{T,\theta}\left( p_T(\theta) \geq {Z_{\alpha}}\right)  =1-\alpha,
\end{align*}
and $\mathbb{P}_{T, \theta}$ corresponds to the marginal posterior $p_T(\theta)$. With a slight abuse of notation, we treat $p_T(\theta)$ as a function of the sample moments $\widehat{m}(\theta)$, rather than of $\theta$ directly. Define this function as $\tilde{p}_T(\widehat{m}(\theta)) := p_T(\theta)$. Then the above condition can be rewritten as
\begin{align}
\label{eq:lemma 1 --40}
     \P_{T,\theta}\left( \tilde{p}_T(\widehat{m}(\theta)) \geq {Z_{\alpha}}\right)  =1-\alpha. 
\end{align} 

To verify equation~\eqref{eq:lemma1-1}, it suffices to show that
\begin{align}
\label{eq:lemma 1 --41}
\int \mathbb{P}_{\mu} \left( \tilde{p}_T(\widehat{m}(\theta(\mu))) \geq Z_{\alpha} \right) f_\mu(\mu) d\mu = 1 - \alpha + o_p(1).
\end{align}

Similar to the above local case, for a given $\mu$, under $\mathbb{P}_\mu$, 
the distribution of $\widehat{m}(\theta(\mu)) - \mu$ is asymptotically 
Gaussian with mean zero and covariance $T^{-1}\Omega(\theta(\mu))$. 
By the assumption $W(\theta(\mu))^{-1} = \Omega(\theta(\mu))$ and 
Theorem~\ref{Theorem:2}, the conditional posterior distribution of 
$\widehat{m}(\theta) - \mu$ given $\mu$ coincides asymptotically with 
this sampling distribution. Since $\pi(\mu) = f_\mu(\mu)$ by assumption, 
integrating over $\mu$ shows that the marginal posterior distribution 
of $\widehat{m}(\theta)$ matches the joint sampling distribution of 
$\widehat{m}(\theta(\mu))$ under $\mathbb{P}$. 
Thus Equation~(\ref{eq:lemma 1 --41}) follows from~(\ref{eq:lemma 1 --40}).
\end{proof}

\subsection{Proof of Theorem \ref{th3}}
\label{proof: theorem3}
\begin{proof}
~\\~\\
\underline{\textbf{Step 1}}\\
Under the information equality, we have {$\sigma_{\eta,\mu}^2 = (\partial \eta(\theta(\mu))/\partial\theta)^{\top} {J}_W(\theta(\mu))^{-1}(\partial \eta(\theta(\mu))/\partial\theta).$} From Assumption \ref{Gaussian}, the frequentist confidence interval is presented as $$[(\partial \eta(\theta(\mu))/\partial\theta)^{\top}\widehat{\theta}(\mu)+\frac{\sigma_{\eta,\mu} z_{\alpha/2}}{\sqrt{T}}, (\partial \eta(\theta(\mu))/\partial\theta)^{\top}\widehat{\theta}(\mu)+ \frac{z_{1-\alpha/2}\sigma_{\eta,\mu}}{{\sqrt{T}}}].$$

\noindent
\underline{\textbf{Step 2}}\\
As the second step, we need to prove that under the information equality, the confidence interval agrees with the frequentist confidence interval.
As suggested by Theorem \ref{Theorem:2}, $p_T(\theta,\mu)$ can be approximated well by $N_T(\theta,\mu)$, and the conditional distribution of $\theta$ on $\mu$ (with density $N_T(\theta|\mu)$) follows a Gaussian distribution with mean $C_{w,\mu}^{-1}C_{m,\mu}$ and variance $(TC_{w,\mu})^{-1}$. As we noticed, $
\partial \eta(\theta(\mu))/\partial \theta^{\top}(TC_{w,\mu})^{-1}\partial \eta(\theta(\mu))/\partial \theta =\sigma_{\eta,\mu}^2/T$ under information equality.
Thus, it suffices to check that the quantiles of $p_T(\theta,\mu)$ and $N_T(\theta,\mu)$ indeed agree.


Let
\begin{align*}
H_{\eta,T}(s,\mu)=&F_{\eta,T}\left(\eta\left(\theta(\mu)\right)+s{ /\sqrt{T}}\right)=\int_{\theta\in\Theta:\eta\left(\theta\right)\leq \eta\left(\theta(\mu)\right)+s/\sqrt{T}}p_{T}(\theta,\mu)/p_T(\mu)d\theta,\\
\widehat{H}_{\eta,T}(s,\mu)=& \int_{\theta\in\Theta:\eta\left(\theta\right)\leq \eta\left(\theta(\mu)\right)+s /\sqrt{T}}N_{T}(\theta,\mu)/N_T(\mu)d\theta,\\
H_{\eta,\infty}(s,\mu)=&\int_{\theta\in\Theta: (\partial \eta(\theta(\mu))/\partial \theta)^{\top}(\theta - \theta(\mu))\leq s /\sqrt{T}}N_{T}(\theta,\mu)/N_T(\mu)d\theta,
\end{align*}
where $N_T(\mu)=\int_{\theta\in \Theta} N_T(\theta,\mu)d\theta$. By definition of total variation of moments norm and Theorem \ref{Theorem:2}, we have
\[
\sup_{s\in \mathcal{S}(\mu),\mu\in \Gamma}\left|H_{\eta,T}(s,\mu)-\widehat{H}_{\eta,T}(s,\mu)\right|\rightarrow_{p}0,
\]  
where $ \mathcal{S}(\mu)$ denotes the support of $H_{\eta,T}(\cdot,\mu)$ such that $ \mathcal{S}(\mu) = \{s\in \mathbb{R}: s = \sqrt{T} (\eta(\theta)-\eta(\theta(\mu))), \theta \in \Theta\}$.

By the uniform
continuity of the integral of the normal density with respect to the
boundary integration, we have
\[
\sup_{s\in \mathcal{S}(\mu),\mu\in \Gamma}\left|\widehat{H}_{\eta,T}(s,\mu)-H_{\eta,\infty}(s,\mu)\right|\rightarrow_{p}0,
\]
which implies that
\[
\sup_{s\in \mathcal{S}(\mu),\mu\in \Gamma}\left|H_{\eta,T}(s,\mu)-H_{\eta,\infty}(s,\mu)\right|\rightarrow_{p}0.
\] 

The convergence of the distribution function implies the convergence of quantiles at continuous points of distribution functions so that $H_{\eta, T}^{-1}(\alpha,\mu)-H_{\eta,\infty}^{-1}(\alpha,\mu)\rightarrow_{p}0,  $ where $H_{\eta,\infty}^{-1}(\alpha,\mu)$ and $H_{\eta,T}^{-1}(\alpha,\mu)$ are defined as the inverse of the function $H_{\eta,T}(s,\mu)$ in terms of $s$ for any fixed $\mu$.

Next, similar to the proof of Theorem 3 in \cite{chernozhukov2003mcmc} and we have that 
\begin{align*}
H_{\eta,\infty}(s,\mu)=&\mathbb{P}_{N_T(\theta | \mu )}\left\{ (\partial \eta(\theta(\mu))/\partial \theta)^{\top} (\theta-\theta(\mu))\leq {  s/\sqrt{T}}\right\} 
\end{align*}
 so that
$H_{\eta,\infty}^{-1}(\alpha,\mu)=(\partial \eta(\theta(\mu))/\partial \theta)^{\top}\sqrt{T} {U}_T(\mu) +q_{\alpha} \sqrt{(\partial \eta(\theta(\mu))/\partial \theta)^{\top}{J}_W(\mu)^{-1}(\partial \eta(\theta(\mu))/\partial \theta)}$  implied by the proof of Theorem \ref{Theorem:2}, where $q_{\alpha}$ is the $\alpha$-quantile of a standard normal distribution. The rest of the results follow from the fact that $H_{\eta,T}^{-1}(\alpha,\mu)= \sqrt{T}(c_{\eta,T}(\alpha,\mu) -\eta({\theta}(\mu)))$ and the delta method. 

{
Recall that  $f_{\mu}(\mu)= \partial F_{\mu}(\mu)/\partial \mu$ be the density corresponding to $\mathbb{P}^*(.)$.
To prove the second statement, we have, 
\begin{align*}
&\lim_{T\rightarrow\infty}\mathbb{P}^*\left\{ \eta(\theta(\tilde{\mu}))\in \cup_{\mu' \in \mathcal{M}} \text{CI}(\mu'), \forall \tilde{\mu} \in\Gamma \right\}
\\=& \lim_{T\rightarrow\infty}\int_{\mu}\mathbb{P}^*_{\theta(\mu),\mu}\left\{ \eta(\theta(\tilde{\mu}))\in \cup_{\mu' \in \mathcal{M}} \text{CI}(\mu') ,\forall \tilde{\mu} \in\Gamma \right\} f_{\mu}(\mu)d\mu \\
\geq & \lim_{T\rightarrow\infty}\int_{\mu}\mathbb{P}^*_{\theta(\mu),\mu}\left\{ \eta(\theta(\tilde{\mu}))\in  \text{CI}(\tilde{\mu}),\forall \tilde{\mu} \in\Gamma\right\} f_{\mu}(\mu)d\mu= 1-\alpha .
\end{align*}
}

\end{proof}


\subsection{Proof of Theorem \ref{th4}}
\label{proof: theorem4}
\begin{proof} 
Keep $\mu \in \Gamma$ throughout this proof.
In view of Assumption of the consistency of ${\Omega}_{T}(\theta(\mu))$, it suffices to show
that
$\|\widehat{J}_{T}^{-1}\left(\theta(\mu)\right)-J_W^{-1}(\theta(\mu))\|\rightarrow_{p}0,
$ and then conclude using the delta method.  
Let $\zeta_\mu(\theta)$ be a function of $\theta$ such that 
$$\zeta_\mu(\theta)=\sqrt{T}\left(\theta-\theta(\mu)\right)- 
\sqrt{T}\underbrace{{J}_W\left(\theta(\mu)\right)^{-1}{\Delta}_{T,W}\left(\theta(\mu)\right)}_{{U}_{T}(\mu)}, $$
 and the localized quasi-posterior density for $\zeta_\mu(\theta)$ is
\[
p_{T}(\zeta_\mu(\theta),\mu)=\frac{1}{\sqrt{T}}p_{T}\left(\zeta_\mu(\theta)/\sqrt{T}+\theta(\mu)+{U}_{T}(\mu),\mu\right).
\]

And similarly, define
$N_{T}(\zeta_\mu(\theta),\mu)=\frac{1}{\sqrt{T}}N_{T}\left(\zeta_\mu(\theta)/\sqrt{T}+\theta(\mu)+{U}_{T}(\mu),\mu\right).$  {Define $H_T$ for the set of $\zeta_\mu(\theta)$ containing the set $\{\theta:\sqrt{T}\|h(\theta,\mu)\|\leq\varepsilon, 
\theta \in \Theta\}$. }
{Denote $\zeta_\mu(\theta)=\left(\zeta_{\mu,1}(\theta),\ldots,\zeta_{\mu,k}(\theta)\right)$ }and $\widetilde{T}_{T}=\left(\widetilde{T}_{T1},\ldots,\widetilde{T}_{Tk}\right)$
where $\widetilde{T}_{T}=\sqrt{T}\left(\widehat{\theta}(\mu)-\theta(\mu)\right)-\sqrt{T}{U}_{T}(\mu)$.


Note also
\[
\begin{aligned}{\widehat{J}_{T}}^{-1}\left(\widehat{\theta}(\mu)\right) =&\int_{{\Theta}}{T}(\theta-\widehat{\theta}(\mu))(\theta-\widehat{\theta}(\mu))^{\top}p_{T}(\theta,\mu)/p_{T}(\mu)d\theta \IF(\mu:p_{T}(\mu)>c)\\
  {=}&\int_{H_{T}}\left(\zeta_\mu(\theta)-\sqrt{T}\left(\widehat{\theta}(\mu)-\theta(\mu)\right)+\sqrt{T}{U}_{T}(\mu)\right) \\& \cdot\left(\zeta_\mu(\theta)-\sqrt{T}\left(\widehat{\theta}(\mu)-\theta(\mu)\right)+\sqrt{T}{U}_{T}(\mu)\right)^{\top}[p_{T}(\zeta_\mu(\theta){,  \mu})/p_{T}(\mu)]d\zeta_\mu(\theta)\IF(\mu:p_{T}(\mu)>c),
\end{aligned}
\]
 and

\[
\tilde{J}_T^{-1}(\theta(\mu))\equiv\int_{H_T}\zeta_\mu(\theta)\zeta_\mu(\theta)^{\top}p_{T}(\zeta_\mu(\theta),\mu)/p_{T}(\mu)d\zeta_\mu(\theta).
\]

\[
\tilde{J}_{T}^{-1,c}\left(\theta(\mu)\right)\equiv\int_{H_T}\zeta_\mu(\theta)\zeta_\mu(\theta)^{\top}p_{T}(\zeta_\mu(\theta),\mu)/p_{T}(\mu)d\zeta_\mu(\theta)\IF(\mu:p_{T}(\mu)>c).
\]

Therefore we have
$\tilde{J}_{T}^{-1}(\theta(\mu)) - \tilde{J}_{T}^{-1,c}(\theta(\mu)) = \int_{H_T}\zeta_\mu(\theta)\zeta_\mu(\theta)^{\top}[p_{T}(\zeta_\mu(\theta),\mu)/p_{T}(\mu)]d\zeta_\mu(\theta)\IF(\mu:p_{T}(\mu)\leq c).$
Due to the condition that 
 $\int 1_{\{p_T(\mu)\leq c\} } p_{T}(\mu) d\mu =o_p(1)$, thus we have
$\|\tilde{J}_{T}^{-1}(\theta(\mu)) - \tilde{J}_{T}^{-1,c}(\theta(\mu)) \| =o_p(1).$

So ${\widehat{J}}_{T}^{-1}(\widehat{\theta}(\mu)) - \tilde{J}_{T}^{-1,c} (\theta(\mu))=-2\int_{H_T}\zeta_\mu(\theta)\widetilde{T}_{T}^{\top}p_{T}(\zeta_\mu(\theta),\mu)/p_{T}(\mu)d\zeta_\mu(\theta){\IF(\mu:p_{T}(\mu)>c)}+\\ 
\int_{H_T}\widetilde{T}_{T}\widetilde{T}_{T}^{\top}p_{T}(\zeta_\mu(\theta),\mu)/p_{T}(\mu)d\zeta_\mu(\theta)\IF(\mu:p_{T}(\mu)>c)$, which will be verified to be ignorable in statement (c), (d), (e) and (f).

{Denote $\Theta(\zeta_\mu(\theta))$ as the set corresponding to $H_T$}.
Recall that 
{\[
{J}_{W}^{-1}\left(\theta(\mu)\right)= \int_{\Theta(\zeta_\mu(\theta))}\zeta_\mu(\theta)\zeta_\mu(\theta)^{\top}N_{T}(\zeta_\mu(\theta),\mu)/N_T(\mu)d\zeta_\mu(\theta)+ o_p(1).
\]}


So 
\begin{eqnarray*} 
J_{W}^{-1}(\theta(\mu)) - \tilde{J}_{T}^{-1}(\theta(\mu))
=&& \int_{H_T}\zeta_\mu(\theta)\zeta_\mu(\theta)^{\top}(N_{T}(\zeta_\mu(\theta),\mu)/N_{T}(\mu)-p_{T}(\zeta_\mu(\theta),\mu)/p_{T}(\mu))d\zeta_\mu(\theta)\\&&+ \int_{H_T^c}\zeta_\mu(\theta)\zeta_\mu(\theta)^{\top}p_{T}(\zeta_\mu(\theta),\mu)/p_{T}(\mu)d\zeta_\mu(\theta).
\end{eqnarray*}
This is verified in (a) and (b). Denote $\widetilde{T}_{Tj}$ as the $j$th element of $\widetilde{T}_{T}$ ($1\leq j\leq k$).





{
\begin{itemize}
\item[(a)]
$ [\sum_{i,j}\{\int_{H_{T}}\zeta_{\mu,i}(\theta)\zeta_{\mu,j}(\theta)(p_{T}(\zeta_\mu(\theta),\mu)/p_{T}(\mu)-N_{T}(\zeta_\mu(\theta),\mu)/N_{T}(\mu)) d\zeta_\mu(\theta)\}^2]^{1/2}=o_{p}(1)$ results from the following steps. First, note that one vector norm inequality that $|v|_2 = (\sum_i v_i^2)^\frac{1}{2} \leq |v|_1 = (\sum_i |v_i|), v\in \mathbb{R}^d$ implies the above term is upper bounded by $$\sum_{i,j} \left|\int_{H_{T}}\zeta_{\mu,i}(\theta)\zeta_{\mu,j}(\theta)\left|{{(p_{T}(\zeta_\mu(\theta),\mu)/p_{T}(\mu)-N_{T}(\zeta_\mu(\theta),\mu) /N_{T}(\mu)) }}\right|d\zeta_\mu(\theta) \right|.$$ For each absolute term within the summation, we see
by the Cauchy-Schwarz inequality or the H\"{o}lder inequality, 
\begin{align*}
   & \left|\int_{H_{T}}\zeta_{\mu,i}(\theta)\zeta_{\mu,j}(\theta)\left|p_{T}(\zeta_\mu(\theta),\mu)/p_{T}(\mu)-N_{T}(\zeta_\mu(\theta),\mu)/N_{T}(\mu)\right|d\zeta_\mu(\theta)\right|
\\\leq & \left[\int_{H_{T}}\zeta_{\mu,i}(\theta)^2 \left|{{(p_{T}(\zeta_\mu(\theta),\mu)/p_{T}(\mu)-N_{T}(\zeta_\mu(\theta),\mu) /N_{T}(\mu)) }}\right|d\zeta_\mu(\theta)\right]^{1/2}\\ &\left[\int_{H_{T}}\zeta_{\mu,j}(\theta)^2\left|{{(p_{T}(\zeta_\mu(\theta),\mu)/p_{T}(\mu)-N_{T}(\zeta_\mu(\theta),\mu) /N_{T}(\mu)) }}\right|d\zeta_\mu(\theta)\right]^{1/2}.
\end{align*}

{  Then we have, 
\begin{align*} 
&\sum_{i,j} \left[\int_{H_{T}}\zeta_{\mu,i}(\theta)^2\left|{{(p_{T}(\zeta_\mu(\theta),\mu)/p_{T}(\mu)-N_{T}(\zeta_\mu(\theta),\mu) /N_{T}(\mu)) }}\right|d\zeta_\mu(\theta)\right]^{1/2}\\ &\left[\int_{H_{T}}\zeta_{\mu,j}(\theta)^2\left|{{(p_{T}(\zeta_\mu(\theta),\mu)/p_{T}(\mu)-N_{T}(\zeta_\mu(\theta),\mu) /N_{T}(\mu)) }}\right|d\zeta_\mu(\theta)\right]^{1/2}
\\ = &  \sum_{i} \left[\int_{H_{T}}\zeta_{\mu,i}(\theta)^2\left|{{(p_{T}(\zeta_\mu(\theta),\mu)/p_{T}(\mu)-N_{T}(\zeta_\mu(\theta),\mu) /N_{T}(\mu)) }}\right|d\zeta_\mu(\theta)\right]^{1/2}\\ &\sum_j \left[\int_{H_{T}}\zeta_{\mu,j}(\theta)^2\left|{{(p_{T}(\zeta_\mu(\theta),\mu)/p_{T}(\mu)-N_{T}(\zeta_\mu(\theta),\mu) /N_{T}(\mu)) }}\right|d\zeta_\mu(\theta)\right]^{1/2}.
\end{align*}

We know
{
\begin{align*}
   & \sum_{i} \left[\int_{H_{T}}\zeta_{\mu,i}(\theta)^2\left|{{(p_{T}(\zeta_\mu(\theta),\mu)/p_{T}(\mu)-N_{T}(\zeta_\mu(\theta),\mu) /N_{T}(\mu)) }}\right|d\zeta_\mu(\theta)\right]^{1/2} \\ \leq &  k \left[\max_{i} 
 \int_{H_{T}}\zeta_{\mu,i}(\theta)^2\left|{{(p_{T}(\zeta_\mu(\theta),\mu)/p_{T}(\mu)-N_{T}(\zeta_\mu(\theta),\mu) /N_{T}(\mu)) }}\right|d\zeta_\mu(\theta) \right]^{1/2}
\end{align*}
}
}
Then it boils down to bound {$k \max_{i} 
 \int_{H_{T}}\zeta_{\mu,i}(\theta)^2\left|{{(p_{T}(\zeta_\mu(\theta),\mu)/p_{T}(\mu)-N_{T}(\zeta_\mu(\theta),\mu) /N_{T}(\mu)) }}\right|d\zeta_\mu(\theta)$}, which is guaranteed by Theorem \ref{Theorem:2} with $\kappa = 2$.

\item[(b)] 
Next we look at ${\sum_{1 \leq i,j\leq k}}|\int_{H_{T}^{c}}\zeta_{\mu,i}(\theta)\zeta_{\mu,j}(\theta)p_{T}(\zeta_\mu(\theta),\mu)/p_{T}(\mu)d\zeta_\mu(\theta)|=o_{p}(1)$
by definition of $p_{T}(\theta, \mu)$ and $J_{T}\left(\theta(\mu)\right)$
being uniformly nonsingular by Lemma \ref{Gaussintegral}. It suffices to look at $k|\int_{H_{T}^{c}}\sum_i \zeta_{\mu,i}^2(\theta)p_{T}(\zeta_\mu(\theta),\mu)/p_{T}(\mu)d\zeta_\mu(\theta)|$.
This is further upper bounded by 
\begin{align*}
&\left|\int_{H_{T}^{c}}\sum_i \zeta_{\mu,i}^2(\theta)p_{T}(\zeta_\mu(\theta),\mu)/p_{T}(\mu)d\zeta_\mu(\theta)\right|\\  \leq &\left|\int_{H_{T}^{c}}\sum_i \zeta_{\mu,i}^2(\theta)|p_{T}(\zeta_\mu(\theta),\mu)/p_{T}(\mu)- N_{T}(\zeta_\mu(\theta),\mu) /N_{T}(\mu)|d\zeta_\mu(\theta)\right| \\ & {+ \left|\int_{H_{T}^{c}}\sum_i \zeta_{\mu,i}^2(\theta)N_{T}(\zeta_\mu(\theta),\mu) /N_{T}(\mu)]d\zeta_\mu(\theta)\right|  }
\end{align*}



\item[(c)] 
By Assumption \ref{Gaussian}, $\|(\widehat{\theta}(\mu)-\theta(\mu)) - U_T(\mu)\| = o_p(1/\sqrt{T})$, so $\|\widetilde{T}_{T}\| = o_p(1)$.

\item[(d)]
$\int_{H_{T}}\left\|\widetilde{T}_{T}\right\|^{2} N_{T}(\zeta_\mu(\theta),\mu) /N_{T}(\mu) d\zeta_\mu(\theta)=o_{p}(1)$\\
since $\|\widetilde{T}_{T}\|^{2} = o_p(1)$ for fixed $k$ and $\int N_{T}(\zeta_\mu(\theta),\mu)/N_{T}(\mu) d\zeta_\mu(\theta) = 1 + o_p(1)$.

\item[(e)] For all $1\leq i,j\leq k$, since $\widetilde{T}_{Tj} = o_p(1)$ for fixed $k$,
\\$\sum_{i,j}|\int_{H_{T}}\zeta_{\mu,i}(\theta)\widetilde{T}_{Tj}\left|p_{T}(\zeta_\mu(\theta),\mu)/p_{T}(\mu)-N_{T}(\zeta_\mu(\theta),\mu) /N_{T}(\mu)\right|d\zeta_\mu(\theta)|=o_{p}(1)$
by Theorem \ref{Theorem:2}.

\item[(f)] For all $1\leq i,j\leq k$, since $\widetilde{T}_{Tj} = o_p(1)$ for fixed $k$,
$\sum_{i,j}|\int_{H_{T}}\zeta_{\mu,i}(\theta)\widetilde{T}_{Tj} N_{T}(\zeta_\mu(\theta),\mu) /N_{T}(\mu)d\zeta_\mu(\theta)|=o_{p}(1)$
by definition of $N_{T}(\zeta_\mu(\theta),\mu)$ and ${J}_W\left(\theta(\mu)\right)$
being uniformly nonsingular, from which the required conclusion follows.
\end{itemize}}

\end{proof}

\subsection{Link to optimal decision rule}
\label{appendix:optimalrule}

We argue in this section that, similar to the results established in \cite{andrews2022optimal}, under model misspecification, we can also establish that the quasi-posterior can be obtained as the limit of a sequence of posteriors under proper priors, and the resulting quasi-Bayes decision rule can correspond to the pointwise limit of the sequence of Bayes decision rules. 
To align with the setup in \cite{andrews2022optimal}, we do not work directly with $\mathbb{P}_{\mu}$ or focus solely on the parameter $\theta(\mu)$. Instead, we define a fixed coarse reference measure $\mathbb{P}_0$ and consider the parameter pair $(\theta(\mu), \mu)$. While the construction in this subsection may appear tedious, it is primarily introduced to facilitate this match.

\begin{Assumption}(Measure in the limit) \label{assum:(Measure in the limit)} 
There exists $\tilde{Z}_t$, which is a maximal subset of $Z_t$ such that the distribution of $\tilde{Z}_t$ does not change with respect to any $(\theta(\mu),\mu) \in \Xi$.
There exists i.i.d. random sequence $\mu_t$ such that $\E_{\mathbb{P}_{\mu_*}}(\mu_t)= \mu_*$. The random variable $(g(Z_t,\theta(\mu)) - \mu_t, \tilde{Z}_t), t=1,\cdots, T$ follows $\mathbb{P}_0$ that is invariant for all plausible  pairs $(\theta(\mu),\mu)\in \Xi$.
\end{Assumption}
 
The specification of the $\mathbb{P}_0$ 
matches the nonparametric Bayesian idea from \cite{gallant2007statistical} such that we can treat $\mu_t$ as a latent variable capturing the misspecification, and $\mathbb{P}_0$ remains invariant with respect to the plausible characteristics. While $\mathbb{P}_\mu$ is a marginal measure over $Z_t$ given $\mu$, $\mathbb{P}_0$ can then be seen as a coarse joint measure concerning $(Z_t, \mu_t)$ and it may not be enough to deduce a marginal distribution for $Z_t$ or $\mu_t$, which is similar to the case discussed in \cite{gallant2007statistical}. 

Assumption \ref{assum:(Measure in the limit)} essentially assumes that the measure (in the limit) of the sample moment can be decomposed into two parts. The plausible term impacts one part, while the remaining part is a measure invariant to $(\theta(\mu),\mu)$. For example, in the plausible IV model from \cite{conley2012plausibly}, we may choose $\mu_t=  \gamma D_t^\top  D_t, \Tilde{Z}_t=(X_t^{\top}, W_t^{\top}, D_t^{\top})^{\top}$ and thus the distribution of $\sum_t [g(Z_t,\theta(\mu)) - \mu_t]$ does not depend on the plausible term $\mu$.  Additionally, it is reasonable to consider
the deviations from the measure of $((g(Z_t,\theta) - \mu_t)^\top, \Tilde{Z}_t^\top )^\top$, i.e., $\mathbb{P}_0$, and in the spirit of \cite{andrews2022optimal} and \cite{kitamura2013robustness} we consider perturbations in the probability measure in Assumption \ref{assum: Differentiability in quadratic mean}.
The subspace of score functions as 
\[
T_{\mu}\left(\mathbb{P}_{0}\right)=\left\{ f\in T\left(\mathbb{P}_{0}\right):\mathbb{E}_{\mathbb{P}_{0}}\left[f(Z_t) \left(g\left(Z_t, \theta(\mu)\right) - \mu_t\right) \right]=0\right\}.
\]
We define $\bar{m}(\mu,\theta)=\E_{\mathbb{P}_0}(f\left(g\left(Z_t, \theta\right) - \mu_t\right))$, then for $f\in T_{\mu}\left(\mathbb{P}_0\right)$, 
 by design $\bar{m}(\mu,\theta(\mu))=0$. 

\begin{Assumption}(Differentiability in quadratic mean) \label{assum: Differentiability in quadratic mean}  $(g(Z_t,\theta) - \mu_t, \Tilde{Z}_t)$ of size $T$ follows distribution {$\mathbb{P}=\mathbb{P}_{T,f}$}, where the sequence
$\mathbb{P}_{T,f}$ converges to  $\mathbb{P}_{0}$, 
\[
\int\left[\sqrt{T}\left(d\mathbb{P}_{T,f}^{1/2}-d\mathbb{P}_0^{1/2}\right)-\frac{1}{2}fd\mathbb{P}_0^{1/2}\right]^{2}\rightarrow0.
\]
\end{Assumption}

While the score $f$ controls the data distribution, our interest
lies in the plausible pair $(\theta(\mu),\mu)$. Denote $\widehat{g}(\theta,\mu) = \frac{1}{\sqrt{T}}\sum_{t=1}^T [g(Z_t,\theta) -\mu]$.

\begin{Assumption}(GP in the limit) \label{assum:(GP in the limit)} 
Assume that under $\mathbb{P}_{T,f}$, $\widehat{g}(\theta(\mu),\mu)$ weakly converges to a Gaussian process with mean function $\bar{m}(\cdot)$ and covariance function $\Sigma(\cdot)$, where the covariance function is continuous and nonsingular for all pairs $(\theta,\mu)$ and is consistently estimable. 
\end{Assumption}
Additionally, denote by $\tilde{g}(\cdot)$ the limiting Gaussian process of $\widehat{g}(\cdot)$ on $(\theta(\mu),\mu)$ pairs, $\mathcal{C}_\Xi$ the space of continuous functions from $\Xi$ to $\mathbb{R}^{q}$. Let $A$ be any linear functional $A: \mathcal{C}_\Xi \rightarrow \mathbb{R}^{q}$ such that $\operatorname{Cov}(A(\tilde{g}(\cdot)), \tilde{g}(\theta,\mu))$ is nonsingular for all $(\theta,\mu)$. Let $h_\perp(\cdot)=\tilde{g}(\cdot)-\operatorname{Cov}(\tilde{g}(\cdot), \xi_A) \operatorname{Var}(\xi_A)^{-1} \xi_A$ and $\mu_\perp(\cdot)=\bar{m}(\cdot)-\operatorname{Cov}(\tilde{g}(\cdot), \xi_A) \operatorname{Var}(\xi_A)^{-1} A(\tilde{g}(\cdot))$, where $\xi_A=A(\tilde{g}(\cdot))$.   

Analogue to arguments in Sections 3.1.2 from \cite{andrews2022optimal}, the likelihood function $\ell\left(\mu_\perp, \theta(\mu), \mu; \tilde{g}(\cdot)\right)$ based on the observed data $\tilde{g}(\cdot)$ factors as
$$
\ell\left(\mu_\perp, \theta(\mu), \mu; \tilde{g}(\cdot)\right)=\ell\left(\mu_\perp, \theta(\mu), \mu; \xi_A\right) \ell(\mu_\perp; h_\perp),
$$
and thus their equation (8) directly implies that a class of proportional priors over $\mu_\perp$ would give rise to quasi-Bayes as a limiting case.

To establish that quasi-Bayes decision rules are pointwise limits of the corresponding Bayes decision rules, it suffices to verify the conditions of Theorem 8 in \cite{andrews2022optimal}. The non-singularity of the covariance function is ensured, for instance, by Assumption \ref{a3}. The remaining conditions of their Theorem 8 require a loss function $\ell(a, \theta, \mu)$ that is uniformly bounded, continuous, and strictly convex in $a$ for each $(\theta, \mu)$, along with a compact and convex action space $\mathcal{A}$. These requirements are met with appropriately chosen loss functions and action spaces.

\section{Examples}\label{Sec:Example}  

\paragraph{\textbf{Difference from \cite{chen2018monte}}.} 
\label{chen}
Section 5.2 of \cite{chen2018monte} develops their methodology within a GMM framework; however, several of the key underlying assumptions do not necessarily hold in our model setup. In particular, regularity conditions (b)--(c) in Proposition 5.3 impose restrictions on certain quadratic forms of the sample moments, which may be incompatible with our framework. This mismatch arises because  
that we allow \(\widehat{m}(\theta)-\mu\) to exhibit different statistical properties across plausible pairs, provided that its mean remains zero.

We illustrate this point via the following example. In the model with $Y_i=X_i^{\top} \beta_i+D_i^{\top} \gamma_i+U_i$ allowing for heterogeneous treatment effects where $\E\left[D_i U_i\right]=0$, $\E\left[X_i U_i\right] \neq 0$, and $\gamma_i$ and $\beta_i$ are jointly independent of $X_i$ and $D_i$ so that $Y_i=X_i^{\top} \beta+D_i^{\top} \gamma+\varepsilon_i$ satisfies $\E\left[X_i \varepsilon_i\right] \neq 0$ and $\E\left[D_i \varepsilon_i\right]=0$ where $\varepsilon_i=U_i+X_i^{\top}\left(\beta_i-\beta\right)+D_i^{\top}\left(\gamma_i-\gamma\right), \beta=\E\left[\beta_i\right]$, and $\gamma=\E\left[\gamma_i\right]$. In this case, $\beta$ should be interpreted as the average treatment effect of $X$ on $Y$, and $\gamma$ should be interpreted as the average partial effect of the IV $D$ on $Y$, and $(\beta,\gamma)$ are jointly identified.

For convenience, consider that both $\beta$ and $\gamma$ are scalars, and then there is a one-to-one mapping between $\gamma$ and $\mu$ such that $\mu(\gamma) =\E\left[D_i^2\right]\gamma$. Additionally, there is a one to one mapping between $\gamma$ and $\beta$ such that for a given $\gamma$, $\beta(\gamma)=\E\left[D_iY_i-D_i^2\gamma\right] /\E\left[D_iX_i\right] $. Therefore, equivalently, we consider the plausible pair to be $(\beta(\gamma),\gamma)$, whose collection should correspond to the identified set $\Theta_I$ explored in \cite{chen2018monte}. Under this setup, the conditions required in Proposition 5.3 of \cite{chen2018monte} may fail to hold. For instance, the covariance of the sample moment now depends on the chosen plausible pair rather than remaining fixed as $\E[g(Z_i, \theta)^2] = \E[D_i^2 \varepsilon_i^2]$, where $\varepsilon_i$ itself varies with the values of $\beta$ and $\gamma$. This dependence violates condition (b) of Proposition 5.3. Moreover, it directly leads to a violation of condition (c) as well, particularly when $\gamma$ is allowed to vary over a relatively large range within $\Theta_I$.

\paragraph{\textbf{Example I} (Linear IV model)} 
\cite{conley2012plausibly} and \cite{armstrong2021sensitivity} examine the linear IV model with potential misspecification. \cite{conley2012plausibly} relax the exclusion restriction and consider plausible exogenous instrumental variables that may directly impact the dependent variables. They model the misspecification in IV models by introducing a parameter $\gamma$ within the first stage regression, gauging the validity of the exclusion restriction, as illustrated in the following set of equations:
\begin{equation*}
	\begin{aligned}
		& Y=\iota_T \alpha  + (X ~\vdots~  W) \beta+D \gamma+\varepsilon; 
	\end{aligned}
\end{equation*}
where $\beta=(\beta_X^\top, \beta_W^\top)^\top$, $Y$ represents a $T \times 1$ vector of outcomes; $X$ refers to an $T \times p_X$ matrix of endogenous variables, with $\E[X^\top \varepsilon] \neq 0$; $W$ refers to an $T \times p_W$ matrix of exogenous variables; $D$ corresponds to an $T \times k$ matrix of instrumental variables, where  $k \geq p_X$ and $\E\left[D^{\top} \varepsilon\right]=0$; $\theta$ is the parameter of interest and $\gamma$ is the parameter that measures the plausibility of the exclusion restriction. Correspondingly, we have that the plausible characteristics $\mu =T^{-1}\E(D^{\top}D)\gamma$.

We illustrate the $\Gamma$  utilizing the IV population moment condition $g(\theta) = {\mathbb{E} Z_t^\top Z_t}\gamma+\mathbb{E} Z_t^\top (X_t ~\vdots~  W_t)(\beta-\beta_0)$,  which implies that the plausible characteristics will be in at most a subspace spanned by columns of ${\mathbb{E} Z_t^\top Z_t}$ and $\mathbb{E} Z_t^\top  (X_t ~\vdots~  W_t)$  that is a measure-zero subset of $\mathbb{R}^q$ with $q=1+q_W +k$ in case of overidentification with $k>p_X$, and the similar arguments also apply to other linear models. Therefore, the posterior of $\mu$ will shrink to a lower-dimensional space than  $\mathbb{R}^q$.  
  

\paragraph{\textbf{Example II} (Nonlinear instrumental variable quantile regression (IVQR))}
\cite{chernozhukov2005iv} (see also, e.g., \cite{chernozhukov2006instrumental}) propose the IVQR, which effectively estimates the treatment effects at various quantiles via instrumental variable regressions. This approach is empirically appealing and has been used in many recent studies, e.g., \cite{glaeser2015entrepreneurship}. \cite{chernozhukov2003mcmc} discuss one IVQR example for the $\tau$-quantile $(0<\tau<1)$. 

The IVQR method (see, e.g., \cite{chernozhukov2005iv}) requires the validity of the IVs and also assumes rank invariance (or rank similarity); these conditions depend on the accurate specification of the model. Allowing for model misspecification would make the procedure more robust.

\paragraph{\textbf{Example III}  (Two-way fixed effect with heterogeneous treatment effect)}
We consider the settings in \cite{de2020two} such that for every $(i, j, t) \in\left\{1, \ldots, N_{j, t}\right\} \times\{1, \ldots, J\} \times\{1, \ldots, T\}$,  $D_{i, j, t}$ and $\left(Y_{i, j, t}(0), Y_{i, j, t}(1)\right)$ respectively denote the treatment status and the potential outcomes without and with treatment of observation $i$ in group $j$ at period $t$. The outcome of observation $i$ in group $j$ and period $t$ is $Y_{i, j, t}=Y_{i, j, t}\left(D_{i, j, t}\right)$. 
We may thus consider a two-way fixed effect model for the observed outcome of individual $i$ in group $j$ at period
$t$, $Y_{i, j, t}=Y_{i, j, t}(1)D_{i, j,t}+Y_{i, j, t}(0)(1-D_{i,j,t})$ with $D_{i,j,t} =D_{j,t}\in\{0,1\}$, as follows 
\begin{align*}
& Y_{i, j, t}  =\alpha_j+\lambda_t+\delta_{j,t} D_{j,t}+\varepsilon_{i, j, t}.
\end{align*}
where $\mathbb{E} (\varepsilon_{i, j, t}|D_{j,t}, t=1,\cdots, T) =0$ 

Under the constant treatment effect assumption $\delta_{j,t}=\delta$, the part of the moment conditions related to a two-stage estimator proposed in \cite{gardner2022two} that directly involves $\delta$ is
$g(\delta) = \frac{1}{N}\sum_{i,j,t} \E((Y_{i,j, t}- (\alpha_j+ \lambda_t+\delta  D_{j,t}) )D_{j,t})=0$ with $N=\sum_{j,t} N_{j,t}$ denoting the number of total observations.  If the model is misspecified such that $\delta_{j,t}$ varies across groups and time periods, then we only 
have $\E((Y_{i, j, t}- (\alpha_j+\lambda_t+\delta_{j,t}  D_{j,t}) )D_{j,t})=0$ and thus $g(\delta)=\mu$ with $\mu=\frac{1}{N}\sum_{i,j,t} \E((\delta_{j,t}-\delta)  D_{j,t})$. If we denote such a pair in terms of $(\delta(\mu), \mu)$, then it is not hard to see that the moment equation admits the relationship $g(\delta(\mu)) = \mu,$ when  $\mu=\frac{1}{N}\sum_{i, j,t} \E((\delta_{j,t}-\delta(\mu) )  D_{j,t})$.

The examples mentioned above resort to the conventional models when the misspecification term is held constant at zero. Our interest, however, lies in the outcomes achieved when the values of $\mu$ are adjusted to permit a degree of model misspecification.

\subsection{Gaussian Location Model and the AK Interval}\label{AK}
We begin with a simplified Gaussian location model to illustrate the connection between \cite{armstrong2021sensitivity}(AK) confidence intervals and the interval we propose. Consider the \textit{limiting experiment} (cf. \ Equation (13) in AK) under the additional restrictions that $\psi$ is scalar:
\[
Y = \psi + b + \varepsilon, \qquad \varepsilon \sim \mathcal{N}(0,\sigma^2).
\]
Here, $\psi$ is the parameter of interest and $b$ represents misspecification, constrained to lie in $[-M, M]$ for a constant $M$.

In this toy setting, the worst-case bias is $\overline{\mathrm{bias}} = M$ and $|k|=1$. Therefore, the AK minimax fixed-length confidence interval for $\psi$ takes the form
\[
Y \pm \mathrm{cv}_\alpha\!\left(\frac{M}{\sigma}\right)\sigma,
\]
where $\mathrm{cv}_\alpha(c)$ is the $1-\alpha$ quantile of $|Z|$, with $Z \sim \mathcal{N}(c, 1)$.


Define $\delta_M$ as the dirac measure at M.
We now show that this AK interval coincides with a Bayesian credible interval under a suitable least favorable prior. Place the priors
\[
\pi(\psi)\propto 1, \qquad
\pi(b)=\tfrac{1}{2}\delta_{M} + \tfrac{1}{2}\delta_{-M},
\]
where the first choice is the improper flat prior and the second is a symmetric prior supported at the boundary values of the set of the misspecification term $b$.

Under these priors, the posterior distribution of $\psi$ given $Y$ is
\[
\pi(\psi\mid Y)
= \tfrac{1}{2}\mathcal{N}(Y-M,\sigma^2)
+ \tfrac{1}{2}\mathcal{N}(Y+M,\sigma^2).
\]

By symmetry, we look at $(1-\alpha)$ credible set (restricted to connected intervals) centered at $Y$,\footnote{If the posterior is unimodal, this happens to be HPD.} i.e.,
\[
C_{\mathrm{Bayes}}(Y) = [Y - c_B,\; Y + c_B],
\]
where $c_B$ solves $
\tfrac{1}{2}\P(|\varepsilon^{-}-Y|\le c_B|Y) + \tfrac{1}{2}\P(|\varepsilon^{+}-Y|\le c_B|Y)
= 1-\alpha$ with $\varepsilon^+\sim \mathcal{N}(Y+M,\sigma^2)$ and $ \varepsilon^-\sim \mathcal{N}(Y-M,\sigma^2)$. 

Equivalently, $c_B$ solves $
\tfrac{1}{2}\P(|z-M|\le c_B) + \tfrac{1}{2}\P(|z+M|\le c_B)
= 1-\alpha$ with $z\sim \mathcal{N}(0,\sigma^2)$. Since symmetry implies that $\P(|z+M|\le c_B)=\P(|z-M|\le c_B)$, the condition reduces to
\[
\P(|z+M|\le c_B)=1-\alpha.
\]
Dividing by $\sigma$, this is exactly the defining equation for $\mathrm{cv}_{\alpha}(M/\sigma)$, hence
\[
c_B = \mathrm{cv}_{\alpha}\!\left(\frac{M}{\sigma}\right)\sigma.
\]

Therefore, the Bayesian credible interval under this least favorable prior coincides exactly with the AK minimax fixed-length confidence interval.
We now extend this argument to the full limiting experiment considered in Section 4.1 in AK.
{ 
Consider the Gaussian limiting experiment
\[
Y = -\Gamma\theta + \mu + \Sigma^{1/2}\varepsilon,\qquad
\varepsilon \sim \mathcal{N}(0,I_q),
\]
where
\[
\theta \in \mathbb{R}^k,\qquad
\Gamma \in \mathbb{R}^{q\times k},\qquad
\Sigma \in \mathbb{R}^{q\times q},\qquad
\mu \in \mathcal{C} \subset \mathbb{R}^q,
\]
and $\mathcal{C}$ is convex and centrosymmetric. Suppose the parameter of interest is of the form $\psi = H\theta$ with $H \in \mathbb{R}^{1\times k}.$ AK considers the linear estimator $w^\top Y$, with $H=-w^\top \Gamma$ and confidence intervals of the form \[
w^\top Y \pm
\mathrm{cv}_\alpha\!\left(
\frac{\overline{\mathrm{bias}}_{\mathcal{C}}(w)}{\sqrt{w^\top\Sigma w}}
\right)
\sqrt{w^\top\Sigma w}.
\]
We thus also consider inference based upon the statistic $w^\top Y$, and show, step by step, that for a given $w$, the AK interval coincides exactly with a QBP credible interval under a suitable least favorable prior.

The misspecification term $\mu$ lies in the known set $\mathcal{C}$, and the worst-case bias of $w^\top Y$ for $\psi$ is
\[
\overline{\mathrm{bias}}_{\mathcal{C}}(w)
:= \sup_{\mu\in\mathcal{C}} |w^\top \mu|.
\]
Convexity and centrosymmetry of $\mathcal{C}$ ensure the existence of least favorable misspecification directions $\pm M_w \in \mathcal{C}$ such that
\[
w^\top M_w = \sup_{\mu\in\mathcal{C}} |w^\top\mu|
= \overline{\mathrm{bias}}_{\mathcal{C}}(w).
\]

As in the scalar model, place an improper flat prior on $\psi$ and the symmetric two-point prior 
$\pi_w(\mu) = \tfrac12 \delta_{M_w} + \tfrac12 \delta_{-M_w}
$ on the misspecification term $\mu$.
Implied directly by the model, we have that 
%
\[
w^{\top}Y \mid \psi,\mu
\sim
\mathcal{N}\!\big(\psi + w^{\top}\mu,\; w^{\top}\Sigma w\big).
\]
Hence, this posterior for $\psi$ is thus a symmetric mixture of two normal distributions with means
\[
m_- := w^\top Y - w^\top M_w
\quad\text{and}\quad
m_+ := w^\top Y + w^\top M_w.
\] 

By symmetry, we look at $(1-\alpha)$ credible set (restricted to connected intervals) centered at $w^\top Y$, 
\[
C_{\mathrm{Bayes}}(w^\top Y) = [w^\top Y - c_B,\; w^\top Y + c_B],
\]
and the half-length $c_B$ must then satisfy
\[
\frac{c_B}{\sqrt{w^\top\Sigma w}}
= \mathrm{cv}_\alpha\!\left(
\frac{\overline{\mathrm{bias}}_{\mathcal{C}}(w)}{\sqrt{w^\top\Sigma w}}
\right),
\]
and hence $
	c_B
	= \mathrm{cv}_\alpha\!\left(
	\frac{\overline{\mathrm{bias}}_{\mathcal{C}}(w)}{\sqrt{w^\top\Sigma w}}
	\right)\sqrt{w^\top\Sigma w}$, which leads to the same interval as AK. 
}

\clearpage

\setcounter{section}{0}\setcounter{subsection}{0}\setcounter{subsubsection}{0}
\setcounter{equation}{0}\setcounter{figure}{0}\setcounter{table}{0}
\setcounter{theorem}{0}\setcounter{lemma}{0}\setcounter{corollary}{0}
\makeatletter
\@ifundefined{c@Theorem}{}{\setcounter{Theorem}{0}}
\@ifundefined{c@Lemma}{}{\setcounter{Lemma}{0}}
\@ifundefined{c@Corollary}{}{\setcounter{Corollary}{0}}
\@ifundefined{c@Assumption}{}{\setcounter{Assumption}{0}}
\@ifundefined{c@algocf}{}{\setcounter{algocf}{0}}
\makeatother

\renewcommand{\thesection}{OM.\arabic{section}}
\renewcommand{\theequation}{OM.\arabic{equation}}
\renewcommand{\thefigure}{OM.\arabic{figure}}
\renewcommand{\thetable}{OM.\arabic{table}}
\renewcommand{\thecorollary}{OM.\arabic{corollary}}
\renewcommand{\thelemma}{OM.\arabic{lemma}}
\renewcommand{\thetheorem}{OM.\arabic{theorem}}
\renewcommand{\theAssumption}{OM.\arabic{Assumption}}
\renewcommand{\thealgocf}{OM.\arabic{algocf}}
\providecommand{\theHsection}{\thesection}\renewcommand{\theHsection}{OM.\arabic{section}}
\providecommand{\theHequation}{\theequation}\renewcommand{\theHequation}{OM.\arabic{equation}}
\providecommand{\theHfigure}{\thefigure}\renewcommand{\theHfigure}{OM.\arabic{figure}}
\providecommand{\theHtable}{\thetable}\renewcommand{\theHtable}{OM.\arabic{table}}
\providecommand{\theHtheorem}{\thetheorem}\renewcommand{\theHtheorem}{OM.\arabic{theorem}}
\providecommand{\theHlemma}{\thelemma}\renewcommand{\theHlemma}{OM.\arabic{lemma}}
\providecommand{\theHcorollary}{\thecorollary}\renewcommand{\theHcorollary}{OM.\arabic{corollary}}
\providecommand{\theHAssumption}{\theAssumption}\renewcommand{\theHAssumption}{OM.\arabic{Assumption}}

\begin{center}
{\Large\bfseries Online Materials for\\ ``Plausible GMM: A Quasi-Bayesian Approach''\par}
\vspace{1em}
{\normalsize Victor Chernozhukov$^a$, Christian B. Hansen$^b$, Lingwei Kong$^c$, Weining Wang$^d$\par}
\vspace{0.5em}
{\footnotesize $^a$ Department of Economics, Massachusetts Institute of Technology.\\
$^b$ Booth School of Business, The University of Chicago.\\
$^c$ Department of Economics, Econometrics and Finance, University of Groningen.\\
$^d$ Department of Economics, University of Bristol.}
\end{center}
\thispagestyle{empty}
\clearpage
\pagestyle{plain}
 
\section{Appendix: verification of assumptions}

\subsection{Verification of Assumption \ref{Gaussian}.}

We impose a high-level expansion assumption in Assumption \ref{Gaussian}. 
We extend the method of obtaining an expansion of a CUE estimator from \cite{newey2004higher} to our high-dimensional misspecified moment settings to verify that assumption. Corollary \ref{coro:Gaussian} delivers this expansion.

\begin{corollary}\label{coro:Gaussian} 
Suppose that the following assumptions hold with $\xi>2$:
\begin{enumerate} 
    \item Assumptions \ref{a1}, Assumption \ref{assum33};
    \item $g(Z_t,\theta)$ is continuous in $\theta \in \Theta$ with probability one;
    \item    for a given $\mu \in \mathcal{M}$, 
    $\sup_{1\leq u \leq q, \theta\in\Theta}\mathbb{E}( \|e_u^\top (g(Z_t, \theta) -\mu)\|^{\xi}) <C $,  $\sup_{\theta\in\Theta}\mathbb{E}( \|(g(Z_t, \theta) -\mu)(g(Z_t, \theta) -\mu)^\top \|^{\xi})< Cq^{\xi/2}$ for some positive constant  $C$, where $e_u$ denotes a $q\times 1$ vector with the $u$th element of $e_u$ being one and the rest entries are zero;
    \item Define $\|.\|_F$ as the Frobenius norm.
    For each given $\mu \in \Gamma$,  $\mathbb{E}\|(g(Z_t, \theta(\mu)) - \mu)(g(Z_t, \theta(\mu)) - \mu)^\top - \Omega(\theta(\mu),\mu) \|_F^2 = O(q)$, 
    {$\mathbb{E}\|g(Z_t, \theta(\mu)) - \mu\|^\xi \lesssim q^{\xi/2} $},  
    {$\mathbb{E}\|(g(Z_t, \theta(\mu)) - \mu)(g(Z_t, \theta(\mu)) - \mu)^\top - \Omega(\theta(\mu),\mu) \|_F^\xi \lesssim q^{\xi/2} $};
         \item $g(Z_t,\theta)$ is 
         continuously differentiable, and
         $G_t(\theta)=\partial g_t(\theta) / \partial \theta$;
{          \item   there exists a neighborhood, $\mathcal{N}(\theta(\mu))$, of $\theta(\mu)$ such that for each {$\theta \in \mathcal{N}(\theta(\mu))$}: there exists a non-negative scalar random variable $b(Z_t)$ such that  $\| g(Z_t, \theta)-g(Z_t, \theta(\mu))  \|\leq b(Z_t) \|\theta -\theta(\mu)\|$,  $\|\partial (g(Z_t, \theta)-\mu) /\partial \theta -\partial (g(Z_t, \theta(\mu))-\mu) /\partial \theta \|\leq b(Z_t) \|\theta -\theta(\mu)\|$,  and $\mathbb{E}[ (b(Z_t))^2 ]<\infty$; 
        } 
\item

$\|\sup_{\theta \in \Theta}(T^{-1}\sum_t \{G_t({\theta}) -\E G_t(\theta)\})\| = o_p(1)$, 
       there exists $C>0$ such that 
        $1/C\leq  \lambda_{\min}(\mathbb{E}[(\partial (g(Z_t, \theta)-\mu)/ \partial \theta^\top)]^\top 
\mathbb{E} [\partial (g(Z_t, \theta)-\mu)/ \partial \theta^\top ])  \leq \lambda_{\max}(\mathbb{E}(\partial (g(Z_t, \theta)-\mu)/ \partial \theta^\top)^\top \mathbb{E}\partial (g(Z_t, \theta)-\mu)/ \partial \theta^\top )  \leq C$;

        \item  
        \label{item-coro appendix: }
        {There exists a constant $\xi_\lambda \in (\xi^{-1} , 1/2)$ such that $ \max \{q  T^{-\xi_\lambda + \xi^{-1}}, q T^{\xi_\lambda - 1/2}\} \rightarrow 0$. }
\end{enumerate}
Then, the following expansion (\ref{eq:expansion}) holds for any fixed $\mu \in \Gamma$,
\begin{equation} \small 
\left\|\widehat{\theta}(\mu)-\theta(\mu)-(G(\theta(\mu))^{\top }\Omega(\theta(\mu),\mu)^{-1} G(\theta(\mu)) )^{-1}G(\theta(\mu))^{\top}\Omega(\theta(\mu),\mu)^{-1} \left(\widehat{m}\left(\theta(\mu)\right)-\mu\right)\right\|  =o_p( {q} T^{-1/2}),\label{eq:expansion}
\end{equation}
where we consider 
$\widehat{\theta}(\mu) =\argmin_{\theta \in \Theta} \{\widehat{m}(\theta)-\mu\}^{\top}{\widehat{\Omega}(\theta, \mu)^{-1}} \{\widehat{m}(\theta)-\mu\},$
 {$\Omega(\theta(\mu), \mu)=\mathbb{E}[(g(Z_t, \theta(\mu)) - \mu )(g(Z_t, \theta(\mu))-\mu)^\top]$ and $\widehat{\Omega}(\theta(\mu), \mu)=T^{-1}\sum_{t=1}^T[(g(Z_t, \theta(\mu)) - \mu )(g(Z_t, \theta(\mu))-\mu)^\top]$}.
Thus, Assumption \ref{Gaussian} is satisfied for this case once the distributional property of $(\widehat{m}(\theta)-\mu)$ leads to the Gaussian limiting distribution, which hold under conditions indicated by, e.g., the central limit theorem proposed in \cite{francq2005central}. 
\end{corollary}
Before the discussion, we first claim several lemmas employed in proving Corollary \ref{coro:Gaussian}, which correspond to Lemmas A1-A3 in \cite{newey2004higher} and are adjusted to fit our conditions. Furthermore, we let $\rho(v)$ be a quadratic scalar function such that $\rho(v) = - v - \frac{1}{2} v^2$ and let $ \rho_j(v)=\partial^j \rho(v) / \partial v^j$, and we let $$\widehat{P}_\mu(\theta,\lambda)=\frac{1}{T}\sum_{t=1}^T \rho(\lambda^{\top} (g(Z_t, \theta) -\mu))= -\lambda^\top \frac{1}{T}\sum_{t=1}^T (g(Z_t, \theta) -\mu)  -\frac{1}{2}\lambda^\top \frac{1}{T}\sum_{t=1}^T (g(Z_t, \theta) -\mu))(g(Z_t, \theta) -\mu)^\top \lambda.$$ 
The estimator for $\theta(\mu)$ for a given $\mu\in \Gamma$ then corresponds to a saddle point problem such that $$\widehat{\theta}(\mu) =\argmin_{\theta\in \Theta} \max_{\lambda \in {{\widehat{\Lambda}_{T,\mu}}}(\theta)  } \widehat{P}_\mu(\theta,\lambda),$$ with $ {{\widehat{\Lambda}_{T,\mu}}}(\theta)  = \{\lambda: \lambda^{\top} (g(Z_t, \theta)-\mu)\in \Lambda_g, 1\leq t\leq T \}$ and $\Lambda_g$ being an open interval containing zero. We use the following notations for a given $\mu \in \Gamma$. We denote  $g_t(\theta)=g(Z_t, \theta)$ and $\widehat{g}_{\mu, t}=g_t(\widehat{\theta}(\mu))$.
Let ${\upsilon}_{\theta\lambda}$ be an augmented parameter vector such that ${\upsilon}_{\theta\lambda} = (\theta^\top, \lambda^\top)^\top$ and for a given $\mu \in \Gamma$, ${\upsilon}_{\theta(\mu)0} = (\theta(\mu)^\top, 0^\top)^\top$. Let $v_t(\mu, {\upsilon}_{\theta\lambda} )=\lambda^{\top} (g_t(\theta)-\mu)$ and $v_{t,1}(\mu, {\upsilon}_{\theta\lambda} )=\partial v_t(\mu, {\upsilon}_{\theta\lambda} ) / \partial {\upsilon}_{\theta\lambda} =\left(\lambda^{\top} G_t(\theta), (g_t(\theta)-\mu)^{\top}\right)^{\top}$. 
Denote $ m_\mu\left(Z_t, {\upsilon}_{\theta\lambda} \right)=\partial \rho\left(v_t(\mu, {\upsilon}_{\theta\lambda} )\right)/\partial {\upsilon}_{\theta\lambda}  =\rho_1\left(v_t(\mu, {\upsilon}_{\theta\lambda} )\right)v_{t,1}(\mu, {\upsilon}_{\theta\lambda} ),$ $m_{\mu,t}({\upsilon}_{\theta\lambda} )=m_\mu\left(Z_t, {\upsilon}_{\theta\lambda} \right)$, $M_\mu=\mathbb{E}\left[ \partial m_{\mu, t}\left( {\upsilon}_{\theta(\mu)0}\right) /\partial {\upsilon}_{\theta\lambda}^\top  \right]$, $\psi_\mu(Z_t)=-M_\mu^{-1} m_\mu\left(Z_t, {\upsilon}_{\theta\lambda}\right)$ and $\widehat{\psi}_\mu= \sum_{t=1}^{T} \psi_\mu(Z_t)/\sqrt{T}$. By construction, we have $M_\mu= -\left(\begin{matrix}
    0 & G(\theta(\mu))^\top \\ G(\theta(\mu)) & \Omega(\theta(\mu),\mu)
\end{matrix} \right)$. 

\begin{lemma}(Generalised High dimensional Fuk-Nagaev Inequality)
\label{l2cont} 
 Consider i.i.d. centered $X_1, ..., X_n$ in $\mathbb{R}^k$. Let $\Sigma := \E[X_1X_1^{\top}]$ and $\omega := \diag(\Sigma)$. Assume that $\E\lVert X_1\rVert^\xi < \infty$ for some $\xi > 2$, then we have for any $t > 0$,
\begin{equation}
\mathrm{P}\bigg\{\big\lVert\sum_{i=1}^n X_i\big\rVert \geq 2\sqrt{n|\omega|_1} + t\bigg\} \leq C_\xi\frac{n\mathrm{E}\lVert X_1\rVert^\xi}{t^\xi} + \exp\bigg(-\frac{t^2}{3n \lambda_{\max}(\Sigma)}\bigg),
\end{equation}
where $C_\xi > 0$ is a constant only depending on $\xi$.
\end{lemma}
\begin{proof}[Proof of Lemma \ref{l2cont}] See Theorem 3.1 in \cite{einmahl2008characterization}. We apply Theorem 3.1 therein with $(B,\|\cdot\|)=\left(\mathbb{R}^k,\|\cdot\|\right)$ where $\eta=\delta=1$.  The unit ball of the dual of $\left(\mathbb{R}^k,\|\cdot\|\right)$ is the set of linear functions $\left\{x=\left(x_1, \ldots, x_k\right)^T \mapsto \sum_{j=1}^k \lambda_j x_j: (\sum_{j=1}^k\left|\lambda_j\right|^2)^{1/2} \leq 1\right\}$, and for $\lambda_1, \ldots, \lambda_k$ with $(\sum_{j=1}^k\left|\lambda_j\right|^2)^{1/2} \leq 1$, with the following step,
$$
\begin{gathered}
\sum_{i=1}^n \mathbb{E}\left[\left(\sum_{j=1}^k \lambda_j X_{i j}\right)^2\right]= \sum_{i=1}^n \E\left[\lambda^{\top}X_iX_i^{\top}\lambda\right]
\leq n \|\lambda\|^2 \lambda_{\max}(\Sigma)\leq n  \lambda_{\max}(\Sigma).
\end{gathered}
$$
 Hence in this case, $\Lambda_n^2$ in Theorem 3.1 \cite{einmahl2008characterization} is bounded by  $n  \lambda_{\max}(\Sigma)$. Additionally, by Jensen's inequality, $\sqrt{n|\omega|_1} \geq \mathbb{E}\lVert\sum_{i=1}^n X_i\rVert$. Therefore, Lemma \ref{l2cont} is implied by Theorem 3.1 \cite{einmahl2008characterization}.
\end{proof}

\begin{lemma}\label{lemA1}
    Under the assumptions of Corollary \ref{coro:Gaussian}, let $b_t=\sup_{\theta\in \Theta}\|g(Z_t, \theta) -\mu \|$ for a $\mu\in \Gamma$, then $\max_{1\leq t\leq T}b_t =O_p( {q}^{{1}/{2}} T^{{1}/{\xi}})$. 
\end{lemma}
\begin{proof}[Proof of Lemma \ref{lemA1}]~\\
Denote $b_{tu}=\sup_{\theta\in\Theta} \|e_u^\top(g(Z_t, \theta) -\mu) \|$, and thus
    $ b_t\leq  \max_{1\leq u\leq q}\sqrt{q} b_{tu}$. The assumption that $\max_{1\leq u\leq q}\mathbb{E}(b_{tu}^{\xi}) <C $ and the Markov inequality imply that $\max_{1\leq u\leq q}\max_{1\leq t\leq T} b_{tu}=O_p(T^{1/\xi})$, and thus the conclusion follows.
\end{proof}

\begin{lemma}\label{lemA2}
    Under the assumptions of Corollary \ref{coro:Gaussian}, $\sup_{\theta\in \Theta, \lambda \in \Lambda_{T,s_\lambda}, 1\leq t\leq T}\|\lambda^\top (g(Z_t, \theta) -\mu)\| = O_p(s_\lambda  {q}^{\frac{1}{2}} T^{-\xi_\lambda+\frac{1}{\xi}}) $ with $\Lambda_{T,s_\lambda} = \{\lambda: \|\lambda \| \lesssim
 s_\lambda T^{-\xi_\lambda} \}$ and $s_\lambda>0$.  
\end{lemma}
\begin{proof}[Proof of Lemma \ref{lemA2}]
    This is a direct result of Lemma \ref{lemA1} and Cauchy-Schwarz inequality such that $\sup_{\theta\in \Theta, \lambda \in \Lambda_{T,s_\lambda}, 1\leq t\leq T}\|\lambda^\top (g(Z_t, \theta) -\mu)\| \leq s_\lambda T^{-\xi_\lambda}\max_{1\leq t\leq T}b_t =s_\lambda T^{-\xi_\lambda} O_p( {q}^{\frac{1}{2}} T^{ \frac{1}{\xi}})  $. 
\end{proof}

\begin{lemma}\label{lemA31}
    Under the assumptions of Corollary \ref{coro:Gaussian}, 
    $\bar{\lambda} = \argmax_{\lambda \in \Lambda_{T,q^{1/2}}}\widehat{P}_{\mu}(\theta(\mu),\lambda)$ exists {with probability approaching one} such that $\bar{\lambda} =O_p(q^{1/2}T^{-1/2})$ and $ \max_{\lambda \in \Lambda_{T,q^{1/2}}}\widehat{P}_{\mu}(\theta(\mu),\lambda) =O_p(qT^{-1}) $ with $\Lambda_{T,q^{1/2}} =  \{\lambda: \|\lambda \| \lesssim
 q^{1/2} T^{-\xi_\lambda} \}$ by setting $s_\lambda =q^{1/2}$ in Lemma \ref{lemA2}.  
\end{lemma}
\begin{proof}[Proof of Lemma \ref{lemA31}]
$ \|\frac{1}{{T}} \sum_t (g(Z_t, \theta(\mu))-\mu)(g(Z_t, \theta(\mu))-\mu)^\top  - \Omega(\theta(\mu), \mu)\|$
\\
$= O_p(q^{1/2}T^{-1/2}) =o_p(1)$ by Lemma \ref{l2cont} and the assumptions  $\mathbb{E}\|(g(Z_t, \theta(\mu)) - \mu)(g(Z_t, \theta(\mu)) - \mu)^\top - \Omega(\theta(\mu), \mu) \|_F^2 = O(q)$,  {$\mathbb{E}\|(g(Z_t, \theta(\mu)) - \mu)(g(Z_t, \theta(\mu)) - \mu)^\top - \Omega(\theta(\mu), \mu) \|_F^\xi \lesssim q^{\xi/2}$}. By the global concavity of $\widehat{P}_{\mu}(\theta(\mu),\lambda)$ we know $\bar{\lambda}$ exists and a first-order Taylor expansion around zero for $\widehat{P}_{\mu}(\theta(\mu),\lambda)$ gives,
\begin{align*} 
0=\widehat{P}_{\mu}(\theta(\mu),0)&\leq \widehat{P}_{\mu}(\theta(\mu),\bar{\lambda})  =  -\bar{\lambda}^\top (\widehat{m}(\theta(\mu))-\mu) - \frac{1}{2} \bar{\lambda}^\top (\frac{1}{T} \sum_t (g_t(\theta(\mu))-\mu)(g_t(\theta(\mu))-\mu)^\top )   \bar{\lambda}  \\ 
& \leq    -\bar{\lambda}^\top (\widehat{m}(\theta(\mu))-\mu) - \frac{1}{2} (\lambda_{\min}(\Omega(\theta(\mu), \mu)) +o_p(1)) \|\bar{\lambda}\|^2 \lesssim_p  \|\bar{\lambda}\| \| \widehat{m}(\theta(\mu))-\mu\| - C \|\bar{\lambda}\|^2,
\end{align*}
where $C$ is a positive constant. We know $
\|\bar{\lambda}\| \leq \| \widehat{m}(\theta(\mu))-\mu\|/C =O_p(q^{1/2} T^{-1/2}),$    
and thus {$
  \sup_{\mu \in \Gamma}\max_{\lambda \in \Lambda_{T,q^{1/2}}}\widehat{P}_{\mu}(\theta(\mu),\lambda) =O_p(qT^{-1}).$}
\end{proof} 
Lemmas \ref{lemA1}-\ref{lemA31} are intermediate results used for the proof of Lemma \ref{lem15}. Lemma \ref{lemA31} provides initial results for the objective function evaluated at the true value $\theta(\mu)$, the rate of which serves as an upper bound in order to derive the rate of  $\|\widehat{m}(\widehat{\theta}(\mu)) -\mu\|$ stated in Lemma \ref{lem15}. The rate of $\|\widehat{m}(\widehat{\theta}(\mu)) -\mu\|$ is instrumental in deducing a rate of $\|\widehat{\theta}(\mu) -\theta(\mu)\|$ in Lemma \ref{lem16} and thus the consistency result in Lemma \ref{lem17}, leading to the eventual asymptotic normality. 
\begin{lemma}\label{lem15}
    Under the assumptions of Corollary \ref{coro:Gaussian}, $\|\widehat{m}(\widehat{\theta}(\mu)) -\mu\| =O_p(qT^{-1/2})$ for any $\mu \in \Gamma$.  
\end{lemma}
\begin{proof}[Proof of Lemma \ref{lem15}]
Since we have a saddle point problem, 
\begin{align*}
\widehat{P}_{\mu}(\widehat{\theta}(\mu), \widehat{\lambda}_g(\epsilon_T) )   \leq  \widehat{P}_{\mu}(\widehat{\theta}(\mu),\widehat{\lambda})  \leq   \max_{\lambda \in \Lambda_{T, q^{1/2}}}  \widehat{P}_{\mu}(\theta(\mu), {\lambda}),  
\end{align*}
where $\widehat{\lambda}_g(\epsilon_T)=-\epsilon_T (\widehat{m}(\widehat{\theta}(\mu))-\mu)$ and $\epsilon_T$ is a positive scalar. We have from Lemma \ref{lemA31} that  $\widehat{P}_{\mu}(\widehat{\theta}(\mu),\widehat{\lambda}) \leq   \max_{\lambda \in \Lambda_{T, q^{1/2}}}  \widehat{P}_{\mu}(\theta(\mu), {\lambda})  =  O_p(qT^{-1}).$ For the term $\widehat{P}_{\mu}(\widehat{\theta}(\mu), \widehat{\lambda}_g(\epsilon_T))$ we have that 
\begin{align*}
\widehat{P}_{\mu}(\widehat{\theta}(\mu),\widehat{\lambda}_g(\epsilon_T)) =\epsilon_T\|\widehat{m}(\widehat{\theta}(\mu)) -\mu\|^2 -1/2 (\widehat{\lambda}_g(\epsilon_T))^\top(\frac{1}{T} \sum_t { (g_t(\widehat{\theta}(\mu))-\mu)} (g_t(\widehat{\theta}(\mu))-\mu)^\top ) \widehat{\lambda}_g(\epsilon_T),
\end{align*}
and thus the fact that $$1/2 (\widehat{\lambda}_g(\epsilon_T))^\top(\frac{1}{T} \sum_t { (g_t(\widehat{\theta}(\mu))-\mu)} (g_t(\widehat{\theta}(\mu))-\mu)^\top) \widehat{\lambda}_g(\epsilon_T)\leq q(\sum_{t=1} b_{t}^2/T)\| \widehat{\lambda}_g(\epsilon_T)\|^2 $$ implied by the proof of Lemma \ref{lemA1} leads to 
\begin{align*}
& \epsilon_T\|\widehat{m}(\widehat{\theta}(\mu)) -\mu\|^2 - O_p(q)\| \widehat{\lambda}_g(\epsilon_T)\|^2 = \epsilon_T\|\widehat{m}(\widehat{\theta}(\mu)) -\mu\|^2 - O_p(q) \epsilon_T^2\|\widehat{m}(\widehat{\theta}(\mu)) -\mu\|^2 
\leq     O_p(qT^{-1})
\end{align*}
{
If we choose $\epsilon_T =T^{-\xi_\lambda}/\|\widehat{m}(\widehat{\theta}(\mu)) -\mu\|$ so that $\widehat{\lambda}_g(T^{-\xi_\lambda}/\|\widehat{m}(\widehat{\theta}(\mu)) -\mu\|)\in \Lambda_{T,q^{1/2}} \cap \widehat{\Lambda}_{T}(\widehat{\theta}(\mu))$ w.p.a.1., the above inequality implies that 
\begin{align*}
     \|\widehat{m}(\widehat{\theta}(\mu)) -\mu\| 
\leq  O_p(q)  O_p(T^{-\xi_\lambda}) +   O_p(qT^{\xi_\lambda-1}).
\end{align*}
Hence, $\|\widehat{m}(\widehat{\theta}(\mu)) -\mu\| =  o_p(T^{-1/\xi})$ and from the assumption $q  T^{-\xi_\lambda + \xi^{-1}}\rightarrow 0$, we know that $\widehat{\lambda}_g(\epsilon_T)$ is an interior point in $\Lambda_{T,q^{1/2}}$ and in $\widehat{\Lambda}_{T}(\widehat{\theta}(\mu))$  w.p.a.1. as long as $\epsilon_T \rightarrow 0$. Now consider one arbitrary drifting to zero sequence $\epsilon_T$, $\epsilon_T\|\widehat{m}(\widehat{\theta}(\mu)) -\mu\|^2 - O_p(q)\epsilon_T^2\|\widehat{m}(\widehat{\theta}(\mu)) -\mu\|^2 =  O_p(qT^{-1})$ implies that  $\epsilon_T \|\widehat{m}(\widehat{\theta}(\mu)) -\mu\|^2 =  O_p(qT^{-1})$ if $\epsilon_T=o_p(q^{-1})$, since this is true for all such sequences, we know $\|\widehat{m}(\widehat{\theta}(\mu)) -\mu\|=O_p(qT^{-1/2})$.
}
\end{proof} 
\begin{lemma}\label{lem16}
    Under the assumptions of Corollary \ref{coro:Gaussian}, $\| \widehat{\theta}(\mu) -\theta(\mu)\| =O_p(qT^{-1/2})$. 
\end{lemma}
\begin{proof}[Proof of Lemma \ref{lem16}]

The differentiability assumption implies that 
$$\widehat{m}(\widehat{\theta}(\mu)) -\widehat{m}({\theta}(\mu)) = T^{-1}\sum_t G_t(\overline{\theta}
(\mu))^{\top}(\widehat{\theta}(\mu) -{\theta}(\mu) ),$$
where $$\|\bar{\theta}(\mu)-{\theta}(\mu)\|\leq \|\widehat{\theta}(\mu)-{\theta}(\mu)\|.$$ Denote $G(\theta) = \mathbb{E}G_t(\theta)$,  and the assumption  $
\|\sup_{\theta\in \Theta}(T^{-1}\sum_t G_t({\theta})^{\top}-\E G_t(\theta)^\top)\| =o_p(1)$ leads to
$\widehat{m}(\widehat{\theta}(\mu)) -\widehat{m}({\theta}(\mu)) = (G(\bar{\theta}
(\mu))+ o_p(1))^{\top}(\widehat{\theta}(\mu) -{\theta}(\mu) ),$ {where $o_p(1)$ is defined in terms of $|.|_2$ norm}.
Thus the assumption that $\lambda_{\min} (G(\bar{\theta}
(\mu))^{\top}G(\bar{\theta}
(\mu)))$ is bounded away from zero implies that
$$\|(\widehat{\theta}(\mu) -{\theta}(\mu) )\|\leq C\|\widehat{m}(\widehat{\theta}(\mu)) -\widehat{m}({\theta}(\mu))\|\leq C\{\|\widehat{m}(\widehat{\theta}(\mu))-\mu\| +\|\widehat{m}({\theta}(\mu))-\mu\|\},$$
then the rate follows as $\|\widehat{m}(\widehat{\theta}(\mu))-\mu\|=O_p(qT^{-1/2})$ indicated by Lemma \ref{lem15} and $\|\widehat{m}({\theta}(\mu))-\mu\|=O_p(q^{1/2}T^{-1/2})$ indicated by Lemma \ref{l2cont}. 
\end{proof}

Consider a first-order Taylor expansion of $\sum_t m(Z_t,{\upsilon}_{\theta\lambda} )/T$ for $\widehat{\upsilon}_{\theta\lambda} =\left(\widehat{\theta}(\mu)^{\top}, \widehat{\lambda}^{\top}\right)^{\top}$ and $ {\upsilon}_{\theta(\mu)0}=\left(\theta(\mu)^{\top}, 0^{\top}\right)^{\top}$, we have $0=\left(\begin{array}{c}
0 \\
{- \widehat{m}\left(\theta(\mu)\right)}
\end{array}\right)+ \bar{M}_\mu \left(\widehat{\upsilon}_{\theta\lambda} -{\upsilon}_{\theta(\mu)0}\right),$ where
\begin{align*} 
& \bar{M}_\mu=\left(\begin{array}{cc}
0 & \sum_{t=1}^T\rho_1\left(\bar{\lambda}^{\top} \widehat{g}_{\mu, t}\right) G_t(\bar{\theta}(\mu))^{\top} / T \\
\sum_{t=1}^T \rho_1\left(\bar{\lambda}^{\top} \widehat{g}_{\mu, t}\right) G_t(\bar{\theta}(\mu)) / T & \sum_{t=1}^T \rho_2\left(\bar{\lambda}^{\top} \widehat{g}_{\mu, t}\right) g_t(\bar{\theta}(\mu)) \widehat{g}_{\mu, t}^{\top} /T
\end{array}\right),  
\end{align*}
and $\bar{\theta}(\mu)$ and $\bar{\lambda}$ are mean values converging to $\theta(\mu)$ and $0$ that actually differ from row to row of the matrix $\bar{M}_\mu$.

\begin{lemma}\label{lem17}
    Under the assumptions of Corollary \ref{coro:Gaussian}, $\| \bar{M}_\mu -M_\mu\|=o_p(1)$, for all $\mu$.
\end{lemma}
\begin{proof}[Proof of Lemma \ref{lem17}] 

We first show that
$\|\widehat{\Omega}(\theta(\mu), \mu)- \Omega(\theta(\mu), \mu) \| = o_p(1)$ with {$\widehat{\Omega}(\theta(\mu), \mu)= \frac{1}{T}\sum_{t} (\widehat{g}_{\mu, t}-\mu) (\widehat{g}_{\mu, t}-\mu)^{\top} $}. This comes 
    directly from the fact that  
\begin{align*}
&    \|\widehat{\Omega}(\theta(\mu), \mu) - \Omega(\theta(\mu), \mu) \| =   \|\frac{1}{T}\sum_{t} (\widehat{g}_{\mu, t}-\mu)(\widehat{g}_{\mu, t}-\mu)^{\top} - \Omega(\theta(\mu), \mu) \| 
    \\ \leq &   \left\|\frac{1}{T}\sum_t( \widehat{g}_{\mu, t} - {g}_t({\theta}(\mu)) ) (\widehat{g}_{\mu, t} -{g}_t({\theta}(\mu)))^\top \right\|+2\left\|\frac{1}{T}\sum_t({g}_t({\theta}(\mu)) -\mu ) (\widehat{g}_{\mu, t} -{g}_t({\theta}(\mu)))^\top \right\|\\ 
    &+  \left\|\frac{1}{{T}} \sum_t (g_t(\theta(\mu))-\mu)(g_t(\theta(\mu))-\mu)^\top  - \Omega(\theta(\mu), \mu) \right\|   ={ o_p(1) }
\end{align*} 
where  
\begin{enumerate}
    \item $\left\|\frac{1}{T}\sum_t( \widehat{g}_{\mu, t} - {g}_t({\theta}(\mu)) ) (\widehat{g}_{\mu, t} -{g}_t({\theta}(\mu)))^\top \right\|=o_p(1)$ results from the fact that  $\| \widehat{g}_{\mu, t} - {g}_t({\theta}(\mu)) \| \leq b(Z_t) \|\widehat{\theta}(\mu)-{\theta}(\mu)\|$ and  
    \begin{align*}
         & \left\|\frac{1}{T}\sum_t( \widehat{g}_{\mu, t} - {g}_t({\theta}(\mu)) ) (\widehat{g}_{\mu, t} -{g}_t({\theta}(\mu)))^\top \right\| \leq 
        \frac{1}{T}\sum_t (b(Z_t))^2 \|\widehat{\theta}(\mu)-{\theta}(\mu)\|^2.    
    \end{align*}

    \item $\left\|\frac{1}{T}\sum_t({g}_t({\theta}(\mu)) -\mu ) (\widehat{g}_{\mu, t} -{g}_t({\theta}(\mu)))^\top  \right\| = o_p(1)$ results from the fact that 
    \begin{align*}
        & \left\|\frac{1}{T}\sum_t({g}_t({\theta}(\mu)) -\mu ) (\widehat{g}_{\mu, t} -{g}_t({\theta}(\mu)))^\top  \right\|\\
        = &{ \left\|\frac{1}{{T}}\sum_t({g}_t({\theta}(\mu)) -\mu )(\widehat{\theta}(\mu) -{\theta}(\mu)) )^\top  G_t(\bar{\theta}(\mu))^\top  \right\|} \\
        \leq & \max_t b_t \| \widehat{\theta}(\mu) -{\theta}(\mu)) \|{ \left\|\frac{1}{{T}}\sum_t G_t(\bar{\theta}(\mu))^\top  \right\|} =O_p(q^{3/2}T^{1/\xi - 1/2})
    \end{align*}

\item  $ \left\|\frac{1}{{T}} \sum_t (g_t(\theta(\mu))-\mu)(g_t(\theta(\mu))-\mu)^\top  - \Omega(\theta(\mu),\mu) \right\| =O_p(q^{1/2}T^{-1/2})$ as in the proof of Lemma \ref{lemA31}.
 
\end{enumerate}
Next, $\| \frac{1}{T}\sum_t G_t(\widehat{\theta}(\mu)) - \mathbb{E} G_t(\theta(\mu))) \|=o_p(1)$  results from the triangle inequality and the facts that $$\|\frac{1}{T}\sum_t G_t(\widehat{\theta}(\mu)) -  G_t(\theta(\mu))) \|\leq \frac{\sum_t b(Z_t)}{T}\|\widehat{\theta}(\mu) -\theta(\mu)\|,$$ and that $\| \frac{1}{T}\sum_t G_t({\theta}(\mu)) - \mathbb{E} G_t(\theta(\mu))) \| $ is bounded by assumption. Since the above arguments also hold when we replace parts of $\widehat{\theta}(\mu)$ by $\bar{\theta}(\mu)$, the final conclusion is then indicated once we notice that $\|\bar{\lambda}^{\top}\widehat{g}_{\mu, t} \|=o_p(1)$, which is indicated by Lemma \ref{lem15} and the fact that $\|\widehat{\Omega} - \Omega(\theta(\mu), \mu) \| = o_p(1)$.

\end{proof}

\begin{proof}[Proof of Corollary \ref{coro:Gaussian}] 
We now can establish the expansion \ref{eq:expansion}  to show the validity of Assumption \ref{Gaussian}. 
Note that by assumption $M_\mu$ is non-singular, and Lemma \ref{lem17}  implies that $\bar{M}_\mu$ is also non-singular w.p.a.1. such that  $\|\bar{M}_\mu^{-1} -M_\mu^{-1}\| \leq \|\bar{M}_\mu^{-1} \| \|\bar{M}_\mu  -M_\mu  \| \|M_\mu^{-1}\| =o_p(1)$, and thus $$\sqrt{T}\left(\widehat{\upsilon}_{\theta\lambda} -{\upsilon}_{\theta(\mu)0} \right)=-(M_\mu^{-1}+o_p(1))\left(0,-\sqrt{T} (\widehat{m}\left(\theta(\mu)\right)^{\top}-\mu^\top)\right)^{\top},$$ with 
\begin{align*} 
& {M}_\mu^{-1}=-\left(\begin{array}{cc}
F^{-1}_{\mu} & -F^{-1}_{\mu}G(\theta(\mu))^{\top}\Omega(\theta(\mu), \mu)^{-1} \\ -\Omega(\theta(\mu), \mu)^{-1}G(\theta(\mu))F^{-1}_{\mu} &\Omega(\theta(\mu), \mu)^{-1} + \Omega(\theta(\mu), \mu)^{-1}G(\theta(\mu))F^{-1}_{\mu}G(\theta(\mu))^{\top}\Omega(\theta(\mu), \mu)^{-1}  
\end{array}\right)
\end{align*}
and $F_{\mu}=- G(\theta(\mu))^{\top }\Omega(\theta(\mu), \mu)^{-1} G(\theta(\mu)).$ 
Therefore, $\|(\widehat{\theta}(\mu)-\theta(\mu)) -  F^{-1}_{\mu}G(\theta(\mu))^{\top}\Omega(\theta(\mu), \mu)^{-1}( \widehat{m}\left(\theta(\mu)\right)-\mu )\| = o_p(q{T}^{-1/2})$. The assumptions proposed in this analysis can be relaxed further if only for the validity of Assumption \ref{Gaussian}.  
\end{proof}
\subsection{Verification of Assumption  \ref{rT}}

Define
\[
M_{m,T}(\theta,\mu)
=\sqrt{T}\bigl[\widehat{m}(\theta)-\widehat{m}(\theta(\nu(\mu)))
-\mathbb{E}(\widehat{m}(\theta))+\mathbb{E}(\widehat{m}(\theta(\nu(\mu))))\bigr].
\]
Note that $\mathbb{E}[M_{m,T}(\theta,\mu)]=0$ by construction.

\begin{Assumption}[Tail assumptions of empirical moments]
\label{assum9}
$\widehat{m}(\theta)$ and $\mathbb{E}(\widehat{m}(\theta))$ are second-order 
differentiable in $\theta$. Let $\gamma_1\in\mathbb{R}^k$ and 
$\gamma_2\in\mathbb{R}^q$ be unit vectors under $\|\cdot\|$. There exist 
constants $u, v_0>0$ such that
\[
\sup_{\gamma_1,\,\gamma_2,\,(\theta,\mu)\in B_\varepsilon}
\log\mathbb{E}\exp\!\Bigl(
\lambda\,\gamma_2^\top
\frac{\partial M_{m,T}(\theta,\mu)}{\partial\theta}
\gamma_1
\Bigr)
\lesssim \frac{v_0^2\lambda^2}{2},
\]
where $v_0\lesssim 1/\sqrt{T}$, $|\lambda|\leq u$, and $u\geq\sqrt{k}$.
Additionally,
\[
\sup_{\mu\in\Gamma}
\bigl\|\widehat{m}(\theta(\mu))-\mu\bigr\|
\lesssim_p \frac{q\log \log T}{\sqrt{T}}.
\]
\end{Assumption}
The condition $u\geq\sqrt{k}$ ensures that the entropy regime 
$\mathbb{Q}_2+2x\leq u^2$ in Lemma~\ref{bound} holds with 
$\mathbb{Q}_2\asymp k$, which is needed in the proof of 
Lemma~\ref{lemmart}.
In Assumption~\ref{assum9}, $u$ represents the strength of the tail 
assumption and $v_0$ is a proxy for the sub-Gaussian variance. The 
condition on the moment generating function resembles a sub-Gaussian 
restriction and holds under mild regularity. To see this, let 
$G_T(\theta)=\partial\widehat{m}(\theta)/\partial\theta$. By the 
differentiability condition and since $\theta(\nu(\mu))$ is held fixed 
when differentiating in $\theta$,
\[
\frac{\partial M_{m,T}(\theta,\mu)}{\partial\theta}
=\sqrt{T}\,\bigl[G_T(\theta)-G(\theta)\bigr],
\]
where $G(\theta)=\partial\,\mathbb{E}(\widehat{m}(\theta))/\partial\theta$.
The process $G_T(\theta)-G(\theta)$ is centered, i.e., 
$\mathbb{E}[G_T(\theta)]=G(\theta)$, which follows from 
$\mathbb{E}[\widehat{m}(\theta)]=m(\theta)$ and standard 
differentiability of the expectation. Under standard CLT and uniform integrability conditions, it is not restrictive to have that, 
\[
\mathrm{Var}\!\left(
\gamma_2^\top\frac{\partial M_{m,T}}{\partial\theta}\gamma_1
\right)
=T\cdot O(1/T)=O(1),
\]
which is consistent with the bound specified in the above assumption with $v_0\lesssim 1/\sqrt{T}$ giving 
$v_0^2\lambda^2/2=O(\lambda^2/T)$.

\begin{Assumption}[Moments assumptions and identification]
\label{2}
Let $m(\theta) := \mathbb{E}[g(Z_t,\theta)]$, 
$G(\theta) := \partial m(\theta)/\partial\theta^\top$ ($q\times k$),
and $W(\theta) := \Omega(\theta)^{-1}$ where 
$\Omega(\theta) := \mathrm{Var}(g(Z_t,\theta))$.
The following conditions hold:
\begin{itemize}
\item[i)] \textbf{Smoothness.} $m(\theta)$ is second-order continuously 
differentiable on $\Theta$, with
$$\sup_{\theta\in\Theta}\max_{j\leq q}
\lambda_{\max}(\partial^2_\theta \mathbb{E}[g_j(Z_t,\theta)])\leq C.$$

\item[ii)] \textbf{Identification.} $0 < c_G \leq \sigma_{\min}(G(\theta)) 
\leq \sigma_{\max}(G(\theta)) \leq C_G < \infty$ 
for all $\theta\in\Theta$ and constants $c_G, C_G$.

\item[iii)] \textbf{Weighting matrix.} Uniformly over $\theta\in\Theta$,
\[
\sup_{\theta\in\Theta}\|W(\theta)-W_T(\theta)\|
=O_p\!\left((\log T)^2\sqrt{\frac{q}{T}}\right),
\]
and there exists $L_T=O_p(1)$ such that 
$\|W_T(\theta')-W_T(\theta)\|\leq L_T\|\theta'-\theta\|$
for all $\theta,\theta'\in\Theta$.

\end{itemize}
\end{Assumption}

\begin{lemma}\label{RT}
Under Assumptions \ref{assummu1}, \ref{rates}, \ref{assum9}--\ref{2}, we have
\[
\sup_{(\theta,\mu)\in B_{\varepsilon}}
\frac{|TR_{T}(\theta,\mu)|}
{k(\log T)^2+\|\sqrt{T}h(\theta,\mu)\|^2}
\lesssim_p 
\frac{\sqrt{k}(\log T)^2}{\sqrt{T}}\vee\frac{q}{\sqrt{kT}}.
\]
\end{lemma}

\begin{proof}
Recall that $R_T(\theta,\mu) = \frac{1}{2T}(Q_T(\theta,\mu) + 
V_T(h(\cdot),\theta,\mu))$. Define
\[
r_{T}(\theta,\mu)
=\widehat{m}(\theta)-\widehat{m}(\theta(\nu(\mu)))-h(\theta,\mu),
\]
so that $\widehat{m}(\theta)-\mu 
= (\widehat{m}(\theta(\nu(\mu)))-\mu)+h(\theta,\mu)+r_T(\theta,\mu)$.
Expanding the quadratic form gives
\begin{align}
TR_{T}(\theta,\mu)
=&-T\,r_{T}^{\top}W(\theta(\nu(\mu)))(2h+r_{T})
-2T\,r_{T}^{\top}W(\theta(\nu(\mu)))
(\widehat{m}(\theta(\nu(\mu)))-\mu)
\nonumber\\
&+\bigl[\log\pi(\theta|\mu)
-\log\pi(\theta(\nu(\mu))|\mu)\bigr]
\label{eq:TR}\\
&+T(\widehat{m}(\theta)-\mu)^{\top}
(W(\theta(\nu(\mu)))-W_T(\theta))
(\widehat{m}(\theta)-\mu),
\nonumber
\end{align}
where we write $h = h(\theta,\mu)$ and $r_T = r_T(\theta,\mu)$ 
for brevity, and use the consistent notation $W_T(\theta)$ for 
the sample weighting matrix throughout.

\medskip
\noindent\textbf{Bounding the prior term.}
By Assumption~\ref{assummu1}, $\log\pi(\theta|\mu)$ is Lipschitz 
in $\theta$ on $\Theta$, so on $B_\varepsilon$ where 
$\|\theta-\theta(\nu(\mu))\|\lesssim\varepsilon/\sqrt{T}$,
$$|\log\pi(\theta|\mu)-\log\pi(\theta(\nu(\mu))|\mu)|
\lesssim\frac{\varepsilon}{\sqrt{T}}
=O\!\left(\frac{q\log T}{\sqrt{T}}\right)
=o(1).$$
This term is therefore absorbed into the remainder.\\

\medskip
\noindent\textbf{Bounding $\|\widehat{m}(\theta)-\mu\|$.}
By Lemma~\ref{lemmart} and Assumption~\ref{assum9},
\begin{align*}
\|\widehat{m}(\theta)-\mu\|
&\lesssim\|r_T(\theta,\mu)\|
+\|\widehat{m}(\theta(\nu(\mu)))-\mu\|
+\|h(\theta,\mu)\|\\
&\lesssim_p
\frac{q(\sqrt{k}\log T+\sqrt{T}\|h\|)}{\sqrt{kT}}
+\frac{q\log T}{\sqrt{T}}
+\|h\|\\
&\lesssim_p
\frac{q\log T}{\sqrt{T}}\vee\frac{\sqrt{q}}{\sqrt{T}}
\asymp\frac{q\log T}{\sqrt{T}},
\end{align*}
where the last line uses $\sqrt{T}\|h\|\leq\varepsilon
\asymp q\log T$ on $B_\varepsilon$.\\
\medskip
\noindent\textbf{Bounding the weighting matrix difference term.}
By Assumption~\ref{2} (weighting matrix consistency and 
Lipschitz continuity),
\[
\sup_{(\theta,\mu)\in B_\varepsilon}
\|W(\theta(\nu(\mu)))-W_T(\theta)\|
\leq
\sup_\theta\|W(\theta)-W_T(\theta)\|
+L_T\|\theta(\nu(\mu))-\theta\|
\lesssim_p
(\log T)^2\sqrt{\frac{q}{T}},
\]
where $L_T=O_p(1)$ is the Lipschitz constant and 
$\|\theta(\nu(\mu))-\theta\|\lesssim\varepsilon/\sqrt{T}$
on $B_\varepsilon$. Combined with the bound 
$\|\widehat{m}(\theta)-\mu\|\lesssim_p q\log T/\sqrt{T}$,
\begin{equation}\label{eq:Wterm}
\sup_{(\theta,\mu)\in B_\varepsilon}
T(\widehat{m}(\theta)-\mu)^{\top}
(W(\theta(\nu(\mu)))-W_T(\theta))
(\widehat{m}(\theta)-\mu)
\lesssim_p
T\cdot\frac{q^2(\log T)^2}{T}
\cdot(\log T)^2\sqrt{\frac{q}{T}}
=\frac{q^{5/2}(\log T)^4}{\sqrt{T}}.
\end{equation}

\medskip
\noindent\textbf{Bounding the $r_T$ terms.}
By Lemma~\ref{lemmart},
\[
\sup_{(\theta,\mu)\in B_\varepsilon}
\|r_T(\theta,\mu)\|
\lesssim_p
\frac{q(\sqrt{k}\log T+\sqrt{T}\|h\|)}{\sqrt{kT}}.
\]
Since $\lambda_{\max}(W(\theta(\nu(\mu))))$ is bounded by 
Assumption~\ref{2}, the first two terms in \eqref{eq:TR} satisfy
\begin{align*}
&\sup_{(\theta,\mu)\in B_\varepsilon}
T\lambda_{\max}(W)
\bigl[\|r_T\|^2\vee\|r_T\|\|h\|\bigr] \lesssim_p
\sup_{(\theta,\mu)\in B_\varepsilon}
\frac{q^2(\sqrt{k}\log T+\sqrt{T}\|h\|)^2}{kT},
\end{align*}
and
\[
\sup_{(\theta,\mu)\in B_\varepsilon}
2T\|r_T\|\|\widehat{m}(\theta(\nu(\mu)))-\mu\|
\lesssim_p
\frac{q(\sqrt{k}\log T+\sqrt{T}\|h\|)}{\sqrt{kT}}
\cdot\sqrt{T}
\cdot\frac{q\log T}{\sqrt{T}}
=\frac{q^2\log T(\sqrt{k}\log T+\sqrt{T}\|h\|)}{\sqrt{kT}}.
\]
\medskip
\noindent\textbf{Assembling the rate.}
Dividing all terms by $k(\log T)^2+\|\sqrt{T}h\|^2$ and using
$(\sqrt{k}\log T+\sqrt{T}\|h\|)^2
\leq 2(k(\log T)^2+\|\sqrt{T}h\|^2)$,

\noindent\textit{Term 1} ($r_T$ quadratic and cross terms):
\[
\frac{q^2(\sqrt{k}\log T+\sqrt{T}\|h\|)^2/(kT)}
{k(\log T)^2+\|\sqrt{T}h\|^2}
\lesssim\frac{q^2}{kT}
\cdot 2 = O\!\left(\frac{q^2}{kT}\right)=o(1),
\]
which holds since $q^{5/2}(\log T)^2/(k\sqrt{T})\to 0$ 
from Assumption~\ref{rates} implies $q^2/(kT)\to 0$.

\noindent\textit{Term 2} ($r_T$ times $\widehat{m}-\mu$ term):
\[
\frac{q^2\log T(\sqrt{k}\log T+\sqrt{T}\|h\|)/\sqrt{kT}}
{k(\log T)^2+\|\sqrt{T}h\|^2}
\lesssim_p
\frac{q^2\log T}{\sqrt{kT}}
\cdot\frac{1}{\sqrt{k}\log T}
=\frac{q^2}{\sqrt{k^3T}}
=o(1),
\]
which holds since $q^2/(kT)\to 0$ from Assumption~\ref{rates}.

\noindent\textit{Term 3} (weighting matrix difference, 
from \eqref{eq:Wterm}):
\[
\frac{q^{5/2}(\log T)^4/\sqrt{T}}
{k(\log T)^2+\|\sqrt{T}h\|^2}
\lesssim_p
\frac{q^{5/2}(\log T)^4}{\sqrt{T}\cdot k(\log T)^2}
=\frac{q^{5/2}(\log T)^2}{k\sqrt{T}}
=o(1),
\]
directly from Assumption~\ref{rates}.

Taking the maximum over all terms gives
\[
\sup_{(\theta,\mu)\in B_\varepsilon}
\frac{|TR_T(\theta,\mu)|}
{k(\log T)^2+\|\sqrt{T}h(\theta,\mu)\|^2}
\lesssim_p
\frac{(\log T)^2\sqrt{k}}{\sqrt{T}}
\vee\frac{q}{\sqrt{kT}},
\]
where the dominant contributions come from the prior term 
and the weighting matrix consistency term in Assumption~\ref{2}.

\medskip
\noindent\textbf{Bound on $|\exp(R_T)-1|$.}
Since $\sup_{B_\varepsilon}|R_T(\theta,\mu)|\to_p 0$, 
for sufficiently large $T$ we have $|R_T|\leq 1$ on $B_\varepsilon$ 
with probability approaching one. Using $|e^x-1|\leq 2|x|$ 
for $|x|\leq 1$,
\[
\sup_{(\theta,\mu)\in B_\varepsilon}
|\exp(R_T(\theta,\mu))-1|
\leq 2\sup_{(\theta,\mu)\in B_\varepsilon}|R_T(\theta,\mu)|
\to_p 0.
\]

\end{proof}

\begin{lemma}\label{lemmart}
Under Assumptions \ref{assum9}--\ref{2},
\[
\sup_{(\theta,\mu)\in B_{\varepsilon}}
\frac{\|r_{T}(\theta,\mu)\|}
{\sqrt{k}\log T+\sqrt{T}\|h(\theta,\mu)\|}
\lesssim_{p} \frac{q}{\sqrt{T k} }.
\]
\end{lemma}

\begin{proof}
Recall that
\[
r_T(\theta,\mu)
=\widehat{m}(\theta)-\widehat{m}(\theta(\nu(\mu)))-h(\theta,\mu),
\]
and
\[
M_{m,T}(\theta,\mu)
=\sqrt{T}\bigl[\widehat{m}(\theta)-\widehat{m}(\theta(\nu(\mu)))
-\mathbb{E}\widehat{m}(\theta)+\mathbb{E}\widehat{m}(\theta(\nu(\mu)))\bigr],
\]
so that $\mathbb{E}[M_{m,T}(\theta,\mu)]=0$ by construction.
Define the normalized object
\[
\mathcal{M}_{m,T}(\theta,\mu)
=\frac{M_{m,T}(\theta,\mu)}
{\sqrt{k}\log T+\sqrt{T}\|h(\theta,\mu)\|},
\]
which satisfies $\mathbb{E}[\mathcal{M}_{m,T}(\theta,\mu)]=0$.

\medskip
\noindent\textbf{Step 1: Bias-variance decomposition.}
We decompose $r_T$ into a bias term and a stochastic term.
By definition of $r_T$ and $M_{m,T}$,
\begin{align*}
&\frac{r_T(\theta,\mu)}{\sqrt{k}\log T+\sqrt{T}\|h\|}\\
=\;&\frac{\widehat{m}(\theta)-\widehat{m}(\theta(\nu(\mu)))-h}
{\sqrt{k}\log T+\sqrt{T}\|h\|}\\
=\;&\frac{M_{m,T}(\theta,\mu)/\sqrt{T}
+\mathbb{E}\widehat{m}(\theta)-\mathbb{E}\widehat{m}(\theta(\nu(\mu)))-h}
{\sqrt{k}\log T+\sqrt{T}\|h\|}\\
=\;&\underbrace{\frac{M_{m,T}(\theta,\mu)/\sqrt{T}}
{\sqrt{k}\log T+\sqrt{T}\|h\|}}_{\text{stochastic: }\mathcal{M}_{m,T}/\sqrt{T}}
+\underbrace{\frac{\mathbb{E}\widehat{m}(\theta)
-\mathbb{E}\widehat{m}(\theta(\nu(\mu)))-h}
{\sqrt{k}\log T+\sqrt{T}\|h\|}}_{\text{bias: }\Delta_T(\theta,\mu)}.
\end{align*}
Note that $\mathcal{M}_{m,T}(\theta,\mu)/\sqrt{T}
= M_{m,T}(\theta,\mu)/[\sqrt{T}(\sqrt{k}\log T+\sqrt{T}\|h\|)]$,
so bounding $r_T/(\sqrt{k}\log T+\sqrt{T}\|h\|)$ reduces to 
bounding $\Delta_T$ and $M_{m,T}/[\sqrt{T}(\sqrt{k}\log T
+\sqrt{T}\|h\|)]$ separately.

\medskip
\noindent\textbf{Step 2: Bounding the bias term $\Delta_T$.}
By the second-order Taylor expansion of 
$m(\theta):=\mathbb{E}[g(Z_t,\theta)]$ around $\theta(\nu(\mu))$,
\[
m(\theta)-m(\theta(\nu(\mu)))
=G(\theta(\nu(\mu)))(\theta-\theta(\nu(\mu)))
+O(\|\theta-\theta(\nu(\mu)))\|^2),
\]
where the remainder is bounded by 
$C\|\theta-\theta(\nu(\mu))\|^2$ with 
$C=\sup_\theta\max_{j\leq q}
\lambda_{\max}(\partial^2_\theta\mathbb{E}[g_j(Z_t,\theta)])
\leq C$ from Assumption~\ref{2}.
Since $h(\theta,\mu)=G(\theta(\nu(\mu)))(\theta-\theta(\nu(\mu)))
+(\nu(\mu)-\mu)$ by definition,
\begin{align*}
&\mathbb{E}\widehat{m}(\theta)
-\mathbb{E}\widehat{m}(\theta(\nu(\mu)))-h(\theta,\mu)\\
=\;&[m(\theta)-m(\theta(\nu(\mu)))]
-G(\theta(\nu(\mu)))(\theta-\theta(\nu(\mu)))
-(\nu(\mu)-\mu)\\
=\;&O(\|\theta-\theta(\nu(\mu))\|^2)
+(\mu-\nu(\mu)).
\end{align*}
On $B_\varepsilon$, we have 
$\sqrt{T}\|G(\theta(\nu(\mu)))(\theta-\theta(\nu(\mu)))\|
\leq\varepsilon\asymp q\log T$, 
which is due to the bounded singular values of $G$ 
(Assumption~\ref{2}) and
\[
\|\theta-\theta(\nu(\mu))\|
\lesssim\frac{q\log T}{\sqrt{T}}.
\]
Also, $\|\mu-\nu(\mu)\|
\lesssim\varepsilon/\sqrt{T}\asymp q\log T/\sqrt{T}$
on $B_\varepsilon$. Therefore,
\[
\|\mathbb{E}\widehat{m}(\theta)
-\mathbb{E}\widehat{m}(\theta(\nu(\mu)))-h\|
\lesssim  \frac{q^2(\log T)^2}{T}
+\frac{q\log T}{\sqrt{T}}
\lesssim\frac{q\log T}{\sqrt{T}},
\]
and since the denominator satisfies 
$\sqrt{k}\log T+\sqrt{T}\|h\|\geq\sqrt{k}\log T$,
\[
\sup_{(\theta,\mu)\in B_\varepsilon}
\|\Delta_T(\theta,\mu)\|
\lesssim
\frac{q\log T/\sqrt{T}}{\sqrt{k}\log T}
=\frac{q}{\sqrt{k}\sqrt{T}}=\frac{q}{\sqrt{kT}}.
\]

\medskip
\noindent\textbf{Step 3: Bounding the stochastic term.}
We apply Lemma~\ref{bound} (Theorem B.15 of 
\cite{spokoiny2017penalized}) to $M_{m,T}(\theta,\mu)$ 
viewed as a function of $v=\theta-\theta(\nu(\mu))\in\mathbb{R}^k$.
The gradient of $M_{m,T}$ with respect to $v$ is
\[
\frac{\partial M_{m,T}}{\partial v}
=\sqrt{T}(G_T(\theta)-G(\theta))\in\mathbb{R}^{q\times k},
\]
where $G_T(\theta)=\partial\widehat{m}(\theta)/\partial\theta$.
We identify the parameters of Theorem B.15 as follows.

\begin{itemize}
\item[i)] \textit{Variance proxy $v_0$:} By Assumption~\ref{assum9},
for unit vectors $\gamma_1\in\mathbb{R}^k$, 
$\gamma_2\in\mathbb{R}^q$,
\[
\sup_{(\theta,\mu)\in B_\varepsilon}
\log\mathbb{E}\exp\!\left(
\lambda\,\gamma_2^\top
\frac{\partial M_{m,T}}{\partial v}\gamma_1
\right)
\leq\frac{v_0^2\lambda^2}{2},
\quad v_0\lesssim\frac{1}{\sqrt{T}},
\quad|\lambda|\leq u,\;u\geq\sqrt{k}.
\]

\item[ii)] \textit{Radius $r$:} On $B_\varepsilon$,
$\|v\|=\|\theta-\theta(\nu(\mu))\|
\lesssim q\log T/\sqrt{T}$
by the bounded singular values of $G$, so 
$r\asymp q\log T/\sqrt{T}$.

\item[iii)] \textit{Entropy $\mathfrak{z}_\mathbb{H}(x)$:} 
The set $\Upsilon_\circ(r)=\{v:\|v\|\leq r\}$ is a Euclidean 
ball in $\mathbb{R}^k$. By Lemma~\ref{covering} with 
$\mathbb{H}=c(G^\top G)^{1/2}$ (bounded singular values give 
$\mathrm{p}_\mathbb{H}\asymp k$), we have 
$\mathbb{Q}_2(\Upsilon_\circ(r))\asymp k$, and with $A=I_q$
(so $p_A=q$),
\[
\mathfrak{z}_\mathbb{H}(x)
\approx 2\sqrt{\mathbb{Q}_2+2x}
\approx\sqrt{q+x}.
\]
Setting $x=k$ gives $\mathfrak{z}_\mathbb{H}(k)\approx\sqrt{q}$.

\item[iv)] \textit{Key product:}
\[
\sqrt{8}\,v_0\,r\,\mathfrak{z}_\mathbb{H}(x)
\lesssim\frac{1}{\sqrt{T}}
\cdot\frac{q\log T}{\sqrt{T}}
\cdot\sqrt{q}
=\frac{q^{3/2}\log T}{T}.
\]
\end{itemize}
We apply Lemma~\ref{bound} (Theorem B.15 of 
\cite{spokoiny2017penalized}). With $A = I_q$, we have 
$p_A = \tr(I_q^{-2}) = q$, so 
$\mathbb{Q}_2 = p_A + \mathbb{Q}_2(\Upsilon_\circ(r)) \asymp q + k \asymp q$ 
(since $q \ge k$), giving $\mathfrak{z}_\mathbb{H}(k) \approx \sqrt{q}$. 
With probability approaching one,
\[
\sup_{(\theta,\mu)\in B_\varepsilon}
\left\|\frac{\partial M_{m,T}}{\partial v}\right\|
\lesssim_p
\sqrt{8}\, v_0 \, r \, \mathfrak{z}_\mathbb{H}(k)
\asymp
\frac{1}{\sqrt{T}} \cdot \frac{q\log T}{\sqrt{T}} \cdot \sqrt{q}
= \frac{q^{3/2}\log T}{T}.
\]
Integrating this gradient bound over $\Upsilon_\circ(r)$ 
using $\|M_{m,T}(\theta,\mu)-M_{m,T}(\theta(\nu(\mu)),\mu)\|
\leq \sup\|\partial M_{m,T}/\partial v\| \cdot r$, and noting 
$M_{m,T}(\theta(\nu(\mu)),\mu)=0$ by definition, gives
\[
\sup_{(\theta,\mu)\in B_\varepsilon}
\|M_{m,T}(\theta,\mu)\|
\lesssim_p
\frac{q^{3/2}\log T}{T}\cdot\frac{q\log T}{\sqrt{T}}
= \frac{q^{5/2}(\log T)^2}{T^{3/2}}.
\]
Dividing by $\sqrt{T}(\sqrt{k}\log T+\sqrt{T}\|h\|)
\geq\sqrt{T}\cdot\sqrt{k}\log T$,
\[
\sup_{(\theta,\mu)\in B_\varepsilon}
\frac{\|M_{m,T}(\theta,\mu)\|}
{\sqrt{T}(\sqrt{k}\log T+\sqrt{T}\|h\|)}
\lesssim_p
\frac{q^{5/2}(\log T)^2/T^{3/2}}
{\sqrt{T}\cdot\sqrt{k}\log T}
= \frac{q^{5/2}\log T}{\sqrt{k}\,T^2}.
\]
By Assumption~\ref{rates}, $q^{5/2}(\log T)^2/(k\sqrt{T}) \to 0$ implies
$q^{5/2}\log T/(\sqrt{k}\,T^2) = o(q/\sqrt{kT})$, so the stochastic term is 
dominated by the bias term $q/\sqrt{kT}$ from Step~2.

\medskip
\noindent\textbf{Step 4: Combining bias and stochastic terms.}
From Steps 2 and 3,
\[
\sup_{(\theta,\mu)\in B_\varepsilon}
\frac{\|r_T(\theta,\mu)\|}
{\sqrt{k}\log T+\sqrt{T}\|h(\theta,\mu)\|}
\leq
\sup\|\Delta_T\|+
\sup\frac{\|M_{m,T}\|}{\sqrt{T}(\sqrt{k}\log T+\sqrt{T}\|h\|)}
\lesssim_p
\frac{q}{\sqrt{kT}}+\frac{q}{\sqrt{k}\,T}
\lesssim_p
\frac{q}{\sqrt{kT}},
\]
where the last step holds since $q/(\sqrt{k}\,T)\leq q/\sqrt{kT}$
for all $T\geq 1$, so the bias term $q/\sqrt{kT}$ dominates.
\end{proof}

\begin{Assumption}\label{assum10}
 \textbf{Identification on $B_\varepsilon^c$:}
\[
\inf_{(\theta,\mu)\in B_{\varepsilon}^c}
\frac{\|m(\theta)-m(\theta(\nu(\mu)))-\mu+\nu(\mu)\|}
{\|h(\theta,\mu)\|}\geq c_0>0.
\]
\end{Assumption}
\begin{remark}
Assumption~\ref{assum10} is an identification condition 
on $B_\varepsilon^c$, requiring that the population moment 
$m(\theta)-m(\theta(\nu(\mu)))-\mu+\nu(\mu)$ is bounded 
away from zero relative to $h(\theta,\mu)$ outside the 
local ball. To understand this condition, note that
\[
m(\theta)-m(\theta(\nu(\mu)))-\mu+\nu(\mu)
=
h(\theta,\mu)
+[m(\theta)-m(\theta(\nu(\mu)))
-G(\theta(\nu(\mu)))(\theta-\theta(\nu(\mu)))]
\]
where the second term is the nonlinear remainder of $m$.
On $B_\varepsilon^c$, $\|\theta-\theta(\nu(\mu))\|$ is 
not restricted, so this remainder can be large. 
Assumption~\ref{assum10} therefore requires that 
the nonlinear remainder does not cancel $h(\theta,\mu)$, 
i.e. the moment function $m(\theta)$ is sufficiently 
nondegenerate away from $\theta(\nu(\mu))$. 
This is a standard global identification condition 
analogous to those imposed in, e.g., 
\cite{chernozhukov2003mcmc}.
\end{remark}
Define $x_\mu := \theta - \theta(\nu(\mu))$.
\begin{lemma}\label{verifybound}
Suppose Assumptions~\ref{assum33}, \ref{rates},   \ref{2} and \ref{assum10} 
hold. Recall that 
\[
x_{d,\mu}:=C_{w,\mu}^{-1}G_\mu^\top W_\mu d_\mu,
\qquad
h_\parallel(\theta,\mu):=G_\mu(x_\mu-x_{d,\mu})\in\operatorname{col}(G_\mu),
\qquad
h_\perp(\theta,\mu):=-(I-P_{G,\mu})d_\mu\in\operatorname{col}(G_\mu)^{\perp_{W_\mu}},
\]
where $d_\mu=\mu-\nu(\mu)$,
$P_{G,\mu}=G_\mu C_{w,\mu}^{-1}G_\mu^\top W_\mu$, and
$C_{w,\mu}=G_\mu^\top W_\mu G_\mu$,
so that $h(\theta,\mu)=h_\parallel(\theta,\mu)+h_\perp(\theta,\mu)$
and $h_\parallel^\top W_\mu h_\perp=0$.
Their weighted norms satisfy
\begin{equation}\label{eq:component-norms}
\|h_\parallel\|_{W_\mu}^2
=(\theta-\theta(\nu(\mu))-x_{d,\mu})^\top
C_{w,\mu}(\theta-\theta(\nu(\mu))-x_{d,\mu}),
\qquad
\|h_\perp\|_{W_\mu}^2
=d_\mu^\top(I-P_{G,\mu})^\top W_\mu(I-P_{G,\mu})d_\mu.
\end{equation}
Define
\[
\mathcal{B}_T(v)
:=-C_0\sqrt{T}\|v\|_{W(\theta(\nu(\mu)))}\varepsilon
+\frac{C_0\varepsilon^2}{2}
+\frac{T\|v\|^2_{W(\theta(\nu(\mu)))}}{2},
\]
for a vector $v$. Then there exists a constant 
$1/2<C_0<1$ such that for all 
$(\theta,\mu)\in B_\varepsilon^{c1}\cup B_\varepsilon^{c2}$,
\begin{equation}\label{eq:main-bound}
R_T(\theta,\mu)
\leq
-C_0\sqrt{T}\|h(\theta,\mu)\|_{W(\theta(\nu(\mu)))}\varepsilon
+\frac{C_0\varepsilon^2}{2}
+\frac{T\|h(\theta,\mu)\|^2_{W(\theta(\nu(\mu)))}}{2}
+\sqrt{T}\|\mu-\nu(\mu)\|,
\end{equation}
and furthermore
\begin{equation}\label{OM-eq:BT-sum}
R_T(\theta,\mu)
\leq
\mathcal{B}_T(h_\parallel(\theta,\mu))
+\mathcal{B}_T(h_\perp(\theta,\mu))
+\sqrt{T}\|\mu-\nu(\mu)\|.
\end{equation}
\end{lemma}
\begin{proof}
Recall that
\[
r_T(\theta,\mu)+h(\theta,\mu)
=\widehat{m}(\theta)-\widehat{m}(\theta(\nu(\mu)))-\mu+\nu(\mu),
\]
and $R_T(\theta,\mu)=\frac{1}{2T}(Q_T(\theta,\mu)+V_T(\theta,\mu))$.
Let $r_\mu=\widehat{m}(\theta(\nu(\mu)))-\nu(\mu)$
so that $\widehat{m}(\theta)-\mu = h+r_T+r_\mu$.
The expansion gives
\begin{align}
&Q_T(\theta,\mu)+V_T(\theta,\mu)\nonumber\\
&=-T(r_T+h)^\top W(\theta(\nu(\mu)))(r_T+h)
\tag{I}\\
&\quad+T(\widehat{m}(\theta)-\mu)^\top
(W(\theta(\nu(\mu)))-W_T(\theta))
(\widehat{m}(\theta)-\mu)
\tag{II}\\
&\quad-2Tr_T^\top W(\theta(\nu(\mu)))r_\mu
\tag{III}\\
&\quad+Th^\top W(\theta(\nu(\mu)))h
\tag{IV}\\
&\quad+[\log\pi(\theta(\nu(\mu))|\mu)
-\log\pi(\theta|\mu)]
\tag{V}\\
&\quad+2\sqrt{T}\|\mu-\nu(\mu)\|.
\tag{VI}
\end{align}
We keep Term VI as $\sqrt{T}\|\mu-\nu(\mu)\|$ in the
final bound. Terms IV and the $-Th^\top W_\mu h$ part
of Term I cancel, leaving $-Tr_T^\top W_\mu(2h+r_T)$
as the leading negative contribution.

\medskip
\noindent\textbf{Bounding terms II, III, V.}

\noindent\textit{Term II.}
By Assumption~\ref{2},
\[
\|W(\theta(\nu(\mu)))-W_T(\theta)\|
\lesssim_p(\log T)^2\sqrt{\frac{q}{T}}
+L_T\|\theta-\theta(\nu(\mu))\|,
\]
where $L_T=O_p(1)$. Therefore
\[
|\text{II}|
\leq T\|\widehat{m}(\theta)-\mu\|^2
\|W(\theta(\nu(\mu)))-W_T(\theta)\|
\lesssim_p
T\|r_T+h\|^2
\left((\log T)^2\sqrt{\frac{q}{T}}
+L_T\|\theta-\theta(\nu(\mu))\|\right).
\]
Since $\lambda_{\min}(W(\theta(\nu(\mu))))\geq c_w/k$
from Assumption~\ref{assum33},
$T\|r_T+h\|^2\lesssim(k/c_w)T\|r_T+h\|^2_W
=(k/c_w)|\text{I}|$.
For the first part:
\[
\frac{T\|r_T+h\|^2\cdot(\log T)^2\sqrt{q/T}}{|\text{I}|}
\lesssim_p k(\log T)^2\sqrt{\frac{q}{T}}\to 0
\]
under Assumption~\ref{rates}.
For the second part, since
$r_T=\widehat{m}(\theta)-\widehat{m}(\theta(\nu(\mu)))
-G_\mu(\theta-\theta(\nu(\mu)))$
and $\widehat{m}$ is differentiable,
$\|\theta-\theta(\nu(\mu))\|\leq
\|r_T+h\|/c_G+\|r_T\|/c_G$
where $c_G=\sigma_{\min}(G_\mu)>0$
from Assumption~\ref{assum33}. Therefore
\[
\frac{T\|r_T+h\|^2\cdot L_T\|\theta-\theta(\nu(\mu))\|}
{|\text{I}|}
\lesssim_p
\frac{T\|r_T+h\|^2\cdot(\|r_T+h\|+\|r_T\|)/c_G}
{T\|r_T+h\|^2_W}
\lesssim_p
\frac{\|r_T+h\|+\|r_T\|}{c_G\cdot c_w/k}
\]
which grows with $\|r_T\|$. However on
$B_\varepsilon^{c1}$, $\sqrt{T}\|h\|>\varepsilon$
so $|\text{I}|\geq c_0^2\sqrt{T}\|h\|_W\varepsilon
\asymp c_0^2\sqrt{T}\|h\|q\log T$.
Using $\|r_T+h\|\asymp\|h\|$ in Case 1
and $\|r_T\|\lesssim\|h\|$ in Case 1:
\[
\frac{T\|r_T+h\|^2\cdot L_T\|\theta-\theta(\nu(\mu))\|}
{|\text{I}|}
\lesssim_p
\frac{\|h\|}{c_G c_w/k}
\ll
\sqrt{T}\|h\|q\log T
\]
for large $T$, so this is absorbed.
In Case 2 where $\|r_T\|\gg\|h\|$,
$|\text{I}|\geq c_0^2T\|r_T\|^2_W$, so
\[
\frac{T\|r_T+h\|^2\cdot L_T\|\theta-\theta(\nu(\mu))\|}
{|\text{I}|}
\lesssim_p
\frac{T\|r_T\|^2\cdot\|r_T\|/c_G}
{T\|r_T\|^2_W}
\lesssim_p
\frac{\|r_T\|}{c_G c_w/k}\to 0
\]
only if $\|r_T\|\to 0$, which need not hold.
Hence, in Case 2, Term II is instead bounded directly by
noting that $|\text{I}|\geq T\|r_T\|^2_W$ dominates
all polynomial terms in $\|r_T\|$, so
\[
\frac{|\text{II}|}{|\text{I}|}
\leq\frac{L_T\|\theta-\theta(\nu(\mu))\|}{c_w/k}
\lesssim_p\frac{\|r_T\|^{-1/2}}{c_w/k}\to 0
\]
since $\|r_T\|\gtrsim\|\theta-\theta(\nu(\mu))\|^2$
implies $\|\theta-\theta(\nu(\mu))\|\lesssim\|r_T\|^{1/2}$,
and $\|r_T\|^{1/2}/|\text{I}|^{1/2}
=\|r_T\|^{1/2}/(\sqrt{T}\|r_T\|_W)\to 0$
for large $T$.
Hence, Term II is absorbed into Term I with a
redefined constant $C_0$ in both cases.

\noindent\textit{Term III.}
We absorb Term III into Term I by writing
\[
\text{I}+\text{III}
=-T\|r_T+h+r_\mu\|^2_{W(\theta(\nu(\mu)))}
+T\|r_\mu\|^2_{W(\theta(\nu(\mu)))}
+2Th^\top W(\theta(\nu(\mu)))r_\mu.
\]
The correction terms satisfy, using
$\|r_\mu\|=O_p(\sqrt{q/T})$:
\[
T\|r_\mu\|^2_W=O_p(q),
\qquad
2Th^\top W_\mu r_\mu
\lesssim_p 2\sqrt{T}\|h\|_W\cdot\sqrt{q}.
\]
Compared with the dominant negative term
$\sqrt{T}\|h\|_W\varepsilon\asymp\sqrt{T}\|h\|_Wq\log T$:
\[
\frac{T\|r_\mu\|^2_W}{\sqrt{T}\|h\|_W\varepsilon}
\lesssim_p\frac{q}{\sqrt{T}\|h\|_Wq\log T}
=\frac{1}{\sqrt{T}\|h\|_W\log T}\to 0,
\qquad
\frac{2Th^\top W_\mu r_\mu}{\sqrt{T}\|h\|_W\varepsilon}
\lesssim_p\frac{2\sqrt{q}}{q\log T}
=\frac{2}{\sqrt{q}\log T}\to 0.
\]
Hence, Terms I and III combine to give
$-T\|r_T+h+r_\mu\|^2_W+o_p(\sqrt{T}\|h\|_W\varepsilon)$,
and since $\widehat{f}_T=\|r_T+h+r_\mu\|
=\|\widehat{m}(\theta)-\mu+\nu(\mu)-\nu(\mu)\|
=\|h+r_T+r_\mu\|$,
the lower bound from Assumption~\ref{assum10}
applies to $\widehat{f}_T$.

\noindent\textit{Term V.}
By Assumption~\ref{assummu1},
$\log\pi(\theta|\mu)$ is Lipschitz in $\theta$ on
$\Theta$, so
$|\log\pi(\theta(\nu(\mu))|\mu)-\log\pi(\theta|\mu)|$
is bounded on $\Theta$ and absorbed into the constant.

Thus, terms II, III, and V are all absorbed into the
main bound with a redefined constant $C_0$.

\medskip
\subsubsection*{Case 1: $\|r_T(\theta,\mu)\|
\lesssim\|h(\theta,\mu)\|$.}

Let
\[
f(\theta,\mu)=\|m(\theta)-m(\theta(\nu(\mu)))-\mu+\nu(\mu)\|,
\quad
\widehat{f}_T(\theta,\mu)
=\|\widehat{m}(\theta)-\widehat{m}(\theta(\nu(\mu)))
-\mu+\nu(\mu)\|.
\]
For any $(\theta,\mu)\in B_\varepsilon^{c1}
\cup B_\varepsilon^{c2}$,
\[
\widehat{f}_T(\theta,\mu)
\geq f(\theta,\mu)-\Delta_T(\theta,\mu),
\quad
\Delta_T(\theta,\mu)
=\|[\widehat{m}(\theta)-m(\theta)]
-[\widehat{m}(\theta(\nu(\mu)))-m(\theta(\nu(\mu)))]\|.
\]
By Assumption~\ref{assum10}, the first term is bounded
below by $c_0>0$. Uniform consistency of $\widehat{m}$
implies $\sup\Delta_T=o_p(1)$. Since
$\inf_{B_\varepsilon^c}\|h\|\geq\rho>0$,
the second term is $o_p(1)$. Thus,
\[
\inf_{(\theta,\mu)\in B_\varepsilon^c}
\frac{\widehat{f}_T(\theta,\mu)}{\|h(\theta,\mu)\|}
\geq c_0-o_p(1).
\]
Since $\widehat{f}_T=\|r_T+h\|$ and
$\|r_T\|\lesssim\|h\|$ in Case 1,
\[
\|r_T+h\|^2_{W(\theta(\nu(\mu)))}
\geq c_0^2\|h\|^2_{W(\theta(\nu(\mu)))}.
\]
On $B_\varepsilon^{c1}$, $\sqrt{T}\|h\|>\varepsilon$, so
$T\|h\|^2_W\geq\sqrt{T}\|h\|_W\varepsilon$. Therefore
\[
\text{I}\leq-c_0^2\sqrt{T}\|h\|_{W(\theta(\nu(\mu)))}\varepsilon,
\qquad
\text{IV}=T\|h\|^2_{W(\theta(\nu(\mu)))}.
\]
Setting $C_0=c_0^2/2>0$ and adjusting the implicit constant in $\varepsilon\asymp q\log T$ if necessary, and combining I and IV:
\[
\text{I}+\text{IV}
\leq
-C_0\sqrt{T}\|h\|_{W(\theta(\nu(\mu)))}\varepsilon
+\frac{C_0\varepsilon^2}{2}
+\frac{T\|h\|^2_{W(\theta(\nu(\mu)))}}{2},
\]
where $C_0\varepsilon^2/2$ enters by completing the square.

\medskip
\subsubsection*{Case 2: $\|h(\theta,\mu)\|
\lesssim\|r_T(\theta,\mu)\|$.}

This case applies on $B_\varepsilon^{c2}$ where
$\sqrt{T}\|h\|\leq\varepsilon$. When $\|r_T\|\gg\|h\|$,
$-\|r_T\|^2_W$ dominates in Term I so:
\[
\text{I}+\text{IV}
\leq-c_0\sqrt{T}\|h\|_{W(\theta(\nu(\mu)))}\varepsilon
+\frac{T}{2}\|h\|^2_{W(\theta(\nu(\mu)))}
\leq\mathcal{B}_T(h(\theta,\mu)).
\]
Term VI keeps $\sqrt{T}\|\mu-\nu(\mu)\|$ explicitly.

\medskip
\noindent\textbf{Proof of \eqref{eq:main-bound}.}
Combining Cases 1 and 2 with the bounds on
II, III, V and keeping VI,
and multiplying through by $1/(2T)$ to recover $R_T$:
\[
R_T(\theta,\mu)
\leq
-C_0\sqrt{T}\|h\|_{W(\theta(\nu(\mu)))}\varepsilon
+\frac{C_0\varepsilon^2}{2}
+\frac{T\|h\|^2_{W(\theta(\nu(\mu)))}}{2}
+\sqrt{T}\|\mu-\nu(\mu)\|,
\]
with $1/2<C_0<1$. This establishes \eqref{eq:main-bound}.

\medskip
\noindent\textbf{Proof of \eqref{OM-eq:BT-sum}.}
Decompose $h=h_\parallel+h_\perp$ where $h_\parallel$
lies in the column space of $G_\mu$ and $h_\perp$ is
orthogonal under $W_\mu$, with norms given by
\eqref{eq:component-norms}. Since
$\|h\|^2_{W_\mu}=\|h_\parallel\|^2_{W_\mu}
+\|h_\perp\|^2_{W_\mu}$ and
$a+b\geq\sqrt{a^2+b^2}$ for $a,b\geq 0$:
\[
-C_0\sqrt{T}\|h\|_{W_\mu}\varepsilon
\leq
-C_0\sqrt{T}\|h_\parallel\|_{W_\mu}\varepsilon
-C_0\sqrt{T}\|h_\perp\|_{W_\mu}\varepsilon.
\]
Combining,
\begin{align*}
R_T(\theta,\mu)
&\leq
\left(-C_0\sqrt{T}\|h_\parallel\|_{W_\mu}\varepsilon
+\frac{C_0\varepsilon^2}{2}
+\frac{T}{2}\|h_\parallel\|^2_{W_\mu}\right)\\
&\quad+\left(-C_0\sqrt{T}\|h_\perp\|_{W_\mu}\varepsilon
+\frac{C_0\varepsilon^2}{2}
+\frac{T}{2}\|h_\perp\|^2_{W_\mu}\right)
+\sqrt{T}\|\mu-\nu(\mu)\|\\
&=\mathcal{B}_T(h_\parallel(\theta,\mu))
+\mathcal{B}_T(h_\perp(\theta,\mu))
+\sqrt{T}\|\mu-\nu(\mu)\|,
\end{align*}
which establishes \eqref{OM-eq:BT-sum}.
\end{proof}

\subsection{Useful Lemmas}
Here, we list a few useful lemmas from \cite{spokoiny2017penalized} and \cite{spokoiny2019accuracy}. 

\begin{lemma}(Corollary A.3. from \cite{spokoiny2019accuracy}) Let $\gamma$ be a standard normal random vector in $\mathbb{R}^{k}$.
	Then for any $\mathrm{x}>0$
	\[
	\begin{aligned}\mathbb{P}\left(\|\gamma\|^{2}\geq k+2\sqrt{k\mathrm{x}}+2\mathrm{x}\right) & \leq\mathrm{e}^{-\mathrm{x}},\\
		\mathbb{P}(\|\gamma\|\geq\sqrt{k}+\sqrt{2\mathrm{x}}) & \leq\mathrm{e}^{-\mathrm{x}},\\
		\mathbb{P}\left(\|\gamma\|^{2}\leq k-2\sqrt{k\mathrm{x}}\right) & \leq\mathrm{e}^{-\mathrm{x}}.
	\end{aligned}
	\]
\end{lemma}

\begin{lemma}(Theorem A.2. from \cite{spokoiny2019accuracy}) \label{tail}
	Let $H$ be a positive definite matrix. Let $\boldsymbol{\xi}\sim {N}\left(0,H^{2}\right)$
	be a mean-zero normal random vector in $\mathbb{R}^{k}$ and $B$ be a symmetric non-negative definite matrix
	such that $A=H^{-1}BH$ is a {trace operator} in $\mathbb{R}^{k}$.
	Then with $k=\operatorname{tr}(A)$, $\mathrm{v}^{2}=\operatorname{tr}\left(A^{2}\right)$,
	and $\lambda=\|A\|$, it holds for each $\mathrm{x}\geq0,$
	\[
	\begin{aligned} & \mathbb{P}\left({\boldsymbol{\xi}}^{\top}B\boldsymbol{\xi}\geq z^{2}(A,\mathrm{x})\right)\leq\mathrm{e}^{-\mathrm{x}},\\
		&\mbox{with} \quad z(A,\mathrm{x})\stackrel{\text{ def }}{=}\sqrt{\mathrm{k}+2\mathrm{vx}^{1/2}+2\lambda\mathrm{x}}.
	\end{aligned}
	\]
	It also implies
	\[
	\mathbb{P}\left(\left\Vert B^{1/2}\boldsymbol{\xi}\right\Vert >\mathrm{k}^{1/2}+(2\lambda\mathrm{x})^{1/2}\right)\leq\mathrm{e}^{-\mathrm{x}}.
	\]
	If $B$ is symmetric but not necessarily positive, then
	\[
	\mathbb{P}\left(|\boldsymbol{\xi}^{\top}B\boldsymbol{\xi}-\mathrm{k}|>2\mathrm{v}\mathrm{x}^{1/2}+2\lambda\mathrm{x}\right)\leq2\mathrm{e}^{-\mathrm{x}}.
	\]
\end{lemma}

\begin{lemma}(Theorem B.15 from \cite{spokoiny2017penalized})\label{bound}
Let $\mathcal{Y}(\boldsymbol{v})$ with $\boldsymbol{v}\in \Upsilon_0(r)=\{\boldsymbol{v}\in \Upsilon:  \| \boldsymbol{v}-\boldsymbol{v}^* \| \leq r\}$ and  $\Upsilon\subseteq\mathbb{R}^{k}$, be a
smooth centered random vector process with values in $\mathbb{R}^{q}$. Let also $\E[\mathcal{Y}(\boldsymbol{v}^{*})]=0$
for the center $\boldsymbol{v}^{*}\in \Upsilon_0(r)$. Without loss of generality, assume
$\boldsymbol{v}^{*}=0$. We aim to bound $\|\mathcal{Y}(\boldsymbol{v})\|$ uniformly over $\boldsymbol{v}$
over a vicinity  $\Upsilon_0(r)$ of $\boldsymbol{v}^{*}$. By $\nabla\mathcal{Y}(\boldsymbol{v})$
we denote the $k\times q$ matrix with entries $\nabla_{\boldsymbol{v}_{i}}\mathcal{Y}(\boldsymbol{v}),i\leq k,j\leq q$.
Suppose that $\mathcal{Y}(\boldsymbol{v})$ satisfies for each $\boldsymbol{\gamma}_{1}\in\mathbb{R}^{k}$
and $\boldsymbol{\gamma}_{2}\in\mathbb{R}^{q}$ with $\left\Vert \boldsymbol{\gamma}_{1}\right\Vert =\left\Vert \boldsymbol{\gamma}_{2}\right\Vert =1$, and there exists a positive constant $v_0$, 
\[
\sup_{\boldsymbol{v}\in\Upsilon} \log \E \exp\left\{ \lambda\boldsymbol{\gamma}_{1}^{\top}\nabla\mathcal{Y}(\boldsymbol{v})\boldsymbol{\gamma}_{2}\right\} \leq\frac{v_{0}^{2}\lambda^{2}}{2},\quad|\lambda|\leq g.
\]

Let $A$ be a matrix fulfilling $1 / 2 \leq\left\|A A^{\top}\right\|\leq 1$. Then for each $r$, it holds
$$
\mathbb{P}\left\{\sup _{\boldsymbol{v} \in \Upsilon_0(r)}\|A \nabla\mathcal{Y}(\boldsymbol{v})\|>\sqrt{8} v_0 r \mathfrak{z}_{\mathbb{H}}(\mathrm{x})\right\} \leq \mathrm{e}^{-\mathrm{x}},
$$
where $\mathfrak{z}_{\mathbb{H}}(\mathrm{x})$ is given by the following  with $\mathbb{Q}_2={p}_A+\mathbb{Q}_2\left(\Upsilon_{\circ}(r)\right)$.
$$
\begin{aligned}
& \mathfrak{z}_{\mathbb{H}}(\mathrm{x})= \begin{cases}2 \sqrt{\mathbb{Q}_2+2 \mathrm{x}}, & \text { if } \mathbb{Q}_2+2 \mathrm{x} \leq \mathrm{g}^2, \\
2 \mathrm{~g}^{-1} \mathrm{x}+\mathrm{g}^{-1} \mathbb{Q}_2+\mathrm{g}, & \text { if } \mathbb{Q}_2+2 \mathrm{x}>\mathrm{g}^2.\end{cases}
\end{aligned}
$$
In the above,  $\mathbb{Q}_2$ relates to the entropy of the set $ \Upsilon_0(r)$, $p_A$ denotes a trace norm of a trace operator $A^{-2}$, and both can be calculated according to section B.4 in \cite{spokoiny2017penalized} as outlined below.

For each $k\leq 1$, by $\mathcal{M}_k$ we denote a $r_k$-net in $\Upsilon^{\circ}(r_0)$ with $r_k=r_02^{-k}$, so that $\Upsilon^{\circ}(r_0) \subseteq \bigcup_{\boldsymbol{v} \in \mathcal{M}_k}  \{\boldsymbol{v}'\in \Upsilon:  \| \boldsymbol{v}'-\boldsymbol{v} \| \leq r_k \} $, then $\mathbb{Q}_2\left(\Upsilon^{\circ}\right) \stackrel{\text { def }}{=} \sum_{k=1}^{\infty} 2^{-k+1} \log \left(2 \mathbb{N}_k\right)$ with $\mathbb{N}_k \stackrel{\text { def }}{=}\left|\mathcal{M}_k\right|$ being the cardinality of $\mathcal{M}_k$. For a positive self-adjoint operator in $\mathbb{R}^{\infty}$, denoted by $\mathbb{H}$, such that $\lambda_{\min}(\mathbb{H})=1$ and $\mathbb{H}^{-2}$ is a trace operator, then $\mathrm{p}_\mathbb{H} \stackrel{\text { def }}{=} \operatorname{tr}\left(\mathbb{H}^{-2}\right)=\sum_{j=1}^{\infty} h_j^{-2}<\infty,$
where $1=h_1 \leq h_2 \leq \cdots$ are the ordered eigenvalues of $H$.

\end{lemma}

\begin{lemma}(Theorem B.8. in \cite{spokoiny2017penalized})\label{covering}
 Suppose that for some $\alpha>1$,
$
\mathrm{p}_{\mathbb{H}}(\alpha) \stackrel{\text { def }}{=} \sum_{j=1}^{\infty} h_j^{-2} \log ^\alpha\left(h_j^2\right)<\infty,
$
then
$
\mathbb{Q}_2(\Upsilon^{\circ}_{\mathbb{H}} (r)) \leq \mathrm{C} \sum_{j=1}^{\infty} h_j^{-1},
$ with $\Upsilon^{\circ}_{\mathbb{H}} (r) = \{\boldsymbol{v}'\in \Upsilon:  \| \mathbb{H}(\boldsymbol{v}'-\boldsymbol{v} )\| \leq r\} ,$ for a fixed center $\boldsymbol{v}$.
\end{lemma}
\begin{lemma} (Lemma A.17 from \cite{spokoiny2019accuracy})\label{Gaussintegral}
Let $\mathcal{T}$ be a linear operator in $\mathbb{R}^k$ with $\|\mathcal{T}\|_{\text {op }} \leq 1$. Let $\boldsymbol{z} \in \mathbb{R}^{k}$ be a unit norm vector: $\|\boldsymbol{z}\|=1$. Define $k=\operatorname{tr}\left(\mathcal{T}^{\top} \mathcal{T}\right)$. For any positive $\mathrm{C}_0, \mathrm{r}_0$ with $1 / 2<\mathrm{C}_0 \leq 1$ and $\mathrm{C}_0 \mathrm{r}_0>2 \sqrt{k+1}+\sqrt{\mathrm{x}}$,
$$
\begin{aligned}
& \mathbb{E}\left\{{|\langle \boldsymbol{z}, \gamma\rangle|^2 } \exp \left(-\mathrm{C}_0 \mathrm{r}_0\|\mathcal{T} \gamma\|+\frac{\mathrm{C}_0 \mathrm{r}_0^2}{2}+\frac{1}{2}\|\mathcal{T} \gamma\|^2\right) \mathbb{I}\left(\|\mathcal{T} \gamma\|>\mathrm{r}_0\right)\right\}  \leq \mathrm{Ce}^{-\left(k+\mathrm{x}\right) / 2}.
\end{aligned}
$$
\end{lemma}

 \clearpage

\section{Algorithms}\label{sub:Algorithms}
The implementation of our quasi-Bayesian procedure, such as the construction of $PR_T$'s for parameters of interest, may require numerical simulations to evaluate quasi-posterior distributions in the absence of an explicit closed-form solution. We can use a Markov Chain Monte Carlo (MCMC) algorithm as described in Algorithm \ref{MCMC algorithm}, 


Algorithm \ref{MCMC algorithm} uses the Metropolis-Hastings algorithm to generate random ($\theta^{(i)},\mu^{(i)}$) pairs with $i=1,\cdots, n$ drawn from quasi-posterior distributions. This extends the MCMC algorithm by Chernozhukov and Hong (2003) to include the plausible characteristic term $\mu$. To minimize the impact of initial value selections, we consider burning periods $n'<n$, and subsequently retain a sequence of generated draws ${(\theta^{(n')},\mu^{(n')}),\cdots, (\theta^{(n)},\mu^{(n)})}$.    
\begin{algorithm}[htbp!]
	\SetAlgoLined
	\KwIn{
		\begin{itemize}
			\item $Q_T(\theta,\mu)$, $ \pi(\theta,\mu)$: the objective function and the prior used to form the target quasi-posterior distribution, $p_T(\theta, \mu) \propto \exp\left\{\frac{1}{2}Q_T(\theta,\mu)\right\}\pi(\theta,\mu)$;
			\item 		$\mathcal{q}(\theta', \mu'|\theta, \mu)$: proposal distribution;
			\item	$(\theta^{0}, \mu^{0})$: initial  values; 
			\item $K_J$: the number of iterations.		
		\end{itemize} 			 
  ~\\}
	\KwOut{\begin{itemize}
			\item 
			Samples $\{\theta^{(i)},\mu^{(i)}\}_{i=1}^{K_J}$ drawn from the target distribution $p_T(\theta, \mu)$.~\\
	\end{itemize}}
	 ~\\
	
	\For{$i \leftarrow 1$ \KwTo $K_J$}{
 \begin{enumerate}
     \item Sample $(\theta', \mu')$ from $\mathcal{q}(\theta', \mu'|\theta^{(i-1)}, \mu^{(i-1)})$;
	\item  Calculate acceptance ratio $$\scriptsize \alpha^{(i)} = \min\left(1, \frac{\exp\left\{\frac{1}{2}Q_T(\theta', \mu')\right\}\pi(\theta', \mu')\mathcal{q}(\theta^{(i-1)}, \mu^{(i-1)}|\theta', \mu')}{ \exp\left\{\frac{1}{2}Q_T(\theta^{(i-1)}, \mu^{(i-1)})\right\}\pi(\theta^{(i-1)}, \mu^{(i-1)})\mathcal{q}(\theta', \mu'|\theta^{(i-1)}, \mu^{(i-1)})}\right);$$
\item Sample $u^{(i)}$ from Uniform(0, 1). Set $(\theta^{(i)}, \mu^{(i)}) \leftarrow (\theta', \mu')$		
		if {$u^{(i)} \leq \alpha^{(i)}$},\\ otherwise,	set $(\theta^{(i)}, \mu^{(i)} )\leftarrow (\theta^{(i-1)}, \mu^{(i-1)})$; 
	
 \end{enumerate}  
	}~\\
	\caption{Random Walk Metropolis-Hastings (MH) Algorithm}
 \label{MCMC algorithm}
\end{algorithm}

\begin{algorithm}[htbp!]  
	\SetAlgoLined
	\KwIn{
		\begin{itemize}
			\item $Q_T(\theta,\mu)$, $ \pi(\theta,\mu)$: the objective function and the prior used to form the target quasi-posterior distribution, $p_T(\theta, \mu) \propto \exp\left\{\frac{1}{2}Q_T(\theta,\mu)\right\}\pi(\theta,\mu)$;
   \item Sequential Monte Carlo (SMC) tuning parameters
\begin{itemize}
    \item $B_N$: the size of initial draws, and also the size of the final outputs
    \item $\{ \phi_j \}_{j=1}^{J_N}$: an increasing sequence of tempering tuning parameters of length 
$J_N$ 

from $0$ to $1$;
    \item $K_j$: the length of the MH algorithm in the mutation step specified below;
\end{itemize}
\item 		$\mathcal{q}_{\mathcal{K}_j}(\theta', \mu'|\theta, \mu)$: (adaptive) proposal distributions that may vary for every iteration $2\leq j\leq J_N$ and depend on tuning parameters $\mathcal{K}_j$ that are pre-specified or 
derived at iteration $j-1$;
		\end{itemize}
  ~\\~\\}
	\KwOut{\begin{itemize}
			\item Samples $\{\theta^{(i)},\mu^{(i)}, W^{(i)}\}_{i=1}^{B_N}$.   
~\\
	\end{itemize}}
	 ~\\
\begin{enumerate}
 \item[(1)] (Initialization) Draw the initial particles from the prior: $(\theta_{1}^{i},\mu_{1}^{i} ) {\sim_{{i.i.d.}}} \pi(\theta,\mu)$ and set equal weights $W_{1}^i=1, i=$ $1, \ldots, B_N$.
    \item[(2)] (Recursion)

    \For{$j \leftarrow 2$ \KwTo $J_N$}
{
\begin{itemize}
\item  (Correction) Let $\tilde{w}_j^i=\left[\exp(Q_T( \theta_{j-1}^i,\mu_{j-1}^i)/2)\right]^{\phi_j-\phi_{j-1}}$ and  set new normalised weights $\tilde{W}_j^i=\frac{\tilde{w}_j^i W_{j-1}^i}{\frac{1}{B_N} \sum_{i=1}^N \tilde{w}_j^i W_{j-1}^i}, i=1, \ldots, B_N$.

\item (Selection) Compute $\operatorname{ESS}_j=B_N /\left(\frac{1}{B_N} \sum_{i=1}^N (\tilde{W}_i^j )^2\right)$. 

\begin{enumerate}
\item   If $\operatorname{ESS}_j \leq B_N/2$, let $\{\hat{\theta}_i, \hat{\mu}_i\}_{i=1}^{B_N}$ denote $B_N$ i.i.d. draws from a multinomial distribution with support points and weights $\left\{\theta_{j-1}^i,\mu_{j-1}^i, \tilde{W}_j^i\right\}_{i=1}^{B_N}$ and 

then set $W_j^i=1$. 
\item  If $\operatorname{ESS}_j > B_N/2$, let $\hat{\theta}_i=\theta_{j-1}^i$, $\hat{\mu}_i=\mu_{j-1}^i$ and $W_j^i=\tilde{W}_j^i$, $i=1, \ldots, B_N$.  
\end{enumerate}

\item (Mutation) For each $i$, run $K_j$ steps of an MH algorithm (see Algorithm \ref{MCMC algorithm}) 

with initial value $\left\{\hat{\theta}_i,\hat{\mu}_i\right\}$ and proposal distribution  $ \mathcal{q}_{\mathcal{K}_j}$ targeting at a 

distribution proportional to  $\exp\left\{\frac{1}{2}Q_T(\theta,\mu)\right\}^{\phi_j}\pi(\theta,\mu)$. 

Let $\left\{ {\theta}_j^i, {\mu}_j^i\right\}$ be the final draw out of the MH algorithm. 
\end{itemize}
	}
     \item[(3)] For $j=J_N$, the final round outputs are $
\{ \theta_{J_N}^i,\mu_{J_N}^i, W_{J_N}^i \}_{i=1}^{B_N}
$, which is then set to be $\{\theta^{(i)},\mu^{(i)}, W^{(i)}\}_{i=1}^{B_N}$.
\end{enumerate}	

	\caption{Sequential Monte Carlo (SMC) Algorithm}
 \label{SMC algorithm}
\end{algorithm}   

\clearpage

\section{Additional Simulation Results}\label{simold}
From the Bayesian perspective, Sections \ref{subsec:simulinear moments} and \ref{subsec:simu Nonlinear moments} compare PGMM posteriors in various settings and validate our results stated in Lemma \ref{coverageprt}.
From the frequentist perspective, we validate our BvM theorem and related coverage results. In particular, with the linear IV model setup, we illustrate the validity of Theorems \ref{Theorem:1}-\ref{Theorem:2}. 
In addition to the linear moment case, Section \ref{sec:simu Nonlinear moments} addresses simulations involving nonlinear moments and illustrates the frequentist justification of the unions of the credible regions constructed via quasi-posterior distributions, as indicated by Theorems \ref{th3}-\ref{th4}, by evaluating the frequentist coverage rates.
\subsection{Linear moments scenario}\label{sec:simulinear moments}
\subsubsection{Bayesian motivation}\label{subsec:simulinear moments}
~\\ 
This subsection compares PGMM posteriors between the case where $\mu\equiv 0$ (no model misspecification) and $\mu$ draws from $F_\mu$. The posteriors behave differently in these two cases; e.g., the posteriors of the plausible characteristics in the presence of model misspecification tend to deviate from zero instead of centering around zero as priors, which suggests that the shape of the PGMM posterior provides meaningful information about the extent of misspecification. 

The simulation exercises in this subsection draw plausible characteristics $\mu$ together with $Z_t$ jointly and continue with the linear IV model from {Section \ref{subsec:Illustration Example 1}}, and for simplicity, we consider the case with one endogenous variable, no exogenous control variables, and a single instrument variable. Therefore, $\theta=
(\alpha, \beta_{X}) \in \Theta$ and moment conditions $g(Z_t, \theta)=(1, D_t^\top)^\top \left(Y_t - \alpha + X_t^\top \beta_{X} \right).$  In this simulation exercise, with a given value of $\gamma$, $Y_t$ is generated with $X_t,D_t$ in the following way: 
\begin{align*}
 {Y}_t  =\alpha + X_t  \beta_X +(1, D_t^\top) \gamma +  {\varepsilon_t},
\end{align*}
where ${\varepsilon}$ consists of independent and identically distributed (i.i.d.) draws from $N(0,\sigma^2)$, and $X_t$ and $D_t$ are the fixed in the simulation using the sample data as described in Section \ref{subsec:Illustration Example 1}. By the fact that $\mu=\mathbb{E}g(Z_t, \theta)$, we have $\gamma = (\mathbb{E}(1, D_t^\top)^\top (1, D_t^\top))^{-1}\mu$, so once we draw $\mu$ from $F_\mu$ the value of $\gamma$ is also determined. In the simulation exercises in this subsection, we thus calibrate $\sigma^2$ to the data employed in  Section \ref{subsec:Illustration Example 1}, draw the value of $\mu = (\mu_1,\mu_2)^\top$ from given $F_\mu$, let $\gamma = (\mathbb{E}(1, D_t^\top)^\top (1, D_t^\top))^{-1}\mu$, $\alpha =0, \beta_X = \mu_2/10$, and then generate $Y_t$ and construct PGMM posteriors based on $Z_t=(Y_t, X_t, D_t)$. 

Figure \ref{fig:simu_linear} plots the results based on the above simulation setup and compares the posteriors in two cases.  {In the upper panel, $\mu$ is fixed at zero, and there is no model misspecification. 
The lower panel uses DGPs with $\mu$ drawn from $\pi(\mu)$; therefore, $\mu$ varies instead of fixing at zeros. 
}

\begin{figure}[htbp!]
    \centering
    \includegraphics[width=\textwidth]{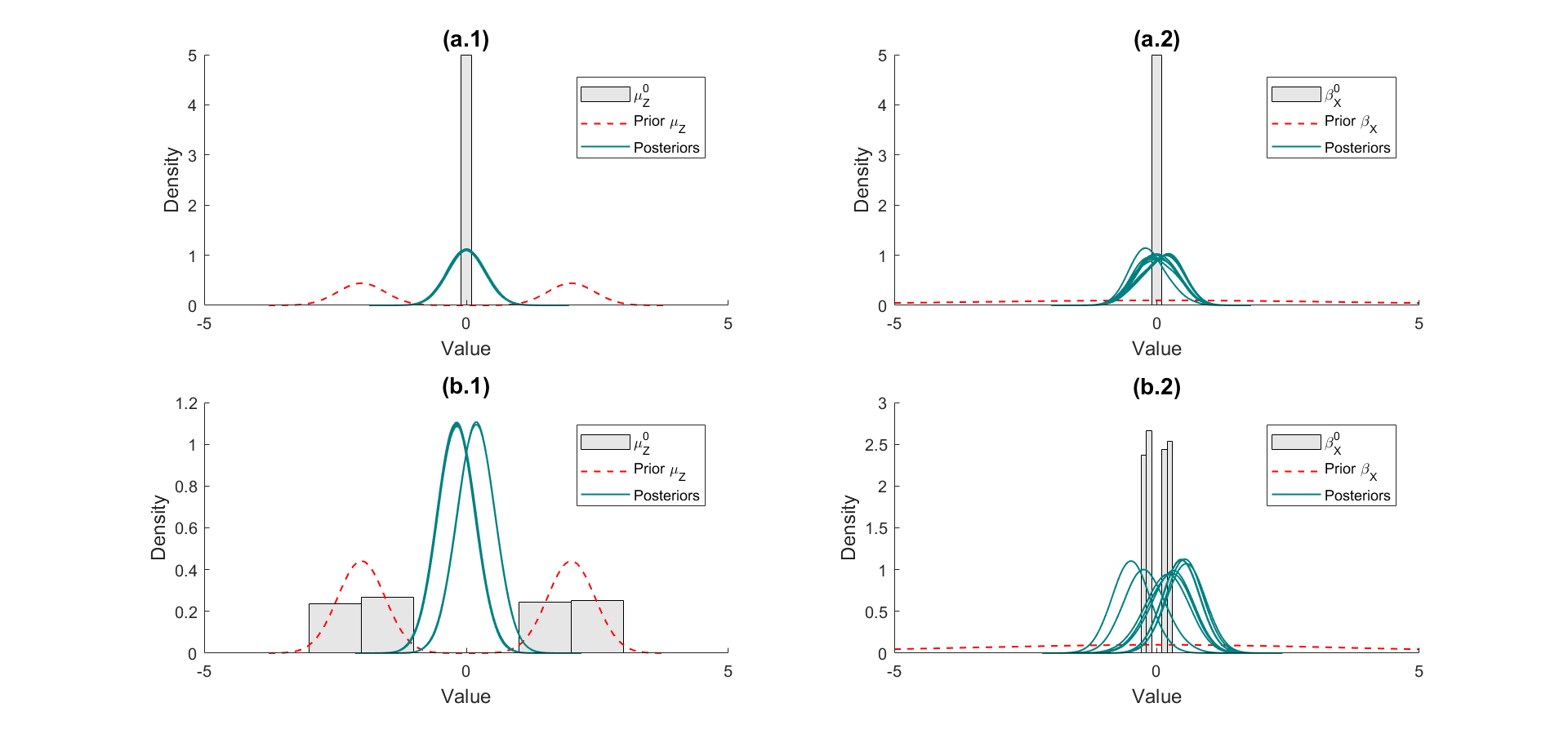}
    \caption{\footnotesize {Marginal priors (red dashed line) and  PGMM marginal posteriors (green solid line) of slope coefficient, $\beta_X$ (left panel), and plausible characteristic, $\mu_Z=\mu_2$ (right panel), resulting from the Section \ref{subsec:simulinear moments} linear IV model simulations.} The upper panel (a.) represents a correctly specified model with $\mu\equiv 0$, i.e., $F_\mu$ assigns all mass to 0. In contrast, the lower panel (b.) reflects a plausible IV setting from \cite{conley2012plausibly} with $\mu$ drawing from the prior of $\mu$, $\pi(\mu)$ in the simulation data generating process (DGP), i.e.,  $F_\mu$ coincides with $\pi(\mu)$. The red dashed curves mark the marginal prior densities. The gray shadowed bars are the histograms of the realized values of $\beta_X, \mu$ used in the DGPs of simulations, i.e., the bar in the subfigure (a.1) locates at point 0 as the value of $\mu$ is fixed at 0 in the upper panel simulation exercise. The green solid curves mark the PGMM marginal posteriors, e.g., the green curves in the subfigure (a.1) correspond to the marginal PGMM posteriors of $\beta_X$.
    }
    \label{fig:simu_linear}
\end{figure}

In addition, Figure \ref{fig:simu_linear} shows that the 
highest posterior regions
of $\beta_X$ tend to cover the true values of $\beta_X$ used in the DGP. {Corresponding to Lemma \ref{coverageprt}, the average coverage rates of the 95\% $PR_T$ containing the values of $\beta_X$ used in the DGP are $0.99$ and $0.96$ for cases (a.) and (b.), respectively.}

\subsubsection{Frequentist validation}
\label{section31}\label{sec:simu linear} ~\\
{Apart from Section \ref{subsec:simulinear moments} and \ref{subsec:simu Nonlinear moments}, the remaining sections consider DGPs with fixed plausible characteristics from a frequentist perspective.} We continue with the linear model from Section \ref{subsec:simulinear moments}. We set $\alpha$ to be zero in this subsection and consider $\theta=
 \beta_{X}  \in \Theta$. With a given value of $\gamma$, $Y_t$ is generated similarly as before:
\begin{align*}
 {Y_t}  =X_t  \theta +D_t\gamma +  {\varepsilon_t}.   \label{eq:simu1}
\end{align*}
In this simulation exercise, we calibrate parameters to the 401(K) data employed in \cite{conley2012plausibly} so that $X_t$ and $D_t$ are fixed subsamples of size 9951 from the 401(K) data with $X_t$ being an indicator for 401(k) participation and $D_t$ being an indicator for 401(k) plan eligibility.  ${\varepsilon}$ consists of independent and identically distributed (i.i.d.) draws from $N(0,\sigma^2)$, and the values of $ \theta$ and $\sigma^2$ are obtained by regressing the net financial assets from the dataset on $X_t$ with $D_t$ as an instrument using the 2SLS estimator and sample variance of the residuals.

\subsubsection{Validity of the Gaussian mixture limiting distribution}
~\\
\noindent
This section validates the Gaussian mixture limiting distribution specified in Theorem \ref{Theorem:2} using the linear IV model mentioned above. The simulated data, $Z_t= (Y_t, X_t, D_t), 1\leq t\leq T$, is generated as follows: $X_t = D_t +v_t, Y_t = X_t \theta + D_t \gamma + \varepsilon_t,$ where $D_t$'s are independently log-normally distributed such that the natural logarithm of $D_t$ follows the standard normal distribution, $v_t$'s and $\varepsilon_t$'s are independently standard normally distributed,  and $Z_t, v_t$, and $\varepsilon_t$ are all independent of each other.

We first proceed with the case $k=q=1$, and the parameter of interest $\theta$ is fixed at zero. We choose $\gamma=1/\sqrt{T}$ with $T$ being the sample size to mimic local misspecification and use a local Gaussian prior for $\mu$ and an independent flat prior for $\theta$ such that $\pi(\theta,\mu)\propto \exp{(-T\mu^\top \mu/2)}$. 
 \begin{figure}[htbp!] 
	\centering 
\includegraphics[width=0.8\textwidth]{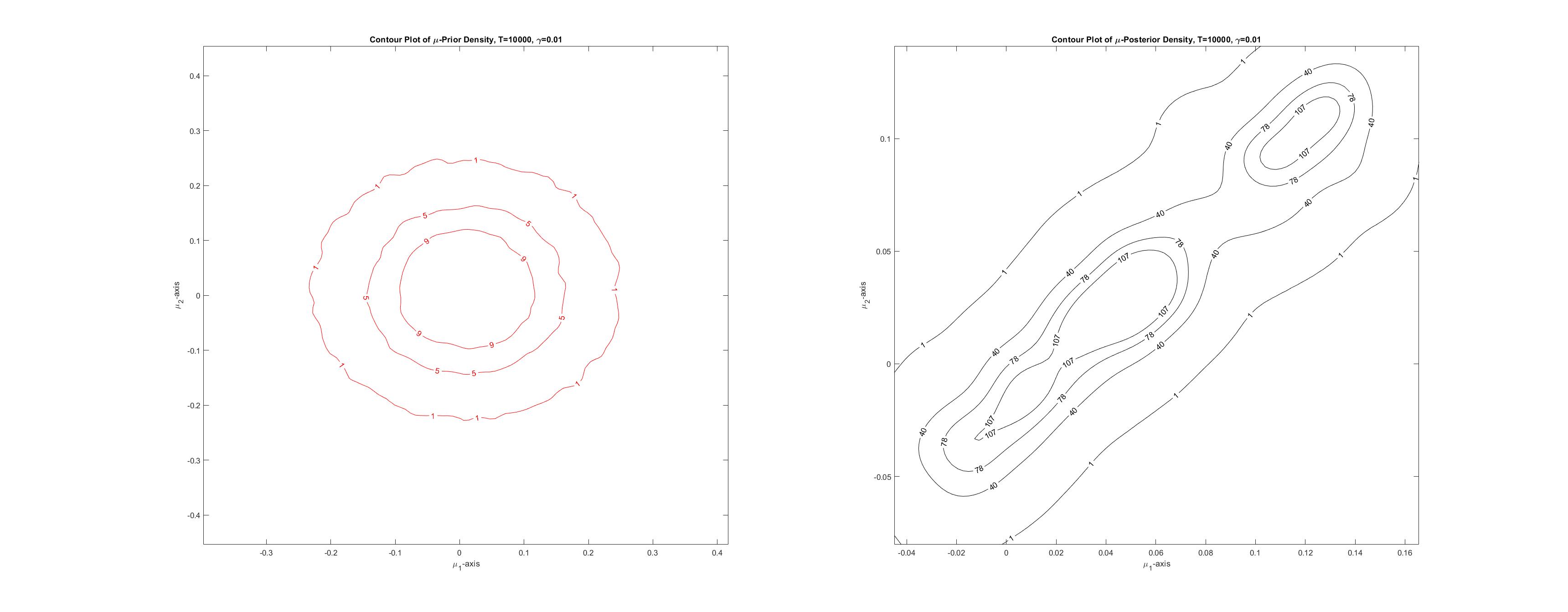}
 \caption{\footnotesize $q=2$, with $\gamma=(0.01, 0.01)^\top$. Left panel:  contour plot of the (marginal) prior density of $\mu=(\mu_1, \mu_2)^\top$; Right panel: contour plot of the (marginal) quasi-posterior density of $\mu=(\mu_1, \mu_2)^\top$.} \label{priors1}
\end{figure}

\begin{sidewaysfigure}
     \centering\textcolor{white}{\rule{1.75\textheight}{0.7\textheight}}
\includegraphics[width=\textwidth, height=0.5\textheight]{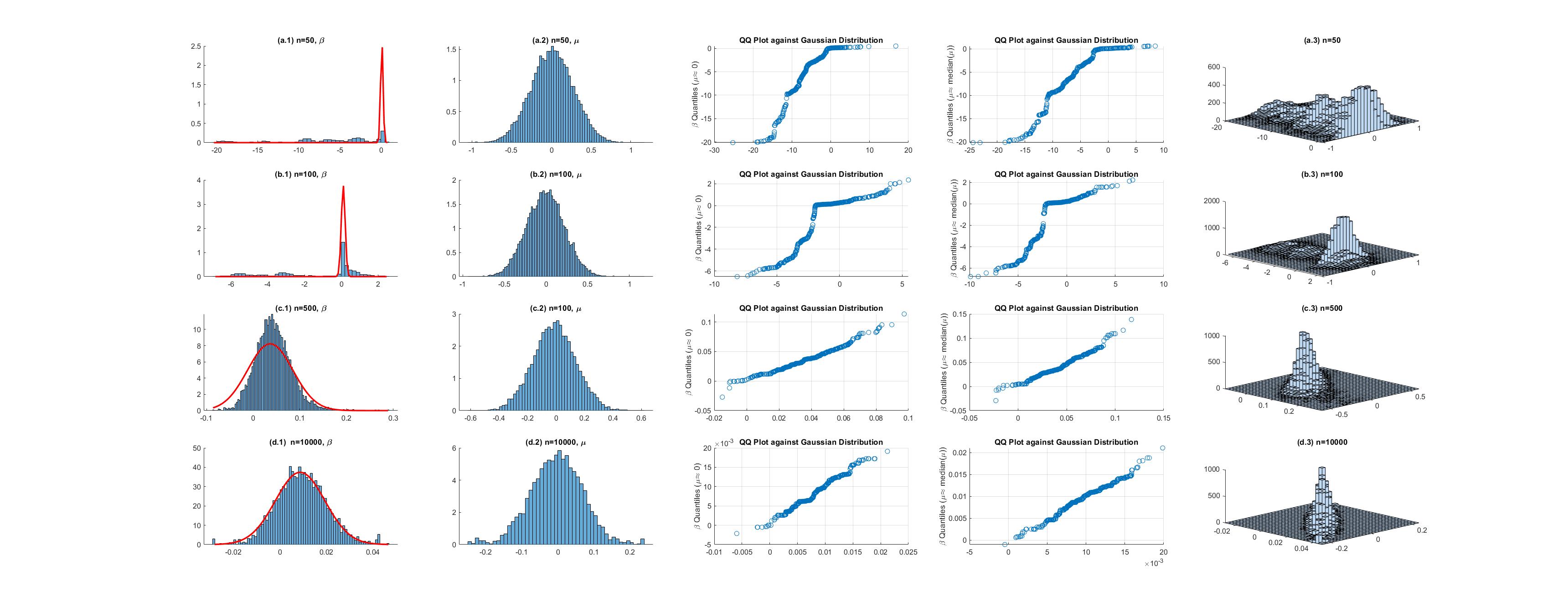}	
	\caption{\footnotesize
The first column displays histograms of randomly drawn $\theta$ values from quasi-posterior distributions, with the red smooth curve representing the limiting density suggested by Theorem \ref{Theorem:2}. In the second column, you can find histograms of random draws of the plausible characteristic $\mu$ from quasi-posterior distributions. The final column exhibits joint histograms of pairs $(\theta,\mu)$. The third and fourth columns display QQ-plots of simulated $\theta$'s corresponding to $\mu$'s with values close to zero and the median of the simulated $\mu$'s, respectively. The sample size varies from 50 (first row) to 10,000 (fourth row).}  \label{fig/linear_IV4} 
\end{sidewaysfigure}

Figures \ref{fig/linear_IV4} plots based on random draws of pairs of $(\theta,\mu)$ from quasi-posterior distributions. These quasi-posterior distributions are constructed using simulated data of various sample sizes. The third and fourth columns of Figure \ref{fig/linear_IV4} illustrate that the Gaussian mixture distribution $N_T(\theta,\mu)$ closely approximates $p_T(\theta,\mu)$ when the sample size is relatively large, as the simulation results show that the conditional quasi-posterior distribution of $\theta$ given $\mu$ closely resembles a Gaussian. These findings are also shown in the QQ plots of Figure \ref{fig/linear_IV4}.

In practice, Theorem \ref{Theorem:2} allows us to utilize the Gaussian mixture limiting distribution directly to draw inference on $\theta$ instead of relying on random draws from $p_T(\theta,\mu)$. This becomes crucial in the presence of a large parameter dimension, where the MCMC algorithm may experience a low acceptance rate. In cases where simulating from $p_T(\theta,\mu)$ is difficult, we can still derive results using simulations from the Gaussian mixture $N_T(\theta,\mu).$

Another interesting observation is that in Figure \ref{fig/linear_IV4}, the fifth column indicates that the region with the highest density of quasi-posterior distributions tends to concentrate in a smaller area than the priors. This phenomenon is illustrated in Figure \ref{priors1}, where we explore a simulation exercise similar to the settings in Figure \ref{fig/linear_IV4}, but with $q=2$ and $\gamma=(0.01,0.01)^\top$. Figure \ref{priors1} demonstrates that in the presence of over-identification, the posterior distribution may concentrate on an area of smaller dimension than the prior distribution.

\subsection{Nonlinear moments scenario}\label{sec:simu Nonlinear moments}
\subsubsection{Bayesian motivation}\label{subsec:simu Nonlinear moments}
~\\
We now consider a simulation exercise with non-smooth moment conditions, i.e., the IVQR model. Inspired by the plausible IV model, we use invalid/plausible IV ($\Tilde{D}_t$) to introduce model misspecification in this simulation exercise such that  $\Tilde{D}_t =D_t + \Tilde{\gamma} \left(\tau-1\left(Y_t \leqslant \alpha_\tau + X_t^\top \beta_{\tau, X} + W_t^\top \beta_{\tau, W} \right)\right)$ with $D_t$ being a valid instrumental variable. Therefore, when $\Tilde{\gamma}=0$ the model is correctly specified, and with non-zero $\Tilde{\gamma}$'s, the moment conditions are misspecified. Specifically, this section focuses upon the Median IV case with $\tau =0.5$ as in Section \ref{subsec:Illustration Example 1} with moment conditions  $$g(Z_t, \theta)=(1, \tilde{D}_t^\top, W_t^\top)^\top\left(\tau-1\left(Y_t \leqslant \alpha_\tau + X_t^\top \beta_{\tau, X} + W_t^\top \beta_{\tau, W} \right)\right),$$ where $Z_t=(Y_t, X_t, W_t, \tilde{D}_t),$
and $\theta =(\alpha_\tau, \beta_{\tau,X}, \beta_{\tau, W}) \in \Theta $. 
\begin{figure}[htbp!]
    \centering
    \includegraphics[width=0.9\textwidth]{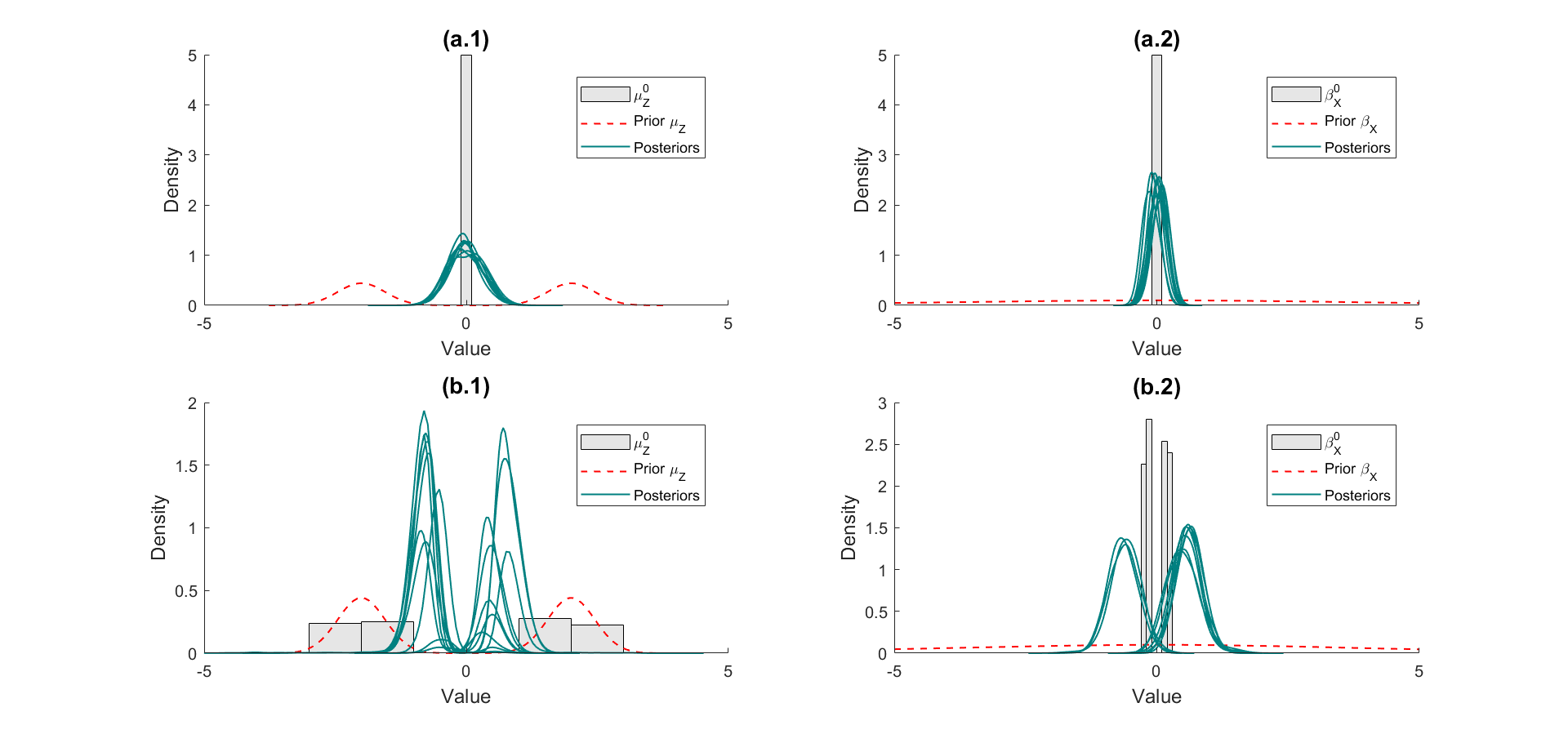}
    \caption{\footnotesize Similar to Figure \ref{fig:simu_linear}, {marginal priors (red dashed line) and  PGMM marginal posteriors (green solid line) of slope coefficient, $\beta_X$ (left panel),  and plausible characteristic, $\mu_D$ (right panel), resulting from the Section \ref{subsec:simu Nonlinear moments}  median IV model simulations.} The upper panel (a.) represents a correctly specified model with $\mu\equiv 0$. In contrast, the lower panel (b.) uses DGPs with $\mu$ drawn from the prior $\pi(\mu)$. The red dashed curves mark the marginal prior densities. The gray shadowed bars are the histograms of the realized values of $\beta_X, \mu$ used in the DGPs of simulations. The green solid curves mark the PGMM marginal posteriors. 
    }
    \label{fig:simu_median}
\end{figure} 

  Figure \ref{fig:simu_median} is created similarly to Figure \ref{fig:simu_linear} with one instrumental variable. In the simulation exercises in this subsection, $X_t$, $W_t$, and $D_t$ are fixed using the sample data as described in Section \ref{subsec:Illustration Example 1}. We first draw $\mu = (\mu_1, \mu_D, \mu_W)$ from $F_\mu$, and the values of $\mu$ directly determines the value of $\tilde{\gamma}$ via $\mu=\mathbb{E}g(Z_t, \theta)$. In this simulation design, $\mu_1$ and $\mu_W$ are always set to zero, and once we draw $\mu$, we choose $\theta = \mu/10$, and $Y_t$ is generated  in the following way: 
\begin{align*}
 {Y}_t  =\alpha_\tau + X_t  \beta_{\tau, X} + W_t \beta_{\tau, W}+  {\varepsilon},
\end{align*}
where ${\varepsilon}$ consists of independent and identically distributed (i.i.d.) draws from $N(0,\sigma^2)$ with $\sigma^2$ calibrated  to the data employed in Section \ref{subsec:Illustration Example 1}.

Similar to Figure \ref{fig:simu_linear}, in the upper panel of Figure \ref{fig:simu_median} $\mu\equiv 0$, while in the lower panel, $\mu$ used in the DGPs are drawn from the prior $\pi(\mu)$.
The observed patterns are similar. The average coverage rates of the 95\% $PR_T$ are 0.99 and 0.92 for cases (a.) and (b.) in Figure \ref{fig:simu_median}, respectively. In the absence of model misspecification in the upper panel of Figure \ref{fig:simu_median}, the marginal posteriors of the plausible characteristics tend to cluster around zero, and the average coverage rate is undoubtedly higher than the nominal rate due to the additional uncertainty introduced. However, they start to deviate from zero as model misspecification is introduced in the lower panel. 
 
\subsubsection{Frequentist validation}
~\\ \noindent
This subsection revisits an median IV simulation example from \cite{chernozhukov2003mcmc} with slight modifications to introduce model misspecification. The Monte Carlo Simulation Example II established by \cite{chernozhukov2003mcmc} is: 
$$ Y_t=\alpha+X_t^{\top} \beta+u_t, 
 u_t=\sigma(X_t) \varepsilon_t, \quad  \sigma(X_t)= (1+\sum_{j=1}^3 X_{t,j}) / 5,$$
where there are no endogenous variables, $X_{t,j}$'s are independently log-normally distributed such that the natural logarithm of $X_{t,j}$ follows a standard normal distribution, $\varepsilon_i$'s are independently standard normally distributed and are independent of $X_{t,j}$'s and $\theta=(\alpha, \beta)$. They consider the following moment conditions for the median $g(Z_t, \theta)  = (1, X_t^\top)^\top\left(0.5-1\left(Y_t \leqslant \alpha + X_t^\top \beta \right)\right).$ 



{
We modify the above DGP to introduce model misspecification in two ways: by adjusting the credibility of the instrument variables and the rank invariance (or similarity) used in the IVQR with discrete (or bounded continuous) treatment variables so that treatment status should not impact the underlying conditional distribution. The former situation resembles the plausible linear IV model in that the exclusion restriction is relaxed. We substitute $u_t$ with $\tilde{u}_t = \sigma(X_t)\varepsilon_t + \gamma X_{t,3}^2$ to include plausible IVs, with $\gamma$ evaluating the credibility of IVs. In the latter case, we consider the following DGP with potentially missing variables $X_t$ such that $Y_t=\alpha +  D_t^{\top}\beta  + \gamma D_t X_t +\varepsilon_t,$ where $D_t$ follows an independent and identically distributed Bernoulli distribution with a success rate of 1/2.  For both cases, we consider the following moment conditions used for estimating the parameters concerning the $\tau$-quantile:  $g(Z_t, \theta_\tau)=(1, X_t^\top, W_t^\top)^\top\left(\tau-1\left(Y_t \leqslant \alpha_\tau + X_t^\top \beta_{\tau}  \right)\right),$ where $Z_t=(Y_t, X_t),$ $\theta =(\alpha_\tau, \beta_{\tau}) \in \Theta $. 

  


Table \ref{table ExampleII-1} is constructed with simulated data generated under the first DGP relaxing the exclusion restriction, while Table \ref{table: ExampleII-2} is built under the latter one. Both tables compare average coverage rates of sets constructed using different approaches containing true parameter values, and they also demonstrate the validity of Theorems \ref{th3} and \ref{th4}.  

}


\begin{table}
\centering
\begin{tabular}{ccccccc|ccccccc}
\toprule
$\gamma$ & $T$ & $\tau$ & Methods & $\beta_{\tau,1}$ & $\beta_{\tau,2}$ & $\beta_{\tau,3}$&{$ \gamma$ }& $T$ & $\tau$ & Methods & $\beta_{\tau,1}$ & $\beta_{\tau,2}$ & $\beta_{\tau,3}$ \\
\hline 
1 & 300 & 0.2 & 0 & 0.946 & 0.941 & 0.029 & 1 & 300 & 0.2 & 1 & 0.985 & 0.981 & 0.924 \\ 
1 & 300 & 0.5 & 0 & 0.675 & 0.676 & 0.34 & 1 & 300 & 0.5 & 1 & 0.988 & 0.97 & 0.921 \\  
1 & 100 & 0.2 & 0 & 0.781 & 0.725 & 0.273 & 1 & 100 & 0.2 & 1 & 0.928 & 0.929 & 0.964 \\ 
1 & 100 & 0.5 & 0 & 0.554 & 0.549 & 0.454 & 1 & 100 & 0.5 & 1 & 0.996 & 0.968 & 0.94 \\ 
0 & 300 & 0.2 & 0 & 0.917 & 0.935 & 0.914 & 0 & 300 & 0.2 & 1 & 0.994 & 0.996 & 0.993 \\ 
0 & 300 & 0.5 & 0 & 0.925 & 0.919 & 0.916 & 0 & 300 & 0.5 & 1 & 0.965 & 0.96 & 0.969 \\  
0 & 100 & 0.2 & 0 & 0.941 & 0.947 & 0.941 & 0 & 100 & 0.2 & 1 & 0.995 & 0.994 & 0.997 \\ 
0 & 100 & 0.5 & 0 & 0.903 & 0.917 & 0.901 & 0 & 100 & 0.5 & 1 & 0.959 & 0.969 & 0.971 \\ 
\hline
\end{tabular}
\caption{\footnotesize This table illustrates the average coverage rates of sets containing the true values of $\beta_{\tau, i}$'s, and the rates are displayed in the columns labeled ``$\beta_{\tau, i}$''. The column labeled ``$\gamma$'' shows the values of $\gamma$ used in the DGPs, labeled ``$T$'' shows the value of the sample size, labeled ``$\tau$''  shows the corresponding quantiles, labeled ``Methods'' outlines the estimation procedure, with the value of $0$ referring to the [CH] method and the value of $1$, referring to the PGMM method.
The [CH] intervals are constructed with flat priors over $\theta_\tau$'s, and the PGMM intervals are constructed with flat priors for $\theta_\tau$'s and  independent local Gaussian priors $N(0, I/T)$ for plausible characteristics.} \label{table ExampleII-1}
\end{table}

Table \ref{table ExampleII-1} compares results with and without accounting for model misspecification. In the simulation exercise outlined in Table \ref{table ExampleII-1}, the parameters $\theta = (\alpha, \beta)$ in the DGP are set equal to the null vector, and thus $\beta_\tau$'s are also null vectors following the settings in \cite{chernozhukov2003mcmc}. The table shows the average coverage rates for the true $\beta_\tau$ value in sets from two quasi-posterior distributions, with and without assuming model misspecification. One set is created using 0.025 to 0.975 sample quantiles obtained from quasi-posterior distributions for $\beta_{\tau, i}$'s without assuming model misspecification (refer to the quasi-Bayesian approach outlined in \cite{chernozhukov2003mcmc}, denoted as [CH]), and another set is constructed from PGMM quasi-posterior. The sets resulting from PGMM approach are built as follows: we select simulated $\beta_{\tau, i}$'s corresponding to pre-selected plausible characteristics $(\mu_1, \mu_2, \mu_3, \mu_4)^\top$ with average values $\sum_i \mu_i/4$ close to the 0.1-0.9 sample quantiles of priors, then we create a set using 0.025 to 0.975 sample quantiles of $\beta_i$'s selected corresponding to a specific value of plausible characteristics, and the final set is the union of all those sets.

In Table \ref{table ExampleII-1}, when $\gamma=0$, both techniques produce comparable outcomes. However, when $\gamma=1$, [CH] sometimes results in a strikingly low coverage rate for $\beta_{\tau, 3}$  (for example, when $\tau=0.2$) while incorporating local plausible characteristics enhances the coverage rate. Table \ref{table: ExampleII-2} explores a different DGP from Table \ref{table ExampleII-1} with $\theta =(0,1)^\top$ in the DGP and thus $\beta_\tau=1$; additionally, Table \ref{table: ExampleII-2} also presents findings regarding the 95\% confidence sets using the IVQR procedure (see \cite{chernozhukov2005iv}). We observe similar patterns in Table \ref{table: ExampleII-2} as in Table \ref{table ExampleII-1}.

%





\begin{table} 
\centering
\begin{tabular}{cccccc| cccccc}
\toprule
$\gamma$ & $T$ & $\tau$ & Methods & $\beta_\tau$ &  & $\gamma$ & $T$ & $\tau$  & Methods & $\beta_\tau$ &   \\ \hline
0 & 300 & 0.5 & -1 & 0.9629 &  & 1 & 300 & 0.5 & -1 & 0 &  \\ 
0 & 300 & 0.5 & 0 & 0.9572 &  & 1 & 300 & 0.5 & 0 & 0.0009 &  \\ 
0 & 300 & 0.5 & 1 & 0.9780 &  & 1 & 300 & 0.5 & 1 & 0.9555 &  \\ 
0 & 300 & 0.8 & -1 & 0.9573 &  & 1 & 300 & 0.8 & -1 & 0 &  \\ 
0 & 300 & 0.8 & 0 & 0.9574 &  & 1 & 300 & 0.8 & 0 & 0.0893 &  \\ 
0 & 300 & 0.8 & 1 & 0.9796 &  & 1 & 300 & 0.8 & 1 & 0.9009 &   \\ 
0 & 100 & 0.5 & -1 & 0.9666 &  & 1 & 100 & 0.5 & -1 & 0.0439 &  \\ 
0 & 100 & 0.5 & 0 & 0.9653 &  & 1 & 100 & 0.5 & 0 & 0.1112 &  \\ 
0 & 100 & 0.5 & 1 & 0.9807 &  & 1 & 100 & 0.5 & 1 & 0.9587 &   \\ 
0 & 100 & 0.8 & -1 & 0.9529 &  & 1 & 100 & 0.8 & -1 & 0.0063 &  \\ 
0 & 100 & 0.8 & 0 & 0.9455 &  & 1 & 100 & 0.8 & 0 & 0.4694 &  \\ 
0 & 100 & 0.8 & 1 & 0.9754 &  & 1 & 100 & 0.8 & 1 & 0.9659 &  \\ 
\hline
\end{tabular}
\caption{\footnotesize Similar to Table \ref{table ExampleII-1}, the value of $-1$ in the column labeled ``Methods'' refers to the IVQR method (see \cite{chernozhukov2005iv}). The [CH] intervals are constructed with flat priors over $\theta$'s, and the PGMM intervals are constructed with flat priors for $\theta$'s and independent local Gaussian priors $N(0, 10I/T)$ for plausible characteristics.}
\label{table: ExampleII-2}
\end{table}

\clearpage
\bibliographystyle{plainnat}
\bibliography{biball}
\end{document}